\documentclass[a4paper,UKenglish,cleveref, autoref, thm-restate, numberwithinsect]{lipics-v2021}

\pdfoutput=1 
\hideLIPIcs  

\usepackage{xspace}
\usepackage{mdframed}
\usepackage{comment}

\usepackage{tikz}

\usetikzlibrary{arrows.meta,shapes.geometric,positioning,decorations.pathmorphing}
\tikzset{
	vert/.style={circle,inner sep=1.5,fill=white,draw=black,minimum size=.3cm},
	vert2/.style={inner sep=1.5,fill=white,draw=black,minimum size=.3cm,fill=lightgray},
	vertb/.style={circle,inner sep=1.5,fill=white,draw=black,minimum size=.3cm, line width=1.5pt},
	vert2b/.style={inner sep=1.5,fill=white,draw=black,minimum size=.3cm,fill=lightgray, line width=1.5pt},
	vertg/.style={circle,inner sep=1.5,fill=white,draw=gray,minimum size=.3cm},
	vert2g/.style={inner sep=1.5,fill=white,draw=gray,minimum size=.3cm,fill=lightgray},
	triac/.style={isosceles triangle,anchor=apex,rotate=90,fill=cyan!20!white,draw=black,minimum size=.9cm},
	trial/.style={isosceles triangle,anchor=apex,rotate=120,fill=cyan!20!white,draw=black,minimum size=.9cm},
	triar/.style={isosceles triangle,anchor=apex,rotate=60,fill=cyan!20!white,draw=black,minimum size=.9cm},
	edge/.style={color=black, line width=1pt},
	edgeb/.style={color=black, line width=1.5pt},
	edgeg/.style={color=gray, line width=1pt},
	edgec/.style={color=black, line width=1pt, decorate, decoration=snake},
	diredge/.style={->,>={Stealth[width=8pt,length=8pt]},color=black, line width=1pt},
    trace/.style={double=none, double distance=18pt, line join=round, line cap=round, draw=green!60!black, line width=2pt},
    trace2/.style={double=none, double distance=18pt, line join=round, line cap=round, draw=red!80!black, line width=2pt}
}

\newcounter{modification}
\newenvironment{modification}[1][]{\refstepcounter{modification}%

\begin{center}
\begin{mdframed}[nobreak=true]
{\normalsize\sffamily{\bfseries Modification~\themodification} (#1){\bfseries .}}

\vspace{.5ex}
\begin{normalsize}
}
{
\end{normalsize}
\end{mdframed}
\end{center}
}

\newcounter{algorithm}
\newenvironment{algorithm}[3][]{\refstepcounter{algorithm}%

\begin{center}
\begin{mdframed}[nobreak=true]
{\normalsize\sffamily{\bfseries Algorithm~\thealgorithm} (#1){\bfseries .}}

\vspace{-2ex}
\begin{normalsize}
\begin{description}
 \item[\textbf{Input:}]  #2
 \item[\textbf{Output:}]  #3
\end{description}

\vspace{-1ex}

\hrule

\vspace{1ex}

}
{
\end{normalsize}
\end{mdframed}
\end{center}
}

\newcommand{\OO}{\mathcal{O}}

\newcommand{\tw}{\mathsf{tw}}
\newcommand{\fvn}{\mathsf{fvn}}
\newcommand{\vcn}{\mathsf{vcn}}

\newcommand{\mw}{\mathsf{mw}}
\newcommand{\introduce}{\mathsf{introduce}}
\newcommand{\forget}{\mathsf{forget}}
\newcommand{\join}{\mathsf{join}}

\newcommand{\bigjoin}{\mathsf{bigjoin}}
\newcommand{\smallintroduce}{\mathsf{smallintroduce}}

\newcommand{\smalljoin}{\mathsf{smalljoin}}
\newcommand{\extendedforget}{\mathsf{extendedforget}}
\newcommand{\void}{\mathsf{void}}
\newcommand{\ptw}{\mathsf{PTW}}
\newcommand{\ltw}{\mathsf{LTW}}
\newcommand{\legal}{\mathsf{legal}}
\newcommand{\bigbag}{\mathsf{bigbag}}
\newcommand{\subbag}{\mathsf{subbag}}
\newcommand{\nbag}{\mathsf{nbag}}
\newcommand{\smallbag}{\mathsf{smallbag}}
\newcommand{\nbagf}{\mathsf{nbag}_{\mathsf{forget}}}
\newcommand{\smallbagf}{\mathsf{forgetbag}}
\newcommand{\smallbagc}{\mathsf{smallbagcandidate}}

\newcommand{\dpath}{\textup{path}}

\newcommand{\true}{\textup{true}}
\newcommand{\false}{\textup{false}}
\newcommand{\slim}{\text{slim}\xspace}
\newcommand{\topheavy}{\text{top-heavy}\xspace}

\newcommand{\In}{\mathsf{In}} 
\newcommand{\Out}{\mathsf{Out}} 
\newcommand{\Rest}{\mathsf{Rest}}

\newcommand{\RemoveBringNeighbor}{\mathsf{RemoveBringNeighbor}}

\newcommand{\Normalize}{\mathsf{Normalize}}
\newcommand{\MoveIntoSubtree}{\mathsf{MoveIntoSubtree}}
\newcommand{\RemoveFromSubtree}{\mathsf{RemoveFromSubtree}}
\newcommand{\MergeFullNodes}{\mathsf{MergeFullNodes}}
\newcommand{\BringNeighborUp}{\mathsf{BringNeighborUp}}
\newcommand{\BringNeighborDown}{\mathsf{BringNeighborDown}}
\newcommand{\MakeSlim}{\mathsf{MakeSlimOne}}
\newcommand{\MakeSlimTwo}{\mathsf{MakeSlimTwo}}
\newcommand{\MakeTopHeavy}{\mathsf{MakeTopHeavy}}

\crefname{modification}{Modification}{Modifications}
\crefname{operation}{Modification}{Modifications}
\crefname{observation}{Observation}{Observations}
\crefname{algorithm}{Algorithm}{Algorithm}

\usepackage{todonotes}

\title{Treewidth Parameterized by Feedback Vertex Number} 

\author{Hendrik Molter}{Department of Computer Science, Ben-Gurion~University~of~the~Negev, 
	Beer-Sheva, 
	Israel}{molterh@post.bgu.ac.il}{https://orcid.org/0000-0002-4590-798X}{}

\author{Meirav Zehavi}{Department of Computer Science, Ben-Gurion~University~of~the~Negev, 
	Beer-Sheva, 
	Israel}{meiravze@bgu.ac.il}{https://orcid.org/0000-0002-3636-5322}{}

\author{Amit Zivan}{Department of Computer Science, Ben-Gurion~University~of~the~Negev, 
	Beer-Sheva, 
	Israel}{amitziv@post.bgu.ac.il}{}{}

\authorrunning{H. Molter, M. Zehavi, and A. Zivan} 

\Copyright{H. Molter, M. Zehavi, and A. Zivan} 



\category{} 

\relatedversion{} 

\funding{This work was supported by the Israel Science Foundation, grant nr.~1470/24, by the European Union's Horizon Europe research and innovation programme under grant agreement~949707, and by the European Research Council, grant nr.~101039913 (PARAPATH).}

\acknowledgements{}

\nolinenumbers 

\EventEditors{John Q. Open and Joan R. Access}
\EventNoEds{2}
\EventLongTitle{42nd Conference on Very Important Topics (CVIT 2016)}
\EventShortTitle{CVIT 2016}
\EventAcronym{CVIT}
\EventYear{2016}
\EventDate{December 24--27, 2016}
\EventLocation{Little Whinging, United Kingdom}
\EventLogo{}
\SeriesVolume{42}
\ArticleNo{23}

\begin{document}

\maketitle

\begin{abstract}
We provide the first algorithm for computing an optimal tree decomposition for a given graph $G$ that runs in single exponential time in the \emph{feedback vertex number} of $G$, that is, in time $2^{\OO(\fvn(G))}\cdot n^{\OO(1)}$, where $\fvn(G)$ is the feedback vertex number of $G$ and $n$ is the number of vertices of $G$. 
On a classification level, this improves the previously known results by Chapelle et al.~[Discrete Applied Mathematics~'17] and Fomin et al.~[Algorithmica~'18], who independently showed that an optimal tree decomposition can be computed in single exponential time in the \emph{vertex cover number} of $G$. 

One of the biggest open problems in the area of parameterized complexity is whether we can compute an optimal tree decomposition in single exponential time in the \emph{treewidth} of the input graph. The currently best known algorithm by Korhonen and Lokshtanov~[STOC~'23] runs in $2^{\OO(\tw(G)^2)}\cdot n^4$ time, where $\tw(G)$ is the treewidth of $G$. Our algorithm improves upon this result on graphs $G$ where $\fvn(G)\in o(\tw(G)^2)$. On a different note, since $\fvn(G)$ is an upper bound on $\tw(G)$, our algorithm can also be seen either as an important step towards a positive resolution of the above-mentioned open problem, or, if its answer is negative, then a mark of the tractability border of single exponential time algorithms for the computation of treewidth.

\end{abstract}


\section{Introduction}

\emph{Treewidth} is (arguably) both the most important and the most well-studied structural parameter in algorithmic graph theory~\cite{Bodlaender93,Bodlaender97,Bodlaender05,Bodlaender06,Bodlaender07,FluschnikMNRZ20,Kloks94}. 
Widely regarded as one of the main success stories, most textbooks in parameterized complexity theory dedicate at least one whole chapter to this concept; see e.g.\ Niedermeier~\cite[Chapter~10]{Nie06}, Flum and Grohe~\cite[Chapters~11~\&~12]{FG06}, Downey and Fellows~\cite[Chapters~10--14]{DF13}, Cygan et al.~\cite[Chapter~7]{Cyg+15}, and Fomin et al.~\cite[Chapter~14]{Fom+19}.
In particular, treewidth is a powerful tool that allows us to leverage the conspicuous observation that many NP-hard graph problems are polynomial-time solvable on trees. 
Here, many fundamental graph problems can be solved by simple greedy algorithms that operate from the leaves to the (arbitrarily chosen) root~\cite{cockayne1975linear,daykin1966algorithms,GJ79,mitchell1979linear}. 
Treewidth, intuitively speaking, measures how close a graph is to a tree. In slightly more formal terms, it quantifies the width of a so-called \emph{tree decomposition}, which generalizes the property that trees can be broken up into several parts by removing single (non-leaf) vertices.\footnote{We give a formal definition of \emph{treewidth} and \emph{tree decomposition} in \cref{sec:prelims}.} Tree decompositions naturally provide a tree-like structure on which dynamic programming algorithms can operate in a straightforward bottom-up fashion, from the leaves to the root. The concept of treewidth has been proven to be extremely useful and has led to a plethora of algorithmic results, where NP-hard graph problems are shown to be efficiently solvable if we know a tree decomposition of small width for the input graph~\cite{arnborg1985efficient,arnborg1991easy,arnborg1989linear,bern1987linear,Bodlaender88,bodlaender2008combinatorial}.
Maybe the most impactful such result is Courcelle's theorem~\cite{borie1992automatic,courcelle1990monadic,courcelle2012graph}, which states that all problems expressible in monadic second-order logic can be solved in linear time on graphs for which we know a tree decomposition of constant width.

In the early 1970s, Bertel\`{e} and Brioschi~\cite{bertele1972nonserial} first observed the above-described situation. They showed that a large class of graph problems can be efficiently solved by ``non-serial'' dynamic programming, as long as the input graph has a bounded \emph{dimension}. Roughly 15 years later, this parameter was proven to be equivalent to treewidth~\cite{arnborg1985efficient,Bodlaender98}. In the meantime, the concept of treewidth was rediscovered numerous times, for example by Halin~\cite{halin1976s} and, more prominently, by Robertson and Seymour~\cite{robertson1984graph,robertson1986graph,robertson1995graph} as an integral part of their graph minor theory, one of the most ground-breaking achievements in the field of discrete mathematics in recent history.

In order to perform dynamic programming on a certain structure, it is crucial that this structure is known to the algorithm. In the case of efficient algorithms for graphs with bounded treewidth, this means that the algorithm needs to have access to the tree decomposition. Since computing a tree decomposition of optimal width is NP-hard~\cite{arnborg1987complexity}, many early works on such algorithms assumed that a tree decomposition (or an equivalent structure) is given as part of the input~\cite{arnborg1985efficient,arnborg1989linear,bern1987linear}. The first algorithm for computing a tree decomposition of width $k$ for an $n$-vertex graph (or deciding that no such tree decomposition exists) had a running time in $\OO(n^{k+2})$~\cite{arnborg1987complexity}. 
Shortly afterwards, an algorithm with a running time in $f(k)\cdot n^2$ was given by Robertson and Seymour~\cite{robertson1995graph}. It consists of two steps. In a first step, a tree decomposition of with at most $4k+3$ is computed in $O(3^{2k}\cdot n^2)$ time (or concluded that the treewidth of the input graph is larger than $k$). In a second, non-constructive step (for which the function $f$ in the running time is not specified), this tree decomposition is improved to one with width at most $k$.
In 1991, Bodlaender and Kloks~\cite{bodlaender1996efficient}\footnote{An extended abstract of the paper by Bodlaender and Kloks~\cite{bodlaender1996efficient} appeared in the proceedings of the 18th International Colloquium on Automata, Languages, and Programming (ICALP) in 1991.}, and Lagergren and Arnborg~\cite{lagergren1991finding} independently discovered algorithms with a running time in $2^{\OO(k^3)}\cdot n$ for the second step. Furthermore, Bodlaender~\cite{bodlaender1996linear} gave an improved algorithm for the first step which allowed to compute a tree decomposition of width $k$ (or concluding that no such tree decomposition exists) in time $2^{\OO(k^3)}\cdot n$.
These results have been extremely influential and allowed (among other things) 
Courcelle's theorem~\cite{borie1992automatic,courcelle1990monadic,courcelle2012graph} to be applicable without the prerequisite that a tree decomposition for the input graph is known.
Roughly 30 years later, in 2023, Korhonen and Lokshtanov~\cite{korhonen2023improved} presented an algorithm with a running time in $2^{\OO(k^2)}\cdot n^4$, which is considered a substantial break-through. Whether a tree decomposition of width $k$ (if one exists) can be computed in \emph{single exponential time} in $k$, that is, a running time in $2^{\OO(k)}\cdot n^{\OO(1)}$, remains a major open question in parameterized complexity theory~\cite{korhonen2023improved}.\footnote{We remark that under the exponential time hypothesis (ETH)~\cite{IP01,IPZ01}, the existence of algorithms with a running time in $2^{o({k})}\cdot n^{\OO(1)}$ can be ruled out~\cite{bonnet25}.}

For many other graph problems, however, single exponential time algorithms in the treewidth of the input graph are known~\cite{arnborg1989linear,Bodlaender88,bodlaender2015deterministic,bodlaender2008combinatorial,cygan2022solving,DF13,FG06,Nie06}.
If we want to use those algorithms to solve a graph problem, computing a tree decomposition of sufficiently small width becomes the bottleneck (in terms of the running time). Since we do not know how to compute a tree decomposition of minimum width in single exponential time, we have to rely on near-optimal tree decompositions if we do not want to significantly increase the running time of the overall computation. Therefore, much effort has been put into finding single exponential time algorithms that compute a tree decomposition with approximately minimum width, more specifically, with a width that is at most a constant factor away from the optimum. The first such algorithm was given by Robertson and Seymour~\cite{robertson1995graph} in 1995, and since then many improvements both in terms of running time and approximation factor have been made~\cite{amir2010approximation,BelbasiF21,BelbasiF22,bodlaender2016c,Korhonen23}. The current best algorithm by Korhonen~\cite{Korhonen23} has an approximation factor of two and runs in $2^{\OO(k)}\cdot n$ time. For an overview on the specific improvements, we refer to the paper by Korhonen~\cite{Korhonen23}.

\subsection{Our Contribution} In this work, we focus on computing optimal tree decomposition. We overcome the difficulties of obtaining a single exponential time algorithm by moving to a larger parameter. The arguably most natural parameter that measures ``tree-likeness'' and is larger than treewidth is the \emph{feedback vertex number}. It is the smallest number of vertices that need to be removed from a graph to turn it into a forest. 
Bodlaender, Jansen, and Kratsch~\cite{bodlaender2013preprocessing} studied feedback vertex number in the context of kernelization algorithms for treewidth computation.
The main result of our work is the following.

\begin{theorem}\label{thm:main}
    Given an $n$-vertex graph $G$ and an integer $k$, we can compute a tree decomposition for $G$ with width $k$, or decide that no such tree decomposition exists, in $2^{\OO(\fvn(G))}\cdot n^{\OO(1)}$ time. Here, $\fvn(G)$ is the feedback vertex number of graph $G$.
\end{theorem}

\Cref{thm:main} directly improves a previously known result, that an optimal tree decomposition for an $n$-vertex graph $G$ can be computed in $2^{\OO(\vcn(G))}\cdot n^{\OO(1)}$ time (if it exists)~\cite{chapelle2017treewidth,fomin2018algorithms}, where $\vcn(G)$ is the \emph{vertex cover number} of $G$, which is the smallest number of vertices that need to be removed from a graph to remove all of its edges, and hence always at least as large as the feedback vertex number. To the best of our knowledge, the only other known result for treewidth computation that uses a different parameter to obtain single exponential running time is by Fomin et al.~\cite{fomin2018algorithms}, who showed that an optimal tree decomposition can be computed in $2^{\OO(\mw(G))}\cdot n^{\OO(1)}$, where $\mw(G)$ is the \emph{modular width} of $G$. We remark that the modular width can be upper-bounded by a (non-linear) function of the vertex cover number, but is incomparable to both the treewidth and the feedback vertex number of a graph.
It follows that our result improves the best-known algorithms for computing optimal tree decompositions (in terms of the exponential part of the running time) for graphs $G$ where 
\[
\fvn(G)\in o\bigl(\max\bigl\{\tw(G)^2,\vcn(G),\mw(G)\bigr\}\bigr),
\]
where $\tw(G)$ denotes the treewidth of $G$. 
On a different note, since $\fvn(G)$ is an upper bound on $\tw(G)$, our algorithm can also be seen either as an important step towards a positive resolution of the open problem of finding an $2^{\OO(\tw(G))}\cdot n^{\OO(1)}$ time algorithm for computing an optimal tree decomposition, or, if its answer is negative, then a mark of the tractability border of single exponential time algorithms for the computation of treewidth.

Our algorithm constitutes a significant extension of a dynamic programming algorithm by Chapelle et al.~\cite{chapelle2017treewidth} for computing optimal tree decompositions which runs in  $2^{\OO(\vcn(G))}\cdot n^{\OO(1)}$ time. The building blocks of the algorithm are fairly simple and easy to implement, however the analysis is quite technical and involved. 
First of all, our algorithm needs to have access to a \emph{minimum feedback vertex set}, that is, a set of vertices of minimum cardinality, such that the removal of those vertices turns the graph into a forest. It is well-known that such a set can be computed in $2.7^{\fvn(G)}\cdot n^{\OO(1)}$ time (randomized)~\cite{li2022detecting} or in $3.6^{\fvn(G)}\cdot n^{\OO(1)}$ deterministic time \cite{kociumaka2014faster}.
As a second step, we apply the kernelization algorithm by Bodlaender, Jansen, and Kratsch~\cite{bodlaender2013preprocessing} to reduce the number of vertices in the input graph to $\OO(\fvn(G)^4)$. We remark that this kernelization algorithm needs a minimum feedback vertex set as part of the input and ensures that this set remains a feedback vertex set for the reduced instance~\cite{bodlaender2013preprocessing}. This second step is neither necessary for the correctness of our algorithm nor to achieve the claimed running time bound, but it ensures that the polynomial part of our running time is bounded by the maximum of the running time of the kernelization algorithm and the polynomial part of the running time needed to compute a minimum feedback vertex set. Neither Bodlaender, Jansen, and Kratsch~\cite{bodlaender2013preprocessing}, Li and Nederlof~\cite{li2022detecting}, nor Kociumaka and Pilipczuk~\cite{kociumaka2014faster} specify the polynomial part of the running time of their respective algorithms. However, an inspection of the analysis for the kernelization algorithm suggests that its running time is in~$\OO(n^5)$ and it is known that a minimum feedback vertex set can be computed in $2^{\OO(\fvn(G))}\cdot n^{2}$ time~\cite{Cao18}.


In the following, we give a brief overview of the main ingredients of our algorithm.
The main purpose of this overview is to convey an intuition and an idea of how the algorithm works; it is not meant to be completely precise. 
The basic terminology from graph theory that we use here is formally introduced and defined in \cref{sec:prelims}.
\begin{itemize}
\item Chapelle et al.~\cite{chapelle2017treewidth} introduced notions and terminology to formalize how a (rooted) tree decomposition interacts with a minimum vertex cover. We adapt and extend those notions for minimum feedback vertex sets. Inspired by the well-known concept of \emph{nice} tree decompositions~\cite{bodlaender1996efficient}, we define a class of rooted tree decompositions that interact with a given feedback vertex set in a ``nice'' way. We call those tree decompositions \emph{$S$-nice}, where $S$ denotes the feedback vertex set.
In $S$-nice tree decompositions, we classify each node $t$ by which vertices of the feedback vertex set are contained in the bag and which are only contained in bags of nodes in the subtree rooted at $t$. We show that nodes of the same class form (non-overlapping) paths in the tree decomposition.
\item We give algorithms to compute candidates of bags for the top and the bottom node of a path that corresponds to a given node class. Here, informally speaking, we differentiate between the cases where nodes (and their neighbors) have ``full'' bags or not, where we consider a bag to be full if adding another vertex to it increases the width of the tree decomposition. We need two novel algorithmic approaches for the two cases.
\item We give an algorithm to compute the maximum width of any bag ``inside'' a path that corresponds to a given node class. This algorithm needs as input the bag of the top node and the bottom node of the path, and the set of vertices that should be contained in some bags of the path.
\item The above-described algorithm allows us to perform dynamic programming on the node classes. Informally, we show that node classes form a partially ordered set. This means that for a given node class, we can look up a partial tree decomposition where the root belongs to a preceding (according to the poset) node class in a dynamic programming table. Then we use the above-described algorithm to extend the looked-up tree decomposition to one where the root has the given node class.
\end{itemize}
We show that in the above-described way, we can find a tree decomposition of width $k$ whenever one exists.
To this end, we introduce (additional) novel classes of tree decompositions that have properties that we can exploit algorithmically. We introduce so-called \emph{\slim} tree decompositions, which, informally speaking, have a minimum number of full nodes and some additional properties. Furthermore, we introduce \emph{\topheavy} tree decompositions, where, informally speaking, vertices are pushed into bags that are as close to the root as possible. We show that there exists tree decompositions with minimum width that are $S$-nice, \slim, and \topheavy.
Lastly, we give a number of ways to modify tree decompositions. We will show that we can modify any optimal tree decomposition that is $S$-nice, slim, and top-heavy into one that our algorithm is able to find.

\subsection{Further Related Work} We already gave an overview of the work related to computing tree decompositions. 
The feedback vertex number is a fundamental graph parameter that has found applications in many contexts. The problem of computing the feedback vertex number is included in Karp's original list of 21 NP-hard problems~\cite{Kar72}. As already mentioned, it can be computed in single exponential time and there is a long lineage of improvements in the running time~\cite{becker2000randomized,guo2006compression,kociumaka2014faster,li2022detecting}. For more specific information on the improvements in the running time, we refer to the paper by Li and Nederlof~\cite{li2022detecting}. Furthermore, many variations of this problem have been studied~\cite{agrawal2018simultaneous,casteigts2021finding,chaudhary2025parameterized,cygan2013subset,kanesh2021parameterized,LokshtanovRS18,misra2012parameterized,misra2012fpt}.
The feedback vertex number has also been considered as a parameter in various contexts. In particular, as mentioned before, Bodlaender, Jansen, and Kratsch~\cite{bodlaender2013preprocessing} studied feedback vertex number in the context of kernelization algorithms for treewidth computation. It has also been used as a parameter for algorithmic and hardness results in various other contexts~\cite{AdigaCS10,casteigts2021finding,EnrightMM23,jansen2013vertex,KratschS10,zehavi2023tournament}.

\section{Algorithm Overview}\label{sec:overview}
In this section, we take a closer look at the different building blocks of our algorithm before we explain each one in full detail. The goal of this section is to provide an intuition about how the algorithm works, that is fairly close to the technical details but not completely precise. We describe the main ideas of each part of the algorithm and how they interact with each other.

\subsection{How Feedback Vertex Sets Interact With Tree Decompositions}
Given a rooted tree decomposition $\mathcal{T}=(T,\{X_t\}_{t\in V(T)})$ of a graph $G$ and a minimum feedback vertex set $S$ for $G$, the feedback vertex set interacts with the bags of $\mathcal{T}$ is a very specific way. 
Given a bag $X_t$ of some node in $t$, we use $X_t^S=X_t\cap S$ to denote the vertices from $S$ that are inside the bag $X_t$. We use $L^S_t\subseteq S$ to denote the set of vertices from $S$ that are \emph{only} in bags $X_{t'}$ ``below'' $X_t$. More formally, bags $X_{t'}$ of nodes $t'$ that are descendants of node $t$. We use $R^S_t = S\setminus (L^S_t\cup X^S_t)$ to denote the set of vertices in $S$ that are \emph{only} in bags of nodes that are outside the subtree of $\mathcal{T}$ rooted at $t$. 
These three sets build a so-called \emph{$S$-trace}, a concept introduced by Chapelle et al.~\cite{chapelle2017treewidth}.
One can show that the nodes in a tree decomposition that share the same $S$-trace partition the tree decomposition into path segments~\cite{chapelle2017treewidth}.
These path segments (or the corresponding $S$-traces) form a partially ordered set which, informally speaking, allows us to design a bottom-up dynamic programming algorithm that constructs a rooted tree decomposition for $G$ from the leaves to the root. 

A similar approach was used by Chapelle et al.~\cite{chapelle2017treewidth} to obtain a dynamic programming algorithm that has single exponential running time in the vertex cover number of the input graph. They use a minimum vertex cover instead of a minimum feedback vertex set to define $S$-traces. We extend this idea to work for feedback vertex sets, which requires a significant extension of almost all parts of the algorithm.

The main idea of the overall algorithm is the following. We give a polynomial-time algorithm to compute a partial rooted tree decomposition containing nodes corresponding to a directed path of an $S$-trace and, informally speaking, some additional nodes that only cover vertices from the forest $G-S$.
We combine this partial rooted tree decomposition with one for the preceding directed paths according to the above-mentioned partial order. In other words, intuitively, we give a subroutine to compute partial tree decompositions for directed paths corresponding to $S$-traces, and try to assemble them to a rooted tree decomposition for the whole graph. We start with the directed paths that do not have any predecessor paths and then try to extend them until we cover all vertices. We save the partial progress we made in a dynamic programming table. The overall algorithm is composed out of fairly simple subroutines, which makes the algorithm easy to implement. However, a very sophisticated analysis is necessary to prove its correctness and the desired running time bound.

In order to bound the number of dynamic programming table look-ups that we need to compute a new entry, we need to be able to assume that the rooted tree decomposition we are aiming to compute behaves nicely (in a very similar sense as in the well-known concept of \emph{nice tree decompositions}~\cite{bodlaender1996efficient}) with respect to the vertices in $S$. To this end, we introduce the concept of \emph{$S$-nice tree decompositions} and we prove that every graph $G$ admits an $S$-nice tree decomposition with minimum width for every choice of $S$. 

\subsection{Computing a ``Local Tree Decomposition'' for a Directed Path}

The main technical difficulty is to compute the part of the rooted tree decomposition that contains the nodes corresponding to a directed path of an $S$-trace. As described before, these parts are the main building blocks with which we assemble our tree decomposition for the whole graph. 

Our approach here is to compute candidates for the bags of the top node and the bottom node of the directed path corresponding to an $S$-trace. Intuitively, since the $S$-trace specifies where in the tree decomposition the other vertices from $S$ are located and since we know that $G-S$ is a forest, we have some knowledge on how these bags need to look like. 
If the bag $X$ we are computing is for a node that, when removed, splits the tree decomposition into few parts, then we can ``guess'' which vertices of $S$ are located in which of the different parts in an optimal tree decomposition. 
To this end, we adapt the concept of ``introduce nodes'', ``forget nodes'', and ``join nodes'' for $S$-nice tree decompositions. A consequence is that if a node does not share its bag $X$ with any neighboring node, then, when the node is removed, the tree decomposition splits into at most three relevant parts (that is, parts that contain bags with vertices from $S$ that are not in $X$).
In other word, the bag $X$ acts as a separator for at most three parts of $S$.
This allows us to compute a candidate for $X$ in a greedy fashion, and argue that we can modify an optimal tree decomposition to agree with our choice for $X$. An important property that we need for this is, that neighboring bags are not ``full'', that is, their size is not $\tw(G)+1$, which gives us flexibility to shuffle around vertices.

However, if the optimal tree decomposition contains a large subtree of nodes that all have the same bag and that are all full, then removing those nodes can split the tree decomposition into an unbounded number of parts that each can contain a different subset of $S$. Here, informally speaking, we cannot afford to ``guess'' anymore how the vertices of $S$ are distributed among the parts and we cannot shuffle around vertices. Hence, we have to find a way to compute a candidate bag for this case where we do not have to change the tree decomposition to agree with our choice. 

To resolve this issue, we need to refine the concept of $S$-nice tree decompositions. We introduce the \emph{\slim $S$-nice tree decompositions} that, informally speaking, ensure subtrees of the tree decomposition where all nodes have the same full bag have certain desirable properties that allow us to compute candidate bags efficiently. We prove that there are \slim $S$-nice tree decomposition of minimal width for all graphs and all choices of $S$. Furthermore, we prove (implicitly) that we can bound the size of the search-space for a potential bag by a function that is single exponential in the feedback vertex number. 
Therefore, we will able to provide small ``advice strings'', that lead our algorithm for this case to a candidate bag and we can prove that there is an advice string that will lead the algorithm to the correct bag.

\subsection{Coordinating Between Different Directed Paths}

Since we need to modify the optimal tree decomposition in order to agree with the bags that we compute, we have to make sure that later modifications do not invalidate earlier choices for bags. To ensure this, we introduce so-called \emph{\topheavy} \slim $S$-nice tree decompositions. Informally speaking, we try to push all vertices that are not in $S$ into bags as close to the root as possible. Intuitively, this ensures that if we make modifications to those tree decompositions to make them agree with our bag choices, those modifications do not affect bags that are sufficiently far ``below'' the parts of the tree decomposition that we modify.

Now, having access to the bag of the top node and the bottom node of a directed path corresponding to some $S$-trace, we have to decide which other vertices (that are not in $S$) we wish to put into bags that are attached to that path. Here, we also have to make sure that our decisions are consistent. We will compute a three-partition of the complete vertex set: one part that represents the vertices that are contained in some bag of the directed path or a bag that is attached to it, one part of vertice that are only in bags ``below'' the directed path, and the last part with the remaining vertices, that is, the vertices that are only in bags outside of the subtree of the tree decomposition that is rooted at the top node of the directed path.
The first described part is needed to compute the partial tree decomposition of the directed path. We provide a polynomial-time algorithm to do this.
Furthermore, the three-partition will allow us to decide which choices for bags to the top and bottom nodes of preceding directed paths are compatible with each other. 

These are all the ingredients we need to design our algorithm.

\subsection{Organization of the Paper}

We organize our paper as follows. In \cref{sec:prelims} we introduce all standard notation and terminology that we use. In \cref{sec:mainconcepts} we present some concepts by Chapelle et al.~\cite{chapelle2017treewidth} that we adopt for our results and we introduce novel concepts, such as $S$-nice tree decompositions, \slim tree decompositions, and \topheavy tree decompositions. Furthermore, we describe our methods to modify optimal tree decompositions such that they agree with our bag choices. In \cref{sec:mainalgo} we present the main dynamic programming algorithm. \cref{sec:bags} is dedicated to the subroutines we use to compute the top bag and the bottom bag of directed paths corresponding to $S$-traces. In \cref{sec:localtw} we give a polynomial-time algorithm to compute the part of the tree decomposition for an directed path, provided that we know the bag of the top node, the bag of the bottom node, and the set of vertices that we wish to be contained in some bag attached to the directed path. In \cref{sec:legal} we describe how we decide whether our decisions for different directed paths are consistent. Finally, in \cref{sec:correct} we prove that our algorithm is correct and we give a running time analysis. In \cref{sec:conclusion}, we conclude and give some future research directions.

\section{Preliminaries}\label{sec:prelims}

We use standard notation and terminology from graph theory~\cite{Die16}. 
Let $G=(V,E)$ with $E\subseteq \binom{V}{2}$ be a graph. We denote by~$V(G)=V$ and~$E(G)=E$ the sets of its vertices and edges, respectively.
We use $n$ to denote the number of vertices of $G$.
We call two vertices $u,v\in V$ \emph{adjacent} if~$\{u,v\}\in E$.
The \emph{open neighborhood} of a vertex $v\in V$ is the set $N_G(v)=\{u \mid \{u,v\}\in E\}$. 
The \emph{degree} of a vertex $v\in V$ is $\deg(v)=|N_G(v)|$.
The \emph{closed neighborhood} of a vertex $v\in V$ is the set $N_G[v]=N_G(V)\cup\{v\}$. 
If a vertex $v$ has degree one, then we call $v$ a \emph{leaf}.
For some vertex set $V'\subseteq V$, we define its open neighborhood as $N_G(V')=\bigcup_{v\in V'}N_G(v)\setminus V'$ and its closed neighborhood as $N_G[V']=N_G(V')\cup V'$.
If $G$ is clear from the context, we drop the subscript $G$ from $N$.
For some vertex set $V'\subseteq V$, we denote by~$G[V']$ the \emph{induced} subgraph of $G$ on the vertex set $V'$, that is, $G[V']=(V',E')$ where $E' = \{\{v,w\}\mid \{v,w\}\in E\wedge v\in V'\wedge w\in V'\}$. We further denote $G-V'=G[V\setminus V']$.
We say that a sequence $P=(\{v_{i-1},v_i\})_{i=1}^\ell$ of edges in $E$ forms a \emph{path} in $G$ from $v_0$ to $v_\ell$ if $v_{i}\neq v_j$ for all $0\le i<j\le \ell$.
We denote with $V(P)=\{v_0,v_1,\ldots,v_\ell\}$ the vertices visited by path $P$.
We say that vertices $u$ and $v$ are \emph{connected} in $G$ if there is a path from $u$ to $v$ in $G$. We say that graph $G$ is connected if for all $u,v\in V$ we have that $u$ and $v$ are connected. We say that $G[V']$ for some $V'\subseteq V$ is a \emph{connected component} of $G$ if $G[V']$ is connected and for all $v\in V\setminus V'$ we have that $G[V'\cup\{v\}]$ is not connected.
We say that a vertex set $V'\subseteq  V$ \emph{separates} a vertex set $V_1\subseteq V\setminus V'$ and a vertex set $V_2\subseteq V\setminus V'$ if $V_1\cap V_2=\emptyset$ and there is no $u\in V_1$ and $v\in V_2$ such that there is a path from $u$ to $v$ in $G-V'$.
We call $G$ a \emph{tree} if for each pair of vertices $u,v\in V$ there is exactly one path from $u$ to $v$ in $G$. We call $G$ a \emph{forest} if for each pair of vertices $u,v\in V$ there is at most one path from $u$ to $v$ in $G$. We call $G$ \emph{trivial} if $|V|=1$.

Let $e=\{u,v\}\in E$. \emph{Contracting} edge $e$ in $G$ results in the graph obtained by deleting $v$ from $G$ and then adding all edges between $u$ and $u'\in N_G(v)$.
A graph $H$ is called a \emph{minor} of $G$ if $H$ can be obtained from $G$ by deleting edges, vertices, and by contracting edges.

Next, we formally define the \emph{treewidth} of a graph and related concepts such as \emph{tree decompositions} and \emph{nice tree decompositions}.

\begin{definition}[Tree Decomposition and Treewidth] \label{def:tree_decomposition}
A {\em tree decomposition} of a graph $G=(V,E)$ is a pair $\mathcal{T}=(T,\{X_t\}_{t\in V(T)})$ consisting of a tree $T$ and a family of \emph{bags} $\{X_t\}_{t\in V(T)}$ with $X_t\subseteq V$, such that:
\begin{enumerate}
    \item $\bigcup_{t\in V(T)}X_t=V$. \label{condition_1_tree_decomposition}
    \item For every $\{u, v\}\in E$, there is a node $t\in V(T)$ such that $\{u, v\}\subseteq X_t$. \label{condition_2_tree_decomposition}
    \item For every $v\in V$, the set $X^{-1}(v)=\{t\in V(T) \mid v\in X_t\}$ induces a subtree of $T$. \label{condition_3_tree_decomposition}
\end{enumerate}
The {\em width} of a tree decomposition $\mathcal{T}=(T,\{X_t\}_{t\in V(T)})$ is $\max_{t\in V(T)}|X_t|-1$. The {\em treewidth} $\tw(G)$ of a graph $G$ is the minimum width over all tree decomposition of $G$. 
\end{definition}
Let $\mathcal{T}=(T,\{X_t\}_{t\in V(T)})$ be a tree decomposition of a graph $G=(V,E)$ with width $k$. To avoid confusion, we refer to the elements of $V(T)$ as \emph{nodes} and call elements of $V$ \emph{vertices}.
Let $X_t$ be the bag of a node $t$ in $\mathcal{T}$, if $|X_t|=k+1$, then we say the bag $X_t$ is \emph{full}, that is, increasing the size of the bag would increase the width of the tree decomposition.
The following lemmas are well-known~\cite{arnborg1990forbidden,bodlaender2006safe,Die16} and give us useful properties of tree decompositions.

\begin{lemma}[\cite{arnborg1990forbidden,Die16}]\label{lem:twminor}
Let $H$ be a minor of a graph $G$. Then $\tw(H)\le\tw(G)$.
\end{lemma}
   
\begin{lemma}[\cite{bodlaender2006safe,Die16}]
\label{lem:cliquebag}
    Let $\mathcal{T}$ be a tree decomposition of a graph $G=(V,E)$. Let $V'\subseteq V$ such that $G[V']$ is complete. Then, there exists a node $t$ in $\mathcal{T}$ such that $V'\subseteq X_t$. 
\end{lemma}

\begin{observation}[\cite{Die16}]\label{lem:cliquebag2}
Let $\mathcal{T}$ be a tree decomposition of a graph $G=(V,E)$ with width at most $k$, and let $t$ be a node in $\mathcal{T}$. Let $G'=(V,E\cup\binom{X_t}{2})$. 
Then $\mathcal{T}$ is a tree decomposition of a graph $G'$ with width at most $k$.
\end{observation}

In many applications, it is very convenient to use rooted tree decompositions with a particular structure. 
A \emph{rooted tree} $T=(V,E,r)$ is a triple where $(V,E)$ is a tree and $r\in V$ is a designated root node. Let $v\in V$ and $u\in N(v)$. If the unique path from $r$ to $u$ is shorter than the unique path from $r$ to $v$, then we say that $u$ is the \emph{parent} of $v$. Otherwise, we say that $u$ is a \emph{child} of $v$. Let $P=\left(\{v_{i-1},v_i\}\right)_{i=1}^\ell$ be a path in $T$ from $v_0$ to $v_\ell$. If $v_{i-1}$ is the parent of $v_i$ for all $1\le i\le \ell$, then we say that $P$ is a \emph{directed path} from $v_0$ to $v_\ell$. If for two vertices $u,v\in V$ there is a directed path from $u$ to $v$, then we say that $u$ is an \emph{ancestor} of $v$ and $v$ is a \emph{descendant} of $u$.
A {\em rooted tree decomposition} is a tree decomposition $\mathcal{T}=(T,\{X_t\}_{t\in V(T)})$ where $T$ is a rooted tree with root node $r\in V(T)$. 

\begin{definition}[Nice Tree Decomposition]\label{def:nicetd}
    A {\em nice} tree decomposition is a rooted tree decomposition $\mathcal{T}=(T,\{X_t\}_{t\in V(T)})$ that satisfies the following properties. The bag of $X_r$ is empty, and the bag $X_\ell$ is empty for every leaf $\ell$ in $T$. Additionally, every internal node $t$ of $T$ is of one of the following types: 
\begin{description}
    \item[Introduce Node:] $t$ has exactly one child $t'$, such that $X_t=X_{t'}\cup \{v\}$ for some $v\notin X_{t'}$.
    \item[Forget Node:] $t$ has exactly one child $t'$, such that $X_t=X_{t'}\setminus \{v\}$ for some $v\in X_j$. 
    \item[Join Node:] $t$ has exactly two children $t_1, t_2$, such that $X_t=X_{t_1}=X_{t_2}$. 
\end{description}
\end{definition}

It is well-known that every tree decomposition can be straightforwardly transformed into a nice tree decomposition with the same width.

\begin{lemma}[\cite{bodlaender1996efficient}]\label{lem:nicetd}
    If a graph $G$ admits a tree decomposition $\mathcal{T}$ with width $k$, then it also admits a nice tree decomposition $\mathcal{T}'$ with width $k$.
\end{lemma}

For our algorithm, we will use the feedback vertex number as a parameter and our algorithm will need access to a feedback vertex set. It is formally defined as follows.

\begin{definition}[Feedback Vertex Set and Feedback Vertex Number]
    A \emph{feedback vertex set} of a graph $G=(V,E)$ is a set $C\subseteq V$ such that $G-C$ is a forest. The \emph{feedback vertex number} $\fvn(G)$ of a graph $G$ is the cardinality of a minimum feedback vertex set.
\end{definition}

It is well-known that a minimum feedback vertex set can be computed in  $2^{\OO(\fvn)}\cdot n^{\OO(1)}$ time. To the best of our knowledge, the following are the currently best-known algorithms.

\begin{theorem}[\cite{kociumaka2014faster,li2022detecting}]\label{thm:computefvs}
    A minimum feedback vertex set of a graph $G$ can be computed in $2.7^{\fvn(G)}\cdot n^{\OO(1)}$ randomized time or in $3.6^{\fvn(G)}\cdot n^{\OO(1)}$ deterministic time.
\end{theorem}


\section{Main Concepts for the Algorithm}\label{sec:mainconcepts}
One of the main ingredients to prove \cref{thm:main} is showing that there are tree decompositions (in particular, also minimum-width ones) that interact with a given vertex set in a very structured way.  To this end, we present several concepts introduced by Chapelle et al.~\cite{chapelle2017treewidth} that we also use for our algorithm in \cref{sec:concepts}.  In particular, we give the definitions of \emph{traces}, \emph{directed paths} thereof~\cite{chapelle2017treewidth}, and several additional concepts. Then, in \cref{sec:snice}, we introduce $S$-nice tree decompositions, which interact with a vertex set $S$ in a ``nice'' way, but are more flexible with vertices not in $S$.
In \cref{sec:mod}, we present several ways to modify tree decompositions. Those are an important tool to show that we can modify an optimal tree decomposition into one that our algorithm finds. Finally, we introduce \slim and \topheavy $S$-nice tree decompositions in \cref{sec:useful}.
Those types of tree decompositions have crucial additional properties that we need to show that our algorithm is correct.
We use the modifications from \cref{sec:mod} to show that we can transform $S$-nice tree decompositions into ones of those types.

\subsection{\boldmath$S$-Traces, Directed Paths, $S$-Children, and $S$-Parents}
\label{sec:concepts}
Let $t$ be a node of a rooted tree decomposition $\mathcal{T}=(T,\{X_t\}_{t\in V(T)})$ of a graph $G=(V,E)$. We denote with~$T_t$ the subtree of $T$ rooted at node $t$. We use the following notation throughout the document:
\begin{itemize}
    \item The set $V_t=\bigcup_{t'\in V(T_t)}X_{t'}$ is the union of all bags in the subtree $T_t$.
    \item The set $L_t=V_t\setminus X_t$ is the set of vertices ``below'' node $t$.
    \item The set $R_t=V\setminus V_t$ is the set of the ``remaining'' vertices, that is, the ones that are neither in the bag of node $t$ nor below $t$.
\end{itemize}
Clearly, for every node $t$, the sets $L_t,X_t,R_t$ are a partition of $V$.
\begin{definition}[\cite{chapelle2017treewidth}) ($S$-Trace] 
\label{def:trace}
    Let $G=(V,E)$ be a graph and tet $S\subseteq V$. Let $t$ be a node of a rooted tree decomposition $\mathcal{T}=(T,\{X_t\}_{t\in V(T)})$ of $G$. 
    The \emph{$S$-trace} of $t$ is the triple $(L_t^S, X_t^S, R_t^S)$ such that $L_t^S=L_t\cap S$, $X_t^S=X_t\cap S$, and $R_t^S=R_t\cap S$. 
    A triple $(L^S,X^S,R^S)$ is an \emph{$S$-trace} if it is the $S$-trace of some node in some tree decomposition of~$G$. 
\end{definition}
We remark that in the work of Chapelle et al.~\cite{chapelle2017treewidth}, the set $S$ in \cref{def:trace} is fixed to be a minimum vertex cover. In our work, we will (later) fix it to be a minimum feedback vertex set. It is easy to see that $S$-traces have the following useful properties.
\begin{observation}[\cite{chapelle2017treewidth}]\label{obs:strace}
Let $G=(V,E)$ be a graph and $S\subseteq V$.
    \begin{enumerate}
        \item Let $t,t'$ be nodes of a rooted tree decomposition $\mathcal{T}=(T,\{X_t\}_{t\in V(T)})$ of a graph $G$. Let $(L_t^S, X_t^S, R_t^S)$ and $(L_{t'}^S, X_{t'}^S, R_{t'}^S)$ be the $S$-traces of $t$ and $t'$, respectively. If $t$ is an ancestor of $t'$, then $L_{t'}^S\subseteq L_t^S$.
        \item Let $L^S,X^S,R^S$ be a partition of the set $S$. The triple $(L^S,X^S,R^S)$ is an $S$-trace if and only if there is no path from a vertex in $L^S$ to a vertex in $R^S$ in $G[S]-X^S$.
    \end{enumerate}
\end{observation}

Chapelle et al.~\cite{chapelle2017treewidth} show that the nodes of a rooted tree decomposition that have the same $S$-trace form a directed path. They prove this for $S$-traces where $S$ is a minimum vertex cover and the rooted tree decomposition is nice, however inspecting their proof reveals that this property also holds for arbitrary sets $S$ and arbitrary rooted tree decompositions. Formally, we have the following.

\begin{lemma}[\cite{chapelle2017treewidth}]\label{lem:dirpath}
    Let $\mathcal{T}=(T,\{X_t\}_{t\in V(T)})$ be a rooted tree decomposition of a graph $G=(V,E)$. Let $S\subseteq V$ and let $(L^S, X^S, R^S)$ be an $S$-trace such that $L^S\ne \emptyset$. Moreover, let $N$ be the set of nodes of $\mathcal{T}$ whose $S$-trace is $(L^S, X^S, R^S)$. Then the subgraph of $T$ induced by $N$ is a directed path.
\end{lemma}
\cref{lem:dirpath} allows us to define so-called \emph{directed paths of $S$-traces}, which will be central to our algorithm. Informally, given an $S$-trace our algorithm will compute candidates for the corresponding directed path, and then use a dynamic programming approach to connect the path to the rest of a partially computed rooted tree decomposition.

\begin{definition}[Directed Path of an $S$-Trace, Top Node, and Bottom Node]\label{def:path_of_a_trace}
    Let $\mathcal{T}=(T,\{X_t\}_{t\in V(T)})$ be a rooted tree decomposition of a graph $G=(V,E)$, $S\subseteq V$ and $(L^S, X^S, R^S)$ an $S$-trace such that $L^S\ne \emptyset$. We call the directed path induced by the set of nodes of $\mathcal{T}$ whose $S$-trace is $(L^S, X^S, R^S)$ \emph{the directed path of} $(L^S, X^S, R^S)$ in $\mathcal{T}$ from node $t_{\max}$ to node $t_{\min}$. We call $t_{\max}$ the \emph{top node} of the directed path and we call $t_{\min}$ the \emph{bottom node} of the directed path.
\end{definition}

We use the above concepts and some additional ones that we introduce below to define a specific type of rooted tree decomposition that interacts with the set $S$ in a very structured way. The concepts are illustrated in \cref{fig:traces}. First, we distinguish nodes that are at the end of a directed path of some $S$-trace.

\begin{definition}[$S$-Bottom Node]\label{def:bottom_node}
    Let $\mathcal{T}=(T,\{X_t\}_{t\in V(T)})$ be a rooted tree decomposition of a graph $G=(V,E)$ and $S\subseteq V$. Let $t$ be a node of $\mathcal{T}$.  We say that $t$ is an \emph{$S$-bottom} node if there is an $S$-trace $(L^S, X^S, R^S)$ such that $L^S\ne \emptyset$, and $t$ is the bottom node $t_{\min}$ of the directed path of $(L^S, X^S, R^S)$ in $\mathcal{T}$. 
\end{definition}
Next, we observe that for every directed path of some $S$-trace, the parent of $t_{\max}$ is also an $S$-bottom node. Note that this implies that a situation as depicted in \cref{fig:tracesb} cannot happen.

\begin{observation}\label{obs:sbottom}
    Let $\mathcal{T}=(T,\{X_t\}_{t\in V(T)})$ be a rooted tree decomposition of a graph $G=(V,E)$, $S\subseteq V$ and $(L^S, X^S, R^S)$ an $S$-trace such that $L^S\ne \emptyset$. Let the directed path of $(L^S, X^S, R^S)$ go from $t_{\max}$ to node $t_{\min}$. We have that if $t_{\max}$ has a parent, then the parent is an $S$-bottom node.
\end{observation}
\begin{proof}
    Assume that $t_{\max}$ has a parent and let $t$ denote the parent of $t_{\max}$. Let $(L_t^S, X_t^S, R_t^S)$ be the $S$-trace of $t$. We have that $(L_t^S, X_t^S, R_t^S)\neq (L^S, X^S, R^S)$ since $t$ is not part of the directed path of $(L^S, X^S, R^S)$. Furthermore, by \cref{obs:strace} we have that $L^S\subseteq L_t^S$ and hence $L_t^S\neq\emptyset$. If $t_{\max}$ is the only child of $t$, then we are done. Assume that $t$ has a child $t'\neq t_{\max}$. Let $(L_{t'}^S, X_{t'}^S, R_{t'}^S)$ be the $S$-trace of $t'$. Assume for contradiction that $(L_t^S, X_t^S, R_t^S)= (L_{t'}^S, X_{t'}^S, R_{t'}^S)$. Then we have that $L_{t'}^S\subseteq R^S$ and hence $L_{t}^S\subseteq R^S$. This together with $L^S\subseteq L_t^S$ implies that $L^S\subseteq R^S$, a contradiction.
\end{proof}

\begin{figure}[t]
\centering
\begin{subfigure}[t]{0.49\textwidth}
\centering
\begin{tikzpicture}[line width=1pt,scale=.8,xscale=1]
\phantom{
\draw[draw, line width=10pt, color=red!60!white] (-3,0) -- (3,-9);
\draw[draw, line width=10pt, color=red!60!white] (3,0) -- (-3,-9);
}

\draw[trace] (0,0) -- (0,-3);
\draw[trace] (-2,-4) -- (-2,-6);
\draw[trace] (-2,-7) -- (-2,-8);
\draw[trace] (2,-4) -- (2,-7);

\node[vert] (v1) at (0,0) {};
\node[vert] (v2) at (0,-1) {};
\node[vert] (v3) at (0,-2) {};
\node[vert2] (v4) at (0,-3) {};

\node[trial] (t1) at (1,-1.5) {};
\draw[edge] (v2) -- (t1.apex);

\node[triac] (t2) at (0,-4) {};
\draw[edge] (v4) -- (t2.apex);

\node[vert] (u1) at (-2,-4) {};
\node[vert] (u2) at (-2,-5) {};
\node[vert2] (u3) at (-2,-6) {};

\node[triar] (t3) at (-3,-6.5) {};
\draw[edge] (u3) -- (t3.apex);

\node[vert] (u4) at (-2,-7) {};
\node[vert2] (u5) at (-2,-8) {};

\node[vert] (w1) at (2,-4) {};
\node[vert] (w2) at (2,-5) {};
\node[vert] (w3) at (2,-6) {};
\node[vert2] (w4) at (2,-7) {};

\node[trial] (t4) at (3,-6.5) {};
\draw[edge] (w3) -- (t4.apex);

\draw[edge,dashed] (0,.75) -- (v1);

\draw[edge] (v1) -- (v2);
\draw[edge] (v2) -- (v3);
\draw[edge] (v3) -- (v4);

\draw[edge] (v4) -- (u1);

\draw[edge] (u1) -- (u2);
\draw[edge] (u2) -- (u3);
\draw[edge] (u3) -- (u4);
\draw[edge] (u4) -- (u5);

\draw[edge] (v4) -- (w1);

\draw[edge] (w1) -- (w2);
\draw[edge] (w2) -- (w3);
\draw[edge] (w3) -- (w4);

\draw[edge,dashed] (u5) -- (-3,-8.5);
\draw[edge,dashed] (u5) -- (-1,-8.5);
\draw[edge,dashed] (w4) -- (2,-7.75);
\end{tikzpicture}
\caption{Illustration of an $S$-nice tree decomposition.}\label{fig:tracesa}
\end{subfigure}
\begin{subfigure}[t]{0.49\textwidth}
\centering
\begin{tikzpicture}[line width=1pt,scale=.8,xscale=1]
\draw[draw, line width=10pt, color=red!60!white] (-3,0) -- (3,-9);
\draw[draw, line width=10pt, color=red!60!white] (3,0) -- (-3,-9);

\draw[trace] (0,0) -- (0,-3) -- (2,-4) -- (2,-7);
\draw[trace] (-2,-4) -- (-2,-6);
\draw[trace] (-2,-7) -- (-2,-8);

\node[vert] (v1) at (0,0) {};
\node[vert] (v2) at (0,-1) {};
\node[vert] (v3) at (0,-2) {};
\node[vert] (v4) at (0,-3) {};

\node[trial] (t1) at (1,-1.5) {};
\draw[edge] (v2) -- (t1.apex);

\node[triac] (t2) at (0,-4) {};
\draw[edge] (v4) -- (t2.apex);

\node[vert] (u1) at (-2,-4) {};
\node[vert] (u2) at (-2,-5) {};
\node[vert2] (u3) at (-2,-6) {};

\node[triar] (t3) at (-3,-6.5) {};
\draw[edge] (u3) -- (t3.apex);

\node[vert] (u4) at (-2,-7) {};
\node[vert2] (u5) at (-2,-8) {};

\node[vert] (w1) at (2,-4) {};
\node[vert] (w2) at (2,-5) {};
\node[vert] (w3) at (2,-6) {};
\node[vert2] (w4) at (2,-7) {};

\node[trial] (t4) at (3,-6.5) {};
\draw[edge] (w3) -- (t4.apex);

\draw[edge,dashed] (0,.75) -- (v1);

\draw[edge] (v1) -- (v2);
\draw[edge] (v2) -- (v3);
\draw[edge] (v3) -- (v4);

\draw[edge] (v4) -- (u1);

\draw[edge] (u1) -- (u2);
\draw[edge] (u2) -- (u3);
\draw[edge] (u3) -- (u4);
\draw[edge] (u4) -- (u5);

\draw[edge] (v4) -- (w1);

\draw[edge] (w1) -- (w2);
\draw[edge] (w2) -- (w3);
\draw[edge] (w3) -- (w4);

\draw[edge,dashed] (u5) -- (-3,-8.5);
\draw[edge,dashed] (u5) -- (-1,-8.5);
\draw[edge,dashed] (w4) -- (2,-7.75);
\end{tikzpicture}
\caption{\cref{obs:sbottom} implies that this cannot happen.}\label{fig:tracesb}
\end{subfigure}
    \caption{Illustrations for the concepts of $S$-traces (\cref{def:trace}) and their directed paths (\cref{def:path_of_a_trace}), $S$-bottom nodes (\cref{def:bottom_node}), and $S$-nice tree decompositions (\cref{def:snicetd}). Subfigure \ref{fig:tracesa} shows an illustration of a part of an $S$-nice tree decomposition. The directed paths of different $S$-traces are surrounded by green lines. The $S$-bottom nodes are depicted as gray squares. The blue triangles illustrate parts of the tree decompositions where all nodes have $S$-traces with $L^S=\emptyset$, and hence those nodes are not part of any directed paths of $S$-traces. Subfigure \ref{fig:tracesb} shows a configuration of directed paths of different $S$-traces that we can rule out by \cref{obs:sbottom}.}\label{fig:traces}
\end{figure}

We now define so-called \emph{$S$-parent} and \emph{$S$-children}. Informally, an $S$-child is a child of a node such that there is a vertex in $S$ that is only contained in nodes in the subtree of the tree decomposition rooted at the $S$-child. The $S$-parent of a node is defined analogously. Formally, we define them as follows.

\begin{definition}[$S$-Parent and $S$-Child]\label{def:sparentchild}
Let $\mathcal{T}=(T,\{X_t\}_{t\in V(T)})$ be a rooted tree decomposition of a graph $G=(V,E)$ and $S\subseteq V$.
Let $t$ and $t'$ be two nodes of $\mathcal{T}$ such that $t$ is the parent of $t'$ in $T$. Let $(L_t^S, X_t^S, R_t^S)$ and $(L_{t'}^S, X_{t'}^S, R_{t'}^S)$ be the $S$-traces of $t$ and $t'$, respectively.
    \begin{itemize}
        \item If $(R_t^S \cup X_t^S)\setminus (L_{t'}^S \cup X_{t'}^S)\neq \emptyset$, then we say $t$ is an \emph{$S$-parent} of $t'$. 
        \item If $(L_{t'}^S \cup X_{t'}^S)\setminus (R_t^S \cup X_t^S) \neq \emptyset$, then we say $t'$ is an \emph{$S$-child} of $t$. 
    \end{itemize}
\end{definition}

\subsection{\boldmath$S$-Nice Tree Decompositions}\label{sec:snice}
%
We now give the definition of a so-called \emph{$S$-nice} tree decomposition. Intuitively, this tree decomposition behaves nicely (in the sense of \cref{def:nicetd}) when interacting with vertices from the vertex set $S$, but is more flexible for the vertices in $V\setminus S$.
Similarly to nice tree decomposition, $S$-nice tree decompositions distinguish three main types of $S$-bottom nodes, analogous to \emph{introduce}, \emph{forget}, and \emph{join} nodes. We associate each $S$-bottom node with $S$-trace $(L^S,X^S,R^S)$ with a so-called \emph{S-operation}, which can be $\introduce(v)$ for some $v\in X^S$, $\forget(v)$ for some $v\in S\setminus X^S$,
or $\join(X^S,X_1^S,X_2^S,L_{1}^S, L_{2}^S)$ for some $X_1^S,X_2^S\subseteq X^S$ and $L_{1}^S,L_{2}^S\subseteq L^S$.
\begin{definition}[$S$-Nice Tree Decomposition and $S$-Operations] \label{def:snicetd}
    A rooted tree decomposition $\mathcal{T}=(T,\{X_t\}_{t\in V(T)})$ of a graph $G=(V,E)$ is \emph{$S$-nice} for $S\subseteq V$ if the following hold. Let $t$ be a node.
    Then the following holds.
    \begin{enumerate}
        \item If $t$ has a parent $t'$, then $X_t\subseteq X_{t'}$ or $X_t\supseteq X_{t'}$, and $-1\le |X_t|-|X_{t'}|\le 1$.
    \end{enumerate}
        If $t$ is an $S$-bottom node in $\mathcal{T}$ with $S$-trace $(L_t^S, X_t^S, R_t^S)$, then additionally one of the following holds.
        \begin{enumerate}\setcounter{enumi}{1}
            \item Node $t$ has exactly one $S$-child $t_1$. 
            We have that
            $X^S_t= X^S_{t_1}\cup \{v\}$ for some $v\in S\setminus X_{t_1}$. 
            
            Then we say that $t$ admits the \emph{$S$-operation} $\introduce(v)$. 
            \label{condition_S_nice_condition_1_a}
            
            \item Node $t$ has exactly one $S$-child $t_1$.
            We have that
            $X^S_t=X^S_{t_1}\setminus \{v\}$ for some $v\in S\cap X_{t_1}$. 
            
            Then we say that $t$ admits the \emph{$S$-operation} $\forget(v)$. 
            \label{condition_S_nice_condition_1_b}
            
            \item Node $t$ has exactly two $S$-children $t_1, t_2$.  
            We have that $X_{t_1}^S\cup X_{t_2}^S\subseteq X_t^S$.

            Let $(L_{t_1}^S, X_{t_1}^S, R_{t_1}^S)$ and $(L_{t_2}^S, X_{t_2}^S, R_{t_2}^S)$ be the two $S$-traces of $t_1$ and $t_2$, respectively.
            We say that $t$ admits the \emph{$S$-operation} $\join(X^S_t,X^S_{t_1}, X^S_{t_2},L_{t_1}^S, L_{t_2}^S)$. 
            \label{condition_S_nice_condition_1_c}
        \end{enumerate} 
\end{definition}

See \cref{fig:tracesa} for an illustration of an $S$-nice tree decomposition. We show that if there is a tree decomposition for a graph $G=(V,E)$, then for every $S\subseteq V$ there is also an $S$-nice tree decomposition with the same width for $G$.

\begin{lemma}\label{lem:snicetd}
    Let $\mathcal{T}$ be a nice tree decomposition for a graph $G=(V,E)$, then for every $S\subseteq V$, the tree decomposition $\mathcal{T}$ is $S$-nice.
\end{lemma}
\begin{proof}
    Let $G=(V,E)$ be a graph and let $\mathcal{T}$ be a nice tree decomposition with width $k$. 
    Let $S\subseteq V$. We show that $\mathcal{T}$ is $S$-nice.
Let $t$ be a node and let $t'$ be the parent node of $t$.      
     Then from \cref{def:nicetd} we clearly have that $X_t\subseteq X_{t'}$ or $X_t\supseteq X_{t'}$, and that $-1\le |X_t|-|X_{t'}|\le 1$. 

    Let $t$ be an $S$-bottom node. 
    Assume that $t$ has exactly one child $t_1$ in $\mathcal{T}$.
    Since $\mathcal{T}$ is nice it follows from \cref{def:nicetd} that $t$ is either a introduce node or a forget node, and hence $X_t\triangle X_{t_1}=\{v\}$ for some $v\in V$.\footnote{We use $\triangle$ to denote the \emph{symmetric difference} of two sets, that is $A\triangle B=(A\setminus B)\cup (B\setminus A)$.}
    If $X_t^S\neq X_{t_1}^S$, then we are done since we must have either $X^S_t=X^S_{t_1}\cup \{v\}$ for some $v\in S\setminus X_{t_1}$ or $X^S_t=X^S_{t_1}\setminus \{v\}$ for some $v\in X^S_{t_1}=S\cap X_{t_1}$. 
    Assume that $X_t^S= X_{t_1}^S$. Since $t$ is an $S$-bottom node we have by \cref{def:bottom_node} that $(L_t^S, X_t^S, R_t^S)\neq (L_{t_1}^S, X_{t_1}^S, R_{t_1}^S)$. Furthermore, we have by \cref{obs:strace} that $L_{t_1}^S\subseteq L_t^S$. It follows that there is some $v\in L_t^S\cap R_{t_1}^S$. This is a contradiction to Condition~\ref{condition_3_tree_decomposition} of \cref{def:tree_decomposition}.
    
    Assume that $t$ has exactly two children $t_1,t_2$ in $\mathcal{T}$. 
    Since $\mathcal{T}$ is nice it follows from \cref{def:nicetd} that $t$ is a join node and $X_t=X_{t_1}=X_{t_2}$, which implies that $X_{t_1}\cup X_{t_2}\subseteq X_t$.
    Hence, if $t_1$ and $t_2$ are either both $S$-children of $t$ or neither of them is an $S$-child of $t$, then we are done. Assume w.l.o.g.\ that $t_1$ is an $S$-child of $t$ and $t_2$ is not. Then by \cref{def:sparentchild} we have that $(L_{t_1}^S \cup X_{t_1}^S)\setminus (R_t^S \cup X_t^S) \neq \emptyset$ and $(L_{t_2}^S \cup X_{t_2}^S)\setminus (R_t^S \cup X_t^S) = \emptyset$. It follows that $L_{t_1}^S \setminus R_t^S \neq \emptyset$ and $L_{t_2}^S \setminus R_t^S = \emptyset$, which in particular means that $L_{t_2}^S \subseteq R_t^S$. We can conclude that $L_{t_2}^S =\emptyset$, otherwise we get a contradiction to Condition~\ref{condition_3_tree_decomposition} of \cref{def:tree_decomposition}. However, then we have that $R_{t_1}^S=R_t^S\cup L_{t_2}^S=R_t^S$ and hence $(L_{t}^S, X_{t}^S, R_{t}^S)=(L_{t_1}^S, X_{t_1}^S, R_{t_1}^S)$. By \cref{def:bottom_node} this is a contradiction to $t$ being an $S$-bottom node.
\end{proof}

Note that, however, not every $S$-nice tree decomposition is a nice tree decomposition, since, for example, we allow $S$-bottom nodes to have an arbitrary amount of children that are not $S$-children. Furthermore, the $S$-children of an $S$-bottom nodes that admits $\join$ as its $S$-operation do not have to have the same bags as their parent.

Next, we show that if we know the $S$-trace of an $S$-bottom node in an $S$-nice tree decomposition, and we know which $S$-operation is admitted by that node, then we can determine the $S$-traces of the $S$-children of that node. We start with the $S$-operation $\introduce(v)$.

\begin{observation}\label{obs:introduce}
    Let $\mathcal{T}=(T,\{X_t\}_{t\in V(T)})$ be an $S$-nice tree decomposition of a graph $G=(V,E)$ and some $S\subseteq V$. Let $t$ be an $S$-bottom node in $\mathcal{T}$ that admits the $S$-operation $\introduce(v)$. Let $t'$ be the $S$-child of $t$.
    Then we have that $L_{t'}^S= L_{t}^S$, $X_{t'}^S= X_{t}^S\setminus \{v\}$, and $R_{t'}^S=R_{t}^S\cup \{v\}$.
\end{observation}
\begin{proof}
    By \cref{def:snicetd} we have that $X_{t'}^S= X_{t}^S\setminus \{v\}$. By \cref{obs:strace} we have that $L_{t'}^S\subseteq L_{t}^S$. Assume that $u\in L_{t}^S$ such that $u\notin L_{t'}^S$. Then we must have that $u\in R_{t'}^S$. It follows that $t$ has another child $t''$ with $S$-trace $(L_{t''}^S, X_{t''}^S, R_{t''}^S)$ such that $u\in L_{t''}^S \cup X_{t''}^S$. However, since $u\in L_{t}^S$ we have that $u\notin R_{t}^S\cup X_{t}^S$ and hence $u\in (L_{t''}^S \cup X_{t''}^S)\setminus (R_{t}^S\cup X_{t}^S)$. By \cref{def:sparentchild}, node $t''$ is an $S$-child of $t$. However, by \cref{def:snicetd} node $t$ has exactly one $S$-child, a contradiction. We can conclude that $L_{t'}^S=L_{t}^S$ and hence $R_{t'}^S=R_{t}^S\cup \{v\}$. 
\end{proof}
Next, we consider the $S$-operation $\forget(v)$.

\begin{observation}\label{obs:forget}
    Let $\mathcal{T}=(T,\{X_t\}_{t\in V(T)})$ be an $S$-nice tree decomposition of a graph $G=(V,E)$ and some $S\subseteq V$. Let $t$ be an $S$-bottom node in $\mathcal{T}$ that admits the $S$-operation $\forget(v)$. Let $t'$ be the $S$-child of $t$.
    Then we have that $L_{t'}^S= L_{t}^S\setminus \{v\}$, $X_{t'}^S= X_{t}^S\cup \{v\}$, and $R_{t'}^S=R_{t}^S$. 
\end{observation}
\begin{proof}
    By \cref{def:snicetd} we have that $X_{t'}^S= X_{t}^S\cup \{v\}$. It follows that $v\notin L_{t'}^S$. By \cref{obs:strace} we have that $L_{t'}^S\subseteq L_{t}^S\setminus\{v\}$.
    Assume that $u\in L_{t}^S\setminus\{v\}$ such that $u\notin L_{t'}^S$. Then we must have that $u\in R_{t'}^S$. It follows that $t$ has another child $t''$ with $S$-trace $(L_{t''}^S, X_{t''}^S, R_{t''}^S)$ such that $u\in L_{t''}^S \cup X_{t''}^S$. However, since $u\in L_{t}^S$ we have that $u\notin R_{t}^S\cup X_{t}^S$ and hence $u\in (L_{t''}^S \cup X_{t''}^S)\setminus (R_{t}^S\cup X_{t}^S)$. By \cref{def:sparentchild}, node $t''$ is an $S$-child of $t$. However, by \cref{def:snicetd} node $t$ has exactly one $S$-child, a contradiction. We can conclude that $L_{t'}^S=L_{t}^S\setminus \{v\}$ and hence $R_{t'}^S=R_{t}^S$. 
\end{proof}
Finally, we consider the case that the admitted $S$-operation is $\join(X^S_t,X_{t_1}^S, X_{t_2}^S,L_{t_1}^S, L_{t_2}^S)$. 

\begin{observation}\label{obs:join}
    Let $\mathcal{T}=(T,\{X_t\}_{t\in V(T)})$ be an $S$-nice tree decomposition of a graph $G=(V,E)$ and some $S\subseteq V$. Let $t$ be an $S$-bottom node in $\mathcal{T}$ that admits the $S$-operation $\join(X^S_t,X_{t_1}^S, X_{t_2}^S,L_{t_1}^S, L_{t_2}^S)$. Let $t_1$ and $t_2$ be the two $S$-children of $t$.
    Then we have the following:
    \begin{itemize}
        \item $L_{t_1}^S\cup L_{t_2}^S= L_{t}^S$, $L_{t_1}^S\cap L_{t_2}^S= \emptyset$, $L_{t_1}^S\neq\emptyset$, and $L_{t_2}^S\neq\emptyset$, and
        \item $R_{t_1}^S=R_{t}^S\cup (X_t^S \setminus X_{t_1}^S)\cup L_{t_2}^S$ and $R_{t_2}^S=R_{t}^S\cup (X_t^S \setminus X_{t_2}^S)\cup L_{t_1}^S$.
    \end{itemize}
\end{observation}
\begin{proof}
We first prove that the first property holds.
By \cref{obs:strace} we have $L_{t_1}^S\cup L_{t_2}^S\subseteq L_{t}^S$. Assume that $v\in L_{t}^S$ such that $v\notin L_{t_1}^S\cup L_{t_2}^S$. 
By \cref{def:snicetd} we have that $X^S_t = X^S_{t_1}\cup X^S_{t_2}$.
Hence, we must have that $v\in R_{t_1}^S$ and $v\in R_{t_2}^S$. It follows that $t$ has another child $t_3$ with $S$-trace $(L_{t_3}^S, X_{t_3}^S, R_{t_3}^S)$ such that $v\in L_{t_3}^S \cup X_{t_3}^S$. However, since $v\in L_{t}^S$ we have that $v\notin R_{t}^S\cup X_{t}^S$ and hence $v\in (L_{t_3}^S \cup X_{t_3}^S)\setminus (R_{t}^S\cup X_{t}^S)$. By \cref{def:sparentchild}, node $t_3$ is an $S$-child of $t$. A contradiction to $t$ having exactly two $S$-children.

Assume that $v\in L_{t_1}^S\cap L_{t_2}^S$. This is a contradiction to Condition~\ref{condition_3_tree_decomposition} of \cref{def:tree_decomposition}.
Note that we cannot have $L_{t_1}^S=\emptyset$ and $L_{t_2}^S=\emptyset$, since then by what we showed above we have $L_{t}^S=\emptyset$, which by \cref{def:bottom_node} is a contradiction to $t$ being an $S$-bottom node. Assume w.l.o.g.\ that $L_{t_1}^S=\emptyset$. Then we must have that $L_{t_2}^S=L_{t}^S$. Since $X_{t_2}^S=X_{t}^S$ we must have that $R_{t_2}^S=R_{t}^S$ and hence $(L_{t}^S, X_{t}^S, R_{t}^S)=(L_{t_1}^S, X_{t_1}^S, R_{t_1}^S)$. By \cref{def:bottom_node} this is a contradiction to $t$ being an $S$-bottom node.

We can conclude that the first property holds. It follows that the second property must also hold.
\end{proof}

From \cref{obs:introduce,obs:forget,obs:join} we can deduce the following simple corollary. This is the main reason that enables us to perform dynamic programming on the $S$-traces.

\begin{corollary}\label{cor:recursion}
    Let $\mathcal{T}=(T,\{X_t\}_{t\in V(T)})$ be an $S$-nice tree decomposition of a graph $G=(V,E)$ and some $S\subseteq V$. Let $t$ be an $S$-bottom node in $\mathcal{T}$. If $t$ has an $S$-child $t'$, then one of the following holds.
    \begin{enumerate}
        \item $L_{t'}^S\subset L_t^S$, or
        \item $L_{t'}^S = L_t^S$ and $X_{t'}^S\subset X_t^S$.
    \end{enumerate}
\end{corollary}

Formally, \cref{cor:recursion} allows us to define the following partial ordering for $S$-traces.
\begin{definition}[Preceding $S$-Traces]\label{def:predecingstraces}
Let $\mathcal{T}=(T,\{X_t\}_{t\in V(T)})$ be an $S$-nice tree decomposition of a graph $G=(V,E)$ and some $S\subseteq V$. Let $(L^S,X^S,R^S)$ and $(\hat{L}^S,\hat{X}^S,\hat{R}^S)$ be two $S$-traces. We say that $(L^S,X^S,R^S)$ is a \emph{direct predecessor} of $(\hat{L}^S,\hat{X}^S,\hat{R}^S)$ in $\mathcal{T}$ if there is an $S$-bottom node in $\mathcal{T}$ with $S$-trace $(\hat{L}^S,\hat{X}^S,\hat{R}^S)$ that has an $S$-child with $S$-trace $(L^S,X^S,R^S)$. The \emph{predecessor} relation on $S$-traces in an $S$-nice tree decomposition $\mathcal{T}$ is defined as the transitive closure of the direct predecessor relation of $\mathcal{T}$.
\end{definition}

\subsection{Modifications for Tree Decompositions}
\label{sec:mod}

In several proofs, we will use several operations to modify $S$-nice tree decompositions. We introduce those modifications in this section. 


First, we introduce two basic modifications. The first one, when applied to any rooted tree decomposition, ensures that afterwards, for any two nodes $t,t'$ such that $t'$ is the parent of $t$ we have that $X_t\subseteq X_{t'}$ or $X_t\supseteq X_{t'}$, and $-1\le |X_t|-|X_{t'}|\le 1$. 
\begin{modification}[$\Normalize$]\label{def:normalize}
    Let $\mathcal{T}=(T,\{X_t\}_{t\in V(T)})$ be a rooted tree decomposition of a graph $G=(V,E)$. The modification $\Normalize$ applies the following changes to~$\mathcal{T}$.
    
    Repeat the following. Let $t,t'$ be two nodes in $\mathcal{T}$ such that $t'$ is the parent of $t$, such that 
    at least one of the following operations applies. Perform the first applicable operation.
    \begin{itemize}
        \item If $X_t\setminus X_{t'}\neq \emptyset$ and $X_{t'}\setminus X_{t}\neq \emptyset$, then subdivide the edge between $t$ and $t'$ in $T$ and let $\hat{t}$ be the new node. Set $X_{\hat{t}}=X_t\cap X_{t'}$ be its bag.
        \item If $|X_t|<|X_{t'}|-1$, then apply $|X_{t'}|-|X_t|-1$ subdivisions to the edge between $t$ and $t'$ in~$T$. The bags of the new nodes are defined as follows. 
        
        Iterate over the new nodes starting at the child of $t'$. Set the current node's bag as the bag of its parent and then remove an arbitrary vertex from $X_{t'}\setminus X_t$ that is contained in the parent's bag. Continue with the bag of the child.  
        \item If $|X_{t'}|<|X_{t}|-1$, then apply $|X_{t}|-|X_{t'}|-1$ subdivisions to the edge between $t$ and $t'$ in~$T$. The bags of the new nodes are defined as follows. 
        
        Iterate over the new nodes starting at the parent of $t$. Set the current node's bag as the bag of its child and then remove an arbitrary vertex from $X_{t}\setminus X_{t'}$ that is contained in the child's bag. Continue with the bag of the parent.
    \end{itemize}
\end{modification}

The modification $\Normalize$ is visualized in \cref{fig:normalize}. It is easy to observe that for every rooted tree decomposition $\mathcal{T}$, we have that after applying $\Normalize$ to it, it remains a rooted tree decomposition and gains the above-claimed properties. In fact, this modification (for $S=\emptyset$) is part of transforming any rooted tree decomposition into a nice tree decomposition~\cite{bodlaender1996efficient,Die16}. 

\begin{figure}[t]
\centering
\begin{tikzpicture}[line width=1pt,scale=.5,yscale=.8]
\phantom{
\node[vert,minimum width=1cm,fill=cyan!30!white] (p) at (0,3.5) {};
\node[vert,minimum width=1cm,fill=cyan!30!white] (m) at (0,0) {};
\node[vert,minimum width=1cm,fill=orange!30!white] (c) at (0,-3.5) {};}
\node[vert,minimum width=1cm,fill=cyan!30!white] (p) at (0,2) {};
\node[vert,minimum width=1cm,fill=orange!30!white] (c) at (0,-2) {};

\node (a) at (1.5,2) {$t'$};
\node (a) at (1.5,-2) {$t$};

\draw[edge] (p) -- (c);
\end{tikzpicture}
\begin{tikzpicture}[line width=1pt,scale=.5,yscale=.8]
\phantom{
\node[vert,minimum width=1cm,fill=cyan!30!white] (p) at (0,3.5) {};
\node[vert,minimum width=1cm,fill=cyan!30!white] (m) at (0,0) {};
\node[vert,minimum width=1cm,fill=orange!30!white] (c) at (0,-3.5) {};
\node (a) at (1.5,2) {$t'$};
}

\node (a) at (0,0) {\huge$\rightarrow$};

\end{tikzpicture}
\begin{tikzpicture}[line width=1pt,scale=.5,yscale=.9]

\node[vert,minimum width=1cm,fill=cyan!30!white] (p) at (0,3.5) {};
\node[vert,minimum width=.8cm,fill=gray!30!white] (m) at (0,0) {};
\node[vert,minimum width=1cm,fill=orange!30!white] (c) at (0,-3.5) {};

\node (a) at (1.5,3.5) {$t'$};
\node (a) at (1.3,0) {$\hat{t}$};
\node (a) at (1.5,-3.5) {$t$};

\draw[edge,dashed] (p) -- (m);
\draw[edge,dashed] (m) -- (c);

\end{tikzpicture}
    \caption{Illustration of the modification $\Normalize$ (\cref{def:normalize}). Nodes $t$ (blue) and $t'$ (orange) have bags $X_t, X_{t'}$ respectively with $X_t\setminus X_{t'}\neq \emptyset$ and $X_{t'}\setminus X_{t}\neq \emptyset$. A new node $\hat{t}$ (gray) is inserted between~$t$ and $t'$ with bag $X_{\hat{t}}=X_t\cap X_{t'}$. The dashed edges represent paths resulting from edge subdivisions where, informally speaking, vertices are inserted or removed one by one from the bags of the nodes along the path to ensure that the sizes of the bags of adjacent nodes differ by at most one.}\label{fig:normalize}
\end{figure}

The second basic modification, intuitively, removes unnecessary full nodes from an $S$-nice tree decomposition.

\begin{modification}[$\MergeFullNodes$]\label{def:mergefullnodes}
Let $\mathcal{T}=(T,\{X_t\}_{t\in V(T)})$ be an $S$-nice tree decomposition of a graph $G=(V,E)$ and some $S\subseteq V$ of width $k$. The modification $\MergeFullNodes$ applies the following changes to~$\mathcal{T}$.

Repeat the following. Let $t,t'$ be two nodes in $\mathcal{T}$, such that~$t'$ is the parent of $t$, $X_t=X_{t'}$ and $|X_t|=|X_{t'}|=k+1$, that is, both bags are the same and full. If $t$ has at most one $S$-child or $t$ is the only $S$-child of $t'$, then remove $t$ and connect all children of $t$ to $t'$.
\end{modification}

The modification $\MergeFullNodes$ is visualized in \cref{fig:mergefullnodes}. We can observe the following. Let $t,t'$ be two nodes in $\mathcal{T}$ such that $t'$ is the parent of~$t$, and~$t$ has at most one $S$-child. 
If $X_t=X_{t'}$ and $|X_t|=|X_{t'}|=k+1$, then in particular $X^S_t=X^S_{t'}$.
If $L^S_t=\emptyset$, then \cref{def:sparentchild} we have that $t$ does not have any $S$-children. It follows that merging $t$ and $t'$ does not add any new $S$-children to $t'$. Otherwise, we have that $L^S_{t}\setminus R^S_{t}\neq\emptyset$ and $R^S_{t'}\subseteq R^S_{t}$. If follows that $L^S_t\setminus R^S_{t'}\neq\emptyset$, and hence by \cref{def:sparentchild} we have that $t$ is an $S$-child of $t'$. Again, it follows that merging $t$ and $t'$ does not add any new $S$-children to $t'$. 
We can conclude that $\MergeFullNodes$ preserves the property of tree decompositions of being $S$-nice. Clearly, it also does not increase the width.

\begin{figure}[t]
\centering
\begin{subfigure}[t]{0.46\textwidth}
\centering
\begin{tikzpicture}[line width=1pt,scale=.5,yscale=.8]
\node[vert,minimum width=1cm,fill=red!30!white,line width=2pt] (p) at (0,2) {};
\node[vert,minimum width=1cm,fill=red!30!white,line width=2pt] (c) at (0,-2) {};

\node[vert,minimum width=1cm,fill=gray!30!white] (c2) at (3,0) {};
\node[vert,minimum width=1cm,fill=gray!30!white] (c3) at (0,-6) {};

\node (a) at (1.5,2) {$t'$};
\node (a) at (1.5,-2) {$t$};

\draw[edge,dashed] (p) -- (0,4);
\draw[edge] (p) -- (c);
\draw[edge,dashed] (p) -- (-1.5,1);
\draw[edge] (p) -- (c2);
\draw[edge] (c) -- (c3);
\draw[edge,dashed] (c) -- (-1.5,-3);
\draw[edge,dashed] (c) -- (1.5,-3);
\end{tikzpicture}
\caption{Node $t$ has at most one $S$-child.}\label{fig:mergea}
\end{subfigure}
\begin{subfigure}[t]{0.46\textwidth}
\centering
\begin{tikzpicture}[line width=1pt,scale=.5,yscale=.8]
\phantom{
\node[vert,minimum width=1cm,fill=gray!30!white] (c3) at (0,-6) {};}

\node[vert,minimum width=1cm,fill=red!30!white,line width=2pt] (p) at (0,2) {};
\node[vert,minimum width=1cm,fill=red!30!white,line width=2pt] (t) at (0,-2) {};
\node[vert,minimum width=1cm,fill=gray!30!white] (c1) at (3,-5) {};
\node[vert,minimum width=1cm,fill=gray!30!white] (c2) at (-3,-5) {};

\draw[edge,dashed] (p) -- (0,4);
\draw[edge,dashed] (p) -- (-1.5,1);
\draw[edge,dashed] (p) -- (1.5,1);
\draw[edge] (p) -- (t);
\draw[edge] (t) -- (c1);
\draw[edge] (t) -- (c2);
\draw[edge,dashed] (t) -- (0,-4);

\node (a) at (1.5,2) {$t'$};
\node (a) at (1.5,-2) {$t$};
\end{tikzpicture}
\caption{Node $t$ is the only $S$-child of $t'$.}\label{fig:mergeb}
\end{subfigure}
    \caption{Illustration of the modification $\MergeFullNodes$ (\cref{def:mergefullnodes}).
    Nodes $t$ and $t'$ (red circles) have the same bags and the bag is full. This is indicated by the thick line. Nodes represented by gray circles are $S$-children of their respective parents. Further children that connected via dashed lines are not $S$-children. In both cases visualized in \cref{fig:mergea,fig:mergeb}, nodes $t$ and $t'$ are merged.}\label{fig:mergefullnodes}
\end{figure}

Now we introduce two further modifications $\MoveIntoSubtree$ and $\RemoveFromSubtree$. In contrast to the modifications introduced so far, $\MoveIntoSubtree$ and $\RemoveFromSubtree$ have some prerequisites and are not always applicable. The modification $\MoveIntoSubtree$, intuitively, moves a $M$ set of vertices (with $M\cap S=\emptyset$) into some subtree of the tree decomposition, if all neighbors of $M$ are already in that subtree. The modification $\RemoveFromSubtree$, intuitively, makes the opposite modification. It moves a set $M$ of vertices out of a subtree if all neighbors of $M$ are also out of the subtree. Formally, the modifications are defined as follows.

\begin{modification}[$\MoveIntoSubtree$]\label{def:move}
      Let $\mathcal{T}=(T,\{X_t\}_{t\in V(T)})$ be an $S$-nice tree decomposition of a graph $G=(V,E)$ and some $S\subseteq V$. Let $t$ be a node in $\mathcal{T}$ and let $M\subseteq V\setminus S$ such that $N(M)\subseteq V_{t}$. The modification $\MoveIntoSubtree(t, M)$ applies the following changes to~$\mathcal{T}$.
     \begin{enumerate}
         \item Remove the vertices of $M$ from each bag of a node that is not in $T_t$. 
        
         \item Let $M'=M\setminus V_t$. If $M'\neq \emptyset$, then let $G'$ be the graph obtained from $G[N[M']]$ after adding edges between every two non-adjacent vertices in $N(M')$. Moreover, let $\mathcal{T}'$ be some optimal tree decomposition for $G'$. By \cref{lem:cliquebag}, there exists a node $r$ in~$\mathcal{T}'$ such that $N(M')\subseteq X_{r}$. Root $\mathcal{T}'$ at $r$ and make $r$ a child of $t$ in $\mathcal{T}$.
         \item Apply $\Normalize$ and $\MergeFullNodes$ to $\mathcal{T}$ (\cref{def:normalize,def:mergefullnodes}).
     \end{enumerate}
 \end{modification}
 
 \begin{modification}[$\RemoveFromSubtree$]\label{def:remove}
      Let $\mathcal{T}=(T,\{X_t\}_{t\in V(T)})$ be an $S$-nice tree decomposition of a graph $G=(V,E)$ and some $S\subseteq V$.
     Let $t$ be a node in $\mathcal{T}$ and let $M\subseteq V\setminus S$ such that $N(M)\subseteq (V\setminus V_{t})\cup X_{t'}$, where~$t'$ is the parent of $t$. 
     The modification $\RemoveFromSubtree(t, M)$ applies the following changes to~$\mathcal{T}$.
         \begin{enumerate}
         \item Remove the vertices of $M$ from each bag of a node in $T_t$. 
        
         \item Let $M'=M\cap (V_t\setminus X_{t'})$. If $M'\neq \emptyset$, then let $G'$ be the graph obtained from $G[N[M']]$ after adding edges between every two non-adjacent vertices in $N(M')$. Moreover, let $\mathcal{T}'$ be some optimal tree decomposition for $G'$. By \cref{lem:cliquebag}, there exists a node $r$ in~$\mathcal{T}'$ such that $N(M')\subseteq X_{r}$. Root $\mathcal{T}'$ at $r$ and make $r$ a child of $t'$ in~$\mathcal{T}$.
         \item Apply $\Normalize$ and $\MergeFullNodes$ to $\mathcal{T}$ (\cref{def:normalize,def:mergefullnodes}).
     \end{enumerate}
 \end{modification}

We can observe that both above-defined operations, when applied to an $S$-nice tree decomposition, produce another $S$-nice tree decomposition that does have at most the same width as the original one. 
Furthermore, it will be crucial that the modification do not change the ``$S$-trace structure'' of the $S$-nice tree decomposition. To formalize this, we define \emph{sibling} $S$-nice tree decompositions, which have the same predecessor relation (\cref{def:preceding}) on their $S$-traces.
\begin{definition}[Sibling Tree Decompositions]\label{def:sibling}
   Let $\mathcal{T}=(T,\{X_t\}_{t\in V(T)})$ and $\mathcal{T}'=(T',\{X'_t\}_{t\in V(T')})$ be two $S$-nice tree decomposition of a graph $G=(V,E)$ and some $S\subseteq V$. We say that $\mathcal{T}$ and $\mathcal{T}'$ are \emph{siblings} if both tree decompositions have the same predecessor relation on $S$-traces.
\end{definition}

Formally, we now show the following.

 \begin{observation} \label{lem:moveremove}
     Let $\mathcal{T}=(T,\{X_t\}_{t\in V(T)})$ be an $S$-nice tree decomposition of a graph $G=(V,E)$ and some $S\subseteq V$ with width $k$. Let $t$ be a node in $\mathcal{T}$ with parent $t'$, let $M\subseteq V\setminus S$ such that $N(M)\subseteq V_{t}$, and let $M'\subseteq V\setminus S$ such that $N(M')\subseteq V\setminus V_{t}$. The following holds.
     \begin{itemize}
     \item Let $\mathcal{T}_{\text{mod}}$ be the result of applying $\MoveIntoSubtree(t, M)$ to $\mathcal{T}$. Then $\mathcal{T}_{\text{mod}}$ is an $S$-nice tree decomposition for $G$ with width at most $k$ that is a sibling of $\mathcal{T}$. 
     \item Let $\mathcal{T}_{\text{mod}}$ be the result of applying $\RemoveFromSubtree(t, M')$ to $\mathcal{T}$. Then $\mathcal{T}_{\text{mod}}$ is an $S$-nice tree decomposition for $G$ with width at most $k$ that is a sibling of $\mathcal{T}$.
     \end{itemize}    
 \end{observation}
 \begin{proof}
We prove the observation for the case that $\MoveIntoSubtree$ is applied. The case of $\RemoveFromSubtree$ can be proven in an analogous way. 
Let $\mathcal{T}_{\text{mod}}$ denote the result of applying $\MoveIntoSubtree(t, M)$ to $\mathcal{T}$.
First, note that the union of all bags of $\mathcal{T}_{\text{mod}}$ is $V$. Second, we have that for every edge $e\in E$, there is a bag in $\mathcal{T}_{\text{mod}}$ that contains the edge. If $e\cap M=\emptyset$ this is obvious, otherwise, if $e$ is an edge between a vertex of $M$ and a vertex in $N(M)$, then there must be a node in $T_t$ whose bag contains $e$. If $e$ is not contained in any such bag, then we must have that $e\subseteq M\setminus V_t$. In this case, there is a bag in the tree decomposition~$\mathcal{T}'$ (referring to the notation in \cref{def:move}) that contains $e$. Note that due to Condition~\ref{condition_3_tree_decomposition} of \cref{def:tree_decomposition} we have that all neighbors of $M\setminus V_{t}$, that is $N(M\setminus V_{t})$ must be contained in $X_t$. It follows that after $\mathcal{T'}$ is connected to~$t$, the result is a tree decomposition. It is $S$-nice and a sibling of $\mathcal{T}$, since no vertices from $S$ were moves, since the application of $\Normalize$ guarantees that the first condition of \cref{def:snicetd} holds, and since $\MergeFullNodes$ preserves $S$-niceness. It remains to show that the width of the tree decomposition does not increase. We argued above that $N(M\setminus V_{t})\subseteq X_t$. By \cref{lem:cliquebag2} we have that $\tw(G')\le \tw(G)$. It follows that none of the bags of nodes in $\mathcal{T}'$ is larger than~$k+1$. We can conclude that $\mathcal{T}_{\text{mod}}$ is an $S$-nice tree decomposition for $G$ with width at most~$k$ that is a sibling of $\mathcal{T}$.
 \end{proof}

Finally, we introduce two modifications $\BringNeighborUp$ and $\BringNeighborDown$ that, intuitively, move a neighbor of a vertex $v$ into a specific node in order to ensure that we can safely remove $v$ from a certain part of the tree decomposition. 

\begin{modification}[$\BringNeighborUp$]\label{def:bringup}
Let $\mathcal{T}=(T,\{X_t\}_{t\in V(T)})$ be an $S$-nice tree decomposition of a graph $G=(V,E)$ and some $S\subseteq V$ with width $k$. Let $t$ be a node in $\mathcal{T}$ such that $|X_t|\le k$. 
Let $v\in X_t\setminus S$ such that $N(v)\cap (V_t\setminus X_t)=\{u\}$ for some $u\in V\setminus S$. The modification $\BringNeighborUp(t, v, u)$ applies the following changes to~$\mathcal{T}$. 
\begin{enumerate}
	\item Replace $v$ with $u$ in all bags of nodes that are in $T_t$ except $X_t$ (since bags are sets, if a bag contains both $v$ and $u$ this implies that $v$ is removed from the bag).
	\item Add $u$ to $X_{t}$.
         \item Apply $\Normalize$ to $\mathcal{T}$ (\cref{def:normalize}).
\end{enumerate}
\end{modification}

\begin{modification}[$\BringNeighborDown$]\label{def:bringdown}
Let $\mathcal{T}=(T,\{X_t\}_{t\in V(T)})$ be an $S$-nice tree decomposition of a graph $G=(V,E)$ and some $S\subseteq V$ with width $k$. Let $t$ be a node in $\mathcal{T}$ such that $|X_t|\le k$. 
Let $v\in X_t\setminus S$ such that $N(v)\cap (V\setminus V_t)=\{u\}$ for some $u\in V\setminus S$. The modification $\BringNeighborDown(t, v, u)$ applies the following changes to~$\mathcal{T}$.
\begin{enumerate}
	\item Replace $v$ with $u$ in all bags of nodes that are not in $T_t$ (since bags are sets, if a bag contains both $v$ and $u$ this implies that $v$ is removed from the bag).
	\item Add $u$ to $X_t$.
         \item Apply $\Normalize$ to $\mathcal{T}$ (\cref{def:normalize}).
\end{enumerate}
\end{modification}

\begin{figure}[t]
\centering
\begin{tikzpicture}[line width=1pt,scale=.5,yscale=.9]

\node[vert,minimum width=1.2cm,fill=red!30!white] (p) at (0,3) {};
\node[vert,minimum width=1.2cm] (c1) at (3,0) {};
\node[vert,minimum width=1.2cm] (c11) at (1,-3) {};
\node[vert,minimum width=1.2cm] (c12) at (5,-3) {};
\node[vert,minimum width=1.2cm] (c2) at (-3,0) {};
\node[vert,minimum width=1.2cm] (c21) at (-3,-3) {};

\node (a) at (1.6,3) {$t$};

\phantom{
\node[vert,minimum width=1.2cm] (c22) at (-5,-3) {};
}

\draw[edge,dashed] (p) -- (0,5);
\draw[edge] (p) -- (c1);
\draw[edge] (p) -- (c2);
\draw[edge] (c1) -- (c11);
\draw[edge] (c1) -- (c12);
\draw[edge,dashed] (c11) -- (1,-5);
\draw[edge,dashed] (c12) -- (5,-5);
\draw[edge] (c2) -- (c21);
\draw[edge,dashed] (c21) -- (-3,-5);

\node[vert,fill=cyan!20!white] (v) at (-2.5,0) {};
\node[vert,fill=cyan!20!white] (v) at (-2.5,-3) {};
\node[vert,fill=cyan!20!white] (v) at (.5,3) {};
\node[vert,fill=cyan!20!white] (v) at (3.5,0) {};
\node[vert,fill=cyan!20!white] (v) at (5.5,-3) {};

\node[vert2,fill=green!30!white] (u1) at (2.5,0) {};
\node[vert2,fill=green!30!white] (u1) at (.5,-3) {};
\end{tikzpicture}
\begin{tikzpicture}[line width=1pt,scale=.5,yscale=.9]

\node (a) at (-6.6,0) {\huge$\rightarrow$};

\node[vert,minimum width=1.2cm,fill=red!30!white] (p) at (0,3) {};
\node[vert,minimum width=1.2cm] (c1) at (3,0) {};
\node[vert,minimum width=1.2cm] (c11) at (1,-3) {};
\node[vert,minimum width=1.2cm] (c12) at (5,-3) {};
\node[vert,minimum width=1.2cm] (c2) at (-3,0) {};
\node[vert,minimum width=1.2cm] (c21) at (-3,-3) {};

\node (a) at (1.6,3) {$t$};

\phantom{
\node[vert,minimum width=1.2cm] (c22) at (-5,-3) {};
}

\draw[edge,dashed] (p) -- (0,5);
\draw[edge] (p) -- (c1);
\draw[edge] (p) -- (c2);
\draw[edge] (c1) -- (c11);
\draw[edge] (c1) -- (c12);
\draw[edge,dashed] (c11) -- (1,-5);
\draw[edge,dashed] (c12) -- (5,-5);
\draw[edge] (c2) -- (c21);
\draw[edge,dashed] (c21) -- (-3,-5);

\node[vert,fill=cyan!20!white] (v) at (.5,3) {};

\node[vert2,fill=green!30!white] (u1) at (-3.5,0) {};
\node[vert2,fill=green!30!white] (u1) at (-3.5,-3) {};
\node[vert2,fill=green!30!white] (u1) at (-.5,3) {};
\node[vert2,fill=green!30!white] (u1) at (2.5,0) {};
\node[vert2,fill=green!30!white] (u1) at (.5,-3) {};
\node[vert2,fill=green!30!white] (u1) at (4.5,-3) {};

\end{tikzpicture}
    \caption{Illustration of the modification $\BringNeighborUp$ (\cref{def:bringup}). Nodes $t$ is represented by a red circle. Before the modification is applied, the bag $t$ is not full. The small blue circle visualizes vertex $v\in X_t\setminus S$ and the green square vertex $u\in N(v)\setminus S$. The left side shows the configuration before the modification is applied, and the right side shows the configuration after the modification is applied.}\label{fig:bringneighbor}
\end{figure}
The modification $\BringNeighborUp$ is visualized in \cref{fig:bringneighbor}. We can observe that both above-defined operations, when applied to an $S$-nice tree decomposition, produce another $S$-nice tree decomposition that does have at most the same width as and is a sibling of the original one. 

 \begin{observation} \label{lem:bringupdown}
     Let $\mathcal{T}=(T,\{X_t\}_{t\in V(T)})$ be an $S$-nice tree decomposition of a graph $G=(V,E)$ and some $S\subseteq V$ with width $k$. Let $t$ be a node in $\mathcal{T}$ such that $|X_t|\le k$ and $t$ is not an $S$-bottom node. Let $v\in X_t\setminus S$ such that $N(v)\cap (V_t\setminus X_t)=\{u\}$ for some $u\in V\setminus S$.
     Let $v'\in X_t\setminus S$ such that $N(v)\cap (V\setminus V_t)=\{u'\}$ for some $u'\in V\setminus S$. The following holds.
     \begin{itemize}
     \item Let $\mathcal{T}_{\text{mod}}$ be the result of applying $\BringNeighborUp(t, v, u)$ to $\mathcal{T}$. Then $\mathcal{T}_{\text{mod}}$ is an $S$-nice tree decomposition for $G$ with width at most $k$ that is a sibling of $\mathcal{T}$.
     \item Let $\mathcal{T}_{\text{mod}}$ be the result of applying $\BringNeighborDown(t, v', u')$ to $\mathcal{T}$. Then $\mathcal{T}_{\text{mod}}$ is an $S$-nice tree decomposition for $G$ with width at most $k$ that is a sibling of $\mathcal{T}$.
     \end{itemize}
 \end{observation}
 \begin{proof}
 We prove the observation for the case that $\BringNeighborUp$ is applied. The case of $\BringNeighborDown$ can be proven in an analogous way.
 Let $\mathcal{T}_{\text{mod}}$ denote the result of applying $\BringNeighborUp(t, v, u)$ to $\mathcal{T}$.
First, note that the union of all bags in $\mathcal{T}_{\text{mod}}$ is $V$. Second, we have that for every edge $e\in E$, there is a bag in $\mathcal{T}_{\text{mod}}$ that contains the edge, since $N(v)\cap (V\setminus V_{t})=\{u\}$ and $u\in X_{t}$. Furthermore, the bags containing $u$ form a subtree in $\mathcal{T}_{\text{mod}}$, and the bags containing~$v$ form a subtree in $\mathcal{T}_{\text{mod}}$.
Note that the only bag to which we have added vertices is~$X_{t}$, to which we add $u$. 
This means that the size of $X_{t}$ is now at most $k+1$.
It is $S$-nice and a sibling of $\mathcal{T}$, since no vertices from $S$ were moves, and since the application of $\Normalize$ guarantees that the first condition of \cref{def:snicetd} holds. We can conclude that $\mathcal{T}_{\text{mod}}$ is an $S$-nice tree decomposition for $G$ with width at most $k$ that is a sibling of $\mathcal{T}$.  
 \end{proof}

\subsection{Slim and Top-Heavy \boldmath$S$-Nice Tree Decompositions}
\label{sec:useful}

In this section, we introduce two special types of $S$-nice tree decomposition that are central to our algorithm. They are called \emph{\slim} and \emph{\topheavy}.  We first present the concept of \slim $S$-nice tree decompositions. Intuitively, they do not contain nodes that unnecessarily have full bags. To this end, we first define \emph{full join trees} of an $S$-nice tree decomposition, which, intuitively, are subtrees of the tree decomposition where all nodes have the same (full) bag and admit a $\bigjoin$ operation.

\begin{definition}[Full Join Tree]\label{def:fulljointree}
Let $\mathcal{T}=(T,\{X_t\}_{t\in V(T)})$ be an $S$-nice tree decomposition of a graph $G=(V,E)$ and some $S\subseteq V$ with width $k$. 
A \emph{full join tree} $T'$ is a non-trivial, (inclusion-wise) maximal subtree of $T$ such that there exist a set $X\subseteq V$ with $|X|=k+1$ such each $t\in (T')$ has bag $X_t=X$.
\end{definition}
Furthermore, we introduce the following terminology. 
Let $C$ be a connected component in $G[V']$ for some $V'\subseteq V$. 
\begin{itemize}
    \item If $V(C)\cap S=\emptyset$, then we call $C$ an \emph{$F$-component}. 
    \item If $V(C)\cap S\neq\emptyset$, then we call $C$ an \emph{$S$-component}.
\end{itemize}
It is easy to see that every connected component of $G[V']$ falls into one of the above categories, for every choice of $V'\subseteq V$. 

Now we are ready to define \emph{\slim} $S$-nice tree decompositions.

\begin{definition}[Slim $S$-Nice Tree Decomposition]\label{def:slimtd}
Let $\mathcal{T}=(T,\{X_t\}_{t\in V(T)})$ be an $S$-nice tree decomposition of a graph $G=(V,E)$ and some $S\subseteq V$ with width $k$. We call $\mathcal{T}$ \emph{\slim} if the following holds.
\begin{enumerate}
\item The root of $\mathcal{T}$ is not an $S$-bottom node.\label{cond:slim:1}
\item For each $S$-bottom node $t$ in $\mathcal{T}$ that has two $S$-children, we have that either $t$ is part of a full join tree, or both $S$-children and the parent of $t$ have bags that are not full and not larger than the bag of $t$.\label{cond:slim:2}
\item For each $S$-bottom node $t$ in $\mathcal{T}$ that has a bag which is not full, we have that for the parent $t'$ of $t$ it holds that $X_t=X_{t'}$ and $t$ is the only $S$-child of $t'$.\label{cond:slim:3}
\item For each $S$-bottom node $t$ in $\mathcal{T}$ that admits $S$-operation $\forget(v)$ for some $v\in S$, we have that if the bag of the $S$-child $t'$ of $t$ is not full, then $t'$ has at most one $S$-child $t''$. If $t''$ exists, then it holds that $X_{t'}=X_{t''}$.\label{cond:slim:35}
\item For each $S$-bottom node $t$ in $\mathcal{T}$ we have that if the bag of the parent $t'$ of $t$ is not full, then the parent of $t'$ is not an $S$-bottom node or $t'$ is the root.\label{cond:slim:4}
\end{enumerate}
Moreover, for each full join tree $T'$ of $\mathcal{T}$, the following holds.
\begin{enumerate}\setcounter{enumi}{5}
\item Each node $t$ in $V(T')$ is an $S$-bottom node that has two $S$-children.\label{cond:slim:5}
\item Let $X$ denote the bag of nodes in $V(T')$, let $t_r$ denote the root of $T'$, let~$t'$ denote the parent of $t_r$ (if it exists), and let $T'_C$ denote the set of $S$-children of nodes in $V(T')$ that are not contained in $V(T')$. For each $v\in X\setminus S$ there exist three different vertices $u_1,u_2,u_3\in N(v)\setminus X$ and three different nodes $t_1,t_2,t_3\in T'_C\cup\{t'\}$ such that \label{cond:slim:6}
\begin{itemize}
\item For all $1\le i\le 3$, vertex $u_i$ is connected to $S$ in $G-X$.\label{cond:slim:6a}

\item For all $1\le i\le 3$, if $t_i\neq t'$, then vertex $u_i$ is contained in $V_{t_1}$. Otherwise vertex $u_i$ is contained in $V\setminus V_{t_r}$.\label{cond:slim:6b}
\end{itemize}
\end{enumerate}
\end{definition}
Before we show that we can make any $S$-nice tree decomposition \slim, we give some intuition on why the properties of \slim $S$-nice tree decompositions are desirable for us.
\begin{enumerate}
\item Condition~\ref{cond:slim:1} allows us to assume that all $S$-bottom nodes have a parent. This is important since, informally speaking, we want to remove all vertices that do not need to be in nodes that are below some $S$-bottom node, and move them above it. We formalize this later in this section when we define \topheavy tree decompositions.
\item Condition~\ref{cond:slim:2} allows us to assume that $S$-bottom nodes with two $S$-children, that is, the ones that admit $S$-operation $\join$, are either part of a full join tree, or both the $S$-children and the parent do not have full bags and those bags. This is important since we will treat these two cases differently in our algorithm. In particular in the latter case, that is, if an $S$-bottom node admits the $S$-operation $\join$ and it is not part of a full join tree, we need the property that both the $S$-children and the parent do not have full bags and that the bags are subsets of the bag of the $S$-bottom node.
\item Condition~\ref{cond:slim:3} allows us to assume that whenever an $S$-bottom node does not have a full bag, then the bag of its parent is also not full, and in particular, it is the same bag. This will make some proofs easier to achieve and easy to obtain by subdividing the edge between the $S$-bottom node and its parent, and giving the new node the same bag as the $S$-bottom node.
\item Condition~\ref{cond:slim:35} allows us to assume that whenever an $S$-bottom node admits the $S$-operation $\forget$ and its $S$-child (which always has a larger bag) it not full, then the bag of the $S$-child of the $S$-child (if it exists) is also not full. As with the previous condition, this is easy to achieve by edge subdivision and it simplifies some of our proofs.
\item Condition~\ref{cond:slim:4} allows us to assume that in a directed path corresponding to an $S$-trace $(L^S,X^S,R^S)$ with $L^S\neq \emptyset$, if the parent of the bottom node has a non-full bag, then the top node is not the parent of the bottom node. This will simplify many steps in our algorithm.
\item Condition~\ref{cond:slim:5} allows us to assume that every $S$-bottom node that is part of a full join tree admits the $S$-operation $\join$.
\item Condition~\ref{cond:slim:6} is the most technical one and also the most important one. Essentially it allows us to assume that every vertex $v\in V\setminus S$ that is contained in a bag $X$ of a node that is part of a full join tree is a neighbor of at least three different $S$-components in $G-X$. This will be crucial for establishing the running time bound of a subroutine of our algorithm.
\end{enumerate}
Now we show that if a graph $G$ admits an $S$-nice tree decomposition with width $k$, then~$G$ also admits a \slim $S$-nice tree decomposition with width $k$. 

\begin{lemma}\label{lem:slim}
Let $\mathcal{T}=(T,\{X_t\}_{t\in V(T)})$ be an $S$-nice tree decomposition of a graph $G=(V,E)$ with width $k$ and some $S\subseteq V$. Then there exist a \slim $S$-nice tree decomposition~$\mathcal{T}'$ of $G$ with width at most $k$.
\end{lemma}

To prove \cref{lem:slim}, we give the following two procedures and prove that they transform an $S$-nice tree decomposition $\mathcal{T}$ for a graph $G=(V,E)$ and some $S\subseteq V$ with width $k$, into a \slim $S$-nice tree decomposition for $G$ with width $k$. 

\begin{modification}[$\MakeSlim$]\label{def:MakeSlim}

Let $\mathcal{T}=(T,\{X_t\}_{t\in V(T)})$ be an $S$-nice tree decomposition of a graph $G=(V,E)$ and some $S\subseteq V$ with width $k$. The modification $\MakeSlim$ applies the following changes to~$\mathcal{T}$.
\begin{enumerate}
\item 
Apply the modification $\MergeFullNodes$ to~$\mathcal{T}$ (\cref{def:mergefullnodes}). \label{slim:1}
\item If the root $t_r$ of $\mathcal{T}$ is an $S$-bottom node, then add a new node $t'$ to $\mathcal{T}$ and make it the parent of $t_r$. Set $X_{t'}=\emptyset$. Apply $\Normalize$. If the new root $t_r'$ is still an $S$-bottom node, add a new node $t''$ with $X_{t''}=\emptyset$ to $\mathcal{T}$ and make it the parent of $t_r'$. \label{slim:2}
\item For each $S$-bottom node $t$ in $\mathcal{T}$ that has two $S$-children, if $t$ is not contained in any full join tree and the parent $t_p$ of $t$ has a full bag, then subdivide the edge between $t$ and $t_p$ and let $t'$ be the new node. Set $X_{t'}$ = $X_t$. \label{slim:3}
\item 
For each full join tree $T'$ in $\mathcal{T}$ perform the following steps once.
Let $X$ denote the bag of nodes in $V(T')$, let $t_r$ denote the root of $T'$, let~$t'$ denote the parent of $t_r$ (if it exists), and let $T'_C$ denote the set of $S$-children of nodes in $V(T')$ that are not contained in $V(T')$.\label{slim:6}
\begin{enumerate}
\item For each $F$-component $C$ in $G-X$ such that $V(C)\cap V_{t^\star}\neq\emptyset$ for some $t^\star\in T'_C$, or, if $t'$ exists, $V(C)\cap ((V\setminus V_{t'})\cup X_{t'})\neq\emptyset$, add a new child $t_0$ to $t_r$ with bag $X$ and apply $\MoveIntoSubtree(t_0,V(C))$ (\cref{def:move}).
Remove $t_0$ and connect all children of $t_0$ to $t_r$.
\label{slim:6a}
\item Assume there exist $v\in X\setminus S$ such that $N(v)\subseteq X$. Then remove $v$ from the bags of all nodes except $t_r$.\label{slim:6aa}
\item Assume there exist $v\in X\setminus S$, $u_1,u_2\in N(v)\setminus X$, and nodes $t_1,t_2\in T_C'\cup\{t'\}$ with $t_1\neq t_2$ and $t'\neq t_2$ such that the following holds.
For all $t''\in (T'_C\cup\{t'\})\setminus\{t_1,t_2\}$ if $t''\neq t'$, then vertex $N(v)\cap (V_{t''}\setminus X)=\emptyset$. Otherwise, $N(v)\cap  (V\setminus V_{t_r})=\emptyset$. 
Then make the following modification. 
%

We replace $T'$ with a rooted path $P$ (that is, the path is rooted at one of its endpoints) on $|T'_C|-1$ nodes as follows. Remove $T'$ from~$\mathcal{T}$. Let $t'_r$ be the root of~$P$. Make $t'$ (if it exists) the parent of $t'_r$, otherwise~$t'_r$ is the root of $\mathcal{T}$.  
\begin{itemize}
\item If $t_1=t'$, then $t_2$ is a child of $t'_r$. Each child of a node in $T'$ that is not an $S$-child is a child of $t'_r$. Each node in $T'_C\setminus\{t_2\}$ is a child of exactly one inner node of $P$ and two nodes are children of the leaf of~$P$. Set the bags of all nodes in $P$ to $X$. 
Denote $t'_0$ the child of $t'_r$ in $P$.
Apply $\RemoveFromSubtree(t'_0,\{v\})$ (\cref{def:remove}). For an illustration see \cref{fig:fulljointree}.

\item Otherwise, $t_1$ and $t_2$ are children of the leaf of $P$. Each child of a node in $T'$ that is not an $S$-child is a child of the leaf of $P$. Each node in $T'_C\setminus\{t_1,t_2\}$ is a child of exactly of the remaining nodes of $P$. Set the bags of all nodes in $P$ to~$X$. 
Denote~$t'_\ell$ the leaf of $P$. 
Apply $\MoveIntoSubtree(t'_\ell,\{v\})$.
\end{itemize}
\label{slim:6b}
\end{enumerate}%
\end{enumerate}
\end{modification}


\begin{figure}[t]
\centering
\begin{tikzpicture}[line width=1pt,scale=.5,yscale=.9]
\node[vert,minimum width=1cm,fill=red!30!white] (p) at (0,4) {};
\node[vert,minimum width=1cm,fill=gray!30!white,line width=2pt] (v1) at (0,0) {};
\node[vert,minimum width=1cm,fill=gray!30!white] (v2) at (4,-1) {};
\node[vert,minimum width=1cm,fill=gray!30!white] (v3) at (8,-2) {};
\node[vert,minimum width=1cm,fill=gray!30!white] (v5) at (16,-4) {};

\node[vert,minimum width=1cm,fill=red!30!white] (w1) at (-3,-3) {};
\node[vert,minimum width=1cm] (w2) at (1,-4) {};
\node[vert,minimum width=1cm] (w3) at (5,-5) {};
\node[vert,minimum width=1cm] (w5) at (13,-7) {};
\node[vert,minimum width=1cm] (w6) at (19,-7) {};

\node (a) at (2.2,4) {$t_1=t'$};
\node (a) at (1.5,.2) {$t'_r$};
\node (a) at (5.5,-.8) {$t'_0$};
\node (a) at (-1.5,-3) {$t_2$};

\node[vert,fill=cyan!20!white] (v) at (0,4) {};
\node[vert,fill=cyan!20!white] (v) at (0,0) {};
\node[vert,fill=cyan!20!white] (v) at (-3,-3) {};

\draw[edge,dashed] (v1) -- (0,-2);
\draw[edge,dashed] (v1) -- (-1,-1.7);
\draw[edge,dashed] (v1) -- (1,-1.7);

\draw[edge] (p) -- (v1);
\draw[edge] (v1) -- (v2);
\draw[edge] (v2) -- (v3);
\draw[edge,dashed] (v3) -- (v5);
\draw[edge] (v1) -- (w1);
\draw[edge] (v2) -- (w2);
\draw[edge] (v3) -- (w3);
\draw[edge] (v5) -- (w5);
\draw[edge] (v5) -- (w6);
\end{tikzpicture}
    \caption{Illustration of Step~\ref{slim:6} in the procedure ``Transformation into Slim $S$-Nice Tree Decomposition''. The gray circles represent nodes of the path $P$ with which the full join tree $T'$ is replaced. The red circles illustrate $t_1$ and $t_2$. The remaining circles illustrate nodes in $T'_C\setminus\{t_2\}$. The dashed lines below $t'_r$ indicate that all non-$S$-children of nodes in $T'$ are attached to this node. The case where $t_1=t'$ is shown. Vertex $v$ is visualized with a small blue circle. It is removed from all bags of nodes of $P$ except $t'_r$. The bag of $t'_r$ remains full, this is visualized by the thick line.}\label{fig:fulljointree}
\end{figure}

Note that since Step~\ref{slim:5} of $\MakeSlim$ can be applied at most once to each $S$-bottom node and Step~\ref{slim:6} of $\MakeSlim$ is applied to each full join tree once, 
we have that the modification always terminates after a finite number of steps. Informally speaking, $\MakeSlim$ takes care of all conditions of \cref{def:slimtd} except Conditions~\ref{cond:slim:3}, \ref{cond:slim:35}, and~\ref{cond:slim:4}, and partially Condition~\ref{cond:slim:2}. Concerning Condition~\ref{cond:slim:2}: $\MakeSlim$ does not guarantee that if an $S$-bottom node $t$ as two $S$-children and is not part of a full join tree, then the bags of the parent and the children are not larger than the bag of $t$. 
For those, we introduce a second modification $\MakeSlimTwo$. The reason for this is that we can use $\MakeSlimTwo$ after applying modifications $\MoveIntoSubtree$, $\RemoveFromSubtree$, $\BringNeighborUp$, or $\BringNeighborDown$ (\cref{def:move,def:remove,def:bringup,def:bringdown}) to a \slim $S$-nice treedecompositoin to ensure that the tree decomposition remains \slim.

\begin{modification}[$\MakeSlimTwo$]\label{def:MakeSlimTwo}
Let $\mathcal{T}=(T,\{X_t\}_{t\in V(T)})$ be an $S$-nice tree decomposition of a graph $G=(V,E)$ and some $S\subseteq V$ with width $k$. The modification $\MakeSlimTwo$ applies the following changes to~$\mathcal{T}$.
\begin{enumerate}
\item For each $S$-bottom node $t$ in $\mathcal{T}$ that has a bag which is not full, if for the parent $t_p$ of $t$ it holds that $X_t\neq X_{t_p}$, then subdivide the edge between $t_p$ and $t$ and let $t'$ be the new node. Set $X_{t'}=X_{t}$.
If $t$ has two $S$-children $t_1$ and $t_2$ and for some $t_i$ with $i\in\{1,2\}$ it holds that $|X_{t_i}|>|X_{t}|$, then subdivide the edge between $t_1$ and $t$ and let $t''$ be the new node. Set $X_{t''}=X_{t}$.
 \label{slim:4}
 \item For each $S$-bottom node $t$ in $\mathcal{T}$ that admits $S$-operation $\forget(v)$ for some $v\in S$, if the bag of the $S$-child $t'$ of $t$ is not full and $t'$ has an $S$-child $t''$ with $X_{t'}\neq X_{t''}$, then subdivide the edge between $t'$ and $t''$ and let $t'''$ be the new node. Set $X_{t'''}=X_{t'}$.
 \label{slim:45}
\item For each $S$-bottom node $t$ in $\mathcal{T}$, if the bag of the parent $t_p$ of $t$ is not full and $t_p$ is a child of an $S$-bottom node, then subdivide the edge between $t_p$ and its parent and let $t'$ be the new node. Set $X_{t'}=X_{t_p}$. \label{slim:5}
\end{enumerate}
\end{modification}

Now we are ready to prove \cref{lem:slim}.
\begin{proof}[Proof of \cref{lem:slim}]
Let $\mathcal{T}=(T,\{X_t\}_{t\in V(T)})$ be an $S$-nice tree decomposition of a graph $G=(V,E)$ with width $k$ and some $S\subseteq V$. 
Apply modification $\MakeSlim$ and $\MakeSlimTwo$ to $\mathcal{T}$ (\cref{def:MakeSlim} and \cref{def:MakeSlimTwo}). We show that afterwards, $\mathcal{T}$ is a \slim $S$-nice tree decomposition of $G$ with width at most $k$.

In Step~\ref{slim:1} of $\MakeSlim$, the modification $\MergeFullNodes$ (\cref{def:mergefullnodes}) is applied to~$\mathcal{T}$. 
Note that afterwards, $\mathcal{T}$ is still an $S$-nice tree decomposition of a graph $G=(V,E)$ with width $k$.
After Steps~\ref{slim:2} and~\ref{slim:3} of $\MakeSlim$, $\mathcal{T}$ is clearly still an $S$-nice tree decomposition of a graph $G=(V,E)$ with width $k$.

We claim that now the Conditions~\ref{cond:slim:1},~\ref{cond:slim:5}, and partially Condition~\ref{cond:slim:2} of \cref{def:slimtd} hold.
Condition~\ref{cond:slim:1} of \cref{def:slimtd} clearly holds after Step~\ref{slim:2}.
Now we show that the Condition~\ref{cond:slim:2} of \cref{def:slimtd} partially holds. Let $t$ be an $S$-bottom node that has two $S$-children and that is not contained in any full join tree. Assume that one of the $S$-children has a full bag. By \cref{def:snicetd} we have that the bag of $t$ is also full and the two bags are the same. Then either that $S$-child is merged with $t$ when the modification $\MergeFullNodes$ is applied, or $t$ is in a full join tree. Hence, assume that both $S$-children of $t$ do not have full bags. Assume that the parent of $t$ has a full bag. If the bag of $t$ is also full, then by \cref{def:snicetd} we have the bags are the same. Then, either the parent is merged with $t$ when the modification $\MergeFullNodes$ is applied, or $t$ is in a full join tree. Assume that the bag of $t$ is not full. Then, because of Step~\ref{slim:3} of $\MakeSlim$, we have that the bag of the parent of $t$ is the same as the bag of~$t$, a contradiction to the assumption that the bag of the parent of $t$ is full. It follows that part Condition~\ref{cond:slim:2} of \cref{def:slimtd} holds. What remains to show is that the bags of the $S$-children and parent of $t$ are also not larger than the bag of $t$. This will be guaranteed after the application of $\MakeSlimTwo$.

We show that Condition~\ref{cond:slim:5} of \cref{def:slimtd} holds. Let $T'$ be a full join tree in $\mathcal{T}$. Let $t\in V(T')$. Assume for contradiction that $t$ does not have two $S$-children. Assume that $t$ has a parent node $t_p$ with $t_p\in V(T')$. Then, $\MergeFullNodes$ merges $t$ and $t_p$, a contradiction.
If $t$ does not have a parent node or the parent of $t$ is not in $V(T')$, then $t$ must have a child node $t_c$ such that $t_c\in V(T')$, because otherwise $T'$ is trivial. If $t_c$ has at most one $S$-child, then $\MergeFullNodes$ merges $t$ and~$t_c$. Hence, assume that $t_c$ has two $S$-children. 
Clearly, we have that $X_t^S=X_{t_c}^S$, $L^S_{t_c}\neq\emptyset$, $L^S_{t_c}\setminus R^S_{t_c}\neq\emptyset$, and $R^S_{t}\subseteq R^S_{t_c}$. If follows that $L^S_{t_c}\setminus R^S_{t}\neq\emptyset$, and hence by \cref{def:sparentchild} we have that $t_c$ is the only $S$-child of $t$. Then, $\MergeFullNodes$ merges~$t$ and $t_c$, a contradiction. It follows that Condition~\ref{cond:slim:5} of \cref{def:slimtd} holds.

After Step~\ref{slim:6} of $\MakeSlim$, we claim that $\mathcal{T}$ is still an $S$-nice tree decomposition of a graph $G=(V,E)$ with width $k$.
To see this, note that Step~\ref{slim:6a} first adds a new child $t_0$ to the root $t_r$ of $T'$ that has the same bag. It follows that $t_0$ is not an $S$-child of the root of $T'$. Then we perform $\MoveIntoSubtree(t_0,C)$ (\cref{def:move}). By \cref{lem:moveremove} we know that afterwards, $\mathcal{T}$ is still an $S$-nice tree decomposition of a graph $G=(V,E)$ with width $k$. Note that since the $C$ is an $F$-component, $t_0$ is still a non-$S$-child of the root of $T'$. 
We can conclude that $\mathcal{T}$ is still an $S$-nice tree decomposition of a graph $G=(V,E)$ with width $k$. Furthermore, since $\MoveIntoSubtree(t_0,C)$ does not increase the size of any bag, we have that Conditions~\ref{cond:slim:1} and~\ref{cond:slim:2} of \cref{def:slimtd} still hold.

In Step\ref{slim:6aa} the whole neighborhood of $v$ is in $X$, this means that we can safely remove $v$ from all bags except one and $\mathcal{T}$ remains a tree decomposition. Since $v\notin S$, we also have that $\mathcal{T}$ remains $S$-nice. Furthermore, since we do not increase the size of any bag, we have that Conditions~\ref{cond:slim:1} and~\ref{cond:slim:2} of \cref{def:slimtd} still hold.


Now consider the case where $t'=t_1$. Step~\ref{slim:6b} first connects non-$S$-children of nodes in $T'$ to $t'_r$. 
The $S$-children of the nodes in $T'$ are moved to new parents such that every parent has two $S$-children. Let $t'_0$ be the child of $t'_r$ in $P$. Now observe that $N(v)\subseteq (V\setminus V_{t'_0})\cup X$ since $v$ has no neighbors in any of the $S$-children of nodes in $P$ except $t_2$, which is not in the subtree rooted in $t'_0$. It follows that the prerequisites for applying $\RemoveFromSubtree$ (\cref{def:remove}) are given. By \cref{lem:moveremove} we know that afterwards, $\mathcal{T}$ is still an $S$-nice tree decomposition of a graph $G=(V,E)$ with width $k$. 
In the case where $t'\neq t_1$ is analogous.
Furthermore, since no bag size is increased, we have that Conditions~\ref{cond:slim:1} and~\ref{cond:slim:2} of \cref{def:slimtd} still hold.

Now assume that Condition~\ref{cond:slim:6} of \cref{def:slimtd} does not hold. Then there is some $v\in X\setminus S$ such that for all $u_1,u_2,u_3\in N(V)\setminus X$ the following holds. There are no three different nodes $t_1,t_2,t_3\in T'_C\cup\{t'\}$ such that the second condition of \cref{def:slimtd}. It is straightforward to check that then Step~\ref{slim:6aa} or Step~\ref{slim:6b} of $\MakeSlim$ applies. A contradiction to the assumption that $\MakeSlim$ was applied to $\mathcal{T}$. 


Conditions~\ref{cond:slim:2} and~\ref{cond:slim:3} of \cref{def:slimtd} clearly holds after Step~\ref{slim:4} of $\MakeSlimTwo$, Condition~\ref{cond:slim:35} of \cref{def:slimtd} clearly holds after Step~\ref{slim:45} of $\MakeSlimTwo$, and Condition~\ref{cond:slim:4} of \cref{def:slimtd} clearly holds after Step~\ref{slim:5} of $\MakeSlimTwo$.
All those steps only subdivide edges and set the bag of the new node to the bag of one of the previous endpoints. Note that all newly introduced nodes have bags that are not full, and hence, Conditions~\ref{cond:slim:5} and~\ref{cond:slim:6} of \cref{def:slimtd} are preserved. It is easy to see that $S$-niceness is preserved as well and the width of the tree decomposition is not increased.
This finishes the proof.
\end{proof}

Furthermore, we can observe that all modifications introduced in \cref{sec:mod} preserve the property of tree decomposition to be slim.
 \begin{observation} \label{lem:preserveslim}
     Let $\mathcal{T}=(T,\{X_t\}_{t\in V(T)})$ be a \slim $S$-nice tree decomposition of a graph $G=(V,E)$ and some $S\subseteq V$ with width $k$. Let $t$ be a node in $\mathcal{T}$ with parent $t_p$, let $M\subseteq V\setminus S$ such that $N(M)\subseteq V_{t}$, and let $M'\subseteq V\setminus S$ such that $N(M')\subseteq (V\setminus V_{t})\cup X_{t_p}$. 
Let $t'$ be a node in $\mathcal{T}$ such that $|X_{t'}|\le k$ and $t'$ is not parent of an $S$-bottom node. Let $v\in X_{t'}\setminus S$ such that $N(v)\cap (V_t\setminus X_{t'})=\{u\}$ for some $u\in V\setminus S$.
     Let $v'\in X_{t'}\setminus S$ such that $N(v')\cap (V\setminus V_{t'})=\{u'\}$ for some $u'\in V\setminus S$.          
     The following holds.
     \begin{itemize}
     \item Let $\mathcal{T}_{\text{mod}}$ be the result of applying $\MoveIntoSubtree(t, M)$ and then $\MakeSlimTwo$ to $\mathcal{T}$. Then $\mathcal{T}_{\text{mod}}$ is a \slim $S$-nice tree decomposition for $G$ with width at most $k$ that is a sibling of~$\mathcal{T}$.
     \item Let $\mathcal{T}_{\text{mod}}$ be the result of applying $\RemoveFromSubtree(t, M')$ and then $\MakeSlimTwo$ to~$\mathcal{T}$. Then $\mathcal{T}_{\text{mod}}$ is a \slim $S$-nice tree decomposition for $G$ with width at most $k$ that is a sibling of~$\mathcal{T}$.
     \item Let $\mathcal{T}_{\text{mod}}$ be the result of applying $\BringNeighborUp(t', v, u)$ and then $\MakeSlimTwo$ to~$\mathcal{T}$. Then $\mathcal{T}_{\text{mod}}$ is a \slim $S$-nice tree decomposition for $G$ with width at most $k$ that is a sibling of~$\mathcal{T}$.
     \item Let $\mathcal{T}_{\text{mod}}$ be the result of applying $\BringNeighborDown(t', v', u')$ and then $\MakeSlimTwo$ to~$\mathcal{T}$. Then $\mathcal{T}_{\text{mod}}$ is a \slim $S$-nice tree decomposition for $G$ with width at most $k$ that is a sibling of~$\mathcal{T}$.
     \end{itemize}
 \end{observation}
\begin{proof}
Note that from \cref{lem:moveremove,lem:bringupdown} it follows that all cases, before the application of $\MakeSlimTwo$, the tree decomposition remains $S$-nice, it remains a sibling of $\mathcal{T}$, and its width remains at most $k$. 
Furthermore, observe that all steps in $\MakeSlimTwo$ only subdivide edges and set the bag of the new node to the bag of one of the previous endpoints. It is easy to see that this preserved $S$-niceness, does not change the predecessor relation on $S$-traces, and does not increase the width of the tree decomposition.

We can observe that the tree decomposition also remains \slim: Condition~\ref{cond:slim:1} of \cref{def:slimtd} clearly still holds.

The modification $\MoveIntoSubtree(t,M)$ (\cref{def:move}) ($\RemoveFromSubtree(t,M')$ (\cref{def:remove})) does not add any new vertices to existing bags, and hence cannot create any new full join trees with existing nodes. Note that the bags of the nodes of the new subtree attached to $t$ (the parent of $t$) does not contain any vertices from $S$. If follows that all full nodes are merged by $\MergeFullNodes$ and no new full join trees are created.
It follows that Conditions~\ref{cond:slim:2},~\ref{cond:slim:5} and~\ref{cond:slim:6} of \cref{def:slimtd} still hold.

The modification $\BringNeighborUp(t', v, u)$ (\cref{def:bringup}) only increases the bag of node $t'$ and it may become full. Note that $u$ is not contained in the parent of $t'$ and $v$ is not contained in any descendant of~$t'$, hence $t'$ is not part of any full join tree and, after the $\Normalize$ step, the bag of the parent as well as the bags of the children of $t'$ are not full.
It follows that Conditions~\ref{cond:slim:2},~\ref{cond:slim:5} and~\ref{cond:slim:6} of \cref{def:slimtd} still hold.

Similarly, the modification $\BringNeighborDown(t', v', u')$ (\cref{def:bringdown}) only increases the bag of node $t'$ and it may become full. Note that $u'$ is not contained in any of the children of $t'$ and $v'$ is not contained in any ancestor of $t'$, hence $t'$ is not part of any full join tree and, after the $\Normalize$ step, the bag of the parent as well as the bags of the children of $t'$ are not full.
It follows that Conditions~\ref{cond:slim:2},~\ref{cond:slim:5} and~\ref{cond:slim:6} of \cref{def:slimtd} still hold.

Finally, because of the application of $\MakeSlimTwo$ we have that Conditions~\ref{cond:slim:3},~\ref{cond:slim:35}, and~\ref{cond:slim:4} of \cref{def:slimtd} still hold. 
\end{proof}

Now we formally define the second special type of $S$-nice tree decompositions. They are called \emph{\topheavy} (for some node $t$) and, intuitively, have the property that if the modifications introduced in \cref{sec:mod} are applied to the tree decomposition outside of the subtree rooted at $t$, then the bags in the subtree rooted at $t$ are unaffected. 

\begin{definition}[Top-Heavy]\label{def:topheavy}
Let $\mathcal{T}=(T,\{X_t\}_{t\in V(T)})$ be a \slim $S$-nice tree decomposition of a graph $G=(V,E)$ with width $k$ and some $S\subseteq V$. Let $t$ be a parent of an $S$-bottom node in $\mathcal{T}$ with $|X_t|\le k$.
We call $\mathcal{T}$ \emph{\topheavy} for $t$ if the following holds.
\begin{enumerate}
\item For each $S$-bottom node $t'$ in $\mathcal{T}$ that is a descendant of $t$ and for each $F$-component $C$ in $G-X_{t'}$ we have that $V(C)\cap V_{t''}=\emptyset$ for each $S$-child $t''$ of $t'$.
\end{enumerate}
Moreover, the following holds with respect to \cref{def:move,def:remove,def:bringup,def:bringdown}.  
Let~${t^\star}$ be a parent of an $S$-bottom node and contained in $T_t$.
Let $t'$ be a node in $\mathcal{T}$ that is not contained in $T_{t^\star}$.
\begin{enumerate}\setcounter{enumi}{1}
\item If there is $M\subseteq V\setminus S$ such that $N(M)\subseteq V_{t'}$, then applying $\MoveIntoSubtree(t', M)$ to $\mathcal{T}$ does not change $T_{t^\star}$ or any bag of a node in $T_{t^\star}$.
\item If there is $M\subseteq V\setminus S$ such that $N(M)\subseteq (V\setminus V_{t'})\cup X_{t'_p}$, where $t'_p$ is the parent of $t'$, then applying $\RemoveFromSubtree(t', M)$ to $\mathcal{T}$ does not change $T_{t^\star}$ or any bag of a node in $T_{t^\star}$.
\item If $|X_{t'}|\le k$ and there is $v\in X_{t'}\setminus S$ such that $N(v)\cap (V_{t'}\setminus X_{t'})=\{u\}$ for some $u\in V\setminus S$, then applying $\BringNeighborUp(t', v, u)$ to $\mathcal{T}$ does not change $T_{t^\star}$ or any bag of a node in~$T_{t^\star}$.
\item If $|X_{t'}|\le k$ and there is $v\in X_{t'}\setminus S$ such that $N(v)\cap (V\setminus V_{t'})=\{u\}$ for some $u\in V\setminus S$, then applying $\BringNeighborDown(t', v, u)$ to $\mathcal{T}$ does not change $T_{t^\star}$ or any bag of a node in~$T_{t^\star}$.
\end{enumerate}
\end{definition}

We show that if a graph $G$ admits a \slim $S$-nice tree decomposition with width $k$, then we can transform it into a sibling \slim $S$-nice tree decomposition for $G$ with width $k$ that is \topheavy for a node $t$. Formally, we prove the following

\begin{lemma}\label{lem:topheavy}
Let $\mathcal{T}=(T,\{X_t\}_{t\in V(T)})$ be a \slim $S$-nice tree decomposition of a graph $G=(V,E)$ with width $k$ and some $S\subseteq V$. Then for every node $t$ in $\mathcal{T}$ with $|X_t|\le k$ that is a parent of an $S$-bottom node there exist a \slim $S$-nice tree decomposition~$\mathcal{T}'$ of $G$ with width at most $k$ that is \topheavy for $t$ and that is a sibling of $\mathcal{T}$.
\end{lemma}


To prove \cref{lem:topheavy}, we give the following procedure and prove that it can be used to transform a \slim $S$-nice tree decomposition $\mathcal{T}$ for a graph $G=(V,E)$ and some $S\subseteq V$ with width $k$, into a \slim \topheavy $S$-nice tree decomposition for $G$ with width $k$ that is a sibling of the original tree decomposition.


\begin{modification}[$\MakeTopHeavy$]\label{def:maketopheavy}
Let $\mathcal{T}=(T,\{X_t\}_{t\in V(T)})$ be a \slim $S$-nice tree decomposition of a graph $G=(V,E)$ and some $S\subseteq V$ with width $k$. Let $t$ be a parent of an $S$-bottom node in $\mathcal{T}$ such that 
$|X_t|\le k$. The modification $\MakeTopHeavy(t)$ applies the following changes to~$\mathcal{T}$.

Subdivide the edge between $t$ and its parent and let $t'$ be the new node.
Set $X_{t'}=X_t$. 
Repeat the first applicable step until none of the steps applies.
\begin{enumerate}
\item Apply $\MakeSlimTwo$.\label{heavy1}
\item For each $S$-bottom node $t'$ in $\mathcal{T}$ that is an $S$-child of $t$ and each $F$-component $C$ in $G-X_{t'}$, apply $\RemoveFromSubtree(t'',V(C))$ for each $S$-child $t''$ of~$t'$.\label{heavy2}
	\item If there is an $F$-component $C$ in $G-X_t$ such that $V(C)\cap V_t\neq\emptyset$, then apply $\RemoveFromSubtree(t,V(C))$ (\cref{def:remove}).
    Subdivide the edge between $t$ and $t'$ and let $t''$ be the new node. Set $X_{t''}$ = $X_t$ and rename $t''$ to $t'$.
    \label{heavy3}
	\item If there is $v\in X_t\setminus S$ with $N(v)\subseteq (V\setminus V_t)\cup X_{t'}$, then apply $\RemoveFromSubtree(t,\{v\})$ (\cref{def:remove}).
    Subdivide the edge between $t$ and $t'$ and let $t''$ be the new node. Set $X_{t''}$ = $X_t$ and rename $t''$ to $t'$.
    \label{heavy4}
	\item If there is $v\in X_{t'}\setminus S$ and $u\in V\setminus S$, such that $N(v)\cap (V_{t'}\setminus X_{t'})=\{u\}$, then apply $\BringNeighborUp({t'}, v, u)$ (\cref{def:bringup}). 
	Subdivide the edge between $t$ and $t'$ and let $t''$ be the new node. Set $X_{t''}$ = $X_t$ and rename $t''$ to $t'$. \label{heavy5}
\end{enumerate}
\end{modification}
We remark that the application of $\MakeSlimTwo$ in $\MakeTopHeavy(t)$ is only necessary to ensure that after applying $\MakeTopHeavy(t)$ to a \slim $S$-nice tree decomposition, the tree decomposition remains \slim. 
Note that since each application of the a step of $\MakeTopHeavy(t)$ removes a vertex from the subtree rooted at $t$, 
we have that the modification always terminates after a finite number of steps.
Now we are ready to prove \cref{lem:topheavy}.
\begin{proof}[Proof of \cref{lem:topheavy}]
Let $t$ be a node in $\mathcal{T}$ with $|X_t|\le k$ that is a parent of an $S$-bottom node.
We apply $\MakeTopHeavy(\hat{t})$ from the leaves up to $t$ to every node $\hat{t}$ in $T_t$ such that $|X_{\hat{t}}|\le k$ and $\hat{t}$ a parent of an $S$-bottom node. 
Note that the edge subdivisions in $\MakeTopHeavy$ guarantee that the bag of the parent of $\hat{t}$ is never full.
To begin, note that from \cref{lem:moveremove,lem:bringupdown} it follows that after any $\MakeTopHeavy$ modification, the tree decomposition remains $S$-nice and a sibling of the original tree decomposition. Since we apply $\MakeSlimTwo$, we have by \cref{lem:preserveslim} that the tree decomposition also remains \slim. 
Since Step~\ref{heavy2} is exhaustively applied, we have that the first condition of \cref{def:topheavy} holds.

Now we prove the following claim.
\begin{claim}
If in an application of $\MakeTopHeavy(\hat{t})$ 
we have that Steps~\ref{heavy3} and~\ref{heavy4} are not applicable, then the following holds.
\begin{itemize}
\item Let $t'$ a node in $\mathcal{T}$ that is not contained in $T_{\hat{t}}$ and let $M\subseteq V\setminus S$ such that $N(M)\subseteq V_{t'}$. Then applying $\MoveIntoSubtree(t', M)$ does not change $T_{\hat{t}}$ or any bag of a node in $T_{\hat{t}}$.
\item Let $t'$ a node in $\mathcal{T}$ that is not contained in $T_{\hat{t}}$ and let $M\subseteq V\setminus S$ such that $N(M)\subseteq (V\setminus V_{t'})\cup X_{t'_p}$, where $t'_p$ is the parent of $t'$. Then applying $\RemoveFromSubtree(t', M)$ does not change $T_{\hat{t}}$ or any bag of a node in $T_{\hat{t}}$.
\end{itemize}
\end{claim}
\begin{claimproof}
To show the first statement of the claim, assume that there is a node $t'$ a node in $\mathcal{T}$ that is not contained in $T_{\hat{t}}$ and let $M\subseteq V\setminus S$ such that $N(M)\subseteq V_{t'}$.
Assume that there is a node $t'$ that is an ancestor of $\hat{t}$.
According to its definition (\cref{def:move}), the modification $\MoveIntoSubtree(t', M)$ removes vertices only from nodes that are not in $T_{t'}$ and hence also not in $T_{\hat{t}}$. Furthermore, it adds a new child to $t'$ including a new subtree rooted at that child. It follows that $T_{\hat{t}}$ is not changed. Lastly, since so far, no bag of a node in $T_{\hat{t}}$ has been changed, the modifications $\Normalize$ and $\MergeFullNodes$ (\cref{def:normalize,def:mergefullnodes}) do not change $T_{\hat{t}}$ or any bag of a node in $T_{\hat{t}}$. We can conclude that applying $\MoveIntoSubtree(t', M)$ does not change any bag of a node in $T_{\hat{t}}$.

Assume that there is a node $t'$ that is not an ancestor of $\hat{t}$.
Again, adding a new child to $t'$ including a new subtree rooted at that child does not change $T_{\hat{t}}$. It remains to show that when all vertices from $M$ are removed from all bags of nodes that are not in $T_{t'}$ no bags of nodes in $T_{\hat{t}}$ are changed. Assume for contradiction that $V_{\hat{t}} \cap M\neq \emptyset$. 

Consider the case where $v\in (V_{\hat{t}} \cap M)\setminus X_{\hat{t}}$. Let $C$ denote the connected component in $G-X_{\hat{t}}$ that contains $v$. Since $X_{\hat{t}}$ separates $V_{\hat{t}}\setminus X_{\hat{t}}$ from $V\setminus V_{\hat{t}}$, we have that $V(C)\subseteq V_{\hat{t}}\setminus X_{\hat{t}}$.
We claim that $V(C)\subseteq M$. Assume for contradiction that $u\in V(C)\setminus M$. Let $P$ be the path from $v$ to $u$ in $C$. Let $\{v',u'\}$ be an edge in $P$ such that $v'\in M$ and $u'\notin M$. Note that since $v\in M$ and $u\notin M$, such an edge must exist. It follows that $u'\in N(M)$ and hence $u'\in (V\setminus V_{t'})\cup X_{t_p'}$ which implies that $u'\in V\setminus V_{\hat{t}}$. This is a contradiction to $V(C)\subseteq V_{\hat{t}}\setminus X_{\hat{t}}$. We can conclude that $V(C)\subseteq M$. Since $M\cap S=\emptyset$, we have that $C$ is an $F$-component in $G-X_t$ with $V(C)\cap V_{\hat{t}}\neq\emptyset$. This is a contradiction to the assumption that Step~\ref{heavy3} of $\MakeTopHeavy({\hat{t}})$ is not applicable. 

Now consider the case where $v\in V_{\hat{t}} \cap M$ and $V_{\hat{t}} \cap M\subseteq X_{\hat{t}}$. Note that since $M\cap S=\emptyset$, we have that $v\in X_{\hat{t}}\setminus S$. Since Step~\ref{heavy4} of $\MakeTopHeavy({\hat{t}})$ is not applicable, we have that $u\in N(v)\setminus (V\setminus V_{\hat{t}})\cup X_{\hat{t}}$. It follows that $u\in V_{\hat{t}}\setminus X_{\hat{t}}$. Furthermore, it implies that $u\notin (V\setminus V_{t'})\cup X_{t'_p}$, and hence, $u\notin N(M)$. Since $u$ and $v$ are neighbors, we must have that $u\in M$. This is a contradiction to $V_{\hat{t}} \cap M\subseteq X_{\hat{t}}$.

We can conclude that $T_{\hat{t}}$ is not changed. Lastly, since so far, no bag of a node in~$T_{\hat{t}}$ has been changed, the modifications $\Normalize$ and $\MergeFullNodes$ (\cref{def:normalize,def:mergefullnodes}) do not change $T_{\hat{t}}$ or any bag of a node in~$T_{\hat{t}}$. We can conclude that applying $\MoveIntoSubtree(t', M)$ does not change any bag of a node in $T_{\hat{t}}$.

To show the second statement of the claim, assume that there is a node $t'$ a node in $\mathcal{T}$ that is not contained in $T_{\hat{t}}$ and let $M\subseteq V\setminus S$ such that $N(M)\subseteq (V\setminus V_{t'})\cup X_{t_p'}$, where $t_p'$ is the parent of $t'$. 
Assume that there is a node $t'$ that is an ancestor of $\hat{t}$.
Consider the modification $\RemoveFromSubtree(t', M)$ (\cref{def:remove}). First, note that when the modification adds a new child to $t_p'$ including a new subtree rooted at that child, it does not change $T_{\hat{t}}$ or any bag of a node in $T_{\hat{t}}$. It remains to show that when all vertices from $M$ are removed from all bags of nodes in $T_{t'}$ no bags of nodes in $T_{\hat{t}}$ are changed. Assume for contradiction that $V_{\hat{t}} \cap M\neq \emptyset$. The argument is exactly the same as previous case where  $t'$ that is not an ancestor of $\hat{t}$.

Assume that there is a node $t'$ that is not an ancestor of $\hat{t}$.
This case is analogous to the previous case where $t'$ that is an ancestor of $\hat{t}$.

We can conclude that $T_{\hat{t}}$ is not changed. Lastly, since so far, no bag of a node in~$T_{\hat{t}}$ has been changed, the modifications $\Normalize$ and $\MergeFullNodes$ (\cref{def:normalize,def:mergefullnodes}) do not change $T_{\hat{t}}$ or any bag of a node in~$T_{\hat{t}}$. We can conclude that applying $\RemoveFromSubtree(t', M)$ does not change any bag of a node in $T_{\hat{t}}$.
\end{claimproof}

Next, we prove the following claim.
\begin{claim}
If in an application of $\MakeTopHeavy({\hat{t}})$ 
we have that Steps~\ref{heavy3},~\ref{heavy4}, and~\ref{heavy5} are not applicable, then the following holds.
\begin{itemize}
\item Let $t'$ a node in $\mathcal{T}$ that is not contained in $T_{\hat{t}}$ such that $|X_{t'}|\le k$, and let $v\in X_{t'}\setminus S$ such that $N(v)\cap (V_{t'}\setminus X_{t'})=\{u\}$ for some $u\in V\setminus S$. Then applying $\BringNeighborUp(t', v, u)$ does not change $T_{\hat{t}}$ or any bag of a node in $T_{\hat{t}}$.
\item Let $t'$ a node in $\mathcal{T}$ that is not contained in $T_{\hat{t}}$ such that $|X_{t'}|\le k$, and let $v\in X_{t'}\setminus S$ such that $N(v)\cap (V\setminus V_{t'})=\{u\}$ for some $u\in V\setminus S$. Then applying $\BringNeighborDown(t', v, u)$ does not change $T_{\hat{t}}$ or any bag of a node in $T_{\hat{t}}$.
\end{itemize}
\end{claim}
\begin{claimproof}
Assume that there is a node $t'$ a node in $\mathcal{T}$ that is not contained in $T_{\hat{t}}$ such that $|X_{t'}|\le k$, and let $v\in X_{t'}\setminus S$ such that $N(v)\cap (V_{t'}\setminus X_{t'})=\{u\}$ for some $u\in V\setminus S$.
Assume that $t'$ is an ancestor of ${\hat{t}}$.
Now assume for contradiction that there is a node $\check{t}$ in $T_{\hat{t}}$ such that the bag $X_{\check{t}}$ is changed when $\BringNeighborUp(t', v, u)$ (\cref{def:bringup}) is applied.
Note that the only node to whose bag a vertex is added is~$t'$, which is not contained in $T_{\hat{t}}$. Hence, we must have that $X_{\check{t}}$ contained $v$ before the application of $\BringNeighborUp(t', v, u)$, and now contains $u$ but does not contain $v$. Due to Condition~\ref{condition_2_tree_decomposition} of \cref{def:tree_decomposition} we have that $v$ was  contained in $X_{\hat{t}}$ before the application of $\BringNeighborUp(t', v, u)$. 
Furthermore, we have that $N(v)\cap (V_{t'}\setminus X_{t'})=\{u\}$ and hence $N(v)\setminus \{u\}\subseteq ((V\setminus V_{\hat{t}})\cup X_{\hat{t}})$. Since Step~\ref{heavy4} of $\MakeTopHeavy({\hat{t}})$ is not applicable, we have that $N(v)\setminus ((V\setminus V_{\hat{t}})\cup X_{\hat{t}})\neq \emptyset$. It follows that $N(v)\setminus ((V\setminus V_{\hat{t}})\cup X_{\hat{t}})= \{u\}$ which means that $N(v)\cap (V_{\hat{t}}\setminus X_{\hat{t}})=\{u\}$. This is a contradiction to the assumption that Step~\ref{heavy5} of $\MakeTopHeavy({\hat{t}})$ is not applicable. 
We can conclude that $T_{\hat{t}}$ is not changed. Lastly, since so far, no bag of a node in~$T_{\hat{t}}$ has been changed, the modification $\Normalize$ (\cref{def:normalize}) does not change $T_{\hat{t}}$ or any bag of a node in~$T_{\hat{t}}$. We can conclude that applying $\BringNeighborUp(t', v, u)$ does not change any bag of a node in $T_{\hat{t}}$. 
Note that $\BringNeighborUp(t', v, u)$ only changes bags of nodes in $T_{t'}$. Hence, if $t'$ is not an ancestor of ${\hat{t}}$, then no bags of nodes in $T_{\hat{t}}$ are changed.

The modification $\BringNeighborDown(t', v, u)$ behaves symmetrically to $\BringNeighborUp(t', v, u)$. Hence, by symmentrical arguments we arrive at the conclusion that applying $\BringNeighborDown(t', v, u)$ does not change any bag of a node in $T_{\hat{t}}$.
\end{claimproof}
The lemma statement now follows from a straightforward inductive argument using the two claims.
\end{proof}

The proof of \cref{lem:topheavy} also implies the following.
\begin{lemma}\label{lem:topheavynomake}
Let $\mathcal{T}=(T,\{X_t\}_{t\in V(T)})$ be a \slim $S$-nice tree decomposition of a graph $G=(V,E)$ with width $k$ and some $S\subseteq V$ that is \topheavy for some node $t$ in $\mathcal{T}$ with $|X_t|\le k$ that is parent of an $S$-bottom node. Let $t'$ be a node in $T_t$ with $|X_{t'}|\le k$ that is parent of an $S$-bottom node.
Then, when applying $\MakeTopHeavy(t')$ to $\mathcal{T}$, none of the steps of $\MakeTopHeavy(t')$ applies.
\end{lemma}

\begin{lemma}\label{lem:topheavynomake2}
Let $\mathcal{T}=(T,\{X_t\}_{t\in V(T)})$ be a \slim $S$-nice tree decomposition of a graph $G=(V,E)$ with width $k$ and some $S\subseteq V$. Let $t$ be some node in $\mathcal{T}$ with $|X_t|\le k$ that is parent of an $S$-bottom node. 
If $\mathcal{T}$ is \topheavy for each $t'$ that is a descendant of $t$ with $|X_{t'}|\le k$ and that is parent of an $S$-bottom node, and none of the steps of $\MakeTopHeavy(t)$ applies, then $\mathcal{T}$ is \topheavy for $t$.
\end{lemma}

\section{The Dynamic Programming Algorithm}\label{sec:mainalgo}
In this section, we present the main dynamic programming table.
Assume from now on that we are given a connected graph $G=(V,E)$ and a non-negative integer~$k$ that form an instance of \textsc{Treewidth}.
Let the vertices in $V$ be ordered in an arbitrary but fixed way, that is, $V=\{v_1,v_2,\ldots,v_{n}\}$.
Furthermore, we denote with $S\subseteq V$ a minimum feedback vertex set of $G$. We assume that $k\le |S|$, since we can trivially obtain a tree decomposition with width $|S|+1$ by taking any tree decomposition with width one for $G-S$ and adding $S$ to every bag. Hence, if $k>|S|$, then $(G,k)$ is a yes-instance of \textsc{Treewidth}. For the rest of the document, we treat $G=(V,E)$, $S$, and $k$ as global variables that all algorithms have access to.

Roughly speaking, the algorithm we propose to prove \cref{thm:main} is a dynamic program on the potential directed paths of the $S$-traces (\cref{def:trace}) of a \slim \topheavy $S$-nice tree decomposition (\cref{def:snicetd,def:slimtd,def:topheavy}) of~$G$. By \cref{lem:slim,lem:topheavy}, we have that if $(G,k)$ is a yes-instance, then such a tree decomposition of width at most $k$ exists. We remark that the main structure of our dynamic program table is similar to the one of Chapelle et al.~\cite{chapelle2017treewidth}. However, there are major differences in how the table entries are computed.
Informally speaking, a state of the dynamic program consists of the following elements.
\begin{itemize}
    \item An $S$-trace $(L^S, X^S, R^S)$.
    \item The ``extended'' $S$-Operation $\tau_+$ of the parent of the top node $t_{\max}$ of the directed path of the $S$-trace. 

    If $t_{\max}$ is the root, then $\tau_+$ is the \emph{dummy $S$-operation} $\void$.
    
    If $L^S=\emptyset$, then $\tau_+$ is the ``extended'' $S$-operation of an $S$-bottom node that has an $S$-child with $S$-trace $(L^S, X^S, R^S)$.
    \item The ``extended'' $S$-Operation $\tau_-$ of the bottom node $t_{\min}$ of the directed path of the $S$-trace.

   If $L^S=\emptyset$, then $\tau_-$ is the \emph{dummy $S$-operation} $\void$.
\end{itemize}
Here, we use extended versions of the $S$-operations introduced in \cref{def:snicetd}. Those will contain additional information that we need to compute the entries of the dynamic program. In the following, we give some intuition on why we need this additional information, and then we give formal definitions.

Given the above-described information, the main strategy of our algorithm is to determine (candidates for) the bags $X_{\max}$ and $X_{\min}$ of the top node $t_{\max}$ and the bottom node $t_{\min}$ of the directed path of the $S$-trace, respectively. Furthermore, we want to identify (candidates for) the bag $X_p$ of the parent of $t_{\max}$ and (candidates for) a set $X_{\dpath}$ of vertices that we want to add to bags of the directed path of the $S$-trace or subtrees of the tree decomposition that are rooted at children of nodes in the directed path (that are not $S$-children). Using the sets $X_{\max}$, $X_p$, $X_{\min}$, and $X_{\dpath}$, we compute the ``local treewidth'' of the directed path of the $S$-trace. Using the dynamic programming table, we look up the ``partial treewidth'' of the subgraphs ``below'' the $S$-children of the $S$-bottom node of the directed path. Using this information, we determine the partial treewidth of the subgraph below the top node of the directed path. For a visualization see \cref{fig:paths}.

\begin{figure}[t]
\centering
\begin{subfigure}[t]{0.64\textwidth}
\centering
\begin{tikzpicture}[line width=1pt,scale=.8,xscale=1]

\draw[trace2] (0,0) -- (0,-3);
\draw[trace] (-2,-4) -- (-2,-6);
\draw[trace] (2,-4) -- (2,-7);

\node[vertg] (v0) at (0,1) {};
\node[vertb] (v1) at (0,0) {};
\node[vertb] (v2) at (0,-1) {};
\node[vertb] (v3) at (0,-2) {};
\node[vert2b] (v4) at (0,-3) {};

\node (tmax) at (-1,0) {$t_{\max}$};
\node (taup) at (.5,1) {$\tau_+$};
\node (tmin) at (-1,-3) {$t_{\min}$};
\node (taum) at (.75,-3) {$\tau_-$};
\node (trace) at (2,0) {$(L^S, X^S, R^S)$};

\node[trial,line width=1.5pt] (t1) at (1,-1.5) {};
\draw[edgeb] (v2) -- (t1.apex);

\node[triac,line width=1.5pt] (t2) at (0,-4) {};
\draw[edgeb] (v4) -- (t2.apex);

\node[vertg] (u1) at (-2,-4) {};
\node[vertg] (u2) at (-2,-5) {};
\node[vert2g] (u3) at (-2,-6) {};

\node (trace1) at (-4,-4) {$(L_1^S, X_1^S, R_1^S)$};

\node[vertg] (w1) at (2,-4) {};
\node[vertg] (w2) at (2,-5) {};
\node[vertg] (w3) at (2,-6) {};
\node[vert2g] (w4) at (2,-7) {};

\node (trace2) at (4,-4) {$(L_2^S, X_2^S, R_2^S)$};

\draw[edgeg,dashed] (0,1.75) -- (v0);

\draw[edgeg] (v0) -- (v1);
\draw[edgeb] (v1) -- (v2);
\draw[edgeb] (v2) -- (v3);
\draw[edgeb] (v3) -- (v4);

\draw[edgeg] (v4) -- (u1);

\draw[edgeg] (u1) -- (u2);
\draw[edgeg] (u2) -- (u3);

\draw[edgeg] (v4) -- (w1);

\draw[edgeg] (w1) -- (w2);
\draw[edgeg] (w2) -- (w3);
\draw[edgeg] (w3) -- (w4);

\draw[edgeg,dashed] (u3) -- (-2,-6.75);
\draw[edgeg,dashed] (w4) -- (2,-7.75);
\end{tikzpicture}
\caption{$S$-operation $\tau_-$ of bottom $t_{\min}$ node is $\join(X^S,X_1^S,X_2^S,L_1^S,L_2^S)$.}\label{fig:pathsa}
\end{subfigure}
\begin{subfigure}[t]{0.34\textwidth}
\centering
\begin{tikzpicture}[line width=1pt,scale=.8,xscale=1]
\draw[trace2] (0,0) -- (0,-3);
\draw[trace] (0,-4) -- (0,-7);

\node[vertg] (v0) at (0,1) {};
\node[vertb] (v1) at (0,0) {};
\node[vertb] (v2) at (0,-1) {};
\node[vertb] (v3) at (0,-2) {};
\node[vert2b] (v4) at (0,-3) {};

\node (tmax) at (-1,0) {$t_{\max}$};
\node (taup) at (.5,1) {$\tau_+$};
\node (tmin) at (-1,-3) {$t_{\min}$};
\node (taum) at (.75,-3) {$\tau_-$};
\node (trace) at (2,0) {$(L^S, X^S, R^S)$};

\node[trial,line width=1.5pt] (t1) at (1,-1.5) {};
\draw[edgeb] (v2) -- (t1.apex);

\node[triar,line width=1.5pt] (t2) at (-1,-3.5) {};
\draw[edgeb] (v4) -- (t2.apex);

\node[vertg] (w1) at (0,-4) {};
\node[vertg] (w2) at (0,-5) {};
\node[vertg] (w3) at (0,-6) {};
\node[vert2g] (w4) at (0,-7) {};

\node (trace1) at (2,-4) {$(L_1^S, X_1^S, R_1^S)$};

\draw[edgeg,dashed] (0,1.75) -- (v0);

\draw[edgeg] (v0) -- (v1);
\draw[edgeb] (v1) -- (v2);
\draw[edgeb] (v2) -- (v3);
\draw[edgeb] (v3) -- (v4);

\draw[edgeg] (v4) -- (w1);

\draw[edgeg] (w1) -- (w2);
\draw[edgeg] (w2) -- (w3);
\draw[edgeg] (w3) -- (w4);

\draw[edgeg,dashed] (w4) -- (0,-7.75);
\end{tikzpicture}
\caption{$S$-operation $\tau_-$ of bottom node $t_{\min}$ is $\introduce(v)$ or $\forget(v)$.}\label{fig:pathsb}
\end{subfigure}
    \caption{Illustrations of the main idea of the dynamic programming algorithm. Let $(L^S, X^S, R^S)$ be an $S$-trace. Let $\tau_+$ be the $S$-operation of the parent of the top node $t_{\max}$ of the directed path of the $S$-trace. Let $\tau_-$ be the $S$-operation of the bottom node $t_{\min}$ of the directed path of the $S$-trace. The directed path of the $S$-trace $L^S, X^S, R^S$ is surrounded by a red line. The $S$-bottom nodes are depicted as gray squares. The directed paths of the $S$-traces of $S$-children of the $S$-bottom node~$t_{\min}$, for which we look up information in the dynamic programming table, are surrounded by green lines. The blue triangles illustrate parts of the tree decompositions where all nodes have $S$-traces with $L^S=\emptyset$, and hence those nodes are not part of any directed paths of $S$-traces. Furthermore, the roots of the blue triangles are not $S$-children of their respective parent nodes. The width of the part of the tree decomposition drawn in thick black lines is computed directly from the bags $X_{\max}$ and $X_{\min}$ of the top node $t_{\max}$ and the bottom node $t_{\min}$ of the directed path, respectively, and the set $X_{\dpath}$ of vertices that we want to add to bags in the directed path or subtrees illustrated by blue triangles. Subfigure \ref{fig:pathsa} shows the situation where the $S$-bottom node $t_{\min}$ has two $S$-children. Subfigure \ref{fig:pathsb} shows the situation where the $S$-bottom node $t_{\min}$ has one $S$-child. }\label{fig:paths}
\end{figure}

Informally speaking, one of the main difficulties in this approach is handling $S$-bottom nodes that are contained in full join trees (\cref{def:fulljointree}), that is, their parent and $S$-children potentially have full bags, and at least one of them has. 
As a simple example that should convey some intuition consider the case where an $S$-bottom node has $\join$ as its $S$-operation. Let $t_{\min}$ be such a node and let $t_1,t_2$ be its two $S$-children.  
In our algorithm, we wish to determine a candidate for the bag $X_{\min}$. However, if we face a yes-instance, a solution may use a different candidate $X^\star_{\min}$ for the bag of $t_{\min}$. In order to prove correctness, we have to argue that we can transform the solution tree decomposition (without increasing its width) to the one that we are computing tree decomposition. 
 Assume two vertices $u,v\in V$ are adjacent, and in a solution $S$-nice tree decomposition we have that $u,v\in V_{t_1}$ and $v\in X^\star_{\min}$. However, we chose $X_{\min}$ such $u\in X_{\min}$ and $u,v\in V_{t_2}$. For a visualization, see \cref{fig:join1}. Now we have to argue that we can do a transformation such as the one visualized in \cref{fig:join1} (from left to right). However, the transformation increases the size of the bag of~$t_2$. We have to make sure that this does not increase the width of the tree decomposition.
 If the bag of node $t_2$ in the solution tree decomposition is not full, then we can perform the transformation illustrated in \cref{fig:join1}. 

\begin{figure}[t]
\centering
\begin{tikzpicture}[line width=1pt,scale=.5,yscale=.9]

\node[vert,minimum width=1.2cm,fill=red!30!white] (p) at (0,3) {};
\node[vert,minimum width=1.2cm,fill=orange!30!white] (c1) at (3,0) {};
\node[vert,minimum width=1.2cm] (c11) at (1,-3) {};
\node[vert,minimum width=1.2cm] (c12) at (5,-3) {};
\node[vert,minimum width=1.2cm,fill=orange!30!white] (c2) at (-3,0) {};
\node[vert,minimum width=1.2cm] (c21) at (-3,-3) {};

\node (a) at (2,3) {$t_{\min}$};
\node (a) at (-1.3,0) {$t_1$};
\node (a) at (4.7,0) {$t_2$};

\phantom{
\node[vert,minimum width=1.2cm] (c22) at (-5,-3) {};
}

\draw[edge,dashed] (p) -- (0,5);
\draw[edge] (p) -- (c1);
\draw[edge] (p) -- (c2);
\draw[edge] (c1) -- (c11);
\draw[edge] (c1) -- (c12);
\draw[edge,dashed] (c11) -- (1,-5);
\draw[edge,dashed] (c12) -- (5,-5);
\draw[edge] (c2) -- (c21);
\draw[edge,dashed] (c21) -- (-3,-5);

\node[vert,fill=cyan!20!white] (v1) at (-2.5,0) {};
\node[vert,fill=cyan!20!white] (v2) at (.5,3) {};
\node[vert,fill=cyan!20!white] (v3) at (3.5,0) {};
\node[vert,fill=cyan!20!white] (v4) at (1.5,-3) {};
\node[vert,fill=cyan!20!white] (v5) at (5.5,-3) {};

\node[vert2,fill=green!30!white] (u1) at (-3.5,0) {};
\node[vert2,fill=green!30!white] (u2) at (-3.5,-3) {};
\end{tikzpicture}
\begin{tikzpicture}[line width=1pt,scale=.5,yscale=.9]

\node (a) at (-6.6,0) {\huge$\rightarrow$};

\node[vert,minimum width=1.2cm,fill=red!30!white] (p) at (0,3) {};
\node[vert,minimum width=1.2cm,fill=orange!30!white] (c1) at (3,0) {};
\node[vert,minimum width=1.2cm] (c11) at (1,-3) {};
\node[vert,minimum width=1.2cm] (c12) at (5,-3) {};
\node[vert,minimum width=1.2cm,fill=orange!30!white] (c2) at (-3,0) {};
\node[vert,minimum width=1.2cm] (c21) at (-3,-3) {};

\node (a) at (2,3) {$t_{\min}$};
\node (a) at (-1.3,0) {$t_1$};
\node (a) at (4.7,0) {$t_2$};

\phantom{
\node[vert,minimum width=1.2cm] (c22) at (-5,-3) {};
}

\draw[edge,dashed] (p) -- (0,5);
\draw[edge] (p) -- (c1);
\draw[edge] (p) -- (c2);
\draw[edge] (c1) -- (c11);
\draw[edge] (c1) -- (c12);
\draw[edge,dashed] (c11) -- (1,-5);
\draw[edge,dashed] (c12) -- (5,-5);
\draw[edge] (c2) -- (c21);
\draw[edge,dashed] (c21) -- (-3,-5);

\node[vert,fill=cyan!20!white] (v1) at (3.5,0) {};
\node[vert,fill=cyan!20!white] (v2) at (1.5,-3) {};
\node[vert,fill=cyan!20!white] (v3) at (5.5,-3) {};

\node[vert2,fill=green!30!white] (u1) at (2.5,0) {};
\node[vert2,fill=green!30!white] (u2) at (-.5,3) {};
\node[vert2,fill=green!30!white] (u3) at (-3.5,0) {};
\node[vert2,fill=green!30!white] (u4) at (-3.5,-3) {};

\end{tikzpicture}
    \caption{Two illustrations (left and right) of the situations when $t_{\min}$ is an $S$-bottom node of an $S$-nice tree decomposition that has $\join$ as its $S$-operation. The two $S$-children of $t_{\min}$ are denoted~$t_1$ and~$t_2$. The red circle visualizes the bag of $t_{\min}$ and the orange circles visualize the bags of $t_1$ and~$t_2$. The green square visualizes a vertex $u\in V$ and the small blue circle a vertex $v\in V$ such that $\{u,v\}\in E$. The left side illustrates part of a tree decomposition where, in particular, $u$ and $v$ are both contained in the bag of $t_1$, $v$ is contained in the bag of $t_{\min}$, and $u$ is not contained in the bag of $t_{\min}$. The right side illustrates part of a tree decomposition where, in particular, $u$ and $v$ are both contained in the bag of $t_2$, $u$ is contained in the bag of $t_{\min}$, and $v$ is not contained in the bag of~$t_{\min}$.}\label{fig:join1}
\end{figure}

If the bag of node $t_2$ in the solution tree decomposition is full (which may be the case if $t_{\min}$ is contained in a full join tree), then we cannot perform the transformation. Then we face a situation as illustrated in \cref{fig:join2}. Here, informally speaking, we need to make a transformation that changes the size of a bag that is (potentially) far away from $t_{\min}$ in the tree decomposition. In order to be able to do this, we need to identify this situation and treat it differently. 
To this end, we introduce a new $S$-operation $\bigjoin$. Informally speaking, if a \slim \topheavy $S$-nice tree decompositions, there is an $S$-bottom node $t$ that admits a $\join$ as its $S$-operation, then from now on we say that node $t$ admits a $\bigjoin$ as its $S$-operation. 

\begin{figure}[t]
\centering
\begin{tikzpicture}[line width=1pt,scale=.5,yscale=.9]

\node[vert,minimum width=1.2cm,fill=red!30!white,line width=2pt] (p) at (0,3) {};
\node[vert,minimum width=1.2cm,fill=orange!30!white,line width=2pt] (c1) at (3,0) {};
\node[vert,minimum width=1.2cm] (c11) at (1,-3) {};
\node[vert,minimum width=1.2cm] (c12) at (5,-3) {};
\node[vert,minimum width=1.2cm,fill=orange!30!white] (c2) at (-3,0) {};
\node[vert,minimum width=1.2cm] (c21) at (-3,-3) {};

\node (a) at (2,3) {$t_{\min}$};
\node (a) at (-1.3,0) {$t_1$};
\node (a) at (4.7,0) {$t_2$};

\phantom{
\node[vert,minimum width=1.2cm] (c22) at (-5,-3) {};
}

\draw[edge,dashed] (p) -- (0,5);
\draw[edge] (p) -- (c1);
\draw[edge] (p) -- (c2);
\draw[edge] (c1) -- (c11);
\draw[edge] (c1) -- (c12);
\draw[edge,dashed] (c11) -- (1,-5);
\draw[edge,dashed] (c12) -- (5,-5);
\draw[edge] (c2) -- (c21);
\draw[edge,dashed] (c21) -- (-3,-5);

\node[vert,fill=cyan!20!white] (v1) at (-2.5,0) {};
\node[vert,fill=cyan!20!white] (v2) at (.5,3) {};
\node[vert,fill=cyan!20!white] (v3) at (3.5,0) {};
\node[vert,fill=cyan!20!white] (v4) at (1.5,-3) {};

\node[vert2,fill=green!30!white] (u1) at (-3.5,0) {};
\node[vert2,fill=green!30!white] (u2) at (-3.5,-3) {};
\end{tikzpicture}
\begin{tikzpicture}[line width=1pt,scale=.5,yscale=.9]

\node (a) at (-6.6,0) {\huge$\rightarrow$};

\node[vert,minimum width=1.2cm,fill=red!30!white,line width=2pt] (p) at (0,3) {};
\node[vert,minimum width=1.2cm,fill=orange!30!white,line width=2pt] (c1) at (3,0) {};
\node[vert,minimum width=1.2cm] (c11) at (1,-3) {};
\node[vert,minimum width=1.2cm] (c12) at (5,-3) {};
\node[vert,minimum width=1.2cm,fill=orange!30!white] (c2) at (-3,0) {};
\node[vert,minimum width=1.2cm] (c21) at (-3,-3) {};

\node (a) at (2,3) {$t_{\min}$};
\node (a) at (-1.3,0) {$t_1$};
\node (a) at (4.7,0) {$t_2$};

\phantom{
\node[vert,minimum width=1.2cm] (c22) at (-5,-3) {};
}

\draw[edge,dashed] (p) -- (0,5);
\draw[edge] (p) -- (c1);
\draw[edge] (p) -- (c2);
\draw[edge] (c1) -- (c11);
\draw[edge] (c1) -- (c12);
\draw[edge,dashed] (c11) -- (1,-5);
\draw[edge,dashed] (c12) -- (5,-5);
\draw[edge] (c2) -- (c21);
\draw[edge,dashed] (c21) -- (-3,-5);

\node[vert,fill=cyan!20!white] (v1) at (1.5,-3) {};

\node[vert2,fill=green!30!white] (u1) at (.5,-3) {};
\node[vert2,fill=green!30!white] (u2) at (2.5,0) {};
\node[vert2,fill=green!30!white] (u3) at (-.5,3) {};
\node[vert2,fill=green!30!white] (u4) at (-3.5,0) {};
\node[vert2,fill=green!30!white] (u5) at (-3.5,-3) {};

\end{tikzpicture}
    \caption{Two illustrations (left and right) of the situations when $t_{\min}$ is an $S$-bottom node of an $S$-nice tree decomposition that has $\join$ as its $S$-operation. For further information on the visualization see the description of \cref{fig:join1}. Additionally, thick circles indicate that the visualized bags are full. Furthermore, in this figure, the right side illustrates part of a tree decomposition where, in particular, $u$ is contained in the bag of $t_2$, $v$ is not contained in the bag of $t_2$, $u$ is contained in the bag of $t_{\min}$, and $v$ is not contained in the bag of~$t_{\min}$, and $u$ and $v$ both are contained in the bag of the child of $t_2$.}\label{fig:join2}
\end{figure}

Finally, notice that $X_{\max}$ is the bag of an $S$-child of an $S$-bottom node. To determine $X_{\max}$, we need to include additional information in the extended version of the $S$-operations, that help us determine the bags of $S$-children of the nodes that admit the extended $S$-operations.
We remark that all additional information contained in the extended versions of $S$-operations is only relevant for the computation of the ``local treewidth'' of the directed path of the $S$-trace.
We give formal definitions of the extended $S$-operations in the next section.

\subsection{Extended \boldmath$S$-Operations}\label{sec:extendedops}
In this section, we describe how we extend the $S$-operations introduced in \cref{def:snicetd}. 
Since we are interested in \slim $S$-nice tree decomposition, these operations will only be defined on \slim $S$-nice tree decompositions. We do not need the ``\topheavy'' property for the definitions.

First of all, recall that we allow the dummy $S$-operation $\void$ if the top node is the root or if the bottom node does not have any $S$-children.
Furthermore, as discussed before, we introduce a \emph{big} version $\bigjoin$ of the $S$-operation $\join$. When we compute bags for $S$-bottom nodes that admit this operation, we will have several possible candidates for the bag, and an additional bit string~$s$ encodes which one we are considering. Furthermore, we need the additional information on which vertex is (potentially) missing in each of the bags of the $S$-children or the parent, in the case that they are not full. This is encoded with three integers $d_0$, $d_1$ and $d_2$. We have the following big $S$-operation.
\begin{itemize}
    \item $\bigjoin(X^S,X_1^S,X_2^S,L_{1}^S, L_{2}^S,s,d_0,d_1,d_2)$ with some vertex sets $X^S,X_1^S,X_2^S,L_1^S,L_2^S\subseteq S$ 
    and some bit string $s\in\{0,1\}^{2\fvn(G)+1}$ and $d_0,d_1,d_2\in\{0,\ldots,k+1\}$. 
\end{itemize}
In order to define which $S$-bottom nodes admit the above-described big $S$-operations, we introduce a function that ``decodes'' its bag. Let
\[
\bigbag : 2^S\times \{0,1\}^{2\fvn(G)+1} \rightarrow 2^V
\]
be a function which takes as input a subset of $S$ (which will be the set $X^S$ of the $S$-bottom node that admits the big $S$-operation) and the bit string $s$ that encodes which bag candidate is chosen.
This will always be a full bag. The function is formally defined in \cref{sec:bags}.
We additionally need to define a function that determines the bags of the children using the additional integers $d_1$ and $d_2$. Let
\[
\subbag : 2^V\times \{0,\ldots,k+1\}\rightarrow 2^V
\]
be a function that takes as input a subset of $V$ (which will be the bag of the $S$-bottom node that admits the $S$-operation $\bigjoin$) and an additional integer that encodes which vertex is missing is the bag of a child (if the integer is zero, this will encode that the bag of the child is the same).
Formally, $\subbag(X,0)=X$ and $\subbag(X,i)=X\setminus\{v\}$, where $v$ is at the $i$th ordinal position in $X$ (recall that $V$ is ordered in an arbitrary but fixed way). 
\begin{definition}[$\bigjoin$]\label{def:bigjoin}
Let $\mathcal{T}=(T,\{X_t\}_{t\in V(T)})$ be a \slim $S$-nice tree decomposition of~$G$. 
Let $t$ be an $S$-bottom node in $\mathcal{T}$ that is contained in a full join tree and admits $S$-operation $\join(X^S_t,X_1^S,X_2^S,L_{1}^S, L_{2}^S)$. 
Let $t_1,t_2$ be the two $S$-children of $t$ and let~$t_0$ be the parent of $t$. If $X_t=\bigbag(X^S_t, s)$ and for all $0\le i\le 2$ we have $X_{t_i}=\subbag(X_t,d_i)$, then we say that node $t$ admits $S$-operation $\bigjoin(X^S_t,X_1^S,X_2^S,L_{1}^S, L_{2}^S,s,d_0,d_1,d_2)$.
\end{definition}
The big $S$-operation $\bigjoin$ is illustrated in \cref{fig:bigops}. Note that at this point, we cannot decide whether an $S$-bottom node admits a $\bigjoin$ $S$-operation since we have not defined the function $\bigbag$ yet. Informally speaking, there will be sufficiently many options for the bit string $s$ such that for every possible bag that a node $t$ that meets the conditions in \cref{def:bigjoin} has, the function $\bigbag$ maps the $X_t^S$ and the bit string $s$ to that bag. Formally, we will show the following.

\begin{lemma}\label{lem:alwaysbigjoin}
Let $\mathcal{T}=(T,\{X_t\}_{t\in V(T)})$ be a \slim $S$-nice tree decomposition of $G$. 
Let~$t$ be an $S$-bottom node in $\mathcal{T}$ that is contained in a full join tree and admits $S$-operation $\join(X_t^S,X_1^S,X_2^S,L_{1}^S, L_{2}^S)$. 
Then there exists $s\in\{0,1\}^{2\fvn(G)+1}$ and $d_0,d_1,d_2\in\{0,\ldots,k+1\}$ such that node $t$ admits $S$-operation $\bigjoin(X^S_t,X_1^S,X_2^S,L_{1}^S, L_{2}^S,s,d_0,d_1,d_2)$.
\end{lemma}
We postpone the proof of \cref{lem:alwaysbigjoin} so \cref{sec:bigops}, where we also formally define the function $\bigbag$.

\begin{figure}[t]
\centering
\begin{tikzpicture}[line width=1pt,scale=.5,yscale=.8]
\phantom{\node[vert,minimum width=1cm,line width=2pt] (p) at (0,7) {};
\node[vert,minimum width=1cm,line width=2pt] (c1) at (3,0) {};
\node[vert,minimum width=1cm,line width=2pt] (c2) at (-3,0) {};
\node[vert,minimum width=1cm,line width=2pt] (c) at (0,-1) {};}

\node[vert,minimum width=1cm] (p) at (0,7) {};
\node[vert,minimum width=1cm,line width=2pt,dashed] (p) at (0,7) {};
\node[vert,minimum width=1cm,fill=red!30!white,line width=2pt] (t) at (0,3) {};
\node[vert,minimum width=1cm] (c1) at (3,0) {};
\node[vert,minimum width=1cm] (c2) at (-3,0) {};
\node[vert,minimum width=1cm,line width=2pt,dashed] (c1) at (3,0) {};
\node[vert,minimum width=1cm,line width=2pt,dashed] (c2) at (-3,0) {};

\draw[edge] (p) -- (t);
\draw[edge] (t) -- (c1);
\draw[edge] (t) -- (c2);

\node (a) at (1.3,3) {$t$};
\end{tikzpicture}
    \caption{Illustration of the big $S$-operation $\bigjoin$ (\cref{def:bigjoin}). Node $t$ admits $\bigjoin(X^S_t,X_1^S,X_2^S,L_{1}^S, L_{2}^S,s,d_0,d_1,d_2)$. The red circle represents $S$-bottom node $t$. The remaining circles are the parent and two $S$-children, respectively. Children of $t$ that are not $S$-children are not depicted. Thick circles represent nodes whose bags are full and dashed circles represent nodes whose bags are potentially full. 
    Since $t$ is contained in a full join tree, at least one of the nodes represented by dashed circles has a full bag.}\label{fig:bigops}
\end{figure}

Now we introduce ``small'' $S$-operations for $S$-bottom nodes that are not contained in full join trees.
Similar as in the case of the big operations, we need some additional information that allows us to compute the bags. 
Here, we add an additional integer $d$ that will guide our algorithm to a specific bag candidate (and candidates for its neighbors). The main reason why we add this information to the state is to guarantee that our algorithm computes the same bag for the bottom node of the directed path as we computed in the predecessor states for the parent of the top node.
%
We have the following \emph{small} $S$-operations.
\begin{itemize}
    \item $\smallintroduce(v, d)$ with some vertex $v\in S$ and some integer $d\in\{1,\ldots, (n+1)^3\}$.
    \item $\smalljoin(X^S,X_1^S,X_2^S,L_{1}^S, L_{2}^S,d)$ with some vertex sets $X^S,X_1^S,X_2^S,L_1^S,L_2^S\subseteq S$ and some integer $d\in\{1,\ldots,(n+1)^3\}$.
\end{itemize}
In order to define which $S$-bottom nodes admit the above-described small $S$-operations, we again use a function that ``decodes'' its bag, and also its neighboring bags. Let
\[
\smallbag : 2^S\times2^S\times2^S\times2^S\times \mathbb{N}\rightarrow  2^V\times 2^V\times 2^V\times 2^V
\]
be a function which takes as input four subsets of $S$  (which, roughly speaking, say which vertices of $S$ are ``above'', ``in'', ``below to the left'' and ``below to the right'' a certain node in the tree decomposition) and the number $d$ that will guide the algorithm to a specific bag for the $S$-bottom node, as well as bags for its neighbors. The function is formally defined in \cref{sec:bags}.

\begin{definition}[$\smallintroduce$]\label{def:smallintro}
    Let $\mathcal{T}=(T,\{X_t\}_{t\in V(T)})$ be a \slim $S$-nice tree decomposition of $G$. 
    Let $t$ be an $S$-bottom node in $\mathcal{T}$ that admits $S$-operation $\introduce(v)$. Let $t_1$ be the $S$-child of $t$ and let $t_0$ be the parent of $t$. If $(X_t,X_{t_0},X_{t_1},X_{t_1})=\smallbag(R^S_t,X^S_t,L^S_t,\emptyset,d)$, then we say that node $t$ admits $S$-operation $\smallintroduce(v,d)$.
\end{definition}

\begin{definition}[$\smalljoin$]\label{def:smalljoin}
    Let $\mathcal{T}=(T,\{X_t\}_{t\in V(T)})$ be a \slim $S$-nice tree decomposition of~$G$. 
    Let $t$ be an $S$-bottom node in $\mathcal{T}$ that admits $S$-operation $\join(X_t^S,X_1^S,X_2^S,L_{1}^S, L_{2}^S)$ and is not contained in a full join tree. Let $t_1,t_2$ be the two $S$-children of $t$ and let~$t_0$ be the parent of $t$. If $(X_t,X_{t_0},X_{t_1},X_{t_2})=\smallbag(R^S_t,X^S_t,L_{1}^S, L_{2}^S,d)$, then we say that node $t$ admits $S$-operation $\smalljoin(X^S_t,X_1^S,X_2^S,L_{1}^S, L_{2}^S,d)$.
\end{definition}

The small $S$-operations are illustrated in \cref{fig:smallops}. Note that in particular, if an $S$-bottom node admits the $S$-operation $\smallintroduce$, then the bags of its parent and its $S$-child are not full. Similarly, if an $S$-bottom node admits the $S$-operation $\smalljoin$, then the bags of its parent and both its $S$-children are not full. 

\begin{figure}[t]
\centering
\begin{subfigure}[t]{0.45\textwidth}
\centering
\begin{tikzpicture}[line width=1pt,scale=.5,yscale=.8]
\phantom{\node[vert,minimum width=1cm,line width=2pt] (p) at (0,7) {};
\node[vert,minimum width=1cm,line width=2pt] (c1) at (3,0) {};
\node[vert,minimum width=1cm,line width=2pt] (c2) at (-3,0) {};
\node[vert,minimum width=1cm,line width=2pt] (c) at (0,-1) {};}

\node[vert,minimum width=1cm] (p) at (0,7) {};
\node[vert,minimum width=1cm,fill=red!30!white,line width=2pt] (t) at (0,3) {};
\node[vert,minimum width=1cm] (c) at (0,-1) {};

\draw[edge] (p) -- (t);
\draw[edge] (t) -- (c);

\node (a) at (1.3,3) {$t$};
\end{tikzpicture}
\caption{Node $t$ admits $\smallintroduce(v,d)$.}\label{fig:smallopsa}
\end{subfigure}
\begin{subfigure}[t]{0.49\textwidth}
\centering
\begin{tikzpicture}[line width=1pt,scale=.5,yscale=.8]
\phantom{\node[vert,minimum width=1cm,line width=2pt] (p) at (0,7) {};
\node[vert,minimum width=1cm,line width=2pt] (c1) at (3,0) {};
\node[vert,minimum width=1cm,line width=2pt] (c2) at (-3,0) {};
\node[vert,minimum width=1cm,line width=2pt] (c) at (0,-1) {};}
\node[vert,minimum width=1cm] (p) at (0,7) {};
\node[vert,minimum width=1cm,fill=red!30!white,line width=2pt] (t) at (0,3) {};
\node[vert,minimum width=1cm] (c1) at (3,0) {};
\node[vert,minimum width=1cm] (c2) at (-3,0) {};

\draw[edge] (p) -- (t);
\draw[edge] (t) -- (c1);
\draw[edge] (t) -- (c2);

\node (a) at (1.3,3) {$t$};
\end{tikzpicture}
\caption{Node $t$ admits $\smalljoin(X^S_t,X_1^S,X_2^S,L_{1}^S, L_{2}^S,d)$.}\label{fig:smallopsd}
\end{subfigure}
    \caption{Illustration of the small $S$-operations (\cref{def:smallintro,def:smalljoin}). Here, the red circles are $S$-bottom nodes. The remaining circles are parents and $S$-children, respectively. Children of the $S$-bottom nodes that are not $S$-children are not depicted. Thick circles represent nodes whose bags may be full. The bags of nodes represented by thin circles are not full. 
    }\label{fig:smallops}
\end{figure}

In the case that an $S$-bottom node $t$ that admits the $S$-operation $\forget(v)$, we have that its $S$-child $t'$ has a potentially larger bag. 
In order to determine the bag of $t$, we will first compute the bag of the $S$-child $t'$. 
Informally speaking, we need some information about the child of $t'$ that is on the path from $t$ to its closest descendant that is an $S$-bottom node.
Let $t''$ be the closest descendant of $t$ that is an $S$-bottom node. If $t''$ has the same bag as $t'$, and this bag is full, then we essentially have to compute the bag of the $S$-bottom node $t''$. Hence, we need a flag $f\in\{\true,\false\}$ that indicates whether we are in this case. And if we are (that is, $f=\true$), then we need all information about the $S$-operation $\tau$ of $t''$. Note that if $t''$ also admits the $S$-operation $\forget(v)$, then we know that the bag of $t''$ is not full, which will be enough.
If we are not in the above-described case (that is, $f=\false$), then, similar as in the other small $S$-operations, we need an additional integer that will guide our algorithm to a specific bag candidate.
To capture all the above-described extra information, we introduce the following extended version of the $S$-operation $\forget(v)$. 
\begin{itemize}
    \item $\extendedforget(v,d,f, \tau)$ with some vertex $v\in S$, $d\in \{1,\ldots,(n+1)^3\}$, $f\in\{\true,\false\}$, and some extended $S$-operation $\tau$ that is different from $\extendedforget$.
\end{itemize}
If the $S$-operation $\tau$ in $\extendedforget(v,d,f,\tau)$ is a big $S$-operation, then we consider $\extendedforget(v,d,f,\tau)$ a big $S$-operation. Otherwise, we consider $\extendedforget(v,d,f,\tau)$ a small $S$-operation.
Similarly to the previously introduced extended $S$-operations, we define an auxiliary function to ``decode'' the bag and the neighboring bags of the above-decribed extended $S$-operation. Let
\[
\smallbagf : 2^S\times 2^S\times 2^S\times \mathbb{N}\times\{\true,\false\}\times\Pi\rightarrow 2^V\times 2^V\times 2^V,
\]
where $\Pi$ denotes the set of all extended $S$-operations. Note that both in the input and output there are only three sets of vertices, since we do not need to consider the case that the $S$-bottom node has two $S$-children. We remark that the function behaves similar to $\bigbag$ if $\tau$ is a big $S$-operation, and similar to $\smallbag$ otherwise. In \cref{sec:bags} we describe how the function is formally defined in both cases.
\begin{definition}[$\extendedforget$]\label{def:extendedforget}
    Let $\mathcal{T}=(T,\{X_t\}_{t\in V(T)})$ be a \slim $S$-nice tree decomposition of $G$. 
    Let $t$ be an $S$-bottom node in $\mathcal{T}$ that admits $S$-operation $\forget(v)$. 
    Let $t'$ be the $S$-child of $t$. 
    Let $\tau$ be an extended $S$-operation.
    \begin{itemize}
        \item If there is a node $t''$ that is the closest descendant of $t$ which is an $S$-bottom node, $|X_t|=k$, $X_t\cup\{v\}=X_{t'}=X_{t''}$, $t''$ admits the non-extended version of $\tau$, $(X_{t'},X_t,X_{t''})=\smallbagf(R^S_{t'},X^S_{t'},L^S_{t'},d,\true,\tau)$,
        then we say that node $t$ admits $S$-operation $\extendedforget(v,d,\true,\tau)$. 
        \item Otherwise, if $(X_{t'},X_t,X_{t''})=\smallbagf(R^S_{t'}, X^S_{t'},L^S_{t'},d,\false,\void)$, then we say that node $t$ admits $S$-operation $\extendedforget(v,d,\false,\void)$.
    \end{itemize}
\end{definition}

The extended $S$-operations $\extendedforget$ is illustrated in \cref{fig:extforget}.
We can observe that if an $S$-bottom node admits some extended $S$-operation $\extendedforget(v,d,f,\tau)$, then $\tau\neq\extendedforget(v',d',f',\tau')$. This follows from the fact that the bag of an $S$-bottom node that admits $S$-operation $\forget(v)$ (and hence potentially an $\extendedforget$ $S$-operation) is never full.
\begin{observation}
 Let $\mathcal{T}=(T,\{X_t\}_{t\in V(T)})$ be a \slim $S$-nice tree decomposition of $G$. 
Let~$t$ be an $S$-bottom node in $\mathcal{T}$ that admits $S$-operation $\extendedforget(v,d,f,\tau)$ for some $v\in S$, $d\in \{1,\ldots,(n+1)^3\}$, $f\in\{\true,\false\}$, and some $S$-operation $\tau$. Then we have that $\tau\neq\extendedforget(v',d',f',\tau')$ for every $v'\in S$, $d'\in \{1,\ldots,(n+1)^3\}$, $f'\in\{\true,\false\}$, and every extended $S$-operation $\tau'$.
\end{observation}

\begin{figure}[t]
\centering
\begin{subfigure}[t]{0.46\textwidth}
\centering
\begin{tikzpicture}[line width=1pt,scale=.5,yscale=.8]
\phantom{\node[vert,minimum width=1cm,line width=2pt] (p) at (0,7) {};
\node[vert,minimum width=1cm,line width=2pt] (c1) at (3,0) {};
\node[vert,minimum width=1cm,line width=2pt] (c2) at (-3,0) {};
\node[vert,minimum width=1cm,line width=2pt] (c) at (0,-1) {};
}

\node[vert,minimum width=1cm,fill=red!30!white] (p) at (0,7) {};
\node[vert,minimum width=1cm,line width=2pt] (t) at (0,3) {};
\node[vert,minimum width=1cm,fill=orange!30!white,line width=2pt] (c) at (0,-5) {};

\draw[edge] (p) -- (t);
\draw[edge,dashed] (t) -- (c);

\node (a) at (1.3,7) {$t$};
\node (a) at (1.4,3) {$t'$};
\node (a) at (1.5,-5) {$t''$};
\end{tikzpicture}
\caption{Note $t$ admits $\extendedforget(v,d,\true,\tau)$.}\label{fig:smallopsb}
\end{subfigure}
\begin{subfigure}[t]{0.46\textwidth}
\centering
\begin{tikzpicture}[line width=1pt,scale=.5,yscale=.8]
\phantom{\node[vert,minimum width=1cm,line width=2pt] (p) at (0,7) {};
\node[vert,minimum width=1cm,line width=2pt] (c1) at (3,0) {};
\node[vert,minimum width=1cm,line width=2pt] (c2) at (-3,0) {};
\node[vert,minimum width=1cm,line width=2pt] (c) at (0,-1) {};}

\node[vert,minimum width=1cm,fill=red!30!white] (p) at (0,7) {};
\node[vert,minimum width=1cm,line width=2pt] (t) at (0,3) {};
\node[vert,minimum width=1cm] (c) at (0,-1) {};
\node[vert,minimum width=1cm,fill=orange!30!white] (c2) at (0,-5) {};
\node[vert,minimum width=1cm,fill=orange!30!white,line width=2pt,dashed] (c2) at (0,-5) {};

\draw[edge] (p) -- (t);
\draw[edge,dashed] (t) -- (c);
\draw[edge,dashed] (c) -- (c2);

\node (a) at (1.3,7) {$t$};
\node (a) at (1.4,3) {$t'$};
\node (a) at (1.5,-5) {$t''$};
\end{tikzpicture}
\caption{Node $t$ admits $\extendedforget(v,d,\false,\void)$.}\label{fig:smallopsc}
\end{subfigure}
    \caption{Illustration of the extended $S$-operation $\extendedforget$ (\cref{def:extendedforget}). Here, the red and orange circles are $S$-bottom nodes. Children of the $S$-bottom nodes that are not $S$-children are not depicted. Thick and dashed circles represent nodes whose bags may be full. The bags of nodes represented by thin circles are not full. 
    The dashed edges in \cref{fig:smallopsb,fig:smallopsc} indicate ancestor-descendant relationships. The node with a dashed circle in \cref{fig:smallopsc} indicates that the bag of this node may be full, but it is different from the bag of the child of node $t$.}\label{fig:extforget}
\end{figure}

Finally, we remark that for the small $S$-operations, we do not have an analog to \cref{lem:alwaysbigjoin}, that is, in a \slim $S$-nice tree decomposition, it is generally not the case that every $S$-bottom node that is not contained in a full join tree admits some small $S$-operation. We have to refine the tree decompositions further to ensure that all $S$-bottom nodes admit some extended (big or small) $S$-operation. This is discussed further in \cref{sec:correct}.

\subsection{Dynamic Programming States and Table}\label{sec:statestable}
Now we define a state of the dynamic program. From now on, we only consider extended $S$-operations, that is, whenever we refer to $S$-operations, we only refer to extended ones. As described at the beginning of the section, it contains an $S$-trace and two $S$-operations, one for the $S$-bottom node of the directed path of the $S$-trace, and one for the parent (if it exists) of the top node of the directed path of the $S$-trace. 
Formally, we define a state for the dynamic program as follows.
\begin{definition}[State, Witness]\label{def:state}
    Let $(L^S, X^S, R^S)$ be an $S$-trace and let $\tau_+$ and~$\tau_-$ be two $S$-operations. We call $\phi=(\tau_-,L^S, X^S, R^S,\tau_+)$ a \emph{state}. 
We say that a \slim $S$-nice tree decomposition $\mathcal{T}$ of $G$ \emph{witnesses} $\phi$ if the following conditions hold.
    \begin{itemize}
    \item If $\tau_+=\void$, then the root of $\mathcal{T}$ has $S$-trace $(L^S, X^S, R^S)$.
        \item If $L^S=\emptyset$, then $\tau_-=\void$ and there exists an $S$-bottom node that admits $S$-operation $\tau_+$ and has an $S$-child with $S$-trace $(L^S, X^S, R^S)$.
        \item If $L^S\neq\emptyset$, then the $S$-bottom node of the directed path of $S$-trace $(L^S, X^S, R^S)$ admits $S$-operation~$\tau_-$ and, if additionally $\tau_+\neq\void$, then the parent of the top node of the directed path of $S$-trace $(L^S, X^S, R^S)$ admits $S$-operation~$\tau_+$.
    \end{itemize} 
    We denote with $\mathcal{S}$ the set of all states.
\end{definition}
Note that \cref{def:state} implies that if a \slim $S$-nice tree decomposition $\mathcal{T}$ of $G$ witnesses $\phi=(\tau_-,L^S, X^S, R^S,\tau_+)$, then we have that $\tau_-=\void$ if and only if $L^S=\emptyset$, since $S$-bottom nodes never admit the $S$-operation $\void$.

Intuitively, we wish to obtain a dynamic programming table $\ptw:\mathcal{S}\rightarrow \{\true,\false\}$ (for ``partial treewidth'') that maps states to $\true$ if and only if there is a tree decomposition~$\mathcal{T}$ of the subgraph ``below'' the top node of the directed path of the $S$-trace contained in the state that has width at most $k$ and that witnesses $\phi$. 
Note that this is not a precise definition, since there may be many vertices in the graph $G$ that we may add to bags in the directed path of the $S$-trace contained in the state, or above, or below it. Informally speaking, we will only add vertices of $G$ to bags in the directed path of the $S$-trace contained in the state if they \emph{need} to be added to those bags. This allows us to compute a function $\ltw:\mathcal{S}\rightarrow \{\true,\false\}$ (for ``local treewidth'') that outputs $\true$ if and only if the size of any bag in the directed path is at most $k+1$. 
We will give a precise definition later. Our definition for $\ptw:\mathcal{S}\rightarrow \{\true,\false\}$ will be somewhat similar to the one by Chapelle et al.~\cite{chapelle2017treewidth}, however, our definition of $\ltw:\mathcal{S}\rightarrow \{\true,\false\}$ will be significantly more sophisticated than the one used by Chapelle et al.~\cite{chapelle2017treewidth}.

Given a state $(\tau_-,L^S, X^S, R^S,\tau_+)$ of the dynamic program, by \cref{obs:introduce,obs:forget,obs:join} we can immediately deduce the $S$-traces of the $S$-children of the bottom node of the directed path of the $S$-trace $(L^S, X^S, R^S)$. 
This gives rise to a ``preceding''-relation on states.

\begin{definition}[Preceding States]\label{def:preceding}
Let $\phi=(\tau_-,L^S, X^S, R^S,\tau_+)$ be a state with $L^S\neq\emptyset$. Denote with $\Pi$ the set of all extended $S$-operations. Then, if $\tau_-=\smallintroduce(v,d)$ or  $\tau_-=\extendedforget(v,d,f,\tau)$, the set $\Psi(\phi)$ of \emph{preceding states} of $\phi$ is defined as follows. 
\begin{itemize}
    \item If $\tau_-=\smallintroduce(v,d)$, then 
        \[
        \Psi(\phi) = \bigl\{ (\tau'_-,L^S, X^S\setminus\{v\}, R^S\cup\{v\},\smallintroduce(v,d)) \mid \tau'_- \in \Pi\bigr\}.
        \]
        \item If $\tau_-=\extendedforget(v,d,f,\tau)$, then
        \begin{align*}
        & \Psi(\phi)=\bigl\{(\tau,L^S\setminus\{v\}, X^S\cup\{v\}, R^S,\extendedforget(v,d,f,\tau))\bigr\} \text{ if } f=\true, \text{ and}\\
        & \Psi(\phi)=\bigl\{(\tau'_-,L^S\setminus\{v\}, X^S\cup\{v\}, R^S,\extendedforget(v,d,f,\tau)) \mid \tau'_- \in \Pi\bigr\} \text{ otherwise.}
        \end{align*}
\end{itemize}

If $\tau_-=\smalljoin(X^S,X_1^S,X_2^S,L_{1}^S, L_{2}^S,d)$ or  $\tau_-=\bigjoin(X^S,X_1^S,X_2^S,L_{1}^S, L_{2}^S,s,d_0,d_1,d_2)$, then the sets $\Psi_1(\phi)$ and $\Psi_2(\phi)$ of \emph{preceding states} of $\phi$ are defined as follows. 
\begin{itemize}
        \item If $\tau_-=\smalljoin(X^S,X_1^S,X_2^S,L_{1}^S, L_{2}^S,d)$, then 
        \begin{align*}
           & \Psi_1(\phi)=\bigl\{(\tau'_-,L^S_1, X_1^S, R^S\cup (X^S \setminus X_1^S)\cup L^S_2,\smalljoin(X^S,X_1^S,X_2^S,L_{1}^S, L_{2}^S,d))\\
            & \hphantom{xxxxxxxxxxxxxxxxxxxxxxxxxxxxxxxxxxxxxxxxxxxxxxxxxx}\mid \tau'_- \in \Pi\bigr\} \text{ and}\\ 
            & \Psi_2(\phi)=\bigl\{(\tau'_-,L^S_2, X_2^S, R^S\cup (X^S \setminus X_2^S)\cup L^S_1,\smalljoin(X^S,X_1^S,X_2^S,L_{1}^S, L_{2}^S,d))\\
            & \hphantom{xxxxxxxxxxxxxxxxxxxxxxxxxxxxxxxxxxxxxxxxxxxxxxxxxx}\mid \tau'_- \in \Pi\bigr\}.
        \end{align*}
        \item If  $\tau_-=\bigjoin(X^S,X_1^S,X_2^S,L_{1}^S, L_{2}^S,s,d_0,d_1,d_2)$, then 
\begin{align*}
            & \Psi_1(\phi)=\bigl\{(\tau'_-,L^S_1, X_1^S, R^S\cup (X^S \setminus X_1^S)\cup L^S_2,\bigjoin(X^S,X_1^S,X_2^S,L_{1}^S, L_{2}^S),s,d_0,d_1,d_2)\\
            & \hphantom{xxxxxxxxxxxxxxxxxxxxxxxxxxxxxxxxxxxxxxxxxxxxxxxxxx}\mid \tau'_- \in \Pi\bigr\} \text{ and}\\ 
            & \Psi_2(\phi)=\bigl\{(\tau'_-,L^S_2, X_2^S, R^S\cup (X^S \setminus X_2^S)\cup L^S_1,\bigjoin(X^S,X_1^S,X_2^S,L_{1}^S, L_{2}^S),s,d_0,d_1,d_2)\\
            & \hphantom{xxxxxxxxxxxxxxxxxxxxxxxxxxxxxxxxxxxxxxxxxxxxxxxxxx}\mid \tau'_- \in \Pi\bigr\}.
        \end{align*}
    \end{itemize}
\end{definition}
Furthermore, by \cref{cor:recursion} we know that the ``preceding''-relation on states is acyclic. Hence, we can look up the value of $\ptw$ for preceding states in our dynamic programming table when computing $\ptw(\phi)$ for some state $\phi$. 
However, \cref{def:preceding} is purely syntactic and not for every pair of a state and its preceding states there is a tree decomposition that witnesses all states at the same time. Hence, we need to check which preceding states are \emph{legal}. To this end, we introduce functions $\legal_1:\mathcal{S}\times\mathcal{S}\rightarrow\{\true,\false\}$ and $\legal_2:\mathcal{S}\times\mathcal{S}\times\mathcal{S}\rightarrow\{\true,\false\}$ that check whether for a pair of states $\phi,\psi$ or a triple of states $\phi,\psi_1,\psi_2$, respectively, the second state or the second and third state are legal predecessors for the first state, respectively. In other words, it checks whether there exists a tree decomposition that witnesses all states and their preceding-relation. We give a formal definition of this function in \cref{sec:legal}.
For the ease of presentation, we state the dynamic program as an algorithm that only decides whether the input graph has treewidth at most $k$ or not. However, the algorithm can easily be adapted to compute and output a corresponding tree decomposition. We comment more on this in \cref{sec:localtw,sec:correct}.
Formally, we define the dynamic programming table as follows.
\begin{definition}[Dynamic Programming Table $\ptw$]\label{def:dp}
The dynamic programming table is a recursive function $\ptw:\mathcal{S}\rightarrow \{\true,\false\}$ which is defined as follows. Let $\ltw:\mathcal{S}\rightarrow \{\true,\false\}$, let $\legal_1:\mathcal{S}\times\mathcal{S}\rightarrow\{\true,\false\}$, and let $\legal_2:\mathcal{S}\times\mathcal{S}\times\mathcal{S}\rightarrow\{\true,\false\}$.
    Let $\phi=(\tau_-,L^S, X^S, R^S,\tau_+)\in \mathcal{S}$ be a state, then make the following case distinction based on~$\tau_-$.
    \begin{itemize}
        \item If $\tau_-=\void$, then
        \[
        \ptw(\phi)=\ltw(\phi).
        \]
    \item If $\tau_-=\smallintroduce(v,d)$ or $\tau_-=\extendedforget(v,d,f,\tau)$, then 
        \[
        \ptw(\phi) = \bigvee_{\psi\in\Psi(\phi)} \bigl(\ltw(\phi)\wedge\ptw(\psi)\wedge\legal_1(\phi,\psi)\bigr).
        \]
        \item If $\tau_-=\smalljoin(X^S,X_1^S,X_2^S,L_{1}^S, L_{2}^S,d)$ or $\tau_-=\bigjoin(X^S,X_1^S,X_2^S,L_{1}^S, L_{2}^S,s,d_0,d_1,d_2)$, then 
        \[
        \ptw(\phi) = \bigvee_{\psi_1\in\Psi_1(\phi)\wedge\psi_2\in\Psi_2(\phi)} \bigl(\ltw(\phi)\wedge\ptw(\psi_1)\wedge\ptw(\psi_2)\wedge\legal_2(\phi,\psi_1,\psi_2)\bigr).
        \]
    \end{itemize}
\end{definition}

\subsection{Local Treewidth}\label{sec:ltwdef}
The function $\ltw$ depends on the following three sets of vertices.
Let $\phi=(\tau_-,L^S, X^S, R^S,\tau_+)$ be a state. 
\begin{itemize}
    \item $X^\phi_{\max}$ are the vertices in $V$ that we want to be in the bag of the top node $t_{\max}$ of the directed path of $(L^S, X^S, R^S)$, or, if $L^S=\emptyset$ in the bag of the $S$-child with $S$-trace $(L^S, X^S, R^S)$ of an $S$-bottom node.
    \item $X^\phi_p$ are the vertices in $V$ that we want to be in the bag of the parent of $t_{\max}$. If $t_{\max}$ is the root, we will assume that $X^\phi_p=\emptyset$.
    \item $X^\phi_{\min}$ are the vertices in $V$ that we want to be in the bag of the bottom node $t_{\min}$ of the directed path of $(L^S, X^S, R^S)$, or if $L^S=\emptyset$ then we set $X_{\min}=X^S$.
    \item $X^\phi_{\dpath}$ are additional vertices in $V\setminus S$ that we want to place in some bag of the directed path of $(L^S, X^S, R^S)$, or in some bag of a subtree rooted in a child of a node in the directed path that is not an $S$-child, or if $L^S=\emptyset$ in some bag of a node below $t_{\max}$ that is not an $S$-child.
\end{itemize}
As an example, if a vertex $v$ has neighbors both in $L^S$ and $R^S$, then there are bags in the subtree below $t_{\min}$ where $v$ already met its neighbors in $L^S$ and there must still be bags above $t_{\max}$ where $v$ meets its neighbors in $R^S$. Hence, to meet Condition~\ref{condition_3_tree_decomposition} of \cref{def:tree_decomposition}, vertex $v$ needs to be in the bags of both $t_{\max}$ and $t_{\min}$ and also every bag between them. The set $X^\phi_{\dpath}$, intuitively, contains all vertices of connected components of $G-S$ such that all their neighbors are in bags between the ones of $t_{\max}$ and $t_{\min}$, and it is not the case that all neighbors are in bags above $t_{\max}$.

In order to define $X^\phi_{\dpath}$, we also need to know which vertices are contained in the bags of the $S$-children of~$t_{\min}$ (if there are any). Let $X^\phi_c$ be the union of the bags of the $S$-children of~$t_{\min}$ and the empty set if $L^S=\emptyset$.
Computing $X^\phi_{\max}$, $X^\phi_p$, $X^\phi_{\min}$, and $X^\phi_c$ for a given state~$\phi$ is the main challenge here and we will dedicate the next section to this. 
We give a formal definition of these three sets and algorithms to compute them in \cref{sec:bags}. 
Given $X^\phi_{\max}$, $X^\phi_p$, and $X^\phi_{\min}$, we define $X^\phi_{\dpath}$ as follows.
\begin{align*}
v\in X^\phi_{\dpath}\Leftrightarrow & \text{ there exists an }F\text{-component } C \text{ in }G-(X^\phi_{\max}\cup X^\phi_p\cup X^\phi_{\min})\text{ such that}\\
& v\in V(C) \text{, } N[V(C)]\cap (X^\phi_c\setminus X^\phi_{\min})=\emptyset \text{, and } N(V(C))\setminus X^\phi_p\neq\emptyset.
\end{align*}

Using these four sets, we are ready to define the function $\ltw:\mathcal{S}\rightarrow \{\true,\false\}$.

\begin{definition}[Local Treewidth $\ltw$]\label{def:ltw}
    Let $\phi=(\tau_-,L^S, X^S, R^S,\tau_+)$ be a state, and let $G'$ be the graph obtained from $G$ by adding edges between all pairs of vertices $u,v\in  X^\phi_{\max}$ with $\{u,v\}\notin E$, all pairs of vertices $u,v\in  X^\phi_p$ with $\{u,v\}\notin E$, and all pairs of vertices $u,v\in X^\phi_{\min}$ with $\{u,v\}\notin E$. Then 
    \[
    \ltw(\phi) = \true \Leftrightarrow  \tw(G'[X^\phi_{\max}\cup X^\phi_p\cup X^\phi_{\min}\cup X^\phi_{\dpath}])\le k.
    \]
\end{definition}

\section{Computing the Top and Bottom Bag of a Directed Path}\label{sec:bags}
For this section, assume we are given a state $\phi=(\tau_-,L^S, X^S, R^S,\tau_+)$ in addition to the graph $G=(V,E)$, integer $k$, and set $S\subseteq V$. Throughout the section, let $t_{\max}$ the top node of the directed path of $(L^S, X^S, R^S)$ and let $t_{\min}$ the bottom node of the directed path of $(L^S, X^S, R^S)$.
We give algorithms to compute the following.
\begin{itemize}
    \item A candidate for the bag $X^\phi_{\max}$ of $t_{\max}$.
        \item A candidate for the bag $X_p^\phi$ of the parent of $t_{\max}$ except for the case where $\tau_+=\void$ (in this case $t_{\max}$ is the root of the tree decomposition).
    \item A candidate for the bag $X^\phi_{\min}$ of $t_{\min}$.
    \item A candidate for the union of the bags $X_c^\phi$ of the $S$-children of $t_{\min}$ except for the case where $\tau_-=\void$ (in this case $t_{\min}$ does not have any $S$-children).
    \item A candidate for the set $X^\phi_{\dpath}$ of vertices that we want to add to bags of the directed path or subtrees of the tree decomposition that are rooted at children of nodes in the directed path.
\end{itemize}

We first discuss the case where we want to compute $X^\phi_{\max}$.  We postpone the discussion for the remaining cases to \cref{sec:pieces}.
We will show how to compute a set~$X$ that we can use in the following way. 
In the case that we want to compute $X^\phi_{\max}$, the set~$X$ will either be $X^\phi_{\max}$, or the bag of the parent of $t_{\max}$ (if it exists). Informally speaking, $X$ will be the bag of the parent, if the parent (potentially) has a larger bag. This can happen if $\tau_+=\smallintroduce(v,d)$, $\tau_+=\smalljoin(X^S,X_1^S,X_2^S,L_{1}^S, L_{2}^S,d)$, or $\tau_+=\bigjoin(X^S,X_1^S,X_2^S,L_{1}^S, L_{2}^S,s,d_0,d_1,d_2)$.
In the case of $\tau_+=\extendedforget(v,d,f,\tau)$, the set $X$ will be the bag of $t_{\max}$.
The set $X$ that we compute is the bag of the nodes visualized with a thick circle in the illustrations in \cref{fig:bigops,fig:smallops,fig:extforget}. 
In the case where $\tau_+=\void$, we will show that we can set $X=X^\phi_{\max}=X^S$. We will give two algorithms, one for the case that $\tau_+=\bigjoin(X^S,X_1^S,X_2^S,L_{1}^S, L_{2}^S,s,d_0,d_1,d_2)$ and one for all other cases.
 
Note that by \cref{obs:sbottom} we have that the parent of $t_{\max}$ is an $S$-bottom node. It is the node that admits the extended $S$-operation $\tau_+$. Let that node be called $t^\star$.
Observe that in all cases except $\tau_+=\void$, the information in the extended $S$-operations allows us to determine from $X$ (and hence also $X^\phi_{\max}$) the bag of $t^\star$, the bag of the parent of $t^\star$, and the bags of the $S$-children of $t^\star$. 

A crucial property of our algorithms is that the computation of $X$ is independent of the $S$-operation $\tau_-$ (except for one special case that we will discuss later). 
This allows us to do the following: Instead of computing $X^\phi_{\min}$, we compute the set $X^\psi_{\max}$ for some $\psi\in\Psi(\phi)$ or some $\psi\in\Psi_1(\phi)$ if $\tau_-$ is a $\smalljoin$ or a $\bigjoin$ $S$-operation (it will become clear that picking some $\psi\in\Psi_2(\phi)$ will yield the same result).

In order to prove that our algorithms are correct, we will exploit the additional properties of \slim \topheavy $S$-nice tree decompositions (in comparison to ``normal'' $S$-nice tree decompositions). 
In \cref{sec:bigops}, we present the algorithm for the case where $\tau_+$ is $S$-operations $\bigjoin$.
Afterwards, we present the algorithm for the case where $\tau_+$ is different from $S$-operation $\bigjoin$ in \cref{sec:smallops}. These algorithms implicitly define the functions $\bigbag$, $\smallbag$, $\nbag$, $\smallbagf$, and $\nbagf$ used in \cref{sec:extendedops}.

\subsection{Big \boldmath$S$-Operations}\label{sec:bigops}
To make notation more convenient, assume that the input state is $\phi=(\tau_-,\hat{L}^S, \hat{X}^S, \hat{R}^S,\tau_+)$. 
Our goal is to compute the bag $X^\phi_{\max}$  for states $\phi$ where $\tau_+$ is a $\bigjoin$ $S$-operation or $\tau_+=\extendedforget(v,d,\true,\tau)$ and $\tau$ is a $\bigjoin$ $S$-operation. In this case we say that $\tau_+$ is a big $S$-operation. 

\subparagraph{Setup.} As described at the beginning of \cref{sec:bags}, we compute a set $X$ of vertices that is a candidate for a \emph{big target node} (or just \emph{target node}) $t^\star$, which is not necessarily the same node as $t_{\max}$. 
Let $\mathcal{T}$ be a \slim $S$-nice tree decomposition of $G$ that contains a directed path for $S$-trace $(\hat{L}^S, \hat{X}^S, \hat{R}^S)$, or if $\hat{L}^S=\emptyset$ (and $\tau_-=\void$) let $\mathcal{T}$ be a \slim $S$-nice tree decomposition of $G$ that contains an $S$-bottom node that has an $S$-child (which we also call~$t_{\max}$) with $S$-trace $(\hat{L}^S, \hat{X}^S, \hat{R}^S)$. 
The \emph{big target node} ${t^\star}$ is defined as follows.
\begin{itemize}
\item If $\tau_+=\bigjoin(X^S_t,X_1^S,X_2^S,L_{1}^S, L_{2}^S,s,d_0,d_1,d_2)$, then $t^\star$ is the parent of $t_{\max}$.
\item If $\tau_+=\extendedforget(v,d,\true,\tau)$ and $\tau=\bigjoin(X^S_t,X_1^S,X_2^S,L_{1}^S, L_{2}^S,s,d_0,d_1,d_2)$, then $t^\star=t_{\max}$.
\end{itemize}
It follows that the target node ${t^\star}$ is always an $S$-bottom node that admits some $S$-operation $\bigjoin(X^S_t,X_1^S,X_2^S,L_{1}^S, L_{2}^S,s,d_0,d_1,d_2)$. Let from now on $s,d_0,d_1,d_2$ be the bit string and the three integers, respectively, that are contained in the $S$-operation of the target node. We describe an algorithm to compute the function $\bigbag(X_{t^\star}^S,s)$ and use it to compute $X$, that is $X=\bigbag(X_{t^\star}^S,s)$. To this end, we need to know the set $X_{t^\star}^S$ from the $S$-trace of ${t^\star}$.
To make notation simpler, we drop the subscript ${t^\star}$ and let $(L^S, X^S, R^S)$ be the $S$-trace of the target node ${t^\star}$. 
\begin{itemize}
    \item 
    Let $\tau_+=\bigjoin(X^S_t,X_1^S,X_2^S,L_{1}^S, L_{2}^S,s,d_0,d_1,d_2)$. Then by \cref{def:bigjoin} and \cref{obs:join} we either have that 
\begin{itemize}
\item $\hat{L}^S=L_{1}^S$, then $(L^S, X^S, R^S)=(\hat{L}^S\cup L_{2}^S, X^S_t, \hat{R}^S\setminus (L_{2}^S\cup X^S_t))$, or we have that 
\item $\hat{L}^S=L_{2}^S$, then $(L^S, X^S, R^S)=(\hat{L}^S\cup L_{1}^S, X^S_t, \hat{R}^S\setminus (L_{1}^S\cup X^S_t))$.
\end{itemize}    
\item Let $\tau_+=\extendedforget(v,d,\true,\tau)$ and $\tau=\bigjoin(X^S_t,X_1^S,X_2^S,L_{1}^S, L_{2}^S,s,d_0,d_1,d_2)$. Then $(L^S, X^S, R^S)=(\hat{L}^S, \hat{X}^S, \hat{R}^S)$.  
\end{itemize}

\subparagraph{Algorithm.}
Now, we describe the algorithm to compute $\bigbag(X^S,s)$. 
The algorithm will maintain a 3-partition $(\Rest,\In,\Out)$ of the vertex set $V$, where, intuitively,
\begin{itemize}
    \item $\In$ contains vertices that we want to be in $X$,
    \item $\Out$ contains vertices that we do not want to be in $X$, and
    \item $\Rest$ contains vertices where we have not decided yet whether we want them to be in $X$ or not.
\end{itemize}

\begin{algorithm}[$\bigbag$]{A set $X^S\subseteq S$, and a bit string $s$.}{A set $X\subseteq V$.}\label{alg:big}
Set $\In=X^S$, $\Out=S\setminus X^S$, and $\Rest=V\setminus S$. Let $i$ be an integer variable. Set $i$ to one. 

Perform the first applicable of the following steps until $\Rest$ is empty or $|\In|=k+1$.  Break all ties in an arbitrary but fixed way.
\begin{enumerate}
\item If there is a vertex $v$ in $\Rest$ that has degree at most one in $G[\Rest]$, such that at least two $S$-components in $G[\Out]$ contain some neighbor of $v$, then do the following.\label{big:step1}
\begin{itemize}
	\item If the $i$th bit of $s$ is one, then remove $v$ from $\Rest$ and add~$v$ to~$\In$. Increase $i$ by one.
	\item 
	Otherwise, remove $v$ from $\Rest$ and add~$v$ to~$\Out$. Increase $i$ by one.
\end{itemize} 
\item If there is a vertex $v$ in $\Rest$ that has degree at most one in $G[\Rest]$, then remove $v$ from $\Rest$ and add~$v$ to~$\Out$.\label{big:step2}
\end{enumerate}
\vspace{-2ex}
Output $\In$.
\end{algorithm}

Since $G[\Rest]$ is initially a forest and we never add new vertices to $\Rest$, we have that as long as $\Rest$ is non-empty, at least one of the above two steps of \cref{alg:big} is applicable. 
We set $X=\bigbag(X^S,s)$. We can observe that the computation takes polynomial time.
\begin{observation}\label{obs:algo1running}
\cref{alg:big} runs in polynomial time.
\end{observation}

\subparagraph{Correctness.}
Now we show that \cref{alg:big} is correct. To this end, we give a proof for \cref{lem:alwaysbigjoin}.
%
%
We first make the following observation.

\begin{lemma}\label{lem:oneneighborbig}
During every iteration of \cref{alg:big}, the following holds. For every connected component $C$ in $G[\Out\setminus S]$ we have that $|N(V(C))\cap \Rest|\le 1$.
\end{lemma}
\begin{proof}
We prove this by induction on the iterations of the algorithm. Initially, the statement clearly holds since $\Out\setminus S=\emptyset$. Let $(\Rest,\In,\Out)$ be a 3-partition of $V$ that appears during some iteration of the algorithm. By induction hypothesis, we have that for every connected component $C$ in $G[\Out\setminus S]$ we have that $|N(V(C))\cap \Rest|\le 1$.
Assume that Step~\ref{big:step1} applies, then the new 3-partition is $(\Rest\setminus \{v\},\In\cup\{v\},\Out)$ or $(\Rest\setminus \{v\},\In,\Out\cup\{v\})$ for some $v\in \Rest$. In the former case, we clearly have for every connected component $C$ in $G[\Out\setminus S]$ that $|N(V(C))\cap (\Rest\setminus\{v\})|\le |N(V(C))\cap \Rest|\le 1$.
In the latter case or if Step~\ref{big:step2} applies, then the new 3-partition is $(\Rest\setminus \{v\},\In,\Out\cup\{v\})$ for some $v\in \Rest$ that has degree at most one in $G[\Rest]$. Let $C$ be a connected component in $G[(\Out\cup\{v\})\setminus S]$. If $C$ does not contain $v$, then $C$ is a connected component in $G[\Out\setminus S]$, and clearly we have that $|N(V(C))\cap (\Rest\setminus\{v\})|\le |N(V(C))\cap \Rest|\le 1$. Now assume that $C$ contains $v$. Let $C_1,\ldots,C_\ell$ be the connected components of $C-\{v\}$. Note that $C_1,\ldots,C_\ell$ are all connected components in $G[\Out\setminus S]$ and hence only have one neighbor in $\Rest$. Since $v\in\Rest$ and $v\in N(V(C_i))$ for all $1\le i\le \ell$, we have that $|N(V(C_i))\cap (\Rest\setminus\{v\})|=0$ for all $1\le i\le \ell$. Recall that $|N(v)\cap\Rest|\le 1$. It follows that $|N(V(C))\cap (\Rest\setminus\{v\})|\le 1$. 
\end{proof}

Now we show the following result, which directly implies \cref{lem:alwaysbigjoin}.

\begin{lemma}\label{lem:bigcorrect1}
Let $\mathcal{T}=(T,\{X_t\}_{t\in V(T)})$ be a \slim $S$-nice tree decomposition of $G$. Let $T'$ be a full join tree in $\mathcal{T}$ where all nodes in $T'$ have bag $X$. Then there exist a bit string $s\in \{0,1\}^{2\fvn(G)+1}$ such that $\bigbag(X\cap S,s)=X$.
\end{lemma}
\begin{proof}
To prove the statement, we first show that whenever on input $(X^S,s)$ \cref{alg:big} moves a vertex $v$ into the set $\Out$ in Step~\ref{big:step2}, then for all \slim $S$-nice tree decompositions $\mathcal{T}$ of $G$ contain a full join tree $T'$ where all nodes have bag $X$ and $X\cap S=X^S$, we have that $v\notin X$. We prove this by induction on the iterations of the algorithm.

Let $(\Rest,\In,\Out)$ be a 3-partition of $V$ that appears during some iteration of \cref{alg:big} on input $(X^S,s)$. By induction hypothesis, we have that $\In\subseteq X$ and $\Out\cap X=\emptyset$. Note that initially, this is clearly true.
Let $v\in \Rest$ be a vertex that will be added to $\Out$ in Step~\ref{big:step2} of \cref{alg:big}.
Assume for contradiction that there is a \slim $S$-nice tree decompositions $\mathcal{T}$ of $G$ that contains a full join tree $T'$ where all nodes have bag $X$ with $X\cap S=X^S$ such that $v\in X$. We have that the degree of $v$ in $\Rest$ is at most one. Since Step~\ref{big:step2} of \cref{alg:big} is not applicable, we have that there is at most one $S$-component $C$ in $G[\Out]$ such that $v\in N(V(C))$.

We estimate the number of $S$-components $C_i$ in $G-X$ such that $v\in N(C_i)$ Let $C_1$ be the $S$-component that contains $C$. Now consider the neighbors of $v$ outside of $X$. We can categorize then as follows. We know that $N(v)\cap\Rest=\{u\}$ for some $u\in \Rest$. Let $C_2$ denote the connected component in $G_X$ that contains $u$. Let $C_i$ with $3\le i\le\ell$ denote the $F$-components in $G[\Out]$ such that $v\in N(C_i)$ for $3\le i\le\ell$. By \cref{lem:oneneighborbig} we have that $|N(C_i)\cap \Rest|\le 1$ for all $3\le i\le\ell$ which means that $N(C_i)\cap \Rest=\{v\}$. If follows that all $C_i$ with $3\le i\le\ell$ are $F$-components in $G-X$. We can conclude that there are at most two $S$-components in $G-X$, namely $C_1$ and $C_2$ such that $v$ in contained in their neighborhood. This is a contradiction to the assumption that the tree decomposition $\mathcal{T}$ is \slim. In particular, it contradicts Condition~\ref{cond:slim:6} of \cref{def:slimtd}.

It follows that for every \slim $S$-nice tree decomposition $\mathcal{T}=(T,\{X_t\}_{t\in V(T)})$ of $G$ and every full join tree $T'$ in $\mathcal{T}$ where all nodes in $T'$ have bag $X$, there exist a bit string $s\in \{0,1\}^{\ell}$ for some $\ell \le n$ such that $\bigbag(X\cap S,s)=X$. It remains to show that $\ell\le 2\fvn(G)+1$. Note that in every execution of Step~\ref{big:step1} of \cref{alg:big} either the size of $\In$ is increased by one, or the number of $S$-components in $G[\Out]$ is decreased by one. We know that $|X|=k+1$ and $k\le|S|=\fvn(G)$. Furthermore, since by definition, the vertex sets of $S$-components are disjoint, we can have at most $|S|=\fvn(G)$ different $S$-components. It follows that after at most $2\fvn(G)+1$ executions of Step~\ref{big:step1} of \cref{alg:big} on input $(X_S,s)$ we have that $\In=X$. Hence, only the first $2\fvn(G)+1$ bits of $s$ are relevant for the computation. 
\end{proof}

It is straightforward to see that \cref{lem:bigcorrect1} directly implies \cref{lem:alwaysbigjoin}.


\subsection{Small \boldmath$S$-Operations}\label{sec:smallops}
As in the previous section, to make notation more convenient, assume that the input state is $\phi=(\tau_-,\hat{L}^S, \hat{X}^S, \hat{R}^S,\tau_+)$. 
Our goal is to compute the bag $X^\phi_{\max}$ for states $\phi$ where $\tau_+$ is not a big $S$-operation, where $\tau_+=\extendedforget(v,d,f,\tau)$ with $f=\false$ (by \cref{def:extendedforget}, we then also always have that $\tau=\void$), and where $\tau_+=\extendedforget(v,d,f,\tau)$ with $f=\true$ and $\tau$ is not a big $S$-operation. 
In this case we say that $\tau_+$ is a small $S$-operation. 

\subparagraph{Setup.} As described at the beginning of \cref{sec:bags}, we compute a set $X$ of vertices that is a candidate for a \emph{small target node} (or just \emph{target node}) ${t^\star}$, which is not necessarily the same node as $t_{\max}$. Formally, the target node is defined as follows.
Let $\mathcal{T}$ be a \slim $S$-nice tree decomposition of $G$ that contains a directed path for $S$-trace $(\hat{L}^S, \hat{X}^S, \hat{R}^S)$, or if $\hat{L}^S=\emptyset$ (and $\tau_-=\void$) let $\mathcal{T}$ be a \slim $S$-nice tree decomposition of $G$ that contains an $S$-bottom node that has an $S$-child (which we also call~$t_{\max}$) with $S$-trace $(\hat{L}^S, \hat{X}^S, \hat{R}^S)$.
We assume that neither $t_{\max}$ nor the parent of $t_{\max}$ are part of a full join tree, as in this case is handled in \cref{sec:bigops}.
The \emph{small target node} ${t^\star}$ is defined as follows.
\begin{itemize}
\item If $\tau_+=\smallintroduce(v,d)$ or $\tau_+=\smalljoin(X^S_t,X_1^S,X_2^S,L_{1}^S, L_{2}^S,d)$, then ${t^\star}$ is the parent of $t_{\max}$.
\item If $\tau_+=\extendedforget(v,d,f,\tau)$ and $\tau$ is not a $\bigjoin$ $S$-operation, then ${t^\star}=t_{\max}$.
\end{itemize}

If $\tau_-=\void$, then $t_{\min}$ is not defined.
Otherwise, \cref{def:slimtd} (the definition of \slim) implies that either $t_{\min}=t_{\max}=t^{\star}$, this happens if $\tau_+=\extendedforget(v,d,f,\tau)$ and bag of $t_{\max}$ is full, or we have that $t_{\max}$ and $t_{\min}$ are different, 
and $t_{\max}$ is not the parent of $t_{\min}$.
In the former case, we have also that $X^\phi_{\max}=X^\phi_{\min}$ hence, we can use the computation of $X^\phi_{\min}$ to obtain $X^\phi_{\max}$. This is explained in \cref{sec:pieces} and uses information in $\tau_-$. This is also the case mentioned in the beginning of the section where the computation of $X^\phi_{\max}$ depends on $\tau_-$.
From now on, we assume that we are not in this case, that is, we assume that if $t_{\min}$ is defined, then $t_{\max}$ and $t_{\min}$ are different 
and $t_{\max}$ is not the parent of $t_{\min}$.

If $t_{\min}$ is defined and $\tau_+\neq\smalljoin(X^S_t,X_1^S,X_2^S,L_{1}^S, L_{2}^S,d)$, then we denote with $t^{\star}_{d_1}$ the closest descendant of $t^\star$ that is parent of an $S$-bottom node. 
If $t_{\min}$ is defined and $\tau_+=\smalljoin(X^S_t,X_1^S,X_2^S,L_{1}^S, L_{2}^S,d)$, then we denote with $t^{\star}_{d_1}$ and $t^{\star}_{d_2}$ the closest descendants of the two $S$-children of $t^\star$, respectively, that are parent of an $S$-bottom node. 
Furthermore, we denote by $t^{\star}_a$ the closest ancestor of $t^\star$ that is parent of an $S$-bottom node. Note that we have that $t^{\star}_{d_1}$ or $t^{\star}_{d_2}$ is the parent of~$t_{\min}$, or, in other words, that the parent of $t_{\min}$ is not $t^\star$. Furthermore, note that in the case of $\tau_+=\smallintroduce(v,d)$ and $\tau_+=\smalljoin(X^S_t,X_1^S,X_2^S,L_{1}^S, L_{2}^S,d)$ we have that~$t^\star$ is an $S$-bottom node and hence $t^{\star}_a$ is the parent of $t^\star$.

To make notation simpler, we drop the subscript ${t^\star}$ and let $(L^S, X^S, R^S)$ be the $S$-trace of the target node ${t^\star}$. We get the following using \cref{obs:introduce,obs:forget,obs:join}.
\begin{itemize}
    \item Let $\tau_+=\smallintroduce(v,d)$. Then $(L^S, X^S, R^S)=(\hat{L}^S, \hat{X}^S\cup\{v\}, \hat{R}^S\setminus\{v\})$.
    \item Let $\tau_+=\smalljoin(X^S_t,X_1^S,X_2^S,L_{1}^S, L_{2}^S,d)$. Then by \cref{def:smalljoin} and \cref{obs:join} we either have that 
\begin{itemize}    
    \item $\hat{L}^S=L_{1}^S$, then $(L^S, X^S, R^S)=(\hat{L}^S\cup L_{2}^S, X^S_t, \hat{R}^S\setminus (L_{2}^S\cup X^S_t))$, or we have that 
    \item $\hat{L}^S=L_{2}^S$, then $(L^S, X^S, R^S)=(\hat{L}^S\cup L_{1}^S, X^S_t, \hat{R}^S\setminus (L_{1}^S\cup X^S_t))$.
    \end{itemize}
    \item If $\extendedforget(v,d,f,\tau)$ and $\tau$ is not a $\bigjoin$ $S$-operation, then $(L^S, X^S, R^S)=(\hat{L}^S, \hat{X}^S, \hat{R}^S)$.
\end{itemize}

Furthermore, we define sets $L_{1}^S$ and $L_{2}^S$ as follows. 
\begin{itemize}
\item If $\tau_+=\smalljoin(X^S_t,X_1^S,X_2^S,L_{1}^S, L_{2}^S,d)$, then $L_{1}^S$ and $L_{2}^S$ are the respective sets in~$\tau_+$.
\item If $\tau_+=\extendedforget(v,d,f,\tau)$ with $S$-operation $\tau=\smalljoin(X^S_t,X_1^S,X_2^S,L_{1}^S, L_{2}^S,d)$, then $L_{1}^S$ and $L_{2}^S$ are the respective sets in $\tau$.

\item In all other cases, we set $L_{1}^S=L^S$ and $L_{2}^S=\emptyset$.
\end{itemize}


Finally, we can observe the following useful properties of target nodes.
\begin{observation}\label{obs:smallneighbors}
Let $\mathcal{T}=(T,\{X_t\}_{t\in V(T)})$ be a \slim $S$-nice tree decomposition of $G$ with width $k$ that contains a directed path of length at least three with $S$-trace $(\hat{L}^S, \hat{X}^S, \hat{R}^S)$, or if $\hat{L}^S=\emptyset$ then $\mathcal{T}$ contains an $S$-bottom node that has an $S$-child with $S$-trace $(\hat{L}^S, \hat{X}^S, \hat{R}^S)$, such that neither $t_{\max}$ nor the parent of $t_{\max}$ are contained in a full join tree. Let ${t^\star}$ denote the target node in $\mathcal{T}$.
\begin{enumerate}
	\item Node ${t^\star}$ has $S$-trace $(L^S, X^S, R^S)$. 
	\item The bag $X_{t_p}$ of the parent $t_p$ of ${t^\star}$ is not full and  
    $X_{t_p}=X_{t^\star_a}\subseteq X_{t^\star}$.
	\item If ${t^\star}$ has exactly one $S$-child $t_1$, then the bag $X_{t_1}$ of $t_1$ is not full and $X_{t_1}\subseteq X_{t^\star}$. 
	\item If ${t^\star}$ has exactly two $S$-children $t_1, t_2$, then 
    the bag $X_{t_1}$ of $t_1$ is not full,
	the bag $X_{t_2}$ of $t_2$ is not full, $X_{t_1}\subseteq X_{t^\star}$, and $X_{t_2}\subseteq X_{t^\star}$. 
\end{enumerate}
\end{observation}
\begin{proof}
Properties one to four follow from the definition of small target nodes and \cref{def:smallintro,def:smalljoin,def:extendedforget,def:slimtd}. For the second property, we need Condition~\ref{cond:slim:3} of \cref{def:slimtd}: Note that if $t_p=t^\star_a$, then we obviously have $X_{t_p}=X_{t^\star_a}$. If $t_p\neq t^\star_a$, then by the definition of the small target node $t^\star$, we have that $t_p$ is an $S$-bottom node that admits $S$-operation $\extendedforget(v,d,f,\tau)$ and $t^\star_a$ is the parent of $t_p$. From Condition~\ref{cond:slim:3} of \cref{def:slimtd} it follows that $X_{t_p}=X_{t^\star_a}$.
To see that $X_{t_p}\subseteq X_{t^\star}$ consider the following: If the bag of $t^\star$ is full, then this follows immediately. If the bag of bag of $t^\star$ is not full and $t^\star$ is an $S$-bottom node, then it follows from Condition~\ref{cond:slim:3} of \cref{def:slimtd}. Finally, if $t^\star$ is not an $S$-bottom node, then by the definition of the small target node we have that $t_p$ is an $S$-bottom node that admits $S$-operation $\extendedforget(v,d,f,\tau)$ and $X_{t_p}\cup\{v\}= X_{t^\star}$.
For the third property consider two cases: We can have that $\tau_+=\smallintroduce(v,d)$ and that $t^\star=t_p$, then the $S$-child of $t^\star$ and we have $X_{t_1}= X_{t^\star}\setminus \{v\}$ with $v\in X_{t^\star}$.
The other case is that we have  $\tau_+=\extendedforget(v,d,f,\tau)$ and $t^\star=t_{\max}$. By the assumption described at the start of the subsection we have that then the bag of $t_{\max}$ is not full.
By Condition~\ref{cond:slim:35} of \cref{def:slimtd} we then have that $X_{t_1}= X_{t^\star}$. The fourth property follows directly Condition~\ref{cond:slim:2}.
\end{proof}

\subparagraph{Algorithm.}
Now, we describe the algorithm to compute $X$. The algorithm, informally speaking, works as follows. 
It uses the additional integer $d$ from the extended $S$-operation to make an initial ``guess'' of up to three vertices which are contained in the bag $X$ of $t^\star$, but not in the bag of the parent or the $S$-children, respectively. By \cref{obs:smallneighbors} we know that if the bag $X$ is full, then those vertices must exist. Using this guess, we compute a candidate for bag $X$ in a greedy fashion.
The guess also lets us determine candidates for the bags of the parent of $t^\star$ and its $S$-children.
We can prove that there is a tree decomposition for $G$ with width $k$ using these candidates which is $S$-nice and \slim, but this tree decomposition is not necessarily \topheavy. However, {\topheavy}ness is crucial for proving correctness of the algorithm. To this end, the algorithm modifies the candidate sets such that there is a tree decomposition for $G$ with width $k$ using the modified sets which is $S$-nice, \slim, and \topheavy for $t^\star_a$. 

We use the integer $d$ to obtain the initial guess as follows. Let $\mathcal{G}=(V\cup\{\bot\})^3$ be the set of all triples of vertices together with a special symbol $\bot$. Assume that the triples in $\mathcal{G}$ are ordered in some arbitrary but fixed way. Note that there are $(n+1)^3$ different triples. Let $(g_1,g_2,g_3)$ be the $d$th triple in $\mathcal{G}$ according to the ordering. We define the sets $D_R,D_{L_1},D_{L_2}$ as follows.
\begin{itemize}
\item If $g_1\neq\bot$, then we set $D_R=\{g_1\}$. Otherwise, we set $D_R=\emptyset$.
\item If $g_2\neq\bot$, then we set $D_{L_1}=\{g_2\}$. Otherwise, we set $D_{L_1}=\emptyset$.
\item If $g_3\neq\bot$, then we set $D_{L_2}=\{g_3\}$. Otherwise, we set $D_{L_2}=\emptyset$.
\end{itemize}

\begin{algorithm}[$\smallbag$]{Sets $R^S,X^S,L_1^S,L_2^S\subseteq S$, and $D_R,D_{L_1},D_{L_2}\subseteq V$.}{Sets $X,X_R,X_{L_1},X_{L_2}\subseteq V$.}\label{alg:smallwrapper}
Set $X=\smallbagc(R^S,X^S,L_1^S,L_2^S,D_R,D_{L_1},D_{L_2})$. 
Set $X_R=X\setminus D_R$, $X_{L_1}=X\setminus D_{L_1}$, and $X_{L_2}=X\setminus D_{L_2}$. 
Perform the first applicable step until no changes occur:
\begin{enumerate}

\label{wrapper2aa}
\item If there is an $F$-component $C$ in $G-X_R$, then remove all vertices in $V(C)$ from the set $X$, $X_R$, $X_{L_1}$, and $X_{L_2}$.\label{wrapper2a}
\item If there is $v\in X_R\setminus S$ such that for every $u\in N(v)$ it holds that $u\in X_R$ or $u$ is not connected to $L_1^S\cup L_2^S\cup((X\setminus X_R)\cap S)$ in $G-X_R$, then remove $v$ from $X$, $X_R$, $X_{L_1}$, and $X_{L_2}$.\label{wrapper2b}
\item If there is $v\in X_R\setminus S$ and $u\in N(v)\setminus (S\cup X_R)$ such that for every $u'\in N(v)\setminus\{u\}$ it holds that $u'\in X_R$ or $u'$ is not connected to $L_1^S\cup L_2^S\cup((X\setminus X_R)\cap S)$ in $G-X_R$, and $u$ is connected to $L_1^S\cup L_2^S\cup((X\setminus X_R)\cap S)$ in $G-X_R$, then replace $v$ with $u$ in all sets $X$, $X_R$, $X_{L_1}$, and $X_{L_2}$.\label{wrapper2c}
\end{enumerate}  
Output $X$, $X_R$, $X_{L_1}$, and $X_{L_2}$.
\end{algorithm}


Note that we slightly overload notation here in comparison to the definition of the function $\smallbag$ in \cref{sec:extendedops} (instead of an integer $d$, it takes the ``decoded'' dets $D_R, D_{L_1}, D_{L_2}$ as inputs).

Now we give a description of the subroutine $\smallbagc$ called in the beginning of \cref{alg:smallwrapper}.
Similar to \cref{alg:big}, the subroutine will maintain a 3-partition $(\Rest,\In,\Out)$ of the vertex set $V$, where, intuitively,
\begin{itemize}
    \item $\In$ contains vertices that we want to be in $X$,
    \item $\Out$ contains vertices that we do not want to be in $X$, and
    \item $\Rest$ contains vertices where we have not decided yet whether we want them to be in $X$ or not.
\end{itemize}
Furthermore, using the sets $L_1^S$, $L_2^S$, and $R^S$, we introduce the following terminology. 
Let $C$ be a connected component in $G[V']$ for some $V'\subseteq V\setminus S$. 
\begin{itemize}
    \item If we have $N(V(C))\cap L_1^S=\emptyset$, $N(V(C))\cap L_2^S=\emptyset$, and $N(V(C))\cap R^S=\emptyset$, 
    then we call~$C$ an \emph{$F$-component}. 
    
    Note that this is not equivalent to the definition of $F$-component in \cref{sec:useful}, but for convenience of presentation, we overwrite the terminology for this section.
    All occurrences of ``$F$-component'' from now on in this section refer to the above definition.
    \item If we have $N(V(C))\cap L_1^S\neq\emptyset$, $N(V(C))\cap L_2^S=\emptyset$, and $N(V(C))\cap R^S=\emptyset$, 
    then we call~$C$ an \emph{$L_1$-component}.
    \item If we have $N(V(C))\cap L_1^S=\emptyset$, $N(V(C))\cap L_2^S\neq\emptyset$, and $N(V(C))\cap R^S=\emptyset$, 
    then we call~$C$ an \emph{$L_2$-component}.
    \item If we have $N(V(C))\cap L_1^S=\emptyset$, $N(V(C))\cap L_2^S=\emptyset$, and $N(V(C))\cap R^S\neq\emptyset$, 
    then we call~$C$ an \emph{$R$-component}.
    \item If none of the above applies, then we call $C$ a \emph{mixed component}.
\end{itemize}
It is easy to see that every connected component of $G[V']$ falls into one of the above categories, for every choice of $V'\subseteq V\setminus S$.

\begin{algorithm}[$\smallbagc$]{Sets $R^S,X^S,L_1^S,L_2^S\subseteq S$, and $D_R,D_{L_1},D_{L_2}\subseteq V$.}{A set $X\subseteq V$.}\label{alg:small}
Set $\In=X^S\cup D_R\cup D_{L_1}\cup D_{L_2}$, $\Out=S\setminus \In$, and $\Rest=V\setminus (\In\cup\Out)$.

Perform the first applicable of the following steps until $\Rest$ is empty.  Break all ties in an arbitrary but fixed way.
\begin{enumerate}
\item If there is a vertex $v$ in $\Rest$ that is contained in a mixed component in $G[(\Out\cup\{v\})\setminus S]$, then remove $v$ from $\Rest$ and add~$v$ to~$\In$.\label{regular:step1}
\item If there is a vertex $v$ in $\Rest$ that is contained in an $H$-component $C_H$ in $G[(\Out\cup\{v\})\setminus S]$ for $H\in\{L_1,L_2,R\}$ and $N(C_H)\cap D_H\neq \emptyset$, then remove $v$ from $\Rest$ and add~$v$ to~$\In$.\label{regular:step2}
\item If there is a vertex $v$ in $\Rest$ that is contained in an $R$-component in $G[(\Out\cup\{v\})\setminus S]$ and that has degree at most one in $G[\Rest]$, then remove $v$ from $\Rest$ and add~$v$ to~$\Out$.\label{regular:step3}
\item If there is a vertex $v$ in $\Rest$ that is contained in an $H$-component in $G[(\Out\cup\{v\})\setminus S]$ for $H\in\{L_1,L_2\}$ and that has degree at most one in $G[\Rest]$, then remove $v$ from $\Rest$ and add~$v$ to~$\Out$.\label{regular:step4}
\item If there is a vertex $v$ in $\Rest$ that has degree one in $G[\Rest]$ and the connected component~$C$ of $G[(\Out\cup\{v\})\setminus S]$ that contains $v$ has the property that $N(V(C))\cap \In\neq \In$, then remove $v$ from $\Rest$ and add~$v$ to~$\Out$.\label{regular:step5}
\item If there is a vertex $v$ in $\Rest$ that has degree at most one in $G[\Rest]$, then remove $v$ from $\Rest$ and add~$v$ to~$\Out$.\label{regular:step6}
\end{enumerate}
Output $\In$.
\end{algorithm}

Since $G[\Rest]$ is initially a forest and we never add new vertices to $\Rest$, we have that as long as $\Rest$ is non-empty, at least one of the above four steps of \cref{alg:small} is applicable. 
We can observe that the computation takes polynomial time.
\begin{observation}\label{obs:algo2running}
\cref{alg:smallwrapper} runs in polynomial time.
\end{observation}

\subparagraph{Correctness.}
Now we show that \cref{alg:smallwrapper} is correct. 
Assume there exists a \slim $S$-nice tree decomposition $\mathcal{T}=(T,\{X_t\}_{t\in V(T)})$ of $G$ such that the following holds.
\begin{itemize}
    \item $\mathcal{T}$ contains a directed path of length at least three with $S$-trace $(\hat{L}^S, \hat{X}^S, \hat{R}^S)$, or if $\hat{L}^S=\emptyset$ then $\mathcal{T}$ contains an $S$-bottom node that has an $S$-child (called~$t_{\max}$) with $S$-trace $(\hat{L}^S, \hat{X}^S, \hat{R}^S)$, such that neither $t_{\max}$ nor the parent of $t_{\max}$ are contained in a full join tree.
    \item $X^S\cup D_R\cup D_{L_1}\cup D_{L_2}\subseteq X_{t^\star}$ and $X_{t^\star}\cap (S\setminus X^S)=\emptyset$, 
    \item if $t^\star$ has a parent $t_p$, then $X_{t_p}=X_{t^\star}\setminus D_R$, 
    \item if $t^\star$ has one $S$-child $t_1$, then $X_{t_1}=X_{t^\star}\setminus D_{L_1}$, and 
    \item if $t^\star$ has two $S$-children $t_1$ and $t_2$, then $X_{t_1}=X_{t^\star}\setminus D_{L_1}$ and $X_{t_2}=X_{t^\star}\setminus D_{L_2}$.
\end{itemize} 
Then we say that the top operation of $\phi$ \emph{contains the correct guess} for $\mathcal{T}$.
Now we show the following. 

\begin{proposition}\label{prop:smallcorrect2}
Assume there exists a \slim $S$-nice tree decomposition $\mathcal{T}=(T,\{X_t\}_{t\in V(T)})$ of $G$ with width $k$ that contains a directed path of length at least three with $S$-trace $(\hat{L}^S, \hat{X}^S, \hat{R}^S)$, or if $\hat{L}^S=\emptyset$ then $\mathcal{T}$ contains an $S$-bottom node that has an $S$-child with $S$-trace $(\hat{L}^S, \hat{X}^S, \hat{R}^S)$, such that neither $t_{\max}$ nor the parent of $t_{\max}$ are contained in a full join tree, and with target node ${t^\star}$ such that the following holds.
\begin{itemize}
\item $|X_{t^\star}|$ is minimal, that is, for all siblings $\mathcal{T}'$ of $\mathcal{T}$ such that all of the above holds and if $\mathcal{T}$ is \topheavy for a node $\hat{t}$, then the subtrees rooted at $\hat{t}$ in $\mathcal{T}$ and $\mathcal{T}'$ are the same, we have that $|X_{t^\star}|\le |X'_{t^\star}|$,
\item if $t^\star_{d_1}$ and $t^\star_{d_2}$, respectively, exist, then $\mathcal{T}$ is \topheavy for $t^\star_{d_1}$ and $t^\star_{d_2}$, respectively, 
    \item $X^S\cup D_R\cup D_{L_1}\cup D_{L_2}\subseteq X_{t^\star}$ and $X_{t^\star}\cap (S\setminus X^S)=\emptyset$, and
    \item $\phi$ contains the correct guess for $\mathcal{T}$.
\end{itemize} 
Let $(X,X_R,X_{L_1},X_{L_2})$ denote the output of $\smallbag(R^S,X^S,L_1^S,L_2^S,D^\phi_R,D^\phi_{L_1},D^\phi_{L_2})$. Then there is a \slim $S$-nice tree decomposition $\mathcal{T}'=(T',\{X'_t\}_{t\in V(T')})$ of $G$ with width at most $k$ 
such that the following holds.
\begin{itemize}
\item $\mathcal{T}'$ is a sibling of $\mathcal{T}$,
\item if $\mathcal{T}$ is \topheavy for a node $\hat{t}$ such that $\hat{t}$ is
\begin{itemize}
    \item not in $T_{t^\star_a}$ and not an ancestor of $t^\star_a$, or
    \item a descendant of $t^\star$,
\end{itemize}
then $\mathcal{T}'$ is also \topheavy for $\hat{t}$ and the subtree rooted at $\hat{t}$ is the same as in $\mathcal{T}$,
\item $\mathcal{T}'$ is \topheavy for $t^{\star}_a$, 
\item $X'_{t^{\star}}=X$ and $X'_{t^{\star}_a}=X_R$, 
\item if $t^{\star}$ has one $S$-child $t_1$, then $X'_{t_1}=X_{L_1}$, and
\item if $t^{\star}$ has two $S$-children $t_1$ and $t_2$, then $X'_{t_1}=X_{L_1}$ and $X'_{t_2}=X_{L_2}$.
\end{itemize}
\end{proposition}

To show \cref{prop:smallcorrect2}, in particular, we need to argue that the subrouting $\smallbagc$ called in the beginning of \cref{alg:smallwrapper}, that is, \cref{alg:small} is correct. To this end, we have the following.

\begin{lemma}\label{prop:smallcorrect}
Assume there exists a \slim $S$-nice tree decomposition $\mathcal{T}=(T,\{X_t\}_{t\in V(T)})$ of $G$ with width $k$ that contains a directed path of length at least three with $S$-trace $(\hat{L}^S, \hat{X}^S, \hat{R}^S)$, or if $\hat{L}^S=\emptyset$ then $\mathcal{T}$ contains an $S$-bottom node that has an $S$-child with $S$-trace $(\hat{L}^S, \hat{X}^S, \hat{R}^S)$, such that neither $t_{\max}$ nor the parent of $t_{\max}$ are contained in a full join tree, and with target node ${t^\star}$ such that the following holds.
\begin{itemize}
    \item $X^S\cup D_R\cup D_{L_1}\cup D_{L_2}\subseteq X_{t^\star}$ and $X_{t^\star}\cap (S\setminus X^S)=\emptyset$, 
    \item if $t^\star$ has a parent $t_p$, then $X_{t_p}=X_{t^\star}\setminus D_R$, 
    \item if $t^\star$ has one $S$-child $t_1$, then $X_{t_1}=X_{t^\star}\setminus D_{L_1}$, and 
    \item if $t^\star$ has two $S$-children $t_1$ and $t_2$, then $X_{t_1}=X_{t^\star}\setminus D_{L_1}$ and $X_{t_2}=X_{t^\star}\setminus D_{L_2}$.
\end{itemize} 
Then there is a \slim $S$-nice tree decomposition $\mathcal{T}'=(T',\{X'_t\}_{t\in V(T')})$ of $G$ with width at most $k$ 
such that the following holds.
\begin{itemize}
\item $\mathcal{T}'$ is a sibling of $\mathcal{T}$,  
\item if $\mathcal{T}$ is \topheavy for a node $\hat{t}$ such that $\hat{t}$ is
\begin{itemize}
    \item not in $T_{t^\star_a}$ and not an ancestor of $t^\star_a$, or
    \item a descendant of $t^\star$,
\end{itemize}
then $\mathcal{T}'$ is also \topheavy for $\hat{t}$ and the subtree rooted at $\hat{t}$ is the same as in $\mathcal{T}$, 
\item $X'_{t^{\star}}=\smallbagc(R^S,X^S,L_1^S,L_2^S,D_R,D_{L_1},D_{L_2})$, and
\item $|X'_{t^{\star}}|<|X_{t^{\star}}|$ or each of the following holds:
\begin{itemize}
    \item if $t^{\star}$ has a parent $t'_p$, then $X'_{t'_p}=X'_{t^{\star}}\setminus D_R$, 
    \item if $t^{\star}$ has one $S$-child $t'_1$, then $X'_{t'_1}=X'_{t^{\star}}\setminus D_{L_1}$, and 
    \item if $t^{\star}$ has two $S$-children $t'_1$ and $t'_2$, then $X'_{t'_1}=X'_{t^{\star}}\setminus D_{L_1}$ and $X'_{t'_2}=X_{t^{\star}}\setminus D_{L_2}$.
\end{itemize}
\end{itemize}
\end{lemma}

We postpone the proof of \cref{prop:smallcorrect} to the end of this subsection and first show how we can use it to prove \cref{prop:smallcorrect2}.

\begin{proof}[Proof of \cref{prop:smallcorrect2}]


Since $\phi$ contains the correct guess for $\mathcal{T}$, we have that if $t^\star$ has a parent $t_p$, then $X_{t_p}=X_{t^\star}\setminus D_R$, if $t^\star$ has one $S$-child $t_1$, then $X_{t_1}=X_{t^\star}\setminus D_{L_1}$, and if $t^\star$ has two $S$-children $t_1$ and $t_2$, then $X_{t_1}=X_{t^\star}\setminus D_{L_1}$ and $X_{t_2}=X_{t^\star}\setminus D_{L_2}$.
It follows that the prerequisites for \cref{prop:smallcorrect} are fulfilled. 

We can conclude that then there there is a \slim $S$-nice tree decomposition $\mathcal{T}'=(T',\{X'_t\}_{t\in V(T')})$ of $G$ with width at most $k$ 
such that the following holds.
\begin{itemize}
\item $\mathcal{T}'$ is a sibling of $\mathcal{T}$,  
\item if $\mathcal{T}$ is \topheavy for a node $\hat{t}$ such that $\hat{t}$ is
\begin{itemize}
    \item not in $T_{t^\star_a}$ and not an ancestor of $t^\star_a$, or
    \item a descendant of $t^\star$,
\end{itemize}
then $\mathcal{T}'$ is also \topheavy for $\hat{t}$ and the subtree rooted at $\hat{t}$ is the same as in $\mathcal{T}$, 
\item $X'_{t^{\star}}=\smallbagc(R^S,X^S,L_1^S,L_2^S,D_R,D_{L_1},D_{L_2})$, and 
\item $|X'_{t^{\star}}|<|X_{t^{\star}}|$ or each of the following holds:
\begin{itemize}
    \item if $t^{\star}$ has a parent $t'_p$, then $X'_{t'_p}=X'_{t^{\star}}\setminus D_R$, 
    \item if $t^{\star}$ has one $S$-child $t'_1$, then $X'_{t'_1}=X'_{t^{\star}}\setminus D_{L_1}$, and 
    \item if $t^{\star}$ has two $S$-children $t'_1$ and $t'_2$, then $X'_{t'_1}=X'_{t^{\star}}\setminus D_{L_1}$ and $X'_{t'_2}=X_{t^{\star}}\setminus D_{L_2}$.
\end{itemize}
\end{itemize}
Note that $|X'_{t^{\star}}|<|X_{t^{\star}}|$ contradicts the minimality requirement of \cref{prop:smallcorrect2}. Hence, we have that each of the following holds:
\begin{itemize}
    \item if $t^{\star}$ has a parent $t'_p$, then $X'_{t'_p}=X'_{t^{\star}}\setminus D_R$, 
    \item if $t^{\star}$ has one $S$-child $t'_1$, then $X'_{t'_1}=X'_{t^{\star}}\setminus D_{L_1}$, and 
    \item if $t^{\star}$ has two $S$-children $t'_1$ and $t'_2$, then $X'_{t'_1}=X'_{t^{\star}}\setminus D_{L_1}$ and $X'_{t'_2}=X_{t^{\star}}\setminus D_{L_2}$.
\end{itemize}

\cref{alg:smallwrapper} starts with computing the bag $X=\smallbagc(R^S,X^S,L_1^S,L_2^S,D_R,D_{L_1},D_{L_2})$.
Now we argue that Steps~\ref{wrapper2a}, \ref{wrapper2b}, and \ref{wrapper2c} make the same changes to the bag candidate $X$ for ${t^{\star}}$, and the bags of its parent and $S$-children, $X_R$, $X_{L_1}$, and $X_{L_2}$, respectively, as the modification $\MakeTopHeavy(t^{\star}_a)$ (\cref{def:maketopheavy}) do when applied to $\mathcal{T}'$. 
By \cref{lem:topheavynomake,lem:topheavynomake2}, applying $\MakeTopHeavy(t^{\star}_a)$ to $\mathcal{T}'$ makes it \topheavy for $t^{\star}_a$. Furthermore, by \cref{lem:preserveslim} we have that $\mathcal{T}'$ remains \slim, $S$-nice, and its width did not increase. 
Since if $t^\star_{d_1}$ and $t^\star_{d_2}$, respectively, exist, then $\mathcal{T}'$ is \topheavy for $t^\star_{d_1}$ and $t^\star_{d_2}$, respectively, there are no changes to $T_{t^{\star}_{d_1}}$ and $T_{t^{\star}_{d_2}}$ or any bag of a node in $T_{t^{\star}_{d_1}}$ and $T_{t^{\star}_{d_2}}$ (if they exist).
Lastly, $\MakeTopHeavy(t^{\star}_a)$ does not move any vertices from $S$, and hence, we have that $\mathcal{T}'$ remains a sibling of $\mathcal{T}$. It remains to show that the bags of ${t^{\star}}$, its $S$-children, and its parent are correctly computed. 
To this end, observe that Step~\ref{heavy1} of \cref{def:maketopheavy} performs $\MakeSlimTwo$ (\cref{def:MakeSlimTwo}). 
The first step of $\MakeSlimTwo$ makes sure that if the bag of an $S$-bottom node is not full, then the parent of the $S$-bottom node has the same bag. Furthermore, in case it is a join node, the $S$-children have no larger bags. Note that the latter case cannot happen, since all bags are modified in the same way and initially the bags of the $S$-children are not larger. The former case, however, can happen if the target node is an $S$-bottom node and its bag size is decreased. This, however, contradicts the minimality requirement of the bag of the target node.
The second step of $\MakeSlimTwo$ applies if target node is not an $S$-bottom node, and the parent is an $S$-bottom node that admits a $\forget$ $S$-operation. 
If the $S$-child of the $S$-bottom node (which is the target node in this case) is not full, then the $S$-child of the $S$-childs has the same bag as the $S$-child of the $S$-bottom node. 
It follows that this step also only makes changes if the size of the bag of the target node decreases, a contradiction to its minimality
The third step of $\MakeSlimTwo$ does not change the bag of the target node, its parent, or its children.
We can conclude that $\MakeSlimTwo$ does not make any changes to the bag of the target node, its children, or its parent.
Next, observe that Step~\ref{heavy2} of \cref{def:maketopheavy} does not make any changes to the bags of ${t^{\star}}$, its $S$-children, and its parent.

Observe that Step~\ref{wrapper2a} of \cref{alg:smallwrapper} mirrors the changes made by Step~\ref{heavy3} of \cref{def:maketopheavy}:
Note that $X_R$ is the bag of the node denoted with $t$ in $\MakeTopHeavy$ (\cref{def:maketopheavy}). Since the edge-subdivision was performed in the beginning of $\MakeTopHeavy$, we have that the node denoted with $t'$ in $\MakeTopHeavy$ is $t^\star_a$. Hence, the $\RemoveFromSubtree$ operation is applied to the parent of $t^\star_a$, which means that the vertices in $V(C)$ are removed from $X$, $X_R$, $X_{L_1}$, and $X_{L_2}$.
Similarly, Step~\ref{wrapper2b} of \cref{alg:smallwrapper} mirrors the changes made by Step~\ref{heavy4} of \cref{def:maketopheavy}:
If for each $u\in N(v)$ it holds that $u\in X_R$ or $u$ is not connected to $L_1^S\cup L_2^S\cup((X\setminus X_R)\cap S)$ in $G-X_R$, then we have that $N(v)$ does not contain any vertices that are not in $X_R$ but need to be a bag in the subtree rooted at $t^\star_a$. Hence, we may assume that all those neighbors that are not in $X_R$, are contained in some bags outside of the subtree rooted at $t^\star_a$. Again, the $\RemoveFromSubtree$ operation is applied to the parent of $t^\star_a$, which means that the vertex $v$ is removed from $X$, $X_R$, $X_{L_1}$, and $X_{L_2}$.
Finally, Step~\ref{wrapper2c} of \cref{alg:smallwrapper} mirrors the changes made by Step~\ref{heavy5} of \cref{def:maketopheavy}:
Here, the analysis is analogous to the previous case with the exception that it only holds for all neightbors of $v$ except $u$. Then the $\BringNeighborUp$ operation is applied to the parent of $t^\star_a$, which means that the vertex $v$ is replaced with $u$ in $X$, $X_R$, $X_{L_1}$, and $X_{L_2}$.
We can conclude that $X=X_{t^{\star}}$ and that if $t^{\star}$ has a parent or $S$-children, their respective bags are also correctly computed.
\end{proof}

Now we prove \cref{prop:smallcorrect}. To do this, we will essentially show that in each step of \cref{alg:small}, we do not make mistakes. This first lemma, intuitively, shows that Step~\ref{regular:step1} of \cref{alg:small}, if applicable, is always correct.
\begin{lemma}\label{lem:regularstep1}
Let $(\Rest,\In,\Out)$ be a 3-partition of $V$ that appears during some iteration of \cref{alg:small}. 
If there exists a \slim $S$-nice tree decomposition $\mathcal{T}=(T,\{X_t\}_{t\in V(T)})$ of $G$ with width $k$ that contains a directed path of length at least three with $S$-trace $(\hat{L}^S, \hat{X}^S, \hat{R}^S)$, or if $\hat{L}^S=\emptyset$ then $\mathcal{T}$ contains an $S$-bottom node that has an $S$-child with $S$-trace $(\hat{L}^S, \hat{X}^S, \hat{R}^S)$, such that neither $t_{\max}$ nor the parent of $t_{\max}$ are contained in a full join tree, and with target node ${t^\star}$ such that $\In\subseteq X_{t^\star}$ and $X_{t^\star}\cap \Out=\emptyset$, then the following holds.
    If there is a vertex $v\in \Rest$ that is contained in a mixed component in $G[(\Out\cup\{v\})\setminus S]$, then $v\in X_{t^\star}$.
\end{lemma}
\begin{proof}
By \cref{obs:smallneighbors} we have that ${t^\star}$ has $S$-trace $(L^S, X^S, R^S)$.
Assume for contradiction that $v\notin X_{t^\star}$. Then there is a mixed component $C$ in $G-(X_{t^\star}\cup S)$. If $N(V(C))\cap R^S\neq\emptyset$, then $X_{t^\star}$ does not separate $L^S$ and $R^S$. This is a contradiction to $(L^S, X^S, R^S)$ being the $S$-trace of ${t^\star}$. If $N(V(C))\cap R^S=\emptyset$, then we have that $N(V(C))\cap L_1^S\neq\emptyset$ and $N(V(C))\cap L_2^S\neq\emptyset$. Then we have by the definition of the target node and $L_2^S$ that ${t^\star}$ is an $S$-bottom node that admits $S$-operation $\smalljoin(X^S_t,X_1^S,X_2^S,L_{1}^S, L_{2}^S,d)$ (this is either $\tau_+$ or the $S$-operation~$\tau$ from $\tau_+=\extendedforget(v,d,f,\tau)$). From \cref{def:snicetd} and \cref{obs:join} follows that~$X_{t^\star}$ separates $L_1^S$ and $L_2^S$, a contradiction.
\end{proof}

Next, we show that the Step~\ref{regular:step2} of \cref{alg:small}, if applicable, is correct.

\begin{lemma}\label{lem:regularstep2}
Let $(\Rest,\In,\Out)$ be a 3-partition of $V$ that appears during some iteration of \cref{alg:small}. 
Assume that there exists a \slim $S$-nice tree decomposition $\mathcal{T}=(T,\{X_t\}_{t\in V(T)})$ of $G$ with width $k$ that contains a directed path of length at least three with $S$-trace $(\hat{L}^S, \hat{X}^S, \hat{R}^S)$, or if $\hat{L}^S=\emptyset$ then $\mathcal{T}$ contains an $S$-bottom node that has an $S$-child with $S$-trace $(\hat{L}^S, \hat{X}^S, \hat{R}^S)$, such that neither $t_{\max}$ nor the parent of $t_{\max}$ are contained in a full join tree, and with target node ${t^\star}$ such that the following holds.
\begin{itemize}
    \item $\In\subseteq X_{t^\star}$ and $X_{t^\star}\cap \Out=\emptyset$, 
    \item if $t^\star$ has a parent $t_p$, then $X_{t_p}=X_{t^\star}\setminus D_R$, 
    \item if $t^\star$ has one $S$-child $t_1$, then $X_{t_1}=X_{t^\star}\setminus D_{L_1}$, and 
    \item if $t^\star$ has two $S$-children $t_1$ and $t_2$, then $X_{t_1}=X_{t^\star}\setminus D_{L_1}$ and $X_{t_2}=X_{t^\star}\setminus D_{L_2}$.
\end{itemize} 
Then, if there is a vertex $v$ in $\Rest$ that is contained in an $H$-component $C_H$ in $G[(\Out\cup\{v\})\setminus S]$ for $H\in\{L_1,L_2,R\}$ and $N(C_H)\cap D_H\neq \emptyset$, we have that $v\in X_{t^\star}$.

\end{lemma}
\begin{proof}
Assume that there is a vertex $v$ in $\Rest$ that is contained in an $R$-component $C_R$ in $G[(\Out\cup\{v\})\setminus S]$ and $N(C_R)\cap D_R\neq \emptyset$. Let $u\in N(C_R)\cap D_R$.
Assume for contradiction that $v\notin X_{t^\star}$. 
It follows from Conditions~\ref{condition_2_tree_decomposition} and~\ref{condition_3_tree_decomposition} of \cref{def:tree_decomposition} and \cref{def:trace} that there must be an ancestor of ${t^\star}$ with a bag that contains~$v$.
Since $u\in D_R$ and $D_R\subseteq \In$ in all iterations, we have that $u$ is contained in $X_{t^\star}$, but not in the bag of the parent of~${t^\star}$. It follows from \cref{def:tree_decomposition} that $u$ is not contained in any bag of an ancestor of~${t^\star}$.
Consider a path $P$ from $u$ to $v$ in $G[(\Out\cup\{v\})\setminus S]$. Let $u'$ be the vertex closest to $v$ in $P$ that is contained in the bag of ${t^\star}$ or in a bag that is a descendant of $t$. Let $v'$ the next vertex in $P$. 
By assumption, $v'$ is neither contained in the bag of ${t^\star}$ nor in the bag of any of the descendants of~${t^\star}$. It follows that there is no bag that contains both $u'$ and $v'$. This, however, is a contradiction to Condition~\ref{condition_2_tree_decomposition} of \cref{def:tree_decomposition}, which requires that there is a node $t'$ in~$\mathcal{T}$ such that $\{u,v\}\in X_{t'}$, since $u'$ and $v'$ are neighbors.
The cases where $H=L_1$ or $H=L_2$ are analogous.
\end{proof}

Before we show that the remaining steps of the algorithm are correct, we first make the following observation. We have made essentially the same obervation for \cref{alg:big} in \cref{lem:oneneighborbig}. The proof here is very similar, but since \cref{alg:small} obviously behaves differently from \cref{alg:big}, some small adjustments need to be made.
\begin{lemma}\label{lem:oneneighbor}
During every iteration of \cref{alg:small}, the following holds. For every connected component $C$ in $G[\Out\setminus S]$ we have that $|N(V(C))\cap \Rest|\le 1$.
\end{lemma}
\begin{proof}
We prove this by induction on the iterations of the algorithm. Initially, the statement clearly holds since $\Out\setminus S=\emptyset$. Let $(\Rest,\In,\Out)$ be a 3-partition of $V$ that appears during some iteration of the algorithm. By induction hypothesis, we have that for every connected component $C$ in $G[\Out\setminus S]$ we have that $|N(V(C))\cap \Rest|\le 1$.
Assume that Step~\ref{regular:step1} or Step~\ref{regular:step2} applies, then the new 3-partition is $(\Rest\setminus \{v\},\In\cup\{v\},\Out)$ for some $v\in \Rest$. Clearly, for every connected component $C$ in $G[\Out\setminus S]$ we have that $|N(V(C))\cap (\Rest\setminus\{v\})|\le |N(V(C))\cap \Rest|\le 1$.
Assume that one of the remaining steps applies, then the new 3-partition is $(\Rest\setminus \{v\},\In,\Out\cup\{v\})$ for some $v\in \Rest$ that has degree at most one in $G[\Rest]$. Let $C$ be a connected component in $G[(\Out\cup\{v\})\setminus S]$. If $C$ does not contain $v$, then $C$ is a connected component in $G[\Out\setminus S]$, and clearly we have that $|N(V(C))\cap (\Rest\setminus\{v\})|\le |N(V(C))\cap \Rest|\le 1$. Now assume that $C$ contains $v$. Let $C_1,\ldots,C_\ell$ be the connected components of $C-\{v\}$. Note that $C_1,\ldots,C_\ell$ are all connected components in $G[\Out\setminus S]$ and hence only have one neighbor in $\Rest$. Since $v\in\Rest$ and $v\in N(V(C_i))$ for all $1\le i\le \ell$, we have that $|N(V(C_i))\cap (\Rest\setminus\{v\})|=0$ for all $1\le i\le \ell$. Recall that $|N(v)\cap\Rest|\le 1$. It follows that $|N(V(C))\cap (\Rest\setminus\{v\})|\le 1$. 
\end{proof}

Now we show that Step~\ref{regular:step3} of \cref{alg:small}, if applicable, is correct.

\begin{lemma}\label{lem:regularstep3}
Let $(\Rest,\In,\Out)$ be a 3-partition of $V$ that appears during some iteration of \cref{alg:small}. 
Assume there exists a \slim $S$-nice tree decomposition $\mathcal{T}=(T,\{X_t\}_{t\in V(T)})$ of $G$ with width $k$ that contains a directed path of length at least three with $S$-trace $(\hat{L}^S, \hat{X}^S, \hat{R}^S)$, or if $\hat{L}^S=\emptyset$ then $\mathcal{T}$ contains an $S$-bottom node that has an $S$-child with $S$-trace $(\hat{L}^S, \hat{X}^S, \hat{R}^S)$, such that neither $t_{\max}$ nor the parent of $t_{\max}$ are contained in a full join tree, and with target node ${t^\star}$ such that the following holds.
\begin{itemize}
    \item $\In\subseteq X_{t^\star}$ and $X_{t^\star}\cap \Out=\emptyset$, 
        \item if $t^\star$ has a parent $t_p$, then $X_{t_p}=X_{t^\star}\setminus D_R$, 
    \item if $t^\star$ has one $S$-child $t_1$, then $X_{t_1}=X_{t^\star}\setminus D_{L_1}$, 
    \item if $t^\star$ has two $S$-children $t_1$ and $t_2$, then $X_{t_1}=X_{t^\star}\setminus D_{L_1}$ and $X_{t_2}=X_{t^\star}\setminus D_{L_2}$,
    \item Steps~\ref{regular:step1} and~\ref{regular:step2} of \cref{alg:small} do not apply, and
    \item there is a vertex $v\in\Rest$ that is contained in an $R$-component in $G[(\Out\cup\{v\})\setminus S]$ in $G[(\Out\cup\{v\})\setminus S]$ and that has degree at most one in $G[\Rest]$.
\end{itemize} 
Then there is a \slim $S$-nice tree decomposition $\mathcal{T}'=(T',\{X'_t\}_{t\in V(T')})$ of $G$ with width at most $k$ 
such that the following holds.
\begin{itemize}
\item $\mathcal{T}'$ is a sibling of $\mathcal{T}$,
\item if $\mathcal{T}$ is \topheavy for a node $\hat{t}$ such that $\hat{t}$ is
\begin{itemize}
    \item not in $T_{t^\star_a}$ and not an ancestor of $t^\star_a$, or
    \item a descendant of $t^\star$,
\end{itemize}
then $\mathcal{T}'$ is also \topheavy for $\hat{t}$ and the subtree rooted at $\hat{t}$ is the same as in $\mathcal{T}$,
\item $\In\subseteq X'_{t^{\star}}$ and $X'_{t^{\star}}\cap (\Out\cup\{v\})=\emptyset$, and
\item $|X'_{t^{\star}}|<|X_{t^{\star}}|$ or each of the following holds:
\begin{itemize}
    \item $|X'_{t^{\star}}|=|X_{t^{\star}}|$,
    \item if $t^{\star}$ has a parent $t'_p$, then $X'_{t'_p}=X'_{t^{\star}}\setminus D_R$, 
    \item if $t^{\star}$ has one $S$-child $t'_1$, then $X'_{t'_1}=X'_{t^{\star}}\setminus D_{L_1}$, and 
    \item if $t^{\star}$ has two $S$-children $t'_1$ and $t'_2$, then $X'_{t'_1}=X'_{t^{\star}}\setminus D_{L_1}$ and $X'_{t'_2}=X_{t^{\star}}\setminus D_{L_2}$.
\end{itemize}
\end{itemize}
\end{lemma}
\begin{proof}
Let $v\in\Rest$ be a vertex that is contained in an $R$-component in $G[(\Out\cup\{v\})\setminus S]$ and that has degree at most one in $G[\Rest]$. 
Consider $\mathcal{T}$ the target node~$t^\star$ in $\mathcal{T}$.
If $v\notin X_{t^\star}$, then we are done. 
Hence, assume that $v\in X_{t^\star}$. 
We will argue that we can modify the tree decomposition~$\mathcal{T}$ without increasing its width using the operations introduced in \cref{sec:mod} such that afterwards, the requirements in the lemma statement hold.

Let $C$ be the $R$-component in $G[(\Out\cup\{v\})\setminus S]$ that contains $v$. 
 Let $C_1,\ldots,C_\ell$ be the connected components of $C-\{v\}$. By \cref{lem:oneneighbor} we have for each $C_i$ with $1\le i\le \ell$ that $|N(V(C_i))\cap\Rest|\le 1$. We can conclude that $N(V(C_i))\cap\Rest=\{v\}$ for each $1\le i\le \ell$. Since $v$ has degree at most one in $G[\Rest]$, we have that $N(V(C))\cap \Rest\subseteq\{u\}$ for some $u\in \Rest$. 

    Let $t_p$ be the parent of ${t^\star}$. Note that $V_{t^\star}\subseteq V_{t_p}$. By \cref{obs:smallneighbors} we have that $X_{t_p}$ is not full.
         Furthermore, we have that $R^S\subseteq (V\setminus V_{{t^\star}})\cup X_{t_p}$. Denote $M=V(C)\setminus \{v\}$. We claim that $N(M)\subseteq (V\setminus V_{t^\star})\cup X_{t_p}$. 
         To this end, we observe the following:
         \begin{itemize}
             \item We have that $N(M)\cap \Out\subseteq R^S$. It follows that $N(M)\cap \Out\subseteq (V\setminus V_{t^\star})\cup X_{t_p}$. 
             \item We have that $N(M)\cap \Rest\subseteq \{v\}$. Since $v\notin D_R$ we have that $v\in X_{t_p}$ and hence $N(M)\cap \Rest\subseteq (V\setminus V_{t^\star})\cup X_{t_p}$.
             \item Finally, note that $N(M)\cap D_R=\emptyset$ since Step~\ref{regular:step2} of of \cref{alg:small} does not apply. Hence, we have that $N(M)\cap \In\subseteq X_{t_p}$. It follows that $N(M)\cap \In\subseteq (V\setminus V_{t^\star})\cup X_{t_p}$
         \end{itemize}
         Since $V=\Rest\cup\In\cup\Out$, we can conclude that $N(M)\subseteq (V\setminus V_{t^\star})\cup X_{t_p}$.

        This allows us to perform a $\RemoveFromSubtree$ operation. We apply the modification $\RemoveFromSubtree(t_p,M)$ to $\mathcal{T}$. Note that this requires $t_p$ to have a parent. If $t_p$ is the root of the tree decomposition, we create a new node which is a copy of $t_p$ and make it the parent of $t_p$ and hence the new root.
        Let $\mathcal{T}'$ denote the resulting tree decomposition.
        By \cref{lem:moveremove} we have that afterwards, $\mathcal{T}'$ is still $S$-nice, has width at most $k$, and is a sibling of $\mathcal{T}$. 
        By \cref{def:topheavy}, we have that
        if $\mathcal{T}$ is \topheavy for a node $\hat{t}$ such that $\hat{t}$ is not in $T_{t^\star_a}$ and not an ancestor of $t^\star_a$, or
         a descendant of $t^\star$,
then $\mathcal{T}'$ is also \topheavy for $\hat{t}$ and the subtree rooted at $\hat{t}$ is the same as in $\mathcal{T}$.
        Furthermore, since $M\cap X_{t^\star}=\emptyset$, we have that the bags of~$t^\star$, its parent, and its $S$-children remain unchanged. 
By \cref{lem:preserveslim}, we can we apply $\MakeSlimTwo$ to ensure that $\mathcal{T}'$ is \slim.
        Moreover, we have that no vertices from $S$ are moved and if the size of $X_{t^\star}$ remains unchanged, also no vertices from $D_R\cup D_{L_1}\cup D_{L_2}$ are moved.
Finally, we have by the definition of $\RemoveFromSubtree(t_p,M)$ (\cref{def:remove}) that $X_{t^\star}\cap \Out=\emptyset$ and, since $M\cap \In=\emptyset$, that $\In\subseteq X_{t^\star}$. 
%
%
%
%
%
Furthermore, we have that $N[M]\subseteq V\setminus (V_{t^\star}\cup X_{t_p})$.
%

 Observe that $N(v)\setminus ((V\setminus V_{t_p})\cup X_{t_p})\subseteq \{u\}$. Assume that it is not, then there is some $u'\neq u$ such that $u'\in N(v)$ and $u'\notin ((V\setminus V_{t_p})\cup X_{t_p})$. It follows that $u'\notin N[M]$. Furthermore, we have that $u'\notin S\setminus R^S$, since otherwise we have that $C$ is a mixed component. We can conclude that $u'\notin \Out$. 
     Furthermore, we must have that $u'\notin\Rest$, since otherwise we get a contradiction to $N(V(C))\cap \Rest\subseteq \{u\}$. We can conclude that $u'\in \In$. From \cref{def:snicetd} follows that $\{u'\}=X_{t_p}\setminus X_{t^\star}$. 
     Then we must have that $\{u'\}=D_{R}$.
     This is a contradiction to the assumption that Step~\ref{regular:step2} of the algorithm does not apply.
We can conclude that $N(v)\setminus ((V\setminus V_{t_p})\cup X_{t_p})\subseteq \{u\}$.

Consider the case that $N(v)\setminus ((V\setminus V_{t_p})\cup X_{t_p})=\emptyset$.
Then we can apply $\RemoveFromSubtree(t_p,\{v\})$ to~$\mathcal{T}'$. 
By \cref{lem:moveremove} we have that afterwards, $\mathcal{T}'$ is still $S$-nice, has width at most $k$, and is a sibling of $\mathcal{T}$. 
        By \cref{def:topheavy}, we have that
        if $\mathcal{T}$ is \topheavy for a node $\hat{t}$ such that $\hat{t}$ is not in $T_{t^\star_a}$ and not an ancestor of $t^\star_a$, or
         a descendant of $t^\star$,
then $\mathcal{T}'$ is also \topheavy for $\hat{t}$ and the subtree rooted at $\hat{t}$ is the same as in $\mathcal{T}$.
By \cref{lem:preserveslim}, we can we apply $\MakeSlimTwo$ to ensure that $\mathcal{T}'$ remains \slim.
Moreover, we again have that no vertices from $S$ are moved and if the size of $X_{t^\star}$ remains unchanged, also no vertices from $D_R\cup D_{L_1}\cup D_{L_2}$ are moved.
Furthermore, we have that the bags of~$t^\star$, its parent, and its $S$-children all receive the same change, namely that the vertex $v$ is removed.


Finally, consider the case that $N(v)\setminus (V_{t^{\star}_a}\cup S)=\{u\}$. 
We subdivide the edge between $t^\star$ and~$t_p$ and let $t'$ be the new node. We set $X_{t'}=X_{t_p}$. Swap the names of $t_p$ and $t'$ such that $t_p$ remains the parent of $t^\star$. Now we can perform the modification $\BringNeighborUp(t',v,u)$ (\cref{def:bringup}) to~$\mathcal{T}'$.
By \cref{lem:bringupdown} we have that afterwards, $\mathcal{T}'$ is still $S$-nice, has width at most $k$, and is a sibling of $\mathcal{T}$. 
        By \cref{def:topheavy}, we have that
        if $\mathcal{T}$ is \topheavy for a node $\hat{t}$ such that $\hat{t}$ is not in $T_{t^\star_a}$ and not an ancestor of $t^\star_a$, or
         a descendant of $t^\star$,
then $\mathcal{T}'$ is also \topheavy for $\hat{t}$ and the subtree rooted at $\hat{t}$ is the same as in $\mathcal{T}$.
By \cref{lem:preserveslim}, we can we apply $\MakeSlimTwo$ to ensure that $\mathcal{T}'$ remains \slim.
Moreover, we again have that no vertices from $S$ are moved and if the size of $X_{t^\star}$ remains unchanged, also no vertices from $D_R\cup D_{L_1}\cup D_{L_2}$ are moved. 
Furthermore, we have that the bags of~$t^\star$, its parent, and its $S$-children all receive the same change, namely that the vertex $v$ is removed and the vertex $u$ is added. 
This finishes the proof.
%
%
%
\end{proof}

Now we show that Step~\ref{regular:step4} of \cref{alg:small}, if applicable, is correct. This proof has many similarities with the one of \cref{lem:regularstep3}, but there are some key differences.

\begin{lemma}\label{lem:regularstep4}
Let $(\Rest,\In,\Out)$ be a 3-partition of $V$ that appears during some iteration of \cref{alg:small}. 
Assume there exists a \slim $S$-nice tree decomposition $\mathcal{T}=(T,\{X_t\}_{t\in V(T)})$ of $G$ with width $k$ that contains a directed path of length at least three with $S$-trace $(\hat{L}^S, \hat{X}^S, \hat{R}^S)$, or if $\hat{L}^S=\emptyset$ then $\mathcal{T}$ contains an $S$-bottom node that has an $S$-child with $S$-trace $(\hat{L}^S, \hat{X}^S, \hat{R}^S)$, such that neither $t_{\max}$ nor the parent of $t_{\max}$ are contained in a full join tree, and with target node ${t^\star}$ such that the following holds.
\begin{itemize}
    \item $\In\subseteq X_{t^\star}$ and $X_{t^\star}\cap \Out=\emptyset$,
            \item if $t^\star$ has a parent $t_p$, then $X_{t_p}=X_{t^\star}\setminus D_R$, 
    \item if $t^\star$ has one $S$-child $t_1$, then $X_{t_1}=X_{t^\star}\setminus D_{L_1}$, 
    \item if $t^\star$ has two $S$-children $t_1$ and $t_2$, then $X_{t_1}=X_{t^\star}\setminus D_{L_1}$ and $X_{t_2}=X_{t^\star}\setminus D_{L_2}$, 
    \item Steps~\ref{regular:step1},~\ref{regular:step2}, and~\ref{regular:step3} of \cref{alg:small} do not apply, and
    \item there is a vertex $v\in\Rest$ that is contained in an $H$-component in $G[(\Out\cup\{v\})\setminus S]$ for $H\in\{L_1,L_2\}$ in $G[(\Out\cup\{v\})\setminus S]$ and that has degree at most one in $G[\Rest]$.
\end{itemize} 
Then there is a \slim $S$-nice tree decomposition $\mathcal{T}'=(T',\{X'_t\}_{t\in V(T')})$ of $G$ with width at most $k$ 
such that the following holds.
\begin{itemize}
\item $\mathcal{T}'$ is a sibling of $\mathcal{T}$,
\item if $\mathcal{T}$ is \topheavy for a node $\hat{t}$ such that $\hat{t}$ is
\begin{itemize}
    \item not in $T_{t^\star_a}$ and not an ancestor of $t^\star_a$, or
    \item a descendant of $t^\star$,
\end{itemize}
then $\mathcal{T}'$ is also \topheavy for $\hat{t}$ and the subtree rooted at $\hat{t}$ is the same as in $\mathcal{T}$,
\item $\In\subseteq X'_{t^{\star}}$ and $X'_{t^{\star}}\cap (\Out\cup\{v\})=\emptyset$, and
\item $|X'_{t^{\star}}|<|X_{t^{\star}}|$ or each of the following holds:
\begin{itemize}
    \item if $t^{\star}$ has a parent $t'_p$, then $X'_{t'_p}=X'_{t^{\star}}\setminus D_R$, 
    \item if $t^{\star}$ has one $S$-child $t'_1$, then $X'_{t'_1}=X'_{t^{\star}}\setminus D_{L_1}$, and 
    \item if $t^{\star}$ has two $S$-children $t'_1$ and $t'_2$, then $X'_{t'_1}=X'_{t^{\star}}\setminus D_{L_1}$ and $X'_{t'_2}=X_{t^{\star}}\setminus D_{L_2}$.
\end{itemize}
\end{itemize}
\end{lemma}
\begin{proof}
Let $v\in\Rest$ be a vertex that is contained in an $H$-component in $G[(\Out\cup\{v\})\setminus S]$ and that has degree at most one in $G[\Rest]$. 
Consider $\mathcal{T}$ the target node~$t^\star$ in $\mathcal{T}$.
If $v\notin X_{t^\star}$, then we are done. 
Hence, assume that $v\in X_{t^\star}$. 
We will argue that we can modify the tree decomposition~$\mathcal{T}$ without increasing its width using the operations introduced in \cref{sec:mod} such that afterwards, the requirements in the lemma statement hold.

Let $C$ be the $H$-component in $G[(\Out\cup\{v\})\setminus S]$ that contains $v$. 
 Let $C_1,\ldots,C_\ell$ be the connected components of $C-\{v\}$. By \cref{lem:oneneighbor} we have for each $C_i$ with $1\le i\le \ell$ that $|N(V(C_i))\cap\Rest|\le 1$. We can conclude that $N(V(C_i))\cap\Rest=\{v\}$ for each $1\le i\le \ell$. Since $v$ has degree at most one in $G[\Rest]$, we have that $N(V(C))\cap \Rest\subseteq\{u\}$ for some $u\in \Rest$. 
    
We analyze the case that $H=L_1$, that is, $C$ is an $L_1$-component in $G[(\Out\cup\{v\})\setminus S]$.
The case where $H=L_2$ is analogous.
By \cref{obs:smallneighbors} we have that ${t^\star}$ has a child $t_1$, such that $X_{t_1}$ is not full and $L^S_1\subseteq V_{t_1}$. 
        Denote $M=V(C)\setminus \{v\}$.
        We claim that $N(M)\subseteq V_{t_1}$. 
 To this end, we observe the following:
         \begin{itemize}
             \item We have that $N(M)\cap \Out\subseteq L_1^S$. It follows that $N(M)\cap \Out\subseteq V_{t_1}$. 
             \item We have that $N(M)\cap \Rest\subseteq \{v\}$. Since $v\notin D_{L_1}$ we have that $v\in X_{t_1}$ and hence $N(M)\cap \Rest\subseteq V_{t_1}$.
             \item Finally, note that $N(M)\cap D_{L_1}=\emptyset$ since Step~\ref{regular:step2} of of \cref{alg:small} does not apply. Hence, we have that $N(M)\cap \In\subseteq X_{t_1}$. It follows that $N(M)\cap \In\subseteq V_{t_1}$
         \end{itemize}
         Since $V=\Rest\cup\In\cup\Out$, we can conclude that $N(M)\subseteq V_{t_1}$.
        

        This allows us to perform a $\MoveIntoSubtree$ operation. We apply the modification $\MoveIntoSubtree(t_1,M)$ to $\mathcal{T}$. 
        Let $\mathcal{T}'$ denote the resulting tree decomposition.
        By \cref{lem:moveremove} we have that afterwards, 
        $\mathcal{T}'$ is still $S$-nice, has width at most $k$, and is a sibling of $\mathcal{T}$. 
        By \cref{def:topheavy}, we have that
        if $\mathcal{T}$ is \topheavy for a node $\hat{t}$ such that $\hat{t}$ is not in $T_{t^\star_a}$ and not an ancestor of $t^\star_a$, or
         a descendant of $t^\star$,
then $\mathcal{T}'$ is also \topheavy for $\hat{t}$ and the subtree rooted at $\hat{t}$ is the same as in $\mathcal{T}$.
        Furthermore, since $M\cap X_{t^\star}=\emptyset$, we have that the bags of~$t^\star$, its parent, and its $S$-children remain unchanged. 
By \cref{lem:preserveslim}, we can we apply $\MakeSlimTwo$ to ensure that $\mathcal{T}'$ remains \slim.
Moreover, we have that no vertices from $S$ are moved and if the size of $X_{t^\star}$ remains unchanged, also no vertices from $D_R\cup D_{L_1}\cup D_{L_2}$ are moved.
Finally, we have by the definition of $\MoveIntoSubtree(t_1,M)$ (\cref{def:move}) that $X_{t^\star}\cap \Out=\emptyset$ and, since $M\cap \In=\emptyset$, that $\In\subseteq X_{t^\star}$. Furthermore, we have that $N[M]\subseteq V_{t_1}$. 
        
     Observe that $N(v)\setminus V_{t_1}\subseteq \{u\}$. Assume that it is not, then there is some $u'\neq u$ such that $u'\in N(v)$ and $u'\notin V_{t_1}$. It follows that $u'\notin N[M]$. Furthermore, we have that $u'\notin S\setminus L_1^S$, since otherwise we have that $C$ is a mixed component. We can conclude that $u'\notin \Out$. 
     Furthermore, we must have that $u'\notin\Rest$, since otherwise we get a contradiction to $N(V(C))\cap \Rest\subseteq \{u\}$. We can conclude that $u'\in \In$. From \cref{def:snicetd} follows that $\{u'\}=X_{t_1}\setminus X_{t^\star}$. 
     Then we must have that $\{u'\}=D_{L_1}$.
     This is a contradiction to the assumption that Step~\ref{regular:step2} of the algorithm does not apply.
We can conclude that $N(v)\setminus V_{t_1}\subseteq \{u\}$.

Consider the case that $N(v)\setminus V_{t_1}=\emptyset$.
We subdivide the edge between $t^\star$ and~$t_1$ and let~$t'$ be the new node. We set $X_{t'}=X_{t_1}$. Swap the names of $t_1$ and $t'$ such that $t_1$ remains the child of $t^\star$. Then we can apply $\MoveIntoSubtree(t',\{v\})$ to~$\mathcal{T}'$. 
By \cref{lem:moveremove} we have that afterwards, $\mathcal{T}'$ is still $S$-nice, has width at most $k$, and is a sibling of $\mathcal{T}$. 
        By \cref{def:topheavy}, we have that
        if $\mathcal{T}$ is \topheavy for a node $\hat{t}$ such that $\hat{t}$ is not in $T_{t^\star_a}$ and not an ancestor of $t^\star_a$, or
         a descendant of $t^\star$,
then $\mathcal{T}'$ is also \topheavy for $\hat{t}$ and the subtree rooted at $\hat{t}$ is the same as in $\mathcal{T}$.
By \cref{lem:preserveslim}, we can we apply $\MakeSlimTwo$ to ensure that $\mathcal{T}'$ remains \slim.
Moreover, we again have that no vertices from $S$ are moved and if the size of $X_{t^\star}$ remains unchanged, also no vertices from $D_R\cup D_{L_1}\cup D_{L_2}$ are moved.
Furthermore, we have that the bags of~$t^\star$, its parent, and its $S$-children all receive the same change, namely that the vertex $v$ is removed.     

Finally, consider the case that $N(v)\setminus V_{t_1}=\{u\}$. 
We subdivide the edge between $t^\star$ and~$t_1$ and let $t'$ be the new node. We set $X_{t'}=X_{t_1}$. Swap the names of $t_1$ and $t'$ such that $t_1$ remains the child of $t^\star$. Now we can perform the modification $\BringNeighborDown(t',v,u)$ (\cref{def:bringdown}) to~$\mathcal{T}'$.
By \cref{lem:bringupdown} we have that afterwards, $\mathcal{T}'$ is still $S$-nice, has width at most $k$, and is a sibling of $\mathcal{T}$. 
        By \cref{def:topheavy}, we have that
        if $\mathcal{T}$ is \topheavy for a node $\hat{t}$ such that $\hat{t}$ is not in $T_{t^\star_a}$ and not an ancestor of $t^\star_a$, or
         a descendant of $t^\star$,
then $\mathcal{T}'$ is also \topheavy for $\hat{t}$ and the subtree rooted at $\hat{t}$ is the same as in $\mathcal{T}$.
Furthermore, we have that $v\notin X_{\star}$ and hence $\In\subseteq X_{t^{\star}}$ and $X_{t^{\star}}\cap (\Out\cup\{v\})=\emptyset$.
By \cref{lem:preserveslim}, we can we apply $\MakeSlimTwo$ to ensure that $\mathcal{T}'$ remains \slim.
Moreover, we again have that no vertices from $S$ are moved and if the size of $X_{t^\star}$ remains unchanged, also no vertices from $D_R\cup D_{L_1}\cup D_{L_2}$ are moved. 
Furthermore, we have that the bags of~$t^\star$, its parent, and its $S$-children all receive the same change, namely that the vertex $v$ is removed and the vertex $u$ is added. 
This finishes the proof.
\end{proof}

Next, we show that the Step~\ref{regular:step5} of \cref{alg:small}, if applicable, is correct. Again, this proof has many similarities to the previous ones.

\begin{lemma}\label{lem:regularstep5}
Let $(\Rest,\In,\Out)$ be a 3-partition of $V$ that appears during some iteration of \cref{alg:small}. 
Assume there exists a \slim $S$-nice tree decomposition $\mathcal{T}=(T,\{X_t\}_{t\in V(T)})$ of $G$ with width $k$ that contains a directed path of length at least three with $S$-trace $(\hat{L}^S, \hat{X}^S, \hat{R}^S)$, or if $\hat{L}^S=\emptyset$ then $\mathcal{T}$ contains an $S$-bottom node that has an $S$-child with $S$-trace $(\hat{L}^S, \hat{X}^S, \hat{R}^S)$, such that neither $t_{\max}$ nor the parent of $t_{\max}$ are contained in a full join tree, and with target node ${t^\star}$ such that the following holds.
\begin{itemize}
    \item $\In\subseteq X_{t^\star}$ and $X_{t^\star}\cap \Out=\emptyset$, 
                \item if $t^\star$ has a parent $t_p$, then $X_{t_p}=X_{t^\star}\setminus D_R$, 
    \item if $t^\star$ has one $S$-child $t_1$, then $X_{t_1}=X_{t^\star}\setminus D_{L_1}$, 
    \item if $t^\star$ has two $S$-children $t_1$ and $t_2$, then $X_{t_1}=X_{t^\star}\setminus D_{L_1}$ and $X_{t_2}=X_{t^\star}\setminus D_{L_2}$, 
    \item Steps~\ref{regular:step1},~\ref{regular:step2},~\ref{regular:step3}, and~\ref{regular:step4} of \cref{alg:small} do not apply, and
    \item there is a vertex $v\in\Rest$ that has degree one in $G[\Rest]$ and the connected component~$C$ of $G[(\Out\cup\{v\})\setminus S]$ that contains $v$ has the property that $N(V(C))\cap \In\neq \In$.
\end{itemize} 
Then there is a \slim $S$-nice tree decomposition $\mathcal{T}'=(T',\{X'_t\}_{t\in V(T')})$ of $G$ with width at most $k$ 
such that the following holds.
\begin{itemize}
\item $\mathcal{T}'$ is a sibling of $\mathcal{T}$,
\item if $\mathcal{T}$ is \topheavy for a node $\hat{t}$ such that $\hat{t}$ is
\begin{itemize}
    \item not in $T_{t^\star_a}$ and not an ancestor of $t^\star_a$, or
    \item a descendant of $t^\star$,
\end{itemize}
then $\mathcal{T}'$ is also \topheavy for $\hat{t}$ and the subtree rooted at $\hat{t}$ is the same as in $\mathcal{T}$,
\item $\In\subseteq X'_{t^{\star}}$ and $X'_{t^{\star}}\cap (\Out\cup\{v\})=\emptyset$, and
\item $|X'_{t^{\star}}|<|X_{t^{\star}}|$ or each of the following holds:
\begin{itemize}
    \item if $t^{\star}$ has a parent $t'_p$, then $X'_{t'_p}=X'_{t^{\star}}\setminus D_R$, 
    \item if $t^{\star}$ has one $S$-child $t'_1$, then $X'_{t'_1}=X'_{t^{\star}}\setminus D_{L_1}$, and 
    \item if $t^{\star}$ has two $S$-children $t'_1$ and $t'_2$, then $X'_{t'_1}=X'_{t^{\star}}\setminus D_{L_1}$ and $X'_{t'_2}=X_{t^{\star}}\setminus D_{L_2}$.
\end{itemize}
\end{itemize}
\end{lemma}
\begin{proof}
    Let $v\in\Rest$ be a vertex that has degree one in $G[\Rest]$ and the connected component~$C$ of $G[(\Out\cup\{v\})\setminus S]$ that contains $v$ has the property that $N(V(C))\cap \In\neq \In$. 
    Consider $\mathcal{T}$ the target node~$t^\star$ in $\mathcal{T}$.
If $v\notin X_{t^\star}$, then we are done. 
Hence, assume that $v\in X_{t^\star}$. 
We will argue that we can modify the tree decomposition~$\mathcal{T}$ without increasing its width using the operations introduced in \cref{sec:mod} such that afterwards, the requirements in the lemma statement hold.
    
    To this end, consider the connected component $C$ in $G[(\Out\cup\{v\})\setminus S]$ that contains $v$. Since Steps~\ref{regular:step1},~\ref{regular:step2},~\ref{regular:step3}, and~\ref{regular:step4} of \cref{alg:small} do not apply, we know that $C$ is an $F$-component. Let $C_1,\ldots,C_\ell$ be the connected components of $C-\{v\}$. By \cref{lem:oneneighbor} we have for each $C_i$ with $1\le i\le \ell$ that $|N(V(C_i))\cap\Rest|\le 1$. We can conclude that $N(V(C_i))\cap\Rest=\{v\}$ for each $1\le i\le \ell$. Since $v$ has degree one in $G[\Rest]$, we have that $N(V(C))\cap \Rest=\{u\}$ for some $u\in \Rest$. 
The fact that $N(V(C_i))\cap\Rest=\{v\}$ for each $1\le i\le \ell$ together with the fact that $C$ is an $F$-component also implies that $N(V(C)\setminus\{v\})\subseteq \In\cup\{v\}\subseteq X_{t^\star}$.   
   Denote $M=V(C)\setminus\{v\}$. Let $t_p$ denote the parent of $t^\star$. 

        We modify $\mathcal{T}$ as follows. Let $w\in \In$ such that $N(V(C))\cap \In\subseteq \In\setminus \{w\}$. We add a new child $t'$ to ${t^\star}$ with bag $X_{t'}=X_{t^\star}\setminus\{w\}$. We can observe the following.
        \begin{itemize}
            \item We have $N[V(C)]\cap X_{t^\star}\subseteq X_{t'}$,
            \item we have $|X_{t'}|<|X_{t^\star}|$, and
            \item we have that bag $t'$ has $S$-trace $(\emptyset,X^S,L^S\cup R^S)$ and hence, by \cref{def:sparentchild}, is not an $S$-child of ${t^\star}$.
        \end{itemize}
          By construction, we have that all neighbors of the set $M$ are contained in the (trivial) subtree rooted at $t'$, since $N(M)\subseteq N[M\cup\{v\}]$ and $N(M)\subseteq X_{t^\star}$, and hence $N(M)\subseteq N[M\cup\{v\}]\cap X_{t^\star}\subseteq X_{t'}$.
        This allows us to perform a $\MoveIntoSubtree$ modification (\cref{def:move}). We apply $\MoveIntoSubtree(t',M)$ to $\mathcal{T}$. Let $\mathcal{T}'$ denote the resulting tree decomposition.
By \cref{lem:moveremove} we have that afterwards, 
$\mathcal{T}'$ is still $S$-nice, has width at most $k$, and is a sibling of $\mathcal{T}$. 
        By \cref{def:topheavy}, we have that
        if $\mathcal{T}$ is \topheavy for a node $\hat{t}$ such that $\hat{t}$ is not in $T_{t^\star_a}$ and not an ancestor of $t^\star_a$, or
         a descendant of $t^\star$,
then $\mathcal{T}'$ is also \topheavy for $\hat{t}$ and the subtree rooted at $\hat{t}$ is the same as in $\mathcal{T}$.
By \cref{lem:preserveslim}, we can we apply $\MakeSlimTwo$ to ensure that $\mathcal{T}'$ remains \slim.
Moreover, we have that no vertices from $S$ are moved and if the size of $X_{t^\star}$ remains unchanged, also no vertices from $D_R\cup D_{L_1}\cup D_{L_2}$ are moved.
Finally, we have by the definition of $\MoveIntoSubtree(t',M)$ (\cref{def:move}) that $X_{t^\star}\cap \Out=\emptyset$ and, since $M\cap \In=\emptyset$, that $\In\subseteq X_{t^\star}$. Furthermore, we have that $N[M]\subseteq V_{t'}$.

        Observe that $N(v)\setminus V_{t'}\subseteq \{u\}$. Assume that it is not, then there is some $u'\neq u$ such that $u'\in N(v)$ and $u'\notin V_{t'}$. It follows that $u'\notin N[M]$ which implies that $u'\notin \Out$. Since $u'\in N[M\cup\{v\}]$ we have that $u'\notin X_{t^\star}$ which implies that $u'\notin \In$. We must have that $u'\in\Rest$. This is a contradiction to $\{u\}=N(V(C))\cap \Rest$.
We can conclude that $N(v)\setminus V_{t'}\subseteq \{u\}$.

Consider the case that $N(v)\setminus V_{t'}=\emptyset$.
We subdivide the edge between $t^\star$ and~$t'$ and let~$t''$ be the new node. We set $X_{t''}=X_{t'}$. Swap the names of $t'$ and $t''$ such that $t'$ remains the child of $t^\star$. Then we can apply $\MoveIntoSubtree(t',\{v\})$ to~$\mathcal{T}'$. 
By \cref{lem:moveremove} we have that afterwards, $\mathcal{T}'$ is still $S$-nice, has width at most $k$, and is a sibling of $\mathcal{T}$. 
        By \cref{def:topheavy}, we have that
        if $\mathcal{T}$ is \topheavy for a node $\hat{t}$ such that $\hat{t}$ is not in $T_{t^\star_a}$ and not an ancestor of $t^\star_a$, or
         a descendant of $t^\star$,
then $\mathcal{T}'$ is also \topheavy for $\hat{t}$ and the subtree rooted at $\hat{t}$ is the same as in $\mathcal{T}$.
By \cref{lem:preserveslim}, we can we apply $\MakeSlimTwo$ to ensure that $\mathcal{T}'$ remains \slim.
Moreover, we again have that no vertices from $S$ are moved and if the size of $X_{t^\star}$ remains unchanged, also no vertices from $D_R\cup D_{L_1}\cup D_{L_2}$ are moved. 
Furthermore, we have that the bags of~$t^\star$, its parent, and its $S$-children all receive the same change, namely that the vertex $v$ is removed.

Finally, consider the case that $N(v)\setminus V_{t'}=\{u\}$. 
We subdivide the edge between $t^\star$ and~$t'$ and let $t''$ be the new node. We set $X_{t''}=X_{t'}$. Swap the names of $t'$ and $t''$ such that $t'$ remains the child of $t^\star$.
Now we can perform the modification $\BringNeighborDown(t',v,u)$ (\cref{def:bringdown}) to~$\mathcal{T}'$.
By \cref{lem:bringupdown} we have that afterwards, $\mathcal{T}'$ is still $S$-nice, has width at most $k$, and is a sibling of $\mathcal{T}$. 
        By \cref{def:topheavy}, we have that
        if $\mathcal{T}$ is \topheavy for a node $\hat{t}$ such that $\hat{t}$ is not in $T_{t^\star_a}$ and not an ancestor of $t^\star_a$, or
         a descendant of $t^\star$,
then $\mathcal{T}'$ is also \topheavy for $\hat{t}$ and the subtree rooted at $\hat{t}$ is the same as in $\mathcal{T}$.
Furthermore, we have that $v\notin X_{t_\star}$ and hence $\In\subseteq X_{t^{\star}}$ and $X_{t^{\star}}\cap (\Out\cup\{v\})=\emptyset$.
By \cref{lem:preserveslim}, we can we apply $\MakeSlimTwo$ to ensure that $\mathcal{T}'$ remains \slim. 
Moreover, we have again that no vertices from $S$ are moved and if the size of $X_{t^\star}$ remains unchanged, also no vertices from $D_R\cup D_{L_1}\cup D_{L_2}$ are moved.
Furthermore, we have that the bags of~$t^\star$, its parent, and its $S$-children all receive the same change, namely that the vertex $v$ is removed and the vertex $u$ is added. 
    This finishes the proof.
\end{proof}

Finally, we show that Step~\ref{regular:step6} of \cref{alg:small} is correct.

\begin{lemma}\label{lem:regularstep6}
Let $(\Rest,\In,\Out)$ be a 3-partition of $V$ that appears during some iteration of \cref{alg:small}. 
Assume there exists a \slim $S$-nice tree decomposition $\mathcal{T}=(T,\{X_t\}_{t\in V(T)})$ of $G$ with width $k$ that contains a directed path of length at least three with $S$-trace $(\hat{L}^S, \hat{X}^S, \hat{R}^S)$, or if $\hat{L}^S=\emptyset$ then $\mathcal{T}$ contains an $S$-bottom node that has an $S$-child with $S$-trace $(\hat{L}^S, \hat{X}^S, \hat{R}^S)$, such that neither $t_{\max}$ nor the parent of $t_{\max}$ are contained in a full join tree, and with target node ${t^\star}$ such that the following holds.
\begin{itemize}
    \item $\In\subseteq X_{t^\star}$ and $X_{t^\star}\cap \Out=\emptyset$, 
                \item if $t^\star$ has a parent $t_p$, then $X_{t_p}=X_{t^\star}\setminus D_R$, 
    \item if $t^\star$ has one $S$-child $t_1$, then $X_{t_1}=X_{t^\star}\setminus D_{L_1}$, 
    \item if $t^\star$ has two $S$-children $t_1$ and $t_2$, then $X_{t_1}=X_{t^\star}\setminus D_{L_1}$ and $X_{t_2}=X_{t^\star}\setminus D_{L_2}$, 
    \item Steps~\ref{regular:step1},~\ref{regular:step2},~\ref{regular:step3},~\ref{regular:step4}, and~\ref{regular:step5} of \cref{alg:small} do not apply, and
    \item there is a vertex $v\in\Rest$ that has degree at most one in $G[\Rest]$.
\end{itemize} 
Then there is a \slim $S$-nice tree decomposition $\mathcal{T}'=(T',\{X'_t\}_{t\in V(T')})$ of $G$ with width at most $k$ 
such that the following holds.
\begin{itemize}
\item $\mathcal{T}'$ is a sibling of $\mathcal{T}$,
\item if $\mathcal{T}$ is \topheavy for a node $\hat{t}$ such that $\hat{t}$ is
\begin{itemize}
    \item not in $T_{t^\star_a}$ and not an ancestor of $t^\star_a$, or
    \item a descendant of $t^\star$,
\end{itemize}
then $\mathcal{T}'$ is also \topheavy for $\hat{t}$ and the subtree rooted at $\hat{t}$ is the same as in $\mathcal{T}$,
\item $\In\subseteq X'_{t^{\star}}$ and $X'_{t^{\star}}\cap (\Out\cup\{v\})=\emptyset$, and
\item $|X'_{t^{\star}}|<|X_{t^{\star}}|$ or each of the following holds:
\begin{itemize}
    \item if $t^{\star}$ has a parent $t'_p$, then $X'_{t'_p}=X'_{t^{\star}}\setminus D_R$, 
    \item if $t^{\star}$ has one $S$-child $t'_1$, then $X'_{t'_1}=X'_{t^{\star}}\setminus D_{L_1}$, and 
    \item if $t^{\star}$ has two $S$-children $t'_1$ and $t'_2$, then $X'_{t'_1}=X'_{t^{\star}}\setminus D_{L_1}$ and $X'_{t'_2}=X_{t^{\star}}\setminus D_{L_2}$.
\end{itemize}
\end{itemize}
\end{lemma}
\begin{proof}
Let $v\in\Rest$ be a vertex that has degree at most one in $G[\Rest]$. 
    Consider $\mathcal{T}$ the target node~$t^\star$ in $\mathcal{T}$.
If $v\notin X_{t^\star}$, then we are done. 
Hence, assume that $v\in X_{t^\star}$. 
We will argue that we can modify the tree decomposition~$\mathcal{T}$ without increasing its width using the operations introduced in \cref{sec:mod} such that afterwards, the requirements in the lemma statement hold.

To this end, consider the connected component $C$ in $G[(\Out\cup\{v\})\setminus S]$ that contains $v$. Since Steps~\ref{regular:step1},~\ref{regular:step2},~\ref{regular:step3}, and~\ref{regular:step4} of \cref{alg:small} do not apply, we know that $C$ is an $F$-component.
     Let $C_1,\ldots,C_\ell$ be the connected components of $C-\{v\}$. By \cref{lem:oneneighbor} we have for each $C_i$ with $1\le i\le \ell$ that $|N(V(C_i))\cap\Rest|\le 1$. We can conclude that $N(V(C_i))\cap\Rest=\{v\}$ for each $1\le i\le \ell$. Since $v$ has degree at most one in $G[\Rest]$, we have that $N(V(C))\cap \Rest\subseteq\{u\}$ for some $u\in \Rest$. 
    The fact that $N(V(C_i))\cap\Rest=\{v\}$ for each $1\le i\le \ell$ together with the fact that $C$ is an $F$-component also implies that $N(V(C)\setminus\{v\})\subseteq \In\cup\{v\}\subseteq X_{t^\star}$.

\begin{enumerate}
        \item Assume that $v$ has degree zero in $G[\Rest]$.
        Then, we have that $N(V(C))\cap \Rest=\emptyset$. Furthermore, we have that $N(V(C))\subseteq \In\cup\{v\}\subseteq X_{t^\star}$. 
        Denote $M=V(C)$. 
        We modify~$\mathcal{T}$ as follows. We add a new child $t'$ to $t^\star$ with bag $X_{t'}=X_{t^\star}$. We can observe that bag~$t'$ has $S$-trace $(\emptyset,X^S,L^S\cup R^S)$ and hence, by \cref{def:sparentchild}, is not an $S$-child of $t^\star$.

By construction, we have that all neighbors of the set $M$ are contained in the (trivial) subtree rooted at $t^\star$.
This allows us to perform a $\MoveIntoSubtree$ modification (\cref{def:move}). We apply $\MoveIntoSubtree(t',M)$ to $\mathcal{T}$. Let $\mathcal{T}'$ denote the resulting tree decomposition.
By \cref{lem:moveremove} we have that afterwards, $\mathcal{T}'$ is still $S$-nice, has width at most $k$, and is a sibling of $\mathcal{T}$. 
        By \cref{def:topheavy}, we have that
        if $\mathcal{T}$ is \topheavy for a node $\hat{t}$ such that $\hat{t}$ is not in $T_{t^\star_a}$ and not an ancestor of $t^\star_a$, or
         a descendant of $t^\star$,
then $\mathcal{T}'$ is also \topheavy for $\hat{t}$ and the subtree rooted at $\hat{t}$ is the same as in $\mathcal{T}$.
By \cref{lem:preserveslim}, we can we apply $\MakeSlimTwo$ to ensure that $\mathcal{T}'$ remains \slim.
Moreover, we have that no vertices from $S$ are moved and if the size of $X_{t^\star}$ remains unchanged, also no vertices from $D_R\cup D_{L_1}\cup D_{L_2}$ are moved.
Furthermore, we have that $v\notin X_{\star}$ and hence $\In\subseteq X_{t^{\star}}$ and $X_{t^{\star}}\cap (\Out\cup\{v\})=\emptyset$. 

        \item Assume that $v$ has degree one in $G[\Rest]$. 
        We first show that we must have $|\In|+1\le k$. This will help us to argue that we can perform a $\BringNeighborDown$ modification if it turns out to be necessary. Recall that this modification requires the bag of the node that it is applied to to be not full.
        Recall that $\Rest\cap S=\emptyset$ and hence $G[\Rest]$ is a forest. Then there exists a vertex $v'\neq v$ that also has degree one in $G[\Rest]$ such that $v$ and~$v'$ are connected in $G[\Rest]$. 
Let $C'$ be the connected component in $G[(\Out\cup\{v'\})\setminus S]$ that contains $v'$. We claim that $V(C)\cap V(C')=\emptyset$. Assume for contradiction that $v''\in V(C)\cap V(C')$. We know that $v\neq v''\neq v'$ since $v\notin V(C')$ and $v'\notin V(C)$. It follows that $v''\in\Out$. Let $C''$ be the connected component in $G[\Out\setminus S]$ that contains $v''$. Then $\{v,v'\}\subseteq N(C'')\cap \Rest$. By \cref{lem:oneneighbor} this is a contradiction. We can conclude that $V(C)\cap V(C')=\emptyset$.
        Since Step~\ref{regular:step5} of \cref{alg:small} does not apply, we have that $\In\subseteq N(V(C))$ and $\In\subseteq N(V(C'))$. Let $P$ denote a path from $v$ to $v'$ in $G[\Rest]$. Consider the graph $G'=G[\In\cup V(C)\cup V(C')\cup V(P)]$. Add edges between all pairs of vertices in $\In$ to $G'$ (such that $G'[\In]$ is complete). Since all vertices in $\In$ are in one bag of $\mathcal{T}'$ we have by \cref{lem:cliquebag2} that $\tw(G')\le k$. Now let $H$ be the graph obtained from $G'$ by contracting $G'[(V(C)\cup V(P))\setminus \{v'\}]$ (that is, exhaustively contracting all edges between vertices in $V(C)$), and by contracting $G'[V(C')]$. Observe that $H$ is a complete graph with $|\In+2|$ vertices, and that $H$ is a minor of $G'$. It follows by \cref{lem:cliquebag} that $\tw(H)=|\In|+1$ and that by \cref{lem:twminor} that $\tw(H)\le\tw(G')$. We can conclude that $|\In|+1\le k$.

We modify $\mathcal{T}$ as follows. We add a new child $t'$ to $t^\star$ with bag $X_{t'}=X_{t^\star}\cap N[V(C)]$. We can observe that bag $t'$ has $S$-trace $(\emptyset,X^S,L^S\cup R^S)$ and hence, by \cref{def:sparentchild}, is not an $S$-child of $t^\star$.
Note that $X_{t'}\subseteq \In\cup\{u,v\}$, since $X_{t^\star}\cap\Out=\emptyset$, $N(V(C))\cap \Rest\subseteq\{u\}$, and $V(C)\cap\Rest=\{v\}$.
We can conclude the following:
\begin{itemize}
    \item If $u\in X_{t'}$, then $N(V(C))\subseteq X_{t'}$.
    \item If $u\notin X_{t'}$, then $|X_{t'}|\le |\In|+1\le k$.
\end{itemize}
In the first of those two cases behaves the same as Case 1 in the beginning of  the proof, where $v$ has degree zero in $G[\Rest]$. We apply $\MoveIntoSubtree(t',V(C))$ to $\mathcal{T}$. The remainder of the proof is analogous.

From now on, assume we are in the case where $u\notin X_{t'}$ and $|X_{t'}|\le k$.
Denote $M=V(C)\setminus\{v\}$. 
By construction, we have that all neighbors of the set $M$ are contained in the (trivial) subtree rooted at $t'$.
 This allows us to perform a $\MoveIntoSubtree$ modification (\cref{def:move}). We apply $\MoveIntoSubtree(t',M)$ to $\mathcal{T}$. Let $\mathcal{T}'$ denote the resulting tree decomposition.
By \cref{lem:moveremove} we have that afterwards, $\mathcal{T}'$ is still $S$-nice, has width at most $k$, and is a sibling of $\mathcal{T}$. 
        By \cref{def:topheavy}, we have that
        if $\mathcal{T}$ is \topheavy for a node $\hat{t}$ such that $\hat{t}$ is not in $T_{t^\star_a}$ and not an ancestor of $t^\star_a$, or
         a descendant of $t^\star$,
then $\mathcal{T}'$ is also \topheavy for $\hat{t}$ and the subtree rooted at $\hat{t}$ is the same as in $\mathcal{T}$.
By \cref{lem:preserveslim}, we can we apply $\MakeSlimTwo$ to ensure that $\mathcal{T}'$ remains \slim.
Moreover, we have that no vertices from $S$ are moved and if the size of $X_{t^\star}$ remains unchanged, also no vertices from $D_R\cup D_{L_1}\cup D_{L_2}$ are moved.
Finally, we have by the definition of $\MoveIntoSubtree(t',M)$ (\cref{def:move}) that $X_{t^\star}\cap \Out=\emptyset$ and, since $M\cap \In=\emptyset$, that $\In\subseteq X_{t^\star}$. Furthermore, we have that $N[M]\subseteq V_{t'}$.

        Observe that $N(v)\setminus V_{t'}\subseteq \{u\}$. Assume that it is not, then there is some $u'\neq u$ such that $u'\in N(v)$ and $u'\notin V_{t'}$. It follows that $u'\notin N[M]$ which implies that $u'\notin \Out$. Since $u'\in N[M\cup\{v\}]$ we have that $u'\notin X_{t^\star}$ which implies that $u'\notin \In$. We must have that $u'\in\Rest$. This is a contradiction to $\{u\}=N(V(C))\cap \Rest$.
Now we can perform the modification $\BringNeighborDown(t',v,u)$ (\cref{def:bringdown}) to~$\mathcal{T}'$.
By \cref{lem:bringupdown} we have that afterwards, $\mathcal{T}'$ is still $S$-nice, has width at most $k$, and is a sibling of $\mathcal{T}$. 
        By \cref{def:topheavy}, we have that
        if $\mathcal{T}$ is \topheavy for a node $\hat{t}$ such that $\hat{t}$ is not in $T_{t^\star_a}$ and not an ancestor of $t^\star_a$, or
         a descendant of $t^\star$,
then $\mathcal{T}'$ is also \topheavy for $\hat{t}$ and the subtree rooted at $\hat{t}$ is the same as in $\mathcal{T}$.
By \cref{lem:preserveslim}, we can we apply $\MakeSlimTwo$ to ensure that $\mathcal{T}'$ remains \slim.
Moreover, we again have that no vertices from $S$ are moved and if the size of $X_{t^\star}$ remains unchanged, also no vertices from $D_R\cup D_{L_1}\cup D_{L_2}$ are moved. 
Furthermore, we have that the bags of~$t^\star$, its parent, and its $S$-children all receive the same change, namely that the vertex $v$ is removed and the vertex $u$ is added. 
    \end{enumerate}
    This finishes the proof.
\end{proof}

\cref{prop:smallcorrect} now follows directly from \cref{lem:regularstep1,lem:regularstep2,lem:regularstep3,lem:regularstep4,lem:regularstep5,lem:regularstep6}.

\subsection{Summary and Putting the Pieces Together}\label{sec:pieces}

Recall from the beginning of \cref{sec:bags} that we wish to compute the following.
\begin{itemize}
    \item A candidate for the bag $X^\phi_{\max}$ of $t_{\max}$.
    \item A candidate for the bag $X_p^\phi$ of the parent of $t_{\max}$ except for the case where $\tau_+=\void$ (in this case $t_{\max}$ is the root of the tree decomposition).
    \item A candidate for the bag $X^\phi_{\min}$ of $t_{\min}$ except for the case where $\tau_-=\void$ (in this case $t_{\min}$ is not defined).
    \item A candidate for the union of the bags $X_c^\phi$ of the $S$-children of $t_{\min}$ except for the case where $\tau_-=\void$ (in this case $t_{\min}$ is not defined).
    \item A candidate for the set $X^\phi_{\dpath}$ of vertices that we want to add to bags of the directed path or subtrees of the tree decomposition that are rooted at children of nodes in the directed path.
\end{itemize}

In \cref{sec:bigops,sec:smallops} we explained how to compute $X^\phi_{\max}$, depending on whether the top operation $\tau_+$ is a 
\begin{enumerate}
\item $\bigjoin$ $S$-operation or $\tau_+=\extendedforget(v,d,\true,\tau)$ and $\tau$ is a $\bigjoin$ $S$-operation, or
\item any of the remaining cases except $\tau_+=\void$.
\end{enumerate}
In the first case, the bag $X^\phi_{\max}$ and the bag $X_p^\phi$ of the parent of $t_{\max}$ can directly be determined via the additional information in the $\bigjoin$ $S$-operation. 
In the second case, \cref{alg:smallwrapper} outputs all necessary information to determine $X^\phi_{\max}$ and the bag of the parent of $t_{\max}$.
Note that in both cases, we also have that $X^\phi_{\max}\subseteq X_p^\phi$ unless $\tau_+=\extendedforget(v,d,f,\tau)$, then we have that $X^\phi_{\max}= X_p^\phi\cup\{v\}$.
In the case where $\tau_+=\void$, we set $X^\phi_{\max}=X^S$ and $X_p^\phi=\emptyset$.

In the case where we wish to determine $X^\phi_{\min}$, we do the following. Note that the parent of the top node of the directed path corresponding to the $S$-trace of a predecessor state is~$t_{\min}$. This means that we can apply our algorithms to a predecessor state $\psi$ to determine~$X^\phi_{\min}$. 
More specifically, we have that $X^\phi_{\min}=X_p^\psi$.
Since the computation of $X_p^\psi$ is independent from the bottom operation of state $\psi$, we can pick an arbitrary predecessor state $\psi$ of $\phi$ for the computation. In the case where $\tau_-=\void$ and $\phi$ does not have any predecessor states, we set $X^\phi_{\min}=X^S$. Similarly, we have that $X^\phi_c=X_{\max}^\psi$ if $\tau_-\neq\void$ and $\tau_-$ is neither $\smalljoin$ nor $\bigjoin$. In the case where $\tau_-$ is a $\smalljoin$ or a $\bigjoin$ $S$-operation, there is a pair we pick an arbitrary pair $\psi_1,\psi_2$ of predecessor states (see \cref{def:predecingstraces}) and we have that $X^\phi_c=X_{\max}^{\psi_1}\cup X_{\max}^{\psi_2}$. If $\tau_-=\void$, then we set $X^\phi_c=\emptyset$.


From the definition of $X^\phi_{\dpath}$ in \cref{sec:ltwdef} it follows that, given $X^\phi_{\max}$, $X^\phi_p$, $X^\phi_{\min}$, and $X^\phi_c$, we can compute $X^\phi_{\dpath}$ in polynomial time.

\section{Computing the Local Treewidth}\label{sec:localtw}
In this section we present a polynomial-time algorithm to compute the function $\ltw$. Formally, we show the following.

\begin{proposition}\label{prop:computeltw}
The function $\ltw:\mathcal{S}\rightarrow\{\true,\false\}$ can be computed in polynomial time.
\end{proposition}

By \cref{def:ltw}, the task of the function $\ltw$ is to compute the treewidth of graph $G'[X^\phi_{\max}\cup X^\phi_p\cup X^\phi_{\min}\cup X^\phi_{\dpath}]$, where $G'$ is the graph obtained from the input graph $G$ by adding edges between all pairs of vertices $u,v\in  X^\phi_{\max}$ with $\{u,v\}\notin E$, all pairs of vertices $u,v\in  X^\phi_p$ with $\{u,v\}\notin E$, and all pairs of vertices $u,v\in X^\phi_{\min}$ with $\{u,v\}\notin E$. Note that by definition we have $X^\phi_{\max}\subseteq X^\phi_p$ or $X^\phi_p\subseteq X^\phi_{\max}$ and that $G[X^\phi_{\max}\cup X^\phi_p\cup X^\phi_{\min}\cup X^\phi_{\dpath}]-(X^\phi_{\max}\cap X^\phi_p\cap X^\phi_{\min})$ is a forest.
The latter will allow us to infer some important properties on the edge set $E(G'[X^\phi_{\max}\cup X^\phi_p\cup X^\phi_{\min}\cup X^\phi_{\dpath}])$.

We will first show how to solve the following simpler problem: Given a graph $G=(V_1\cup V_2,E)$ such that $G[V_1]$ is complete and $G[V_2]$ is a forest, compute an optimal tree decomposition for $G$. We show that this problem can be solved in polynomial time and we will use this as a subroutine in the algorithm for computing $\ltw$. 
We can observe the following. Since $G[V_1]$ is complete, we have by \cref{lem:twminor,lem:cliquebag} that $\tw(G)\ge |V_1|-1$. Furthermore, since $G[V_2]$ is a forest, we have that $\tw(G)\le |V_1|+1$, since we can obtain a tree decomposition for $G$ by taking any tree decomposition for $G[V_2]$ (of width one) and adding $V_1$ to every bag. Hence, the task of the algorithm is to decide which of the three options applies, and to construct a tree decomposition of the corresponding width. Formally, we show the following.

\begin{lemma}\label{lem:twoneclique}
Let $G=(V_1\cup V_2,E)$ be a graph such that $G[V_1]$ is complete and $G[V_2]$ is a forest. A tree decomposition for $G$ with minimum width can be computed in polynomial time.
\end{lemma}
\begin{proof}
We compute a tree decomposition for $G$ as follows. We start with a root node that has the set $V_1$ as its bag. For every connected component $C$ in $G-V_1$ we create a child  node $t_C$ with bag $X_{t_C}=V_1\cap N(V(C))$. 

Consider a connected component $C$. We know that $C$ is a tree. Note that we must have a bag of size at least $|X_{t_C}|+1$, since contracting $C$ to a single vertex creates a complete subgraph with $|X_{t_C}|+1$ vertices and by \cref{lem:twminor,lem:cliquebag} we have that the treewidth of $G$ is at least $|X_{t_C}|$. Denote $s=|X_{t_C}|$. We must have a bag of size at least $s+1$.

If there is a $v\in V(C)$ such that for each connected component $C'$ in $C-\{v\}$ we have that $N(C')\subseteq (X_{t_C}\cup\{v\})\setminus \{u_{C'}\}$ for some $u_{C'}\in X_{t_C}$, then we do the following.
We compute a tree decomposition $\mathcal{T}_{C'}$ for each $G[V(C')\cup\{v\}]$ and add $X_{t_C}\setminus \{u_{C'}\}$ to every bag. We pick an arbitrary root of $\mathcal{T}_{C'}$ such that its bag contains $v$, and make it a child of $t_C$.
We add $v$ to~$X_{t_C}$.
It is straightforward to check that all conditions of \cref{def:tree_decomposition} are fulfilled.
Note that the bag of $t_C$ as well as all bags of nodes in $\mathcal{T}_{C'}$ have size at most $s+1$. Hence, the width of the constructed tree decomposition is optimal.

Now consider the case that there is no vertex $v\in V(C)$ such that for each connected component $C'$ in $C-\{v\}$ we have that $N(C')\subseteq (X_{t_C}\cup \{v\})\setminus \{u_{C'}\}$ for some $u_{C'}\in X_{t_C}$. 
In order to create a valid tree decomposition, we need to add a vertex from $V(C)$ to the bag $X_{t_C}$ before we can remove any vertex, since otherwise we obtain a contradiction to at least one of the conditions of \cref{def:tree_decomposition}.
For each vertex $v$ that we could add to to $X_{t_C}$, we have that there is a connected component $C'$ in $C-\{v\}$ such that $N(C')=X_{t_C}\cup\{v\}$. By repeating the previous argument, we get that we need to add at least one additional vertex from $V(C')$ to $X_{t_C}$ before we can remove any vertex. It follows that the minimum bag size is at least $s+2$. 
This allows us to simply take any optimal tree decomposition for $C$ and add $X_{t_C}$ to every bag. 

It is straightforward to see that repeating the above described process for every connected component $C$ in $G-V_1$ yields an optimal tree decomposition for $G$.
\end{proof}

Now we prove \cref{prop:computeltw} with the help of \cref{lem:twoneclique}.


\begin{proof}[Proof of \cref{prop:computeltw}]
First, note that the sets $X^\phi_{\max}$, $X^\phi_p$, and $X^\phi_{\min}$ can be computed in polynomial time by \cref{obs:algo1running,obs:algo2running}. By the definition of $X^\phi_{\dpath}$ given in \cref{sec:ltwdef} we have that we can compute $X^\phi_{\dpath}$ from $X^\phi_{\max}$, $X^\phi_p$, and $X^\phi_{\min}$ in polynomial time.

We compute a tree decomposition $\mathcal{T}$ for $G'[X^\phi_{\max}\cup X^\phi_p\cup X^\phi_{\min}\cup X^\phi_{\dpath}]$ as follows. 
For simplicity, we denote $H=G'[X^\phi_{\max}\cup X^\phi_p\cup X^\phi_{\min}\cup X^\phi_{\dpath}]$, $X_1=X^\phi_{\max}\cup X^\phi_p$, and $X_2=X^\phi_{\min}$.
We start $\mathcal{T}$ with a root node $t_1$ that has the set $X_1$ as its bag and has one child node $t_2$ that has the set~$X_2$ as its bag. We call $F=H-(X_1\cup X_2)$ the \emph{remaining forest}.


We iterative turn $\mathcal{T}$ into a tree decomposition for the whole graph $H$ as follows.
For convenience, refer with $X_1$ and $X_2$ to the bags of $t_1$ and $t_2$, respectively. Along the way we will modify $H$ until it is empty. 
We first repeat the following as long as it applies (we refer to this as Step~0):
\begin{itemize}
\item If there is a connected component $C$ in $F$ such that $N(V(C))\subseteq X_\ell$ for some $\ell\in\{1,2\}$, then we compute a rooted tree decomposition $\mathcal{T}_C$ for $G''[X_\ell\cup V(C)]$, where $G''$ is obtained from $G$ by adding edges between all pairs of vertices $u,v\in  X_\ell$ with $\{u,v\}\notin E$, using \cref{lem:twoneclique}. By \cref{lem:twminor,lem:cliquebag} we can assume that the root of $\mathcal{T}_C$ has bag $X_\ell$. We make the root a child of $t_\ell$ and remove $C$ from the remaining forest $F$ and from $H$.

\end{itemize}
Now, while $X_1\triangle X_2\neq \emptyset$, we perform the first applicable step.
\begin{enumerate}
\item If there is a connected component $C$ in $F$ such that $|N(V(C))\cap (X_1\triangle X_2)|=1$ and $N(V(C))\subseteq X_\ell$ for some $\ell\in\{1,2\}$, then we compute a rooted tree decomposition $\mathcal{T}_C$ for $G''[X_\ell\cup V(C)]$, where $G''$ is obtained from $G$ by adding edges between all pairs of vertices $u,v\in  X_\ell$ with $\{u,v\}\notin E$, using \cref{lem:twoneclique}. By \cref{lem:twminor,lem:cliquebag} we can assume that the root of $\mathcal{T}_C$ has bag $X_\ell$. We make the root a child of $t_\ell$ and remove $C$ from the remaining forest $F$ and from $H$. \label{ltw:1}
\item If there is a vertex $v\in X_\ell\setminus X_{3-\ell}$  for some $\ell\in\{1,2\}$ such that $N(v)\subseteq X_\ell$ (where the neighborhood is taken in the modified graph $H$), then subdivide the edge in $\mathcal{T}$ between~$t_\ell$ and $t_{3-\ell}$. Call the new node $t_\ell$ and remove $v$ from its bag $X_\ell$ and from $H$.\label{ltw:2}
\item If there is a vertex $v\in X_\ell\setminus X_{3-\ell}$  for some $\ell\in\{1,2\}$ such that $N(v)\setminus X_\ell=\{u\}$ (where the neighborhood is taken in the modified graph $H$), then subdivide the edge in $\mathcal{T}$ between~$t_\ell$ and $t_{3-\ell}$. Call the new node $t_\ell$ and add $u$ to its bag $X_\ell$. 
Subdivide the edge in $\mathcal{T}$ between (the new) $t_\ell$ and $t_{3-\ell}$. Call the new node $t_\ell$ and remove $v$ from its bag $X_\ell$ and from $H$.\label{ltw:3}
\end{enumerate}
When $X_1\triangle X_2= \emptyset$, add an edge between nodes $t_1$ and $t_2$. Note that then $X_1=X_2$. If the remaining forest is not empty, we compute a rooted tree decomposition $\mathcal{T}_F$ for $G''[X_1\cup V(F)]$, where $G''$ is obtained from $G$ by adding edges between all pairs of vertices $u,v\in  X_1$ with $\{u,v\}\notin E$, using \cref{lem:twoneclique}. By \cref{lem:twminor,lem:cliquebag} we can assume that the root of $\mathcal{T}_F$ has bag $X_1$. We make the root a child of $t_1$. We refer to this as Step 4.

We first argue that while $X_1\triangle X_2\neq \emptyset$, one of the three steps applies. 
To this end we prove the following claim.
\begin{claim}\label{claim:ltw1}
Let $F=(V,E)$ be a forest and $D\subseteq V$. Then one of the following holds.
\begin{itemize}
\item $D=\emptyset$.
\item There is some $v\in D$ with $\deg(v)\le 1$.
\item There is a connected component $C\in F-D$ such that $|N(V(C))\cap D|=1$.
\end{itemize}
\begin{claimproof}
Suppose that $D\neq\emptyset$ and for all $v\in D$ it holds that $\deg(v)>1$.
Since $F$ is a forest, we have that $\sum_{v\in V} \deg(v)=2|E|\le 2|V|-2$. 
Note that $F[D]$ is not a connected component of $F$, since $\sum_{v\in D} \deg(v)\ge 2|D|$ and $F[D]$ is a forest.
Let $F'$ be the forest obtained from $F$ by removing all connected components $C\in F-D$ such that $|N(V(C))\cap D|=0$. Since $F[D]$ is not a connected component of $F$, we have that $F'\neq F[D]$, that is, there are connected components $C$ in $F'-D$ and for each one it holds that $|N(V(C))\cap D|\ge 1$.
We can write 
\[\sum_{v\in V(F')} \deg(v)=\sum_{v\in D} \deg(v) + \sum_{C \text{ is conn.\ comp.\ in }F'-D} (2|V(C)|-2 + |N(V(C))\cap D|).
\]
Here, we use that the connected components in $F'-D$ are trees, and is well-known that for trees $T$ it holds that $\sum_{v\in V(T)}\deg(v)=2|V(T)|-2$.
Assume for contradiction that $|N(V(C))\cap D|\ge 2$ for all connected components $C$ in $F'-D$. Then we have
\[\sum_{v\in V(F')} \deg(v)\ge 2|D| + \sum_{C \text{ is conn.\ comp.\ in }F'-D} 2|V(C)|= 2|V(F')|.
\]
This is a contradiction to $\sum_{v\in V(F')} \deg(v)\le 2|V(F')|-2$. Hence, the claim statement holds.
\end{claimproof}
\end{claim}

Now it is important to note, that by definition of $X^\phi_{\max}$, $X^\phi_p$, $X^\phi_{\min}$, and $X^\phi_{\dpath}$, we have that $G[X^\phi_{\max}\cup X^\phi_p\cup X^\phi_{\min}\cup X^\phi_{\dpath}]-(X^\phi_{\max}\cap X^\phi_p\cap X^\phi_{\min})$ is a forest.
Denote $V^\star=(X^\phi_{\max}\cup X^\phi_p\cup X^\phi_{\min}\cup X^\phi_{\dpath})\setminus (X^\phi_{\max}\cap X^\phi_p\cap X^\phi_{\min})$.
Note that $F^\star=G[V^\star]$ is a forest and $X_1\triangle X_2\subseteq V^\star$. 
Hence, we can apply \cref{claim:ltw1} to conclude that whenever $X_1\triangle X_2\neq\emptyset$, then we have 
\begin{itemize}
\item There is some $v$ in $X_1\triangle X_2$ with $\deg(v)_{F^\star}\le 1$. 

If $\deg(v)_{F^\star}=0$, then we have that Step~\ref{ltw:2} applies.
If $N(v)=\{u\}$ and $u,v\in X_\ell$ for some $\ell\in\{1,2\}$, then we have that Step~\ref{ltw:2} applies.
Otherwise, we have that Step~\ref{ltw:3} applies.
\item There is some connected component $C$ in $F$ such that $|N(V(C))\cap (X_1\triangle X_2)|=1$.
Then we have that Step~\ref{ltw:1} applies.
\end{itemize}
If $X_1\triangle X_2=\emptyset$, then Step~4 applies.

It is straightforward to observe that the number of steps is polynomial. In Step~\ref{ltw:1} we reduce the size of $F$. In Step~\ref{ltw:2} we reduce the size of $X_1\triangle X_2$. In Step~\ref{ltw:3}, if $u$ is a vertex of $F$, we reduce the size of $F$, otherwise we reduce the size of $X_1\triangle X_2$.

It remains to show that the produced tree decomposition respects all conditions of \cref{def:tree_decomposition} and that it has minimum width. First of all, it is easy to see that the union of all bags is $X^\phi_{\max}\cup X^\phi_p\cup X^\phi_{\min}\cup X^\phi_{\dpath}$. Second, we have that Condition~\ref{condition_2_tree_decomposition} of \cref{def:tree_decomposition} holds since whenever we remove a vertex from the current constructed bag $X_\ell$ in Steps~\ref{ltw:2} and~\ref{ltw:3}, it has already been in the same bag with all its neighbors. For Steps~0,~\ref{ltw:1}, and~4 this follows from \cref{lem:twoneclique}. 
Condition~\ref{condition_3_tree_decomposition} of \cref{def:tree_decomposition} holds since we never remove and then re-add the same vertex in Steps~\ref{ltw:2} and~\ref{ltw:3}. For Steps~0,~\ref{ltw:1}, and~4 this follows again from \cref{lem:twoneclique}.

Last, we show that the width of the constructed tree decomposition is $\tw(G'[X^\phi_{\max}\cup X^\phi_p\cup X^\phi_{\min}\cup X^\phi_{\dpath}])$. 
We have by \cref{lem:twoneclique,lem:twminor} that the bags created in Step~0 have size at most $\tw(G'[X^\phi_{\max}\cup X^\phi_p\cup X^\phi_{\min}\cup X^\phi_{\dpath}])+1$.

Next, we show that if Step~\ref{ltw:3} applies for the first time, then $\tw(G'[X^\phi_{\max}\cup X^\phi_p\cup X^\phi_{\min}\cup X^\phi_{\dpath}])\ge\max\{|X_1|,|X_2|\}$. At the end of the proof we discuss the case where Step~\ref{ltw:3} never applies.
Assume w.l.o.g.\ that $|X_1|\ge |X_2|$.
Since $H[X_1]$ and $H[X_2]$ are complete when Step~\ref{ltw:3} applies for the first time, by \cref{lem:cliquebag} we must have a bag that contains $X_1$, and we must have a bag that contains $X_2$.
Since Step~\ref{ltw:1} does not apply, we have that $X_1\setminus X_2\neq\emptyset$ and $X_2\setminus X_1\neq\emptyset$. 
Consider a shortest path in the tree decomposition from a node $t_1$ with bag $X_1$ to a node $t_2$ with bag $X_2$ and assume w.l.o.g.\ that the tree decomposition is nice (\cref{def:nicetd}). Assume that for contradiction that the second node of the path, that is, the one that is visited by the path after $t_1$ has a bag $X^\star$ that is a strict subset of $X_1$. Let $v\in X_1\setminus X^\star$. If $v\in X_2\cap X_1$ we get a contradiction to Condition~\ref{condition_3_tree_decomposition} of \cref{def:tree_decomposition}. Hence, we have that $v\in X_1\setminus X_2$.
Since Step~\ref{ltw:2} does not apply, we have that vertex $v$ has a neighbor in $V(F)\cup (X_2\setminus X_1)$. Let that neighbor be $u$. 
By Condition~\ref{condition_2_tree_decomposition} of \cref{def:tree_decomposition} we have that there is some bag that contains both $u$ and $v$. 
If $u\in X_2\setminus X_1$, then this contradicts Condition~\ref{condition_2_tree_decomposition} of \cref{def:tree_decomposition}.
Now assume that $u\in V(F)$ and let $C$ be the connected component in $F$ that contains $u$. Since Step~\ref{ltw:1} does not apply, we have that $N(V(C))\cap X_2\neq\emptyset$.
It follows from Condition~\ref{condition_2_tree_decomposition} and~\ref{condition_3_tree_decomposition} of \cref{def:tree_decomposition} that the bags containing vertices from $V(C)$ must form a subtree in the tree decomposition. Hence, $X^\star$ must contain a vertex from $V(C)$, a contradiction to the assumption that $X^\star$ that is a strict subset of $X_1$. Hence, $X^\star$ that is a strict superset of $X_1$ (the bags are not the same, since we assumed that the path is a shortest path). It follows that $\tw(H)\ge \max\{|X_1|,|X_2|\}$. From now on, denote $x^\star=\max\{|X_1|,|X_2|\}$ for the sets $X_1$ and $X_2$ when Step~\ref{ltw:3} applies for the first time.


Observe that after each step we have that $|X_\ell|\le x^\star$ for each $\ell\in\{1,2\}$.
We show that none of the steps create a bag whose size is larger than $x^\star+1$.
If Step~\ref{ltw:1} applies, then there is some connected component $C$ in $F$ such that $|N(V(C))\cap (X_1\triangle X_2)|=1$.
Let $N(V(C))\subseteq X_\ell$ and let $N(V(C))\cap (X_1\triangle X_2)=\{v\}$. Note that $H[V(C)\cup\{v\}]$ is a tree and hence as treewidth one. It follows that we can obtain a tree decomposition for $H[X_\ell\cup V(C)]$ by taking an optimal tree decomposition for $H[V(C)\cup\{v\}]$ and adding the set $X_\ell\setminus \{v\}$ to every bag. The maximum bag size is $|X_\ell|+1$. It follows that in Step~\ref{ltw:1}, no bag with size larger than $x^\star+1$ is created.
Step~\ref{ltw:2} never increases the maximium bag size. As discussed above, Step~\ref{ltw:3} increases the bag size to at most $x^\star+1$. 

It remains to show that in Step~4, we do not create bags of size larger than $x^\star+1$. We show this by proving that for every connected component $T$ in the remaining forest (when Step~4 is executed) we have that $T$ has at least one neighbor in $X_1\cup X_2\setminus S$. Then, using the algorithm described in the proof of \cref{lem:twoneclique}, we can obtain a tree decomposition that does not create bags which are larger than $x^\star+1$.
To this end, first observe that no vertices are added to the remaining forest in any of the steps. Second, we claim that if a vertex $w$ is in the remaining forest when Step~4 is executed, then none of its neighbors has been removed from $H$ when Step~4 is executed. Assume for contradiction that there is a neighbor of $w$ that has been removed. In Step~\ref{ltw:1} only whole connected components from the remaining forest are removed. In Step~\ref{ltw:2}, we only remove vertices that do not have neighbors in the remaining forest. Finally, in Step~\ref{ltw:3} we have that if a neighbor $v$ from $w$ is removed, then $v$ has only one neighbor in the remaining forest, denoted by $u$ in the description of Step~\ref{ltw:3}. Hence, we must have that $u=w$. Then $u$ is added to the bag $X_\ell$, and hence $u=w$ is not part of the remaining forest when Step~4 is executed, a contradiction.
Let let $T$ a connected component (tree) in the remaining forest when Step~4 is executed. From the two previous observations it follows that $T$ also is a tree (but not necessarily a connected component) in the remaining forest after Step~0. Consider the case that $T$ has neighbors in $X_1\setminus X_2$ and hence in $X_1\setminus S$ (after Step~0). Since those neighbors are not removed if $T$ remains in the remaining forest. Now consider the case that $T$ has a neighbor that is also in the remaining forest (after Step~0) and hence, not in $S$. Then by the obervations above, this neighbor is not removed from $H$. Since it is not in the remaining forest when Step~4 is executed, it must be in $X_1$. We can conclude that $T$ has at least one neighbor in $X_1\cup X_2\setminus S$.
 Hence, the bag size of our constructed tree decomposition is optimal.

 Finally, consider the case that Step~\ref{ltw:3} never applies. Note that Step~\ref{ltw:1} never adds or removes vertices from $X_\ell$ with $\ell\in\{1,2\}$. Step~\ref{ltw:2} only removes vertices from $X_\ell$ with $\ell\in\{1,2\}$. It follows, that there are never vertices added to $X_\ell$ with $\ell\in\{1,2\}$. We claim that then it is the case that after Step~0, the remaining forest is empty. Assume for contradiction that it is not. Let $T$ be a connected component of the remaining forest after the execution of Step~0. Then we have that $T$ has neighbors in $X_1\setminus X_2$ and in $X_2\setminus X_1$. Note that Step~\ref{ltw:1} only removes connected components from the remaining forest that have at most one neighbor in $X_1\triangle X_2$. Step~\ref{ltw:2} does not remove vertices from the remaining forest. When Step~4 is executed, $T$ is still in the remaining forest and we have that $X_1=X_2$. Since we never add vertices to $X_1$ or $X_2$, we have that not all neighbors of $T$ are in $X_1=X_2$. This contradicts the correctness of the algorithm, which we have proven before. It follows that after the execution of Step~0, the remaining forest is empty. This further implies that also Step~\ref{ltw:1} is never executed. Since Step~\ref{ltw:2} never increases the maximum bag size, we can conlude that also in this case, the constructed tree decomposition is optimal.
\end{proof}

We remark that the above described algorithm explicitly constructs a tree decomposition that it can output. We describe in \cref{sec:correct} how to use this to construct a tree decomposition for the whole input graph $G$.

\section{Legal Predecessor States}\label{sec:legal}

In this section, we formally define the functions $\legal_1$ and $\legal_2$. We show that we can compute the functions using the algorithms introduced in \cref{sec:bags}. Recall that the functions $\legal_1$ and $\legal_2$, intuitively, should check whether a state $\phi$ is compatible with its predecessor state(s). 
%
%
To this end, for a state $\phi=(\tau_-,L^S, X^S, R^S,\tau_+)$, we define three sets $V^\phi_L$, $V^\phi_X$, and $V^\phi_R$ which, intuitively, determine which vertices are in bags ``below'' $t_{\min}$, which vertices are in bags ``between'' $t_{\max}$ and $t_{\min}$, and the remaining vertices.
We use the five sets $X^\phi_{\max}$, $X^\phi_p$, $X^\phi_{\min}$,$X^\phi_c$, and $X^\phi_{\dpath}$ to determine those sets. Using those sets, we define $V^\phi_L$, $V^\phi_X$, and $V^\phi_R$ as follows.
\begin{itemize}
\item $v\in V^\phi_L\Leftrightarrow v\in L^S\cup (X^\phi_c\setminus X^\phi_{\min})$ or $v$ is connected to some $u\in L^S\cup (X^\phi_c\setminus X^\phi_{\min})$ in $G-X^\phi_{\min}$.
\item $V^\phi_X=X^\phi_{\max}\cup X^\phi_p\cup X^\phi_{\min}\cup X^\phi_{\dpath}$.
\item $V^\phi_R=V\setminus (V^\phi_L\cup V^\phi_X)$.
\end{itemize}

Now we are ready to define $\legal_1$ and $\legal_2$ as follows.

\begin{definition}[$\legal_1$]\label{def:legal1}
Let $\phi$ be a state and let $\psi\in\Psi(\phi)$.
Then $\legal_1(\phi,\psi)=\true$ if and only if all of the following three conditions hold.
\begin{enumerate}
\item $V^\psi_L\cup V^\psi_X\subseteq V^\phi_L\cup V^\phi_X$.
\item $N[V^\psi_R]\cap V^\phi_L=\emptyset$.
\item $(V^\psi_L\cup V^\psi_X)\cap V^\phi_X\subseteq X^\phi_{\min}$.
\end{enumerate}
\end{definition}

\begin{definition}[$\legal_2$]\label{def:legal2}
Let $\phi$ be a state let $\psi_1\in\Psi_1(\phi)$, and let $\psi_2=\in\Psi_2(\phi)$.
Then $\legal_2(\phi,\psi_1,\psi_2)=\true$ if and only if all of the following conditions hold.
\begin{enumerate}
\item $V^{\psi_1}_L\cup V^{\psi_1}_X\cup V^{\psi_2}_L\cup V^{\psi_2}_X \subseteq V^\phi_L\cup V^\phi_X$.
\item $N[V^{\psi_1}_R]\cap V^\phi_L\subseteq V^{\psi_2}_L\cup V^{\psi_2}_X$.
\item $N[V^{\psi_2}_R]\cap V^\phi_L\subseteq V^{\psi_1}_L\cup V^{\psi_1}_X$.
\item $(V^{\psi_1}_L\cup V^{\psi_1}_X)\cap V^\phi_X\subseteq X^\phi_{\min}$.
\item $(V^{\psi_2}_L\cup V^{\psi_2}_X)\cap V^\phi_X\subseteq X^\phi_{\min}$.
\item $(V^{\psi_1}_L\cup V^{\psi_1}_X)\cap (V^{\psi_2}_L\cup V^{\psi_2}_X)\subseteq X^\phi_{\min}$.
\end{enumerate}
\end{definition}

Note that we can clearly compute the sets $V^\phi_L$, $V^\phi_X$, and $V^\phi_R$ from the sets  $X^\phi_{\max}$, $X^\phi_p$, $X^\phi_{\min}$, and $X^\phi_{\dpath}$ in polynomial time. The sets $X^\phi_{\max}$, $X^\phi_p$, and $X^\phi_{\min}$ can be computed in polynomial time by \cref{obs:algo1running,obs:algo2running}. By the definition of $X^\phi_{\dpath}$ given in \cref{sec:ltwdef} we have that we can compute $X^\phi_{\dpath}$ from $X^\phi_{\max}$, $X^\phi_p$, and $X^\phi_{\min}$ in polynomial time. Hence, we have the following.

\begin{observation}\label{obs:legalpoly}
The functions $\legal_1:\mathcal{S}\times \mathcal{S}\rightarrow \{\true,\false\}$ and $\legal_2:\mathcal{S}\times \mathcal{S}\times \mathcal{S}\rightarrow \{\true,\false\}$ can be computed in polynomial time.
\end{observation}

Now we show that if we have a \slim $S$-nice tree decomposition of $G$ with width~$k$ that witnesses a state $\phi=(\tau_-,L^S, X^S, R^S,\tau_+)$ and the precedessor state(s) of $\phi$ and is \topheavy for the parent of the parent of the top node of the directed path of $S$-trace $(L^S, X^S, R^S)$, then we have, depending on the number of predecessor states, that $\legal_1(\phi,\psi)=\true$ and $\legal_2(\phi,\psi_1,\psi_2)=\true$.
To this end, we first prove the following.

\begin{lemma}\label{lem:vl}
Let $\mathcal{T}$ be a \slim $S$-nice tree decomposition of $G$ with width~$k$ that witnesses $\phi=(\tau_-,L^S, X^S, R^S,\tau_+)$. 
If $L^S\neq\emptyset$, then let $t_{\min}$ and $t_{\max}$ be the bottom node and the top node of the directed path of $S$-trace $(L^S, X^S, R^S)$, respectively. 
If $L^S=\emptyset$, then let $t_{\max}$ be the $S$-child with $S$-trace $(L^S, X^S, R^S)$ of the $S$-bottom node that admits $S$-operation $\tau_+$. 
Let $t_p$ be the parent of $t_{\max}$ (if it exists). 
The bags of $t_{\min}$ and $t_{\max}$ are $X^\phi_{\min}$ and $X^\phi_{\max}$, respectively. If $t_{\min}$ is not defined, then $X^\phi_{\min}=X^S$.
If $t_p$ exists, its bag is $X^\phi_p$, otherwise $X^\phi_p=\emptyset$.
Let the tree decomposition $\mathcal{T}$ be \topheavy for the parent of $t_p$. 
   It holds that $v\in V^\phi_L$ if and only if $L^S\neq \emptyset$ and $t_{\min}$ has an $S$-child $t$ such that $v\in V_t\setminus X^\phi_{\min}$.
\end{lemma}
\begin{proof}
Note that if $L^S=\emptyset$, then by definition we have that $V^\phi_L=\emptyset$ and $t_{\min}$ is not defined. Hence, in this case the statement holds. From now on, assume that $L^S\neq\emptyset$.
We first show that if $v\in V^\phi_L$, then $t_{\min}$ has an $S$-child $t$ such that $v\in V_t\setminus X^\phi_{\min}$. By definition of $V^\phi_L$ we have that $v\notin X^\phi_{\min}$ and $v\in L^S\cup (X^\phi_c\setminus X^\phi_{\min})$ or $v$ is connected to some $u\in L^S\cup (X^\phi_c\setminus X^\phi_{\min})$ in $G-X^\phi_{\min}$.
If $v\in L^S$, then there must exist exactly one $S$-child $t$ such that $v\in V_t\setminus X^\phi_{\min}$, otherwise, since $v\notin X^\phi_{\min}$, we have a contradiction to Condition~\ref{condition_3_tree_decomposition} of \cref{def:tree_decomposition}. The same argument holds if $v\in X^\phi_c\setminus X^\phi_{\min}$.
Now assume that $v$ is connected to some $u\in L^S\cup (X^\phi_c\setminus X^\phi_{\min})$ in $G-X^\phi_{\min}$. If $u\in L^S$, then by the same arguments as above, there exist exactly one $S$-child $t$ such that $u\in V_t\setminus X^\phi_{\min}$. Again, the same also holds if $u\in X^\phi_c\setminus X^\phi_{\min}$, that is, then there exist exactly one $S$-child $t$ such that $u\in V_t\setminus X^\phi_{\min}$. Furthermore, in both cases we also have that $u\notin V\setminus V_{t_{\min}}$, since otherwise, we have a contradiction to Condition~\ref{condition_3_tree_decomposition} of \cref{def:tree_decomposition}.
Let $C$ be the connected component that contains $v$ (and also $u$ if $u$ exists) in $G-X^\phi_{\min}$. We have that $V(C)$ contains a vertex that is only contained in bags in the subtree of the tree decomposition that is rooted in $t$. Hence, we must have that $V(C)\subseteq V_t\setminus X^\phi_{\min}$, otherwise we get a contradiction to Conditions~\ref{condition_2_tree_decomposition} and~\ref{condition_3_tree_decomposition} of \cref{def:tree_decomposition}. We can conclude that $v\in V_t\setminus X^\phi_{\min}$.

Now suppose that $t_{\min}$ has an $S$-child $t$ such that $v\in V_t\setminus X^\phi_{\min}$. We show that then we have $v\in V^\phi_L$. Note that in particular, $v\notin X^\phi_{\min}$. Let $C$ be the connected component that contains $v$ in $G-X^\phi_{\min}$. Suppose towards contradiction that $C$ is an $F$-component. Note that $t_{\min}$ is an $S$-bottom node. From the definition of \topheavy{ness} (\cref{def:topheavy}) we have that for each $F$-component $C'$ in $G-X^\phi_{\min}$ it holds that $V(C')\cap V_{t'}=\emptyset$ for each $S$-child $t'$ of $t_{\min}$. In particular, this implies that $V(C)\cap V_t=\emptyset$. This is a contradiction to $v\in V_t\setminus X^\phi_{\min}$. It follows that $C$ is not an $F$-component and hence $V(C)\cap S\neq\emptyset$. Let $u\in V(C)\cap S$. We have that $u\notin X^S$ and $u\in V_t$. It follows that $u\in L^S$. Hence, we have that $v$ is connected to some $u\in L^S$ in $G-X^\phi_{\min}$. By definition, we then have that $v\in V^\phi_L$.
\end{proof}

\begin{lemma}\label{lem:vx}
Let $\mathcal{T}$ be a \slim $S$-nice tree decomposition of $G$ with width~$k$ that witnesses $\phi=(\tau_-,L^S, X^S, R^S,\tau_+)$. 
If $L^S\neq\emptyset$, then let $t_{\min}$ and $t_{\max}$ be the bottom node and the top node of the directed path of $S$-trace $(L^S, X^S, R^S)$, respectively. 
If $L^S=\emptyset$, then let $t_{\max}$ be the $S$-child with $S$-trace $(L^S, X^S, R^S)$ of the $S$-bottom node that admits $S$-operation $\tau_+$. 
Let $t_p$ be the parent of $t_{\max}$ (if it exists). 
The bags of $t_{\min}$ and $t_{\max}$ are $X^\phi_{\min}$ and $X^\phi_{\max}$, respectively. If $t_{\min}$ is not defined, then $X^\phi_{\min}=X^S$.
If $t_p$ exists, its bag is $X^\phi_p$, otherwise $X^\phi_p=\emptyset$.
Let the tree decomposition $\mathcal{T}$ be \topheavy for the parent of $t_p$. 
    It holds that $v\in V^\phi_X$ if and only if $v\in X^\phi_p$ or both of the following conditions hold.
    \begin{itemize}
        \item $v\in V_{t_{\max}}$, and
        \item if $L^S\neq\emptyset$, then for each $S$-child $t$ of $t_{\min}$ it holds that $v\notin V_t\setminus X^\phi_{\min}$.
    \end{itemize}
\end{lemma}
\begin{proof}
    First, assume that $v\in V^\phi_X$. By definition, we have that $v\in X^\phi_{\max}\cup X^\phi_p\cup X^\phi_{\min}\cup X^\phi_{\dpath}$. If $v\in X^\phi_p$, then we are done. Hence, assume that $v\in (X^\phi_{\max}\cup X^\phi_{\min}\cup X^\phi_{\dpath})\setminus X^\phi_p$.
    If $x\in X^\phi_{\min}$, then both conditions in the lemma statement are trivially fulfilled. If $x\in X^\phi_{\max}$, then the first condition is trivially fulfilled. To see that the second condition is fulfilled, assume that $L^S\neq\emptyset$ and assume for contradiction that $t_{\min}$ has an $S$-child $t$ such that $v\in V_t\setminus X^\phi_{\min}$. Then we have a contradiction to Condition~\ref{condition_3_tree_decomposition} of \cref{def:tree_decomposition}.
    Finally, suppose that $v\in X^\phi_{\dpath}$. 
    Then by definition, there exists an $F$-component $C$ in $G-(X^\phi_{\max}\cup X^\phi_p\cup X^\phi_{\min})$ such that $v\in V(C)$, $N[V(C)]\cap (X^\phi_c\setminus X^\phi_{\min})=\emptyset$, and $N(V(C))\setminus X^\phi_p\neq\emptyset$. Assume for contradiction that $v\notin V_{t_{\max}}$. Then we must have that $V(C)\subseteq V\setminus V_{t_{\max}}$, otherwise we get a contradiction to Conditions~\ref{condition_2_tree_decomposition} and~\ref{condition_3_tree_decomposition} of \cref{def:tree_decomposition}. However, since $N(V(C))\setminus X^\phi_p\neq\emptyset$, we now have a contradiction to Condition~\ref{condition_2_tree_decomposition} of \cref{def:tree_decomposition}. We can conclude that $v\in V_{t_{\max}}$. 
Now assume that $L^S\neq\emptyset$ and assume for contradiction that $t_{\min}$ has an $S$-child $t$ such that $v\in V_t\setminus X^\phi_{\min}$. By \cref{lem:vl} this means that $v\in V^\phi_L$ and hence, $v\in L^S\cup (X^\phi_c\setminus X^\phi_{\min})$ or $v$ is connected to some $u\in L^S\cup (X^\phi_c\setminus X^\phi_{\min})$ in $G-X^\phi_{\min}$. 
Since $v$ is contained in an $F$-component in $G-(X^\phi_{\max}\cup X^\phi_p\cup X^\phi_{\min})$ and since $N[V(C)]\cap (X^\phi_c\setminus X^\phi_{\min})=\emptyset$, we have that $(L^S\cup (X^\phi_c\setminus X^\phi_{\min}))\cap V(C)=\emptyset$.
It follows that there is a vertex in $(X^\phi_{\max}\cup X^\phi_p)\setminus X^\phi_{\min}$ that is connected to some $u\in L^S\cup (X^\phi_c\setminus X^\phi_{\min})$ in $G-X^\phi_{\min}$. This is a contradiction to Conditions~\ref{condition_2_tree_decomposition} and~\ref{condition_3_tree_decomposition} of \cref{def:tree_decomposition}. We can conclude that for each $S$-child $t$ of $t_{\min}$ it holds that $v\notin V_t\setminus X^\phi_{\min}$.

    For the other direction, assume that $v\in X^\phi_p$ or both of the following conditions hold: $v\in V_{t_{\max}}$, and if $L^S\neq\emptyset$, then for each $S$-child $t$ of $t_{\min}$ it holds that $v\notin V_t\setminus X^\phi_{\min}$. If $v\in X^\phi_p$, then by definition we have that $v\in V^\phi_X$. Hence, assume that $v\notin X^\phi_p$, $v\in V_{t_{\max}}$, and if $L^S\neq\emptyset$, then for each $S$-child $t$ of $t_{\min}$ it holds that $v\notin V_t\setminus X^\phi_{\min}$. If $v\in X^\phi_{\max}\cup X^\phi_{\min}$, then we also have by definition that $v\in V^\phi_X$. Hence, assume that $v\in V_{t_{\max}}\setminus (X^\phi_{\max}\cup X^\phi_{\min})$. Let $C$ denote the connected component of $G-(X^\phi_{\max}\cup X^\phi_p\cup X^\phi_{\min})$ that contains $v$. First, assume for contradiction that $V(C)\cap S\neq\emptyset$. If $V(C)\cap R^S\neq\emptyset$, we immediately get a contradiction to Conditions~\ref{condition_2_tree_decomposition} and~\ref{condition_3_tree_decomposition} of \cref{def:tree_decomposition}. Hence, assume that $V(C)\cap L^S\neq\emptyset$. Then we have for the connected component $C'$ in $G-X^\phi_{\min}$ that contains $v$ also that $V(C')\cap L^S\neq\emptyset$. By definition, we then have $v\in V^\phi_L$. By \cref{lem:vl} then have that $t_{\min}$ has an $S$-child $t$ such that $v\in V_t\setminus X^\phi_{\min}$. This is a contradiction to the assumption that for each $S$-child $t$ of $t_{\min}$ it holds that $v\notin V_t\setminus X^\phi_{\min}$. We can conclude that the connected component $C$ of $G-(X^\phi_{\max}\cup X^\phi_p\cup X^\phi_{\min})$ that contains $v$ is an $F$-component. Now we prove that $v\in X^\phi_{\dpath}$ and hence, $v\in V^\phi_X$. To this end, assume for contradiction that $N[V(C)]\cap (X^\phi_c\setminus X^\phi_{\min})\neq\emptyset$ or $N(V(C))\setminus X^\phi_p=\emptyset$. First, assume that $N[V(C)]\cap (X^\phi_c\setminus X^\phi_{\min})\neq\emptyset$. Then we have that $V(C)$ contains or is connected to a vertex that is only contained in bags in the subtree of the tree decomposition that is rooted in an $S$-child $t$ of $t_{\min}$. Hence, we must have that $V(C)\subseteq V_t\setminus X^\phi_{\min}$, otherwise we get a contradiction to Conditions~\ref{condition_2_tree_decomposition} and~\ref{condition_3_tree_decomposition} of \cref{def:tree_decomposition}. This, however, implies in particular that $v\in V_t\setminus X^\phi_{\min}$, which is a contradiction to the assumption that for each $S$-child $t$ of $t_{\min}$ it holds that $v\notin V_t\setminus X^\phi_{\min}$. We can conclude that $N[V(C)]\cap (X^\phi_c\setminus X^\phi_{\min})=\emptyset$. Finally, assume for contradiction that $N(V(C))\setminus X^\phi_p=\emptyset$. Then we have that $C$ is an $F$-component in $G-X^\phi_p$. Note that since $G$ is connected, we have that $N(V(C))\neq\emptyset$ and hence, $X^\phi_p\neq\emptyset$. Hence, we have that $t_{\max}$ has a parent $t_p$ with bag $X^\phi_p$ that is an $S$-bottom node. From the definition of \topheavy{ness} (\cref{def:topheavy}) we have that for each $F$-component $C'$ in $G-X^\phi_p$ it holds that $V(C')\cap V_{t'}=\emptyset$ for each $S$-child $t'$ of $t_p$. In particular, this implies that $V(C)\cap V_{t_{\max}}=\emptyset$. This is a contradiction to $v\in V_{t_{\max}}$. We can conclude that $N(V(C))\setminus X^\phi_p\neq\emptyset$ and hence,  $v\in X^\phi_{\dpath}$, which implies that  $v\in V^\phi_X$.
    \end{proof}

\begin{corollary}\label{lem:legal0}
Let $\mathcal{T}$ be a \slim $S$-nice tree decomposition of $G$ with width~$k$ that witnesses $\phi=(\tau_-,L^S, X^S, R^S,\tau_+)$. 
If $L^S\neq\emptyset$, then let $t_{\min}$ and $t_{\max}$ be the bottom node and the top node of the directed path of $S$-trace $(L^S, X^S, R^S)$, respectively. 
If $L^S=\emptyset$, then let $t_{\max}$ be the $S$-child with $S$-trace $(L^S, X^S, R^S)$ of the $S$-bottom node that admits $S$-operation $\tau_+$. 
Let $t_p$ be the parent of $t_{\max}$ (if it exists). 
The bags of $t_{\min}$ and $t_{\max}$ are $X^\phi_{\min}$ and $X^\phi_{\max}$, respectively. If $t_{\min}$ is not defined, then $X^\phi_{\min}=X^S$.
If $t_p$ exists, its bag is $X^\phi_p$, otherwise $X^\phi_p=\emptyset$.
Let the tree decomposition $\mathcal{T}$ be \topheavy for the parent of $t_p$. 
It holds that 
$V_{t_{\max}}\cup X^\phi_p=V^\phi_L\cup V^\phi_X$.
\end{corollary}

Now we are ready to prove that the requirements for $\legal_1(\phi,\psi)=\true$ and $\legal_2(\phi,\psi_1,\psi_2)=\true$, respectively, are given if there is a \slim $S$-nice tree decomposition $\mathcal{T}=(T,\{X_t\}_{t\in V(T)})$ of $G$ with width~$k$ that witnesses $\phi$ and $\psi$ or $\psi_1$ and $\psi_2$ respectively.

\begin{lemma}\label{lem:legal1}
Let $\phi=(\tau_-,L^S, X^S, R^S,\tau_+)$ be a state and let $\psi=(\tau'_-,\hat{L}^S, \hat{X}^S, \hat{R}^S,\tau_-)\in\Psi(\phi)$ such that the following holds.
There is a \slim $S$-nice tree decomposition $\mathcal{T}=(T,\{X_t\}_{t\in V(T)})$ of $G$ with width~$k$ that witnesses $\phi$ and $\psi$. Let $t_{\min}$ and $t_{\max}$ be the bottom and top node, respectively, of the directed path of $S$-trace $(L^S, X^S, R^S)$. Let $t_p$ be the parent of $t_{\max}$. 
The tree decomposition $\mathcal{T}$ is \topheavy for the parent of $t_p$.
Then we have the following.
\begin{enumerate}
\item $V^\psi_L\cup V^\psi_X\subseteq V^\phi_L\cup V^\phi_X$.
\item $N[V^\psi_R]\cap V^\phi_L=\emptyset$.
\item $(V^\psi_L\cup V^\psi_X)\cap V^\phi_X= X^\phi_{\min}$.
\end{enumerate}
\end{lemma}
\begin{proof}
We first show that the first condition holds. By \cref{lem:legal0} we have that $V^\phi_L\cup V^\phi_X=V_{t_{\max}}\cup X^\phi_p$ and $V^\psi_L\cup V^\psi_X=V_{t'_{\max}}\cup X^\psi_p$, where $t'_{\max}$ is the top node of the directed path of $S$-trace $(\hat{L}^S, \hat{X}^S, \hat{R}^S)$. Since $t_{\max}$ is an ancestor of $t'_{\max}$, the first condition clearly holds.

To show that the second condition holds, assume for contradiction that $v\in N[V^\psi_R]\cap V^\phi_L$. 
First, note that $V^\psi_R\cap V^\phi_L=\emptyset$, since by \cref{lem:vl,lem:vx} we have that $V^\phi_L\subseteq V^\psi_L\cup V^\psi_X$.
Therefore, we must have that $v\in V^\phi_L$ and there is a $u\in V^\psi_R$ such that $\{u,v\}\in E$. By \cref{lem:legal0} we have that $u$ is not contained in any bag in $T_{t'_{\max}}$ and also not in the bag of the parent of $t'_{\max}$. Note that $t'_{\max}$ is the only $S$-child of $t_{\min}$. By \cref{lem:vl} we have that $v$ not contained in any bag of a node that is not in $T_{t'_{\max}}$. This is a contradiction to Condition~\ref{condition_2_tree_decomposition} of \cref{def:tree_decomposition}.

Finally, we show that the third condition holds. Recall that $V^\phi_X=X^\phi_{\max}\cup X^\phi_p\cup X^\phi_{\min}\cup X^\phi_{\dpath}$, $V^\psi_L\cup V^\psi_X=V_{t'_{\max}}\cup X^\psi_p$, and $X^\psi_p=X^\phi_{\min}$. It follows that $X^\phi_{\min}\subseteq (V^\psi_L\cup V^\psi_X)\cap V^\phi_X$.
Assume for contradiction that there is $v\in (V^\psi_L\cup V^\psi_X)\cap V^\phi_X$ and $v\notin X^\phi_{\min}$. Then we must have that $v\in V_{t'_{\max}}$ and that $v\in V^\phi_X\setminus X^\phi_{\min}$. If $v\in X^\phi_{\max}\cup X^\phi_p$ then we have a contradiction to Condition~\ref{condition_3_tree_decomposition} of \cref{def:tree_decomposition}, since $t_{\min}$ separates $t'_{\max}$ from $t_{\max}$ and $t_p$ in $\mathcal{T}$. 
In particular, we can conclude that $v\notin X^\phi_p$ and hence, by \cref{lem:vx}, we have that $v\in V_{t_{\max}}$, and for each $S$-child $t$ of $t_{\min}$ it holds that $v\notin V_t\setminus X^\phi_{\min}$. Since $t'_{\max}$ is an $S$-child of $t_{\min}$, this is a contradiction to $v\in V_{t'_{\max}}$.
\end{proof}

\begin{lemma}\label{lem:legal2}
Let $\phi=(\tau_-,L^S, X^S, R^S,\tau_+)$ be a state let $\psi_1=(\tau'_-,\hat{L}^S, \hat{X}^S, \hat{R}^S,\tau_-)\in\Psi_1(\phi)$, and let $\psi_2=(\tau''_-,\check{L}^S, \check{X}^S, \check{R}^S,\tau_-)\in\Psi_2(\phi)$ such that the following holds.
There is a \slim $S$-nice tree decomposition $\mathcal{T}=(T,\{X_t\}_{t\in V(T)})$ of $G$ with width~$k$ that witnesses $\phi$, $\psi_1$, and $\psi_2$. Let $t_{\min}$ and $t_{\max}$ be the bottom and top node, respectively, of the directed path of $S$-trace $(L^S, X^S, R^S)$. Let $t_p$ be the parent of $t_{\max}$. 
The tree decomposition $\mathcal{T}$ is \topheavy for the parent of $t_p$.
Then we have the following.
\begin{enumerate}
\item $V^{\psi_1}_L\cup V^{\psi_1}_X\cup V^{\psi_2}_L\cup V^{\psi_2}_X \subseteq V^\phi_L\cup V^\phi_X$.
\item $N[V^{\psi_1}_R]\cap V^\phi_L\subseteq V^{\psi_2}_L\cup V^{\psi_2}_X$.
\item $N[V^{\psi_2}_R]\cap V^\phi_L\subseteq V^{\psi_1}_L\cup V^{\psi_1}_X$.
\item $(V^{\psi_1}_L\cup V^{\psi_1}_X)\cap V^\phi_X= X^\phi_{\min}$.
\item $(V^{\psi_2}_L\cup V^{\psi_2}_X)\cap V^\phi_X= X^\phi_{\min}$.
\item $(V^{\psi_1}_L\cup V^{\psi_1}_X)\cap (V^{\psi_2}_L\cup V^{\psi_2}_X)= X^\phi_{\min}$.
\end{enumerate}
\end{lemma}
\begin{proof}
We first show that the first condition holds. By \cref{lem:legal0} we have that $V^\phi_L\cup V^\phi_X=V_{t_{\max}}\cup X^\phi_p$, $V^{\psi_1}_L\cup V^{\psi_1}_X=V_{t'_{\max}}\cup X^{\psi_1}_p$, where $t'_{\max}$ is the top node of the directed path of $S$-trace $(\hat{L}^S, \hat{X}^S, \hat{R}^S)$, and $V^{\psi_2}_L\cup V^{\psi_2}_X=V_{t''_{\max}}\cup X^{\psi_2}_p$, where $t''_{\max}$ is the top node of the directed path of $S$-trace $(\check{L}^S, \check{X}^S, \check{R}^S)$. 
Since $t_{\max}$ is an ancestor of both $t'_{\max}$ and $t''_{\max}$, the first condition clearly holds.

To show that the second condition holds, assume for contradiction that $v\in N[V^{\psi_1}_R]\cap V^\phi_L$ and $v\notin V^{\psi_2}_L\cup V^{\psi_2}_X$. 
First, note that if $v\in V^{\psi_1}_R\cap V^\phi_L$, then by \cref{lem:vl,lem:vx} we have that $v\in V^{\psi_2}_L\cup V^{\psi_2}_X$
Hence, assume that $v\in V^\phi_L$ and there is a $u\in V^{\psi_1}_R$ such that $\{u,v\}\in E$. By \cref{lem:legal0} we have that $u$ is not contained in any bag in $T_{t'_{\max}}$ and also not in the bag of the parent of $t'_{\max}$.  
By \cref{lem:legal0}, since $v\notin V^{\psi_2}_L\cup V^{\psi_2}_X$, we have that $v$ is not contained in any bag in $T_{t''_{\max}}$ and also not in the bag of the parent of $t''_{\max}$. 
Note that $t'_{\max}$ and $t''_{\max}$ are the only two $S$-child of $t_{\min}$.
Furthermore, by \cref{lem:vl}, since $v\in V^\phi_L$, we have that $v$ not contained in any bag of a node that is not in $T_{t'_{\max}}$. 
This is a contradiction to Condition~\ref{condition_2_tree_decomposition} of \cref{def:tree_decomposition}.

Showing that the third condition holds can be done in a symmetrical way as the second condition.

Next, we show that the fourth condition holds. Recall that $V^\phi_X=X^\phi_{\max}\cup X^\phi_p\cup X^\phi_{\min}\cup X^\phi_{\dpath}$, $V^{\psi_1}_L\cup V^{\psi_1}_X=V_{t'_{\max}}\cup X^{\psi_1}_p$, and $X^{\psi_1}_p=X^\phi_{\min}$. It follows that $X^\phi_{\min}\subseteq (V^{\psi_1}_L\cup V^{\psi_1}_X)\cap V^\phi_X$.
Assume for contradiction that there is $v\in (V^{\psi_1}_L\cup V^{\psi_1}_X)\cap V^\phi_X$ and $v\notin X^\phi_{\min}$. 
Then we must have that $v\in V_{t'_{\max}}$ and that $v\in V^\phi_X\setminus X^\phi_{\min}$. 
If $v\in X^\phi_{\max}\cup X^\phi_p$ then we have a contradiction to Condition~\ref{condition_3_tree_decomposition} of \cref{def:tree_decomposition}, since $t_{\min}$ separates $t'_{\max}$ from $t_{\max}$ and $t_p$ in $\mathcal{T}$. 
In particular, we can conclude that $v\notin X^\phi_p$ and hence, by \cref{lem:vx}, we have that $v\in V_{t_{\max}}$, and for each $S$-child $t$ of $t_{\min}$ it holds that $v\notin V_t\setminus X^\phi_{\min}$. Since $t'_{\max}$ is an $S$-child of $t_{\min}$, this is a contradiction to $v\in V_{t'_{\max}}$.
We can conclude that $(V^{\psi_1}_L\cup V^{\psi_1}_X)\cap V^\phi_X= X^\phi_{\min}$.

Showing that the fifth condition holds can be done in a symmetrical way as the fourth condition.

Finally, we show that the last condition holds. 
Recall that $V^{\psi_1}_L\cup V^{\psi_1}_X=V_{t'_{\max}}\cup X^{\psi_1}_p$, $V^{\psi_2}_L\cup V^{\psi_2}_X=V_{t''_{\max}}\cup X^{\psi_2}_p$, and $X^{\psi_1}_p=X^{\psi_2}_p=X^\phi_{\min}$. 
It follows that $X^\phi_{\min}\subseteq (V^{\psi_1}_L\cup V^{\psi_1}_X)\cap (V^{\psi_2}_L\cup V^{\psi_2}_X)$.
Assume for contradiction that there is $v\in (V^{\psi_1}_L\cup V^{\psi_1}_X)\cap (V^{\psi_2}_L\cup V^{\psi_2}_X)$ and $v\notin X^\phi_{\min}$. 
Then we must have that $v\in V_{t'_{\max}}$ and that $v \in V_{t''_{\max}}$. 
Note that $t_{\min}$ separates each node in $T_{t'_{\max}}$ from each node in $T_{t''_{\max}}$ in $\mathcal{T}$. 
This is a contradiction to Condition~\ref{condition_3_tree_decomposition} of \cref{def:tree_decomposition}. We can conclude that $(V^{\psi_1}_L\cup V^{\psi_1}_X)\cap (V^{\psi_2}_L\cup V^{\psi_2}_X)= X^\phi_{\min}$.
\end{proof}

\section{Correctness and Running Time Analysis}\label{sec:correct}

In this section, we finally prove that our overall dynamic programming algorithm is correct and give a running time analysis. The results from this section will imply \cref{thm:main}. To this end, we first show that if $\ptw(\phi)=\true$ for some $\phi=(\tau_-,L^S, X^S, R^S,\tau_+)$, then we can compute a partial tree decomposition for the input graph, intuitively, until the top node of the directed path of $S$-trace $(L^S, X^S, R^S)$. For the formal statement, we use the sets $V^\phi_L$, $V^\phi_X$, and $V^\phi_R$ defined in \cref{sec:legal}. We show the following.

\begin{proposition}\label{prop:correctness1}
Let $\phi=(\tau_-,L^S, X^S, R^S,\tau_+)$ be a state such that $\tau_-=\void$ if and only if $L^S=\emptyset$, and if $\tau_+=\void$ then $R^S=\emptyset$.
If $\ptw(\phi)=\true$, then there exists a rooted tree decomposition $\mathcal{T}$ of width at most $k$ for $G[V^\phi_L\cup V^\phi_X]$ such that 
the root of $\mathcal{T}$ has bag $X^\phi_p$. 
\end{proposition}
\begin{proof}
We prove this by induction on the predecessor relation (\cref{def:preceding}) on states. For the base case, assume that $\phi$ does not have any preceding states. Then we have that $\tau_-=\void$ and hence, we also have that $L^S=\emptyset$. By \cref{def:dp} we have that $\ptw(\phi)=\ltw(\phi)$ and hence $\ltw(\phi)=\true$. 
By \cref{def:ltw} this means that for the graph $G'$ obtained from~$G$ by adding edges between all pairs of vertices $u,v\in  X^\phi_{\max}$ with $\{u,v\}\notin E$,  all pairs of vertices $u,v\in  X^\phi_p$ with $\{u,v\}\notin E$, and all pairs of vertices $u,v\in X^\phi_{\min}$ with $\{u,v\}\notin E$ we have that $\tw(G'[X^\phi_{\max}\cup X^\phi_p\cup X^\phi_{\min}\cup X^\phi_{\dpath}])\le k$. 
By \cref{lem:cliquebag} we have that then there is a rooted tree decomposition $\mathcal{T}$ for $G'[X^\phi_{\max}\cup X^\phi_p\cup X^\phi_{\min}\cup X^\phi_{\dpath}]$ with width at most $k$ such that the root $t$ of~$\mathcal{T}$ has bag $X_{t}=X^\phi_p$.
Note that since $L^S=\emptyset$, we have by definition that $V^\phi_L=\emptyset$.
We can conclude that $X^\phi_{\max}\cup X^\phi_p\cup X^\phi_{\min}\cup X^\phi_{\dpath}=V^\phi_L\cup V^\phi_X$.
This finishes the base case.

 

Now assume that $\phi$ has predecessor states. Then we have that $L^S\neq \emptyset$ and hence, $\tau_-\neq\void$. We make a case distinction on $\tau_-$.
\begin{enumerate}
\item Assume that $\tau_-=\smallintroduce(v,d)$ or $\tau_-=\extendedforget(v,d,f,\tau)$. 
Then by \cref{def:dp} we have that 
        \[
        \ptw(\phi) = \bigvee_{\psi\in\Psi(\phi)} \bigl(\ltw(\phi)\wedge\ptw(\psi)\wedge\legal_1(\phi,\psi)\bigr).
        \]
Let $\psi\in\Psi(\phi)$ such that $\ptw(\psi)=\true$ and $\legal_1(\phi,\psi)=\true$. 
Then by induction hypothesis, there exists a rooted tree decomposition $\mathcal{T}$ for $G[V^\psi_L\cup V^\psi_X]$ with width at most $k$ such that the root $t_r$ of $\mathcal{T}$ has bag $X^\psi_p$. 
Furthermore, we have that $\ltw(\phi)=\true$. 
By \cref{def:ltw} this means that for the graph $G'$ obtained from $G$ by adding edges between all pairs of vertices $u,v\in  X^\phi_{\max}$ with $\{u,v\}\notin E$,  all pairs of vertices $u,v\in  X^\phi_p$ with $\{u,v\}\notin E$, and all pairs of vertices $u,v\in X^\phi_{\min}$ with $\{u,v\}\notin E$ we have that $\tw(G'[X^\phi_{\max}\cup X^\phi_p\cup X^\phi_{\min}\cup X^\phi_{\dpath}])\le k$.
By \cref{lem:cliquebag} we have that then there is a rooted tree decomposition $\mathcal{T}'$ for $G'[X^\phi_{\max}\cup X^\phi_p\cup X^\phi_{\min}\cup X^\phi_{\dpath}]$ such that the root $t'_r$ of~$\mathcal{T}'$ with width at most $k$ such that the root $t_r'$ of~$\mathcal{T}'$ has bag $X_{t'_r}=X^\phi_p$.
Furthermore, by \cref{lem:cliquebag} we can assume that $\mathcal{T}'$ has a leaf $t'_\ell$ such that $X_{t'_\ell}=X^\phi_{\min}$.
Note that~$\mathcal{T}'$ is also a tree decomposition for $G[X^\phi_{\max}\cup X^\phi_p\cup X^\phi_{\min}\cup X^\phi_{\dpath}]$ since it is a subgraph of $G'[X^\phi_{\max}\cup X^\phi_p\cup X^\phi_{\min}\cup X^\phi_{\dpath}]$. 

Now we create a rooted tree decomposition $\mathcal{T}^\star$ for $G[V^\phi_L\cup V^\phi_X]$ as follows. 
We connect the root $t_r$ of $\mathcal{T}$ with the leaf $t'_\ell$ of $\mathcal{T}'$, making $t'_\ell$ the parent of $t_r$. Recall that by definition we have $X^\phi_{\min}=X^\psi_p$ and hence the bag of $t_r$ in $\mathcal{T}$ is the same as the bag of $t'_\ell$ in $\mathcal{T}'$. The root of $\mathcal{T}^\star$ is $t'_r$ with bag~$X^\phi_p$.
It remains to show that $\mathcal{T}^\star$ is indeed a rooted tree decomposition for $G[V^\phi_L\cup V^\phi_X]$ of width at most $k$. By construction, it is clear that all bags of nodes in $\mathcal{T}^\star$ have size at most $k+1$. 

We first show that the union of all bags of nodes in $\mathcal{T}^\star$ is $V^\phi_L\cup V^\phi_X$.
By \cref{def:legal1} we have that $V^\psi_L\cup V^\psi_X\subseteq V^\phi_L\cup V^\phi_X$ and that $V^\psi_R\cap V^\phi_L=\emptyset$. Let $v\in V^\phi_L\cup V^\phi_X$. If $v\in V^\phi_X$, then $v$ is a vertex of $G[X^\phi_{\max}\cup X^\phi_p\cup X^\phi_{\min}\cup X^\phi_{\dpath}]$ and hence is in a bag of $\mathcal{T}$. If $v\in V^\phi_L$, then $v\notin V^\psi_R$ and hence $v\in V^\psi_L\cup V^\psi_X$. It follows that $v$ is a vertex of $G[V^\psi_L\cup V^\psi_X]$ and hence is in a bag of $\mathcal{T}'$.
We can conclude that the union of all bags of nodes in $\mathcal{T}^\star$ is $V^\phi_L\cup V^\phi_X$.

Now we argue that Condition~\ref{condition_2_tree_decomposition} of \cref{def:tree_decomposition} is fulfilled. Assume that there is $u,v\in V^\phi_L\cup V^\phi_X$ with $\{u,v\}\in E$. If $u,v\in V^\psi_L\cup V^\psi_X$, then the $\mathcal{T}'$ contains a node with a bag that contains both $u$ and $v$ and hence also $\mathcal{T}^\star$ contains such a node. Assume $u,v\in V^\psi_R$. By \cref{def:legal1} we have that then $u,v\notin V^\phi_L$. If $u,v\in V^\phi_X$, then $\mathcal{T}$ contains a node with a bag that contains both $u$ and $v$ and hence also $\mathcal{T}^\star$ contains such a node.
Lastly, assume w.l.o.g.\ that $v\in V^\psi_L\cup V^\psi_X$ and $u\notin V^\psi_L\cup V^\psi_X$. Then $u\in V^\psi_R$ and $v\in N(V^\psi_R)$. By \cref{def:legal1} we have that then $u,v\in V^\phi_X$. By the same argument as above we get that $\mathcal{T}^\star$ contains a node whose bag contains both $u$ and $v$.

Lastly, we argue that Condition~\ref{condition_3_tree_decomposition} of \cref{def:tree_decomposition} is fulfilled. Assume for contradiction that is it not. 
Then there is some $v\in V^\phi_L\cup V^\phi_X$ such that
the nodes with bags containing $v$ in $\mathcal{T}'$ do not form a subtree. Recall that $\mathcal{T}^\star$ is composed of a tree decomposition $\mathcal{T}'$ for $G[V^\psi_L\cup V^\psi_X]$ and a tree decomposition $\mathcal{T}$ for $G[V^\phi_X]$. Hence, we must have that $v\in V^\psi_L\cup V^\psi_X$ and $v\in V^\phi_X$. By \cref{def:legal1} we have that then $v\in X^\phi_{\min}$. Recall that 
$X^\psi_p=X^\phi_{\min}$ and the root $t_r$ of $\mathcal{T}$ has bag $X^\psi_p$ and a leaf $t'_\ell$ of $\mathcal{T}'$ has bag $X^\phi_{\min}$ and that in $\mathcal{T}^\star$, we have that $t'_\ell$ is the parent of $t_r$. It follows that the nodes with bags containing $v$ form a subtree in $\mathcal{T}^\star$.

\item Assume that $\tau_-=\smalljoin(X^S,X_1^S,X_2^S,L_{1}^S, L_{2}^S,d)$ or $\tau_-=\bigjoin(X^S,X_1^S,X_2^S,L_{1}^S, L_{2}^S,s,d_0,d_1,d_2)$. 
Then by \cref{def:dp} we have that 
        \[
        \ptw(\phi) = \bigvee_{\psi_1\in\Psi_1(\phi)\wedge\psi_2\in\Psi_2(\phi)} \bigl(\ltw(\phi)\wedge\ptw(\psi_1)\wedge\ptw(\psi_2)\wedge \legal_2(\phi,\psi_1,\psi_2)\bigr).
        \]
The arguments here are essentially analogous to the first case.
Let $\psi_1\in\Psi_1(\phi)$ and $\psi_2\in\Psi_2(\phi)$ such that $\ptw(\psi_1)=\true$, $\ptw(\psi_2)=\true$, and $\legal_2(\phi,\psi_1,\psi_2)=\true$. 
Then by induction hypothesis, there exists a rooted tree decomposition $\mathcal{T}_1$ for $G[V^{\psi_1}_L\cup V^{\psi_1}_X]$ with width at most $k$ such that the root $t_1$ of $\mathcal{T}_1$ has bag $X^{\psi_1}_p$.
Furthermore, there exists a rooted tree decomposition $\mathcal{T}_2$ for $G[V^{\psi_2}_L\cup V^{\psi_2}_X]$ with width at most $k$ such that the root $t_2$ of $\mathcal{T}_2$ has bag $X^{\psi_2}_p$. 
Moreover, we have that $\ltw(\phi)=\true$. 
By \cref{def:ltw} this means that for the graph $G'$ obtained from $G$ by adding edges between all pairs of vertices $u,v\in  X^\phi_{\max}$ with $\{u,v\}\notin E$,  all pairs of vertices $u,v\in  X^\phi_p$ with $\{u,v\}\notin E$, and all pairs of vertices $u,v\in X^\phi_{\min}$ with $\{u,v\}\notin E$ we have that $\tw(G'[X^\phi_{\max}\cup X^\phi_p\cup X^\phi_{\min}\cup X^\phi_{\dpath}])\le k$.
By \cref{lem:cliquebag} we have that then there is a rooted tree decomposition $\mathcal{T}'$ for $G'[X^\phi_{\max}\cup X^\phi_p\cup X^\phi_{\min}\cup X^\phi_{\dpath}]$ such that the root $t'_r$ of~$\mathcal{T}'$ with width at most $k$ such that the root $t_r'$ of~$\mathcal{T}'$ has bag $X_{t'_r}=X^\phi_p$.
Furthermore, by \cref{lem:cliquebag} we can assume that $\mathcal{T}'$ has a leaf $t'_\ell$ such that $X_{t'_\ell}=X^\phi_{\min}$.
Note that~$\mathcal{T}'$ is also a tree decomposition for $G[X^\phi_{\max}\cup X^\phi_p\cup X^\phi_{\min}\cup X^\phi_{\dpath}]$ since it is a subgraph of $G'[X^\phi_{\max}\cup X^\phi_p\cup X^\phi_{\min}\cup X^\phi_{\dpath}]$. 

Now we create a rooted tree decomposition $\mathcal{T}^\star$ for $G[V^\phi_L\cup V^\phi_X]$ as follows. 
We connect the roots $t_1$ and $t_2$ of $\mathcal{T}_1$ and $\mathcal{T}_2$, respectively, with the leaf $t'_\ell$ of $\mathcal{T}'$, making $t'_\ell$ the parent of $t_1$ and $t_2$. Recall that by definition we have $X^\phi_{\min}=X^{\psi_1}_p=X^{\psi_2}_p$ and hence the bags of $t_1$ and $t_2$ in $\mathcal{T}_1$ and $\mathcal{T}_2$, respectively, are the same as the bag of $t'_\ell$ in $\mathcal{T}'$. The root of $\mathcal{T}^\star$ is $t'_r$ with bag~$X^\phi_p$.
It remains to show that $\mathcal{T}^\star$ is indeed a rooted tree decomposition for $G[V^\phi_L\cup V^\phi_X]$ of width at most $k$. By construction, it is clear that all bags of nodes in $\mathcal{T}^\star$ have size at most $k+1$. 

We first show that the union of all bags of nodes in $\mathcal{T}^\star$ is $V^\phi_L\cup V^\phi_X$.
By \cref{def:legal2} we have that $V^{\psi_1}_L\cup V^{\psi_1}_X\cup V^{\psi_2}_L\cup V^{\psi_2}_X\subseteq V^\phi_L\cup V^\phi_X$. Let $v\in V^\phi_L\cup V^\phi_X$. If $v\in V^\phi_X$, then $v$ is a vertex of $G[X^\phi_{\max}\cup X^\phi_p\cup X^\phi_{\min}\cup X^\phi_{\dpath}]$ and hence is in a bag of $\mathcal{T}$. 
Assume that $v\in V^\phi_L$.
If $v\in V^{\psi_1}_L\cup V^{\psi_1}_X$, then $v$ is a vertex of $G[V^{\psi_1}_L\cup V^{\psi_1}_X]$ and hence is in a bag of $\mathcal{T}_1$.
If $v\in V^{\psi_1}_R$, then by \cref{def:legal2} we have that $v\in V^{\psi_2}_L\cup V^{\psi_2}_X$. It follows that $v$ is a vertex of $G[V^{\psi_2}_L\cup V^{\psi_2}_X]$ and hence is in a bag of $\mathcal{T}_2$.
We can conclude that the union of all bags of nodes in $\mathcal{T}^\star$ is $V^\phi_L\cup V^\phi_X$.

Now we argue that Condition~\ref{condition_2_tree_decomposition} of \cref{def:tree_decomposition} is fulfilled. Assume that there is $u,v\in V^\phi_L\cup V^\phi_X$ with $\{u,v\}\in E$. If $u,v\in V^{\psi_1}_L\cup V^{\psi_1}_X$, then the $\mathcal{T}_1$ contains a node with a bag that contains both $u$ and $v$ and hence also $\mathcal{T}^\star$ contains such a node. 
If $u,v\in V^{\psi_2}_L\cup V^{\psi_2}_X$, then the $\mathcal{T}_2$ contains a node with a bag that contains both $u$ and $v$ and hence also $\mathcal{T}^\star$ contains such a node.
If $u,v\in V^\phi_X$, then $\mathcal{T}$ contains a node with a bag that contains both $u$ and $v$ and hence also $\mathcal{T}^\star$ contains such a node.
Lastly, assume w.l.o.g.\ that $v\in V^{\psi_1}_L\cup V^{\psi_1}_X$ and $u\notin V^{\psi_1}_L\cup V^{\psi_1}_X$.
Then $u\in V^{\psi_1}_R$, and $v\in N[V^{\psi_1}_R]$. 
Consider the case that $v\in V^\phi_L$.
By \cref{def:legal2} we have that then $v\in V^{\psi_2}_L\cup V^{\psi_2}_X$ and hence, $v\in X^\phi_{\text{min}}$. It follows that $v\in V^\phi_X$. Hence, from here this is the same case as if we initially considered that $v\in V^\phi_X$.
If $u\in V^{\psi_2}_L\cup V^{\psi_2}_X$ or $u\in V^\phi_X$, then we are in one of the previous cases. Hence, assume that $u\notin V^{\psi_2}_L\cup V^{\psi_2}_X\cup V^\phi_X$.
It follows that $u\in V^{\psi_2}_R$ and $u\in V^\phi_L$. 
By \cref{def:legal2} we have that $u\in V^{\psi_1}_L\cup V^{\psi_1}_X$, a contradiction to our initial assumption that $u\notin V^{\psi_1}_L\cup V^{\psi_1}_X$.


Lastly, we argue that Condition~\ref{condition_3_tree_decomposition} of \cref{def:tree_decomposition} is fulfilled. Assume for contradiction that is it not. 
Then there is some $v\in V^\phi_L\cup V^\phi_X$ such that
the nodes with bags containing $v$ in $\mathcal{T}'$ do not form a subtree. Recall that $\mathcal{T}^\star$ is composed of a tree decomposition $\mathcal{T}_1$ for $G[V^{\psi_1}_L\cup V^{\psi_1}_X]$, a tree decomposition $\mathcal{T}_2$ for $G[V^{\psi_2}_L\cup V^{\psi_2}_X]$, and a tree decomposition $\mathcal{T}$ for $G[V^\phi_X]$. 
We have the following cases.
Assume $v\in V^{\psi_1}_L\cup V^{\psi_1}_X$ and $v\in V^{\psi_2}_L\cup V^{\psi_2}_X$. By \cref{def:legal2} we have that then $v\in X^\phi_{\min}$.
Recall that 
$X^{\psi_1}_p=X^{\psi_2}_p=X^\phi_{\min}$ and the root $t_1$ of $\mathcal{T}_1$ has bag $X^{\psi_1}_p$, the root $t_1$ of $\mathcal{T}_2$ has bag $X^{\psi_2}_p$, and a leaf $t'_\ell$ of $\mathcal{T}'$ has bag $X^\phi_{\min}$. Furthermore, in $\mathcal{T}^\star$ we have that $t'_\ell$ is the parent of $t_1$ and $t_1$. It follows that the nodes with bags containing $v$ form a subtree in $\mathcal{T}^\star$.
The cases where $v\in V^{\psi_1}_L\cup V^{\psi_1}_X$ and $v\in V^\phi_X$, or $v\in V^{\psi_2}_L\cup V^{\psi_2}_X$ and $v\in V^\phi_X$ are analogous.
\end{enumerate}
This finishes the proof.
\end{proof}

We remark that in the proof of \cref{prop:correctness1}, we explicitly construct a tree decomposition. Hence, we can save this tree decomposition in the dynamic programming table alongside the Boolean value that indicates whether the tree decomposition has width at most $k+1$ or not. This allows the algorithm to output a tree decomposition for the whole input graph $G$.

For the other direction of the correctness, we prove the following. We first show that there exists a slim $S$-nice tree decomposition of optimal width where every $S$-bottom node admits an extended $S$-operation. Then, intuitively, for every directed path in that tree decomposition, there is a state that is witnessed, and via those states our algorithm will find this tree decomposition and output $\true$. 
To this end, we do the following. Given a slim $S$-nice tree decomposition $\mathcal{T}$ of optimal width, we take an arbitrary but fixed linear order to the $S$-traces that is a linearization of their processor relation on $\mathcal{T}$. Recall that by definition, siblings preserve the predecessor relation.
By induction on this linear order, we show that we can transform $\mathcal{T}$ into a slim $S$-nice tree decomposition $\mathcal{T}'$ of the same or smaller with, that is a sibling of $\mathcal{T}$, such that each $S$-bottom node admits some extended $S$-operation.
The following lemma gives us the induction basis and step.

\begin{lemma}\label{lem:witnesstd}
    Assume there exists a \slim $S$-nice tree decomposition $\mathcal{T}$ for $G$ with width at most~$k$ that has the following properties. Let $(L^S, X^S, R^S)$ be an $S$-trace with $L^S\neq \emptyset$.
\begin{itemize}
    \item $\mathcal{T}$ contains an $S$-bottom $t$ node that has $S$-trace $(L^S, X^S, R^S)$.
    \item For each descendant $t'$ of $t$ that is an $S$-bottom node we have that $t'$ admit some extended $S$-operation, and if $t'$ is contained in a full join tree~$T$, then $\mathcal{T}$ is \topheavy for the parent of the root of~$T$. Otherwise, $\mathcal{T}$ is \topheavy for the parent of $t'$.
\end{itemize}
Then there exists a \slim $S$-nice tree decomposition $\mathcal{T}'$ for $G$ with width at most~$k$ such that the following holds.
\begin{itemize}
    \item $\mathcal{T}'$ is a sibling of $\mathcal{T}$ and $t$ is the $S$-bottom node in $\mathcal{T}'$ with $S$-trace $(L^S, X^S, R^S)$,
    \item if $\mathcal{T}$ is \topheavy for a node $\hat{t}$ such that $\hat{t}$ is the parent of an $S$-bottom node, and
\begin{itemize}
    \item not in $T_{t}$ and not an ancestor of $t$, or
    \item a descendant of $t$,
\end{itemize}
then $\mathcal{T}'$ is also \topheavy for $\hat{t}$ and the subtree rooted at $\hat{t}$ is the same as in $\mathcal{T}$,
    \item All descendants of $t$ that are $S$-bottom nodes admit some extended $S$-operation.
    \item Node $t$ admits some extended $S$-opertation $\tau$.
    \item If $t$ is contained in a full join tree~$T$ in $\mathcal{T}'$, then $\mathcal{T}'$ is \topheavy for the parent of the root of~$T$. Otherwise, $\mathcal{T}'$ is \topheavy for the parent of $t$.
\end{itemize}
\end{lemma}
\begin{proof}
We first consider the case that $t$ does not have any descendants that are $S$-bottom nodes. Then we have that each $S$-child of $t$ has an $S$-trace $(\hat{L}^S, \hat{X}^S, \hat{R}^S)$ with $\hat{L}=\emptyset$.
By \cref{obs:introduce,obs:forget,obs:join} we have that $t$ admits (non-extended) $S$-operation $\forget(v)$ with $\{v\}=L^S$. Note that this implies that the bag of $t$ is not full and hence, $t$ is not part of a full join tree.
Furthermore, we have that $t$ has exactly one $S$-child with $S$-trace $(\emptyset,X^S\cup\{v\},R^S)$.
Let $\tau=\extendedforget(v,d,\false,\void)$ such that $d$ is the correct guess for $t$ in $\mathcal{T}$.
Assume that the size of the bag of the $S$-child of $t$ is minimal, that is, for all siblings $\mathcal{T}'$ of $\mathcal{T}$ such that all of the above holds and if $\mathcal{T}$ is \topheavy for a node $\hat{t}$, then the subtrees rooted at $\hat{t}$ in $\mathcal{T}$ and $\mathcal{T}'$ are the same, we have that the bag size is not smaller.
Note that if $\mathcal{T}$ does not have this property, we can replace it with a sibling that has this property.
Let $\phi=(\void,\emptyset,X^S\cup\{v\},R^S,\tau)$ be the state used to obtain an input for \cref{alg:smallwrapper}.
By \cref{prop:smallcorrect2} and \cref{def:extendedforget} we have that there is a \slim $S$-nice tree decomposition $\mathcal{T}'=(T',\{X'_t\}_{t\in V(T')})$ of $G$ with width at most $k$ that is a sibling of $\mathcal{T}$ such that $t$ admits extended $S$-operation $\tau$ in $\mathcal{T}'$. Furthermore, we have that $\mathcal{T}'$ is \topheavy for the parent of $t$ and if $\mathcal{T}$ is \topheavy for a node $\hat{t}$ such that $\hat{t}$ is not in $T_{t}$ and not an ancestor of $t$, or a descendant of $t$,
then $\mathcal{T}'$ is also \topheavy for $\hat{t}$ and the subtree rooted at $\hat{t}$ is the same as in~$\mathcal{T}$. This finishes the case that $t$ does not have any descendants that are $S$-bottom nodes.

Now consider the case that $t$ has descendants that are $S$-bottom nodes. We now investigate the following subcases:
\begin{enumerate}
    \item Node $t$ is a leaf of a full join tree.
    \item Node $t$ is an inner node of a full join tree.
    \item Node $t$ is the parent of a full join tree.
    \item Node $t$ is not contained in a full join tree nor is it the parent of a full join tree, and has two $S$-children.
    \item Node $t$ is not contained in a full join tree nor is it the parent of a full join tree, and has one $S$-child.
\end{enumerate}
Assume that we are in the first case. 
We know that $\mathcal{T}$ is \topheavy for the parent of the $S$-bottom nodes which are the closest descendants of $t$. Note that those descendants cannot be part of the same full join tree as $t$, since $t$ is a leaf of a full join tree.
We first apply the modification $\MakeTopHeavy$ (\cref{def:maketopheavy}) to the parent of the root of the full join tree. Let the resulting tree decomposition be $\mathcal{T}'$. By \cref{def:topheavy} we have that this does not change any subtrees rooted at nodes $\hat{t}$ for which $\mathcal{T}$ is \topheavy and which are not ancestors of $t$. 
Furthermore, by \cref{lem:moveremove,lem:bringupdown,lem:preserveslim}, we have that $\MakeTopHeavy$ does not change any $S$-traces, therefore $\mathcal{T}'$ is \slim, $S$-nice, and a sibling of $\mathcal{T}$.
Moreover, since $\mathcal{T}$ is \topheavy for all nodes that are parents of $S$-bottom nodes that are descendants of $t$, we have by \cref{def:topheavy} that all subtrees rooted at those nodes are unchanged. It follows that all descendants of $t$ in $\mathcal{T}'$ that are $S$-bottom nodes still admit the same extended $S$-operations. Finally, the $\MakeTopHeavy$ modification applies $\RemoveFromSubtree$ to the root of the full join tree and the parent of the root. Note that if Step~\ref{heavy2} of $\MakeTopHeavy$ is applied to a node in the full join tree, the vertices which are moved not from the bag, since all nodes in the full join tree have the same bag.  Furthermore, it applies $\RemoveBringNeighbor$ to the parent of the parent of the root. Since all nodes in the full join tree have the same bag, this implies that either all bags remain unchanged (in particular full) or at least one vertex is removed from all bags except the root of the full join tree. In particular, then none of the bags are full except the root of the (former) full join tree. If the full join tree remains, the arguments to finish the proof are the same as in the next (second) case. If the bags in the full join tree (except the root) are not full after the modification, then the full join tree ceases to exist. Depending on whether $t$ as one or two $S$-children, the arguments to finish the proof are the same as in the fourth or fifth case.

Assume that we are in the second case. Note that each node in the full join tree is an $S$-bottom node. Hence, $t$ has a descendant that is an $S$-bottom node and contained in the same full join tree as $t$. It follows that $\mathcal{T}$ is \topheavy for the parent of the full join tree. By \cref{lem:alwaysbigjoin} we have that $t$ admits some $\bigjoin$ $S$-operation $\tau$ in $\mathcal{T}$. Hence, we can set $\mathcal{T}'=\mathcal{T}$ and have the lemma statement.

Assume that we are in the third case. Then by \cref{def:slimtd}, \cref{lem:alwaysbigjoin}, and \cref{def:extendedforget} we must have that $t$ admits an $\extendedforget(v,d,\true,\tau)$ $S$-operation for some $v$ and $d$, where $\tau$ is a $\bigjoin$ $S$-operation. Furthermore, we have that $\mathcal{T}$ is \topheavy for $t$. From Condition~\ref{cond:slim:3} of \cref{def:slimtd} we get the parent of $t$ has the same bag as $t$ and has $t$ as its only $S$-child. 
We apply the modification $\MakeTopHeavy$ (\cref{def:maketopheavy}) to the parent of $t$. Let the resulting tree decomposition be $\mathcal{T}'$. By \cref{def:topheavy} we have that this does not change any subtrees rooted at nodes $\hat{t}$ for which $\mathcal{T}$ is \topheavy and which are not ancestors of $t$. 
Furthermore, by \cref{lem:moveremove,lem:bringupdown,lem:preserveslim}, we have that $\MakeTopHeavy$ does not change any $S$-traces, therefore $\mathcal{T}'$ is \slim, $S$-nice, and a sibling of $\mathcal{T}$.
Moreover, since $\mathcal{T}$ is \topheavy for all nodes that are parents of $S$-bottom nodes that are descendants of $t$, we have by \cref{def:topheavy} that all subtrees rooted at those nodes are unchanged. It follows that all descendants of $t$ in $\mathcal{T}'$ that are $S$-bottom nodes still admit the same extended $S$-operations.
Finally, note that $\MakeTopHeavy$ does not change the bag of $t$, since $t$ has the same bag as its parent. It follows that $t$ still admits $\extendedforget(v,d,\true,\tau)$.


Assume that we are in the fourth case. Let $(\hat{L}^S, \hat{X}^S, \hat{R}^S)$ be the $S$-trace of one of the $S$-children of $t$. Let $\tau_-$ be the $S$-operation that is admitted by the closest descendant of $t$ that is an $S$-bottom node.
We have that in $t$ admits the non-extended $S$-operation $\join$.
Assume that the size of the bag of $t$ is minimal, that is, for all siblings $\mathcal{T}'$ of $\mathcal{T}$ such that all of the above holds and if $\mathcal{T}$ is \topheavy for a node $\hat{t}$, then the subtrees rooted at $\hat{t}$ in $\mathcal{T}$ and $\mathcal{T}'$ are the same, we have that the bag size is not smaller.
Note that if $\mathcal{T}$ does not have this property, we can replace it with a sibling that has this property.
Let $\tau_+$ the corresponding small $S$-operation that contains the correct guess for $t$ in $\mathcal{T}$ and let $\phi=(\tau_-,\hat{L}^S, \hat{X}^S, \hat{R}^S,\tau_+)$ be the state used to obtain an input for \cref{alg:smallwrapper} (when applied to compute the bags of the directed path of $S$-trace $(\hat{L}^S, \hat{X}^S, \hat{R}^S)$).
    By \cref{prop:smallcorrect2} and \cref{def:state} we have that there is a \slim $S$-nice tree decomposition $\mathcal{T}'=(T',\{X'_t\}_{t\in V(T')})$ of $G$ with width at most $k$ that is a sibling of $\mathcal{T}$ such that $t$ admits the extended $S$-operation $\tau_+$ in $\mathcal{T}'$. 
    Furthermore, we have that $\mathcal{T}'$ is \topheavy for the parent of $t$ and if $\mathcal{T}$ is \topheavy for a node $\hat{t}$ such that $\hat{t}$ is not in $T_{t}$ and not an ancestor of $t$, or a descendant of $t$,
then $\mathcal{T}'$ is also \topheavy for $\hat{t}$ and the subtree rooted at $\hat{t}$ is the same as in~$\mathcal{T}$. This implies that all descendants of $t$ that are $S$-bottom nodes still admit the same extended $S$-operations as in~$\mathcal{T}$.


Assume that we are in the fifth case. Let $(\hat{L}^S, \hat{X}^S, \hat{R}^S)$ be the $S$-trace of the $S$-child of $t$. 
If $\hat{L}^S\neq\emptyset$ and the directed path with $S$-trace $(\hat{L}^S, \hat{X}^S, \hat{R}^S)$ has at least length three, then this case is essentially the same as the previous one. Assume this is the case.
Let $\tau_-$ be the $S$-operation that is admitted by the closest descendant of $t$ that is an $S$-bottom node.
We have that in $t$ admits the non-extended $S$-operation $\introduce$ or $\forget$.
Assume that the size of the bag of $t$ or the $S$-child of $t$, respectively, is minimal, that is, for all siblings $\mathcal{T}'$ of $\mathcal{T}$ such that all of the above holds and if $\mathcal{T}$ is \topheavy for a node $\hat{t}$, then the subtrees rooted at $\hat{t}$ in $\mathcal{T}$ and $\mathcal{T}'$ are the same, we have that the bag size is not smaller.
Note that if $\mathcal{T}$ does not have this property, we can replace it with a sibling that has this property.
Let $\tau_+$ the corresponding small $S$-operation that contains the correct guess for $t$ in $\mathcal{T}$ and let $\phi=(\tau_-,\hat{L}^S, \hat{X}^S, \hat{R}^S,\tau_+)$ be the state used to obtain an input for \cref{alg:smallwrapper}.
In case $t$ admits a (non-extended) $\forget(v)$ $S$-operation, we set $\tau_+=\extendedforget(v,d,\true,\tau_-)$ if the $S$-child of $t$ has a full bag and the closest descendant of $t$ that is an $S$-bottom node has the same full bag. Otherwise, we set $\tau_+=\extendedforget(v,d,\false,\void)$.
    By \cref{prop:smallcorrect2} and \cref{def:state} we have that there is a \slim $S$-nice tree decomposition $\mathcal{T}'=(T',\{X'_t\}_{t\in V(T')})$ of $G$ with width at most $k$ that is a sibling of $\mathcal{T}$ such that $t$ admits the extended $S$-operation $\tau_+$ in~$\mathcal{T}'$. 
    Furthermore, we have that $\mathcal{T}'$ is \topheavy for the parent of $t$ and if $\mathcal{T}$ is \topheavy for a node $\hat{t}$ such that $\hat{t}$ is not in $T_{t}$ and not an ancestor of $t$, or a descendant of $t$,
then $\mathcal{T}'$ is also \topheavy for $\hat{t}$ and the subtree rooted at $\hat{t}$ is the same as in~$\mathcal{T}$. This implies that all descendants of $t$ that are $S$-bottom nodes still admit the same extended $S$-operations as in~$\mathcal{T}$.
Finally, consider the case that the directed path with $S$-trace $(\hat{L}^S, \hat{X}^S, \hat{R}^S)$ consists of one vertex. This case is somewhat similar to case three. By \cref{def:snicetd,def:slimtd}, we must have that $t$ has one $S$-child $t'$ which is an $S$-bottom node and which has a full bag, and the bag of $t$ itself is not full. It follows that $t$ admits a (non-extended) $\forget(v)$ $S$-operation.
Let $\tau_-$ denote the extended $S$-operation of $t'$. Then by \cref{def:extendedforget} and \cref{sec:bags} we have that $t$ admits the extended $S$-operation $\extendedforget(v,d,\true,\tau_-)$ for any $d\in \mathbb{N}$.
Furthermore, we have that $\mathcal{T}$ is \topheavy for $t$. From Condition~\ref{cond:slim:3} of \cref{def:slimtd} we get the parent of $t$ has the same bag as $t$ and has $t$ as its only $S$-child. 
We apply the modification $\MakeTopHeavy$ (\cref{def:maketopheavy}) to the parent of $t$. Let the resulting tree decomposition be $\mathcal{T}'$. By \cref{def:topheavy} we have that this does not change any subtrees rooted at nodes $\hat{t}$ for which $\mathcal{T}$ is \topheavy and which are not ancestors of $t$. 
Furthermore, by \cref{lem:moveremove,lem:bringupdown,lem:preserveslim}, we have that $\MakeTopHeavy$ does not change any $S$-traces, therefore $\mathcal{T}'$ is \slim, $S$-nice, and a sibling of $\mathcal{T}$.
Moreover, since $\mathcal{T}$ is \topheavy for all nodes that are parents of $S$-bottom nodes that are descendants of $t$, we have by \cref{def:topheavy} that all subtrees rooted at those nodes are unchanged. It follows that all descendants of $t$ in $\mathcal{T}'$ that are $S$-bottom nodes still admit the same extended $S$-operations.
Finally, note that $\MakeTopHeavy$ does not change the bag of $t$, since $t$ has the same bag as its parent. It follows that $t$ still admits $\extendedforget(v,d,\true,\tau)$.
This finishes the proof.
\end{proof}

Now we use \cref{lem:witnesstd} to argue that we can obtain a \slim $S$-nice tree decomposition $\mathcal{T}$ for $G$ such that every $S$-bottom node admits an extended $S$-operation.

\begin{corollary}\label{cor:witnesstd}
    Assume $G$ has treewidth at most $k$. Then there exists a \slim $S$-nice tree decomposition $\mathcal{T}$ for $G$ with width at most~$k$ such that every $S$-bottom node admits an extended $S$-operation.
\end{corollary}
\begin{proof}
    We can obtain a tree decomposition with the desired property by using \cref{lem:witnesstd} iteratively as follows. We start with a \slim $S$-nice tree decomposition $\mathcal{T}$ for $G$ of width at most $k$. 
    We take an arbitrary but fixed linear order to the $S$-traces that is a linearization of their predecessor relation (\cref{def:predecingstraces}) on $\mathcal{T}$. Recall that by \cref{def:sibling}, siblings preserve the predecessor relation and hence also the set of $S$-bottom nodes.
    Now we iterate over the $S$-bottom nodes of $\mathcal{T}$ in the above described order and apply \cref{lem:witnesstd} each time. This way we create a modified tree decomposition that is \slim and $S$-nice.
    Furthermore, we always produce a sibling, we obtain {\topheavy}ness for the parents of the $S$-bottom nodes that we already processed, and since we follow a linearization of the predecessor relation, the $S$-bottom nodes we already processed are never ancestors of the ones that are not processed yet. Hence, the {\topheavy}ness ensures that we do not revert any changes. Once this process finishes, we have a \slim $S$-nice tree decomposition $\mathcal{T}$ for $G$ with width~$k$ such that every $S$-bottom node admits an extended $S$-operation.
\end{proof}




\cref{cor:witnesstd} allows us to assume that each $S$-bottom node admits some extended $S$-operation. Under this assumption we can prove the following in a straightforward way.

\begin{proposition}\label{prop:correctness2}
Assume there exists a \slim $S$-nice tree decomposition $\mathcal{T}$ for $G$ with width~$k$ such that every $S$-bottom node in $\mathcal{T}$ admits an extended $S$-operation. If $\mathcal{T}$ witnesses a state $\phi$, then $\ptw(\phi)=\true$.
\end{proposition}
\begin{proof}
We prove this by induction on the predecessor relation (\cref{def:preceding}) on states. Let $\phi$ be a state that is witnessed by $\mathcal{T}$. For the base case, assume that $\phi$ does not have any preceding states. Then we have that $\tau_-=\void$ and hence $X^\phi_{\min}=X^S$. By the same arguments as in the base case of the proof of \cref{lem:witnesstd} we have that $\tau_+$ is an $\extendedforget$ $S$-operation
By \cref{def:dp} we have that $\ptw(\phi)=\ltw(\phi)$.
By \cref{def:extendedforget} we have that the bag of $t_{\max}$ of the $S$-child of the $S$-bottom node that admits $\tau_+$ is~$X^{\phi}_{\max}$, and the bag of the parent $t_p$ of $t_{\max}$ is $X^{\phi}_p$. 
Note that $X^{\phi}_{\min}\subseteq X^{\phi}_{\max}$. Hence, by \cref{lem:cliquebag2} we have that $\mathcal{T}$ is also a tree decomposition for~$G'$, where $G'$ obtained from $G$ by adding edges between all pairs of vertices $u,v\in X^{\phi}_{\max}$ with $\{u,v\}\notin E$,  all pairs of vertices $u,v\in  X^{\phi}_p$ with $\{u,v\}\notin E$, and all pairs of vertices $u,v\in X^{\phi}_{\min}$ with $\{u,v\}\notin E$. It follows by \cref{lem:twminor} that $\tw(G'[X^{\phi}_{\max}\cup X^{\phi}_p\cup X^{\phi}_{\min}\cup X^{\phi}_{\dpath}])\le k$ and hence $\ltw(\phi)=\true$. This finishes the base case.


Now assume that $\phi$ has predecessor states. Then we have that $L^S\neq \emptyset$ and hence, $\tau_-\neq\void$. We make a case distinction on $\tau_-$.
\begin{enumerate}
\item Assume that $\tau_-=\smallintroduce(v,d)$ or $\tau_-=\extendedforget(v,d,f,\tau)$. 
Then by \cref{def:dp} we have that 
        \[
        \ptw(\phi) = \bigvee_{\psi\in\Psi(\phi)} \bigl(\ltw(\phi)\wedge\ptw(\psi)\wedge\legal_1(\phi,\psi)\bigr).
        \]
Let~$\psi$ be a predecessor state of ${\phi}$ that is witnessed by $\mathcal{T}$. By induction hypothesis we have that $\ptw(\psi)=\true$. 

By \cref{def:bigjoin,def:smallintro,def:smalljoin,def:extendedforget}  we have that the bag of~$t_{\min}$ is~$X^{\phi}_{\min}$, the bag of~$t_{\max}$ is~$X^{\phi}_{\max}$, and the bag of the parent $t_p$ of $t_{\max}$ is $X^{\phi}_p$. 
By \cref{lem:cliquebag2} we have that $\mathcal{T}$ is also a tree decomposition for~$G'$, where $G'$ obtained from $G$ by adding edges between all pairs of vertices $u,v\in X^{\phi}_{\max}$ with $\{u,v\}\notin E$,  all pairs of vertices $u,v\in  X^{\phi}_p$ with $\{u,v\}\notin E$, and all pairs of vertices $u,v\in X^{\phi}_{\min}$ with $\{u,v\}\notin E$. It follows by \cref{lem:twminor} that $\tw(G'[X^{\phi}_{\max}\cup X^{\phi}_p\cup X^{\phi}_{\min}\cup X^{\phi}_{\dpath}])\le k$ and hence $\ltw({\phi})=\true$. 
By \cref{def:legal1} and \cref{lem:legal1} we have that $\legal_1({\phi},\psi)=\true$.
We can conclude that $\ptw({\phi})=\true$.
\item Assume that $\tau_-=\smalljoin(X^S,X_1^S,X_2^S,L_{1}^S, L_{2}^S,d)$ or $\tau_-=\bigjoin(X^S,X_1^S,X_2^S,L_{1}^S, L_{2}^S,s,d_0,d_1,d_2)$. 
Then by \cref{def:dp} we have that 
        \[
        \ptw(\phi) = \bigvee_{\psi_1\in\Psi_1(\phi)\wedge\psi_2\in\Psi_2(\phi)} \bigl(\ltw(\phi)\wedge\ptw(\psi_1)\wedge\ptw(\psi_2)\wedge \legal_2(\phi,\psi_1,\psi_2)\bigr).
        \]
The arguments here are essentially analogous to the first case.
Let~$\psi_1$ and $\psi_2$ be predecessor states of ${\phi}$ that is witnessed by $\mathcal{T}$. By induction hypothesis we have that $\ptw(\psi_1)=\true$ and $\ptw(\psi_2)=\true$.

By \cref{def:bigjoin,def:smallintro,def:smalljoin,def:extendedforget}  we have that the bag of~$t_{\min}$ is~$X^{\phi}_{\min}$, the bag of~$t_{\max}$ is~$X^{\phi}_{\max}$, and the bag of the parent $t_p$ of $t_{\max}$ is $X^{\phi}_p$. 
By \cref{lem:cliquebag2} we have that $\mathcal{T}'$ is also a tree decomposition for~$G'$, where $G'$ obtained from $G$ by adding edges between all pairs of vertices $u,v\in X^{\phi}_{\max}$ with $\{u,v\}\notin E$,  all pairs of vertices $u,v\in  X^{\phi}_p$ with $\{u,v\}\notin E$, and all pairs of vertices $u,v\in X^{\phi}_{\min}$ with $\{u,v\}\notin E$. It follows by \cref{lem:twminor} that $\tw(G'[X^{\phi}_{\max}\cup X^{\phi}_p\cup X^{\phi}_{\min}\cup X^{\phi}_{\dpath}])\le k$ and hence $\ltw({\phi})=\true$. 
By \cref{def:legal2} and \cref{lem:legal2} we have that $\legal_2(\phi,\psi_1,\psi_2)=\true$.
We can conclude that $\ptw(\phi)=\true$.
\end{enumerate}
This finishes the proof.
\end{proof}

From \cref{prop:correctness1,prop:correctness2} and the definition of $V^\phi_L$, $V^\phi_X$, and $V^\phi_R$ in \cref{sec:legal} we get the following.

\begin{corollary}\label{cor:correct}
There is a tree decomposition for $G$ with width at most $k$ if and only if
\[
\bigvee_{\phi\in\Phi}\ptw(\phi)=\true,
\]
where $\Phi$ contains all states of the form $(\tau_-,S\setminus X^S, X^S, \emptyset,\void)$ for some extended $S$-operation $\tau_-$.
\end{corollary}
\begin{proof}
    Assume that $\bigvee_{\phi\in\Phi}\ptw(\phi)=\true$. Let $\phi\in\Phi$ such that $\ptw(\phi)=\true$. We have that $\phi=(\tau_-,S\setminus X^S, X^S,\emptyset,\void)$ for some extended $S$-operation $\tau_-$. By the definition of $V^\phi_L$, $V^\phi_X$, and $V^\phi_R$ in \cref{sec:legal} and \cref{lem:vl,lem:vx}, we get that $V^\phi_L\cup V^\phi_X=V$, and $V^\phi_R=\emptyset$. By \cref{prop:correctness1} we have that then $G$ admits a tree decomposition with width at most $k$.

    Now assume that there is a tree decomposition for $G$ with width at most $k$. Then by \cref{cor:witnesstd} we have that there exists a \slim $S$-nice tree decomposition for $G$ with width at most $k$ such that every $S$-bottom node admits an extended $S$-operation.
    The root has $S$-trace $(S\setminus X^S, X^S,\emptyset)$ for some $X^S\subseteq S$. Let $\tau_-$ be the extended $S$-operation that is admitted by the closest descendant of the root of $\mathcal{T}$ is an $S$-bottom node. Let $\phi=(\tau_-,S\setminus X^S, X^S,\emptyset,\void)$. By \cref{prop:correctness2} we have that $\ptw(\phi)=\true$. Furthermore, we have that $\phi\in\Phi$. It follows that $\bigvee_{\phi\in\Phi}\ptw(\phi)=\true$.
\end{proof}

Finally, we analyze the size of the dynamic programming table and the time needed to compute one entry.

\begin{proposition}\label{prop:runningtime}
The dynamic programming table $\ptw$ has size $2^{\OO(\fvn(G))}$ and each entry can be computed in $2^{\OO(\fvn(G))}$ time.
\end{proposition}
\begin{proof}
Note that here, we assume that we apply the kernelization algorithm by Bodlaender, Jansen, and Kratsch~\cite{bodlaender2013preprocessing} to reduce the number of vertices in the input graph to $\OO(\fvn(G)^4)$ while leaving the size of a minimum feedback vertex set size unchanged.
The dynamic programming table has an entry for every state $\phi$. A state consists of two $S$-operations and an $S$-trace. We first estimate the number of $S$-operations.
\begin{itemize}
\item There is one dummy $S$-operation $\void$.
\item The number of different $S$-operations $\smallintroduce(v,d)$ is in $\OO(\fvn(G)^5)$, since $v\in S$ and $d\in\{1,\ldots,(n+1)^3\}$ with $n\in \OO(\fvn(G)^4)$.
\item The number of different $S$-operations $\smalljoin(X^S,X_1^S,X_2^S,L_{1}^S, L_{2}^S,d)$ is in $2^{\OO(\fvn(G))}$, since $X^S,X_1^S,X_2^S,L_{1}^S, L_{2}^S\subseteq S$ and $d\in\{1,\ldots,(n+1)^3\}$ with $n\in \OO(\fvn(G)^4)$.
\item The number of different $S$-operations $\bigjoin(X^S,X_1^S,X_2^S,L_{1}^S, L_{2}^S,s,d_0,d_1,d_2)$ is in $2^{\OO(\fvn(G))}$, since $X^S,X_1^S,X_2^S,L_{1}^S, L_{2}^S\subseteq S$, $s\in\{0,1\}^{2\fvn(G)+1}$, and $d_0,d_1,d_2\in\{0,\ldots,k+1\}$ (recall that $k\in \OO(\fvn(G))$).
\end{itemize}

We can conclude that there are $2^{\OO(\fvn(G))}$ many $S$-operations that are different from $\extendedforget(v,d,f,\tau)$.
If follows that the number of $S$-operations $\extendedforget(v,d,f,\tau)$ is also in $2^{\OO(\fvn(G))}$, since $v\in S$, $d\in\{1,\ldots,(n+1)^3\}$ with $n\in \OO(\fvn(G)^4)$, $f\in\{\true,\false\}$, and $\tau$ is an $S$-operation that is different from $\extendedforget(v',d',f',\tau')$.

The number of $S$-traces is clearly in $2^{\OO(\fvn(G))}$. We can conclude that the overall size of the dynamic programming table $\ptw$ is in $2^{\OO(\fvn(G))}$.

To compute an entry of $\ptw$, we have to iterate over all predecessor states, of which by the arguments above, there are $2^{\OO(\fvn(G))}$ many. For each one we look up a table entry in $\OO(1)$ time.
Furthermore, we have to compute $\ltw(\phi)$, which by \cref{prop:computeltw} can be done in $n^{\OO(1)}$ time with $n\in \OO(\fvn(G)^4)$.
Finally, by \cref{obs:legalpoly} we have that $\legal_1(\phi,\psi)$ and $\legal_2(\phi,\psi_1,\psi_2)$ can be computed in $n^{\OO(1)}$ time with $n\in \OO(\fvn(G)^4)$.
It follows that each entry can be computed in $2^{\OO(\fvn(G))}$ time.
\end{proof}

\cref{thm:main} now follows from \cref{cor:correct}, \cref{prop:runningtime}, the fact that a minimum feedback vertex set can be computed in $2.7^{\fvn(G)}\cdot n^{\OO(1)}$ time (randomized)~\cite{li2022detecting} or in $3.6^{\fvn(G)}\cdot n^{\OO(1)}$ deterministic time \cite{kociumaka2014faster}, and that we apply the kernelization algorithm by Bodlaender, Jansen, and Kratsch~\cite{bodlaender2013preprocessing} to reduce the number of vertices in the input graph to $\OO(\fvn(G)^4)$ while leaving the minimum feedback vertex set size unchanged.



\section{Conclusion}\label{sec:conclusion}
In this paper, we showed that a minimum tree decomposition for a graph $G$ can be computed in single exponential time in the feedback vertex number of the input graph, that is, in $2^{\OO(\fvn(G))}\cdot n^{\OO(1)}$ time. 
This improves the previously known result that a minimum tree decomposition for a graph $G$ can be computed in $2^{\OO(\vcn(G))}\cdot n^{\OO(1)}$ time~\cite{chapelle2017treewidth,fomin2018algorithms}. 
We believe that this can also be seen either as an important step towards a positive resolution of the open problem of finding a $2^{\OO(\tw(G))}\cdot n^{\OO(1)}$ time algorithm for computing an optimal tree decomposition, or, if its answer is negative, then a mark of the tractability border of single exponential time algorithms for the computation of treewidth.

A natural future research direction is to explore to which extent our techniques can be used to obtain a single exponential time algorithm for treewidth computation for other parameters that measure the vertex deletion distance to a graph class where the treewidth is small or can be computed in polynomial time. Natural candidates would be e.g.\ series-parallel graphs (graphs of treewidth two) or chordal graphs. We remark that our algorithm needs to know the deletion set, so in order to obtain a single exponential running time, we need to be able to compute the deletion set fast enough. For a minimum vertex deletion set to chordal graphs, it is currently not known whether it can be computed in single exponential time~\cite{cao2016chordal}. However, the chordal vertex deletion number is incomparable to the treewidth, and hence any polynomial time algorithm for graphs with a constant-size cordal vertex deletion set would be interesting.


\bibliography{bibliography}


\end{document}